\documentclass[10pt,reqno,a4paper]{amsart}

\usepackage{amsmath}
\usepackage{amsfonts}
\usepackage{amsthm}
\usepackage{amstext}
\usepackage{enumerate}
\usepackage{graphicx}
\usepackage{verbatim}

     \newcommand{\pr}{{\operatorname{pr}}}
     
     \newcommand{\supp}{\operatorname{supp}}
     
     \newcommand{\Div}{{\operatorname{div}}}
     
     \newcommand{\dist}{{\operatorname{dist}}}

     \newcommand{\Ran}{{\operatorname{Ran}}}

\newcommand{\Opl}{{\operatorname{Op^l\!}}}
\newcommand{\Opr}{{\operatorname{Op^r\!}}}
\newcommand{\Opw}{{\operatorname{Op\!^w\!}}}
     \newcommand{\N}{{\mathbb{N}}}
     \newcommand{\R}{{\mathbb{R}}}
     \newcommand{\Z}{{\mathbb{Z}}}
     \newcommand{\C}{{\mathbb{C}}}
     \newcommand{\T}{{\mathbb{T}}}

     \newcommand{\A}{\mbox{\boldmath $a$}}

\newcommand{\sph}{{\rm sph}}
\newcommand{\Ai}{{\rm Ai}}
\newcommand{\slim}{{\rm s-}\lim}
\newcommand{\wlim}{{\rm w-}\lim}
\newcommand{\bd}{{\rm bd}}
\newcommand{\rot}{{\rm rot}}

\newcommand{\e}{{\rm e}}
\renewcommand{\i}{{\rm i}}
\renewcommand{\d}{{\rm d}}
\newcommand{\dol}{{\rm dol}}
\newcommand{\righ}{{\rm right}}
\newcommand{\lef}{{\rm left}}
\newcommand{\unif}{{\rm unif}}
\newcommand{\diag}{{\rm diag}}
\newcommand{\aux}{{\rm aux}}

\renewcommand{\sc}{{\rm sc}}
\newcommand{\sr}{{\rm sr}}

\newcommand{\crt}{{\rm crt}}
\renewcommand{\Re}{{\rm Re}\,}
\renewcommand{\Im}{{\rm Im}\,}

\newcommand{\crit}{{\rm tp}}

\newcommand{\japxi}{{\langle \xi\rangle}}

     \theoremstyle{plain}
     \newtheorem{thm}{Theorem}[section]
     \newtheorem{prop}[thm]{Proposition}
     \newtheorem{lemma}[thm]{Lemma}
      \newtheorem{cor}[thm]{Corollary}
     \theoremstyle{definition}

     \newtheorem{example}[thm]{Example}
     
     \newtheorem{cond}[thm]{Condition}

     \newtheorem{remark}[thm]{Remark}
     \newtheorem{remarks}[thm]{Remarks}
\newtheorem*{remarks*}{Remarks}
\newtheorem*{remark*}{Remark}

     \numberwithin{equation}{section}
\title[Zero]{Quantum scattering at low energies}

\author{J. Derezi\'nski}
\address[J. Derezi\'nski]{Dept. of Math. Methods in Physics, Warsaw University\\ 
Ho\.za 74, 00-682, Warszawa, Poland} 
\email{Jan.Derezinski@fuw.edu.pl}

\author{E. Skibsted}
\address[E. Skibsted]{Institut for  Matematiske
Fag \\
Aarhus Universitet\\ Ny Munkegade  8000 Aarhus C, 
Denmark}
\email{skibsted@imf.au.dk}
\thanks{The authors were  partially supported by  MaPhySto -- A
Network in Mathematical Physics and Stochastics, funded by
The Danish National Research Foundation
and  by  the Postdoctoral Training Program 
HPRN-CT-2002-0277. The research of J.D.
 was also supported by the Polish grants
 SPUB127 and  2 P03A 027 25. Part of the research was done during a visit of
 both authors to the Erwin
 Schr\"odinger Institute.
 One of us (E. S.)       thanks
H. Tamura for drawing our attention to the paper \cite{Kv}.}

\begin{document}

\begin{abstract} For a class of negative  slowly decaying
  potentials,  including $V(x):=-\gamma|x|^{-\mu}$ with $0<\mu<2$,
 we study the quantum
  mechanical scattering theory in the low-energy regime. Using modifiers of
 the
 Isozaki--Kitada type we show that scattering theory is well behaved on
   the {\it whole}  continuous spectrum of the
  Hamiltonian, including the energy $0$.
  We show that the $S$--matrices are well-defined and
  strongly continuous down to 
 the zero energy threshold. Similarly, we prove that
the wave matrices and generalized eigenfunctions are  norm
continuous down
 to the zero energy if we use  appropriate weighted spaces.
These results  are
  used to derive (oscillatory) asymptotics of the standard 
short-range and Dollard type
  $S$--matrices for the subclasses of potentials where both kinds of
 $S$-matrices 
  are
  defined.
For potentials whose leading part is $-\gamma|x|^{-\mu}$
we show that the location of 
  singularities of the kernel of $S(\lambda)$ experiences  an abrupt
  change from passing from positive energies $\lambda$ to the limiting
  energy $\lambda=0$. This change corresponds to the behaviour of the
 classical orbits. Under stronger  conditions we extract the 
 leading term of the asymptotics of the kernel of $S(\lambda)$ at its
 singularities; this  leading term defines a Fourier
integral operator in the sense of H\"ormander \cite{Ho4}.  
\end{abstract}

\maketitle
\tableofcontents

\section{Introduction and results} \label{Introduction}

Scattering theory of 2-body systems, both classical and quantum, both short-
and long-range, is nowadays a well understood subject
\cite{Ho2,II,IK1,IK2,Y2,DG}.  In particular, for
large natural classes of potentials we know a lot about the
properties of wave and scattering matrices at positive energies. Zero -- the
only 
threshold energy --  in most works on the subject is avoided, since scattering
at zero energy is much more difficult to describe and strongly depends on the
choice of the potential.

In this paper we consider a  class of potentials
that have an especially well behaved, nontrivial and interesting low energy
scattering theory. Precise
conditions used in our paper are described in Subsection  \ref{Conditions}. 
Roughly speaking, 
 the potentials that we consider have a  dominant negative 
radial term $V_1(x)$ 
similar to $-\gamma|x|^{-\mu}$ with $\gamma>0$ and $0<\mu<2$,
plus a faster decaying perturbation.

Similar classes
of potentials appeared in the literature already in \cite{Ge}. A systematic
study of such 2-body systems at low energies was undertaken in
\cite{FS}, where
 a complete
expansion of the resolvent at the zero-energy threshold was obtained,
and in \cite{DS1}, where
classical low-energy scattering theory was developed.
This paper can be viewed as a continuation of \cite{FS,DS1}. 

In this paper we show that quantum scattering theory for such potentials is
well behaved down to the energy zero. In particular, we study 
appropriately defined wave
matrices and scattering matrices for a fixed energy. We show that they 
have limits at zero energy. Our results were  partly announced in  \cite{DS2}.

For positive energies most (but probably not all) of 
our results are  contained in the literature, 
scattered in many sources. Our material about
 the zero energy case is new.

In the introduction we will first review scattering for
positive energies for a rather general class of potentials.
Then we will describe a simplified version of
 the main  results of our paper, which 
 concerns scattering at low energies for a more restrictive
 class of potentials.

\subsection{Classical orbits at positive energies}
\label{s1.1} 

 For the presentation of known results about positive energies
 we assume that the potentials satisfy the
following condition:
\begin{cond}\label{symbol} $V=V_1+V_3$ is a sum of real measurable functions on $\R^d$
  such that $V_1$ is  smooth and for some $\mu>0$,
\begin{equation}
\partial_x^\alpha V_1(x)=O\big (|x|^{-\mu-|\alpha|}\big ),\ \
|\alpha|\geq0,  \end{equation} 
$V_3$ is compactly supported and
$V_3(H_0+1)^{-1}$  is a
compact operator on the Hilbert space $L^2(\R^d)$. Here
$H_0:=2^{-1}p^2$ with  $p:= -\i \nabla_x$. The  Hamiltonian 
$H=H_0+V$ does not have positive eigenvalues.
\end{cond}

Let us first consider the classical Hamiltonian
 $h_1(x,\xi):=\frac12\xi^2+V_1(x)$
on the phase space 
 $\R^d\times\R^d$, using $h_0(x,\xi):=\frac12\xi^2$ as the
 free Hamiltonian. (The analysis of the classical case
is   needed in the quantum case).
One can prove that for any $\xi\in{\mathbb R}^d$, $\xi\neq0$,
 and $x$ in an appropriate
outgoing/incoming region the following problem admits a solution
(strictly speaking, meaning one solution for $t\to +\infty$ and one for
$t\to -\infty$):
\begin{equation}\label{eq:mixed conditions22}
\begin{cases}
\ddot y(t) =-\nabla V_1(y(t)),\\
y(\pm 1)=x,\\
\xi=\lim_{t\to {\pm}\infty}\dot y(t).
\end{cases}
\end{equation}
One obtains a family $y^\pm(t,x,\xi)$ of solutions smoothly depending on
parameters. All (positive energy) scattering  orbits, i.e. orbits satisfying
$\lim_{t\to \pm \infty}
|y(t)|=\infty$, are of this form (the energy is
$\lambda=\frac12\xi^2$). Using these solutions, in an appropriate
incoming/outgoing region one can construct a solution
$\phi^\pm(x,\xi)$ to
the eikonal  equation
\begin{equation}
\frac12\left(\nabla_x\phi^\pm(x,\xi)\right)^2+V_1(x)=\frac12\xi^2
\end{equation}
satisfying $\nabla_x\phi^\pm(x,\xi)=\dot y(\pm1,x,\xi)$.

\subsection{Wave and scattering matrices at positive energies}

Let us turn to the quantum case.
Following
Isozaki-Kitada,  see \cite{IK1}, \cite{IK2},
\cite{Y2} and  \cite{RY},
 one can use the functions
$\phi^\pm(x,\xi)$ in the quantum case to construct appropriate modifiers,
which can be taken to be
\begin{equation}
J^\pm f(x):=(2\pi)^{-d}\int \e^{\i\phi^\pm(x,\xi)-\i\xi \cdot y}a^\pm(x,\xi)
f(y)\d x\d\xi.
\end{equation}
Here $a^\pm(x,\xi)$ is an appropriate cut-off supported in the domain of the
definition of $\phi^\pm$, equal to one
in the incoming/outgoing region.
Then one constructs modified wave operators 
\begin{equation}
W^\pm f:=\lim_{t\to\pm\infty}\e^{\i tH}J^\pm\e^{-\i tH_0}f,\;
  \hat f\in C_c(\R^d\setminus\{0\}),\label{wave}\end{equation}
 and the modified scattering operator 
\begin{equation}S= W^{+*}W^-.\label{scat}\end{equation} We remark that
$W^\pm $ are isometric with range given by the projection ont the continuous
spectrum of $H$
\[1_{\mathrm
  c}(H)L^2(\R^d)=1_{]0,\infty[}(H)L^2(\R^d).\]
(Whence $S$ is unitary).

The free Hamiltonian $H_0$ can be diagonalized by the direct integral
 \begin{equation}\label{decom}
   \mathcal H_0=\int_{0}^\infty \oplus L^2(S^{d-1})\;\d \lambda, 
 \end{equation}  \begin{equation}
  \label{eq:restricF}
{\mathcal F}_0(\lambda)f(\omega)=(2\lambda)^{(d-2)/4} \hat f(\sqrt
{2\lambda}\omega),\;\;f\in L^2(\R^d). 
\end{equation} Here $\hat f$  refers to the $d$--dimensional
Fourier transform.
The operator $\mathcal F_0(\lambda)$ can be interpreted as a bounded operator
from the weighted space $L^{2,s}(\R^d):=\langle
x\rangle^{-s}L^2(\R^d)$, $s>\frac12$, to $L^2(S^{d-1})$. One can ask whether the
wave and scattering operators can be restricted to a fixed energy
$\lambda$.

 This question is conceptually simpler in the case of the scattering
operator $S$. Due to the intertwining property  $W^\pm H_0=HW^\pm$ it satisfies $SH_0=H_0S$, so abstract theory guarantees the
existence of a decomposition
\[S\simeq\int_{]0,\infty[}\oplus S(\lambda)\d\lambda,\]
where $S(\lambda)$ are unitary operators on $L^2(S^{d-1})$
defined for almost all $\lambda$. One can prove that under Condition
\ref{symbol} $S(\lambda)$ can be chosen to be a strongly continuous function
 (which fixes uniquely $S(\lambda)$ for all
$\lambda\in]0,\infty[$). $S(\lambda)$  is called the {\em scattering matrix at
      the energy $\lambda$}.

The case of wave operators is somewhat more complicated. By the
intertwining property it is natural to
use the direct integral decomposition (\ref{decom}) only from the right and the 
question is whether we can give a rigorous meaning to $W^\pm{\mathcal
  F}_0(\lambda)^*$. Again, under the condition (\ref{symbol}) one can show that
there exists a unique strongly continuous function
$]0,\infty[\ni\lambda\mapsto W^\pm(\lambda)$ with values in the space
of bounded operators
    from $L^2(S^{d-1})$ to $L^{2,-s}(\R^d)$ with $s>\frac12$ such that for
    $f\in L^{2,s}(\R^d)$
\[W^\pm f=\int_{]0,\infty[}W^\pm(\lambda)
{\mathcal F}_0(\lambda)f\d\lambda.\]
 The operator $W^\pm(\lambda)$ is called the  {\em wave matrix at energy
 $\lambda$}.
 One
can also  extend the domain of $W^\pm(\lambda)$ so that it
can act on the delta-function at $\omega\in S^{d-1}$, denoted $\delta_\omega$.
Now $w^\pm(\omega,\lambda):=W^\pm(\lambda)\delta_\omega$  is an element of
$L^{2,-p}(\R^d)$ for $p>\frac{d}{2}$. It  satisfies
\begin{equation}\label{eq:eigenaa}
\left(-\frac12\Delta+V(x)-\lambda\right) w^\pm(\omega,\lambda)=0.\end{equation}
It behaves in the outgoing/incoming region as a plane wave. It will be called
the {\em generalized eigenfunction of $H$  at   energy $\lambda$
  and at  asymptotic
  normalized velocity  $\omega$}; this terminology is   justified
in Subsection \ref{s1.1bb}.

\subsection{Short-range and Dollard wave and scattering operators}

Let us recall that (\ref{wave}) and (\ref{scat})
are   only one of possible definitions
of wave and scattering operators. In
the short-range case, that is $\mu>1$, the usual definitions  are
\begin{eqnarray}
W_\sr^\pm f&:=&\lim_{t\to\pm\infty}\e^{\i tH}\e^{-\i tH_0}f,\label{wave1}
\\
S_\sr&:=& W_\sr^{+*}W_\sr^-.\label{scat1}\end{eqnarray}
The operators $W^\pm$  and $W_\sr^\pm$ differ by a momentum-dependent phase
factor: 
\begin{eqnarray}W^\pm&=&W_\sr^\pm \e^{\i\psi_\sr^\pm(p)},\\
S &=&\e^{-\i\psi_\sr^{+}(p)}S_\sr \e^{\i\psi_\sr^-(p)}.\label{factor}\end{eqnarray}

Similarly, in the case $\mu>\frac12$ one can use the so-called
Dollard construction:
\begin{eqnarray}
  W^{\pm}_{\dol}f&:=&\lim _{t\to {\pm}\infty}e^{\i tH}U_{\dol}(t)f,\\
U_{\dol}(t)&:=&\e^{-\i \int_{0}^t(p^2/2+V(sp)1_{\{|sp|\geq R_0\}})\;\d
    s},\;R_0>0,\\
S_\dol&:=& W_\dol^{+*}W_\dol^-.\label{scat1a}\end{eqnarray}
Analogously, we have
\begin{eqnarray}W^\pm&=&W_\dol^\pm \e^{\i\psi_\dol^\pm(p)},\\
S &=&\e^{-\i\psi_\dol^{+}(p)}S_\sr \e^{\i\psi_\dol^-(p)}.\label{factora}\end{eqnarray}

\subsection{Asymptotic normalized velocity operator}
\label{s1.1bb}

The reader may ask 
  why  $W^\pm$, $W_\sr^\pm$,
 $W_\dol^\pm$ are all called  wave operators. 
In fact, it is natural to define a whole 
family of wave operators associated with a
 given Schr\"odinger operator. In this subsection we briefly describe a
  possible definition of this family, following  essentially  \cite{De,DG}.

Suppose that $V$ satisfies \eqref{symbol} (or even much weaker
 conditions). Then it can be shown that there exists the following operator:
\begin{equation}
v^\pm:=\mathrm{s}-\lim_{t\to\pm\infty}\pm\e^{\i tH}\hat x\e^{-\i
 tH}1_{\mathrm c}(H),\;\hat x=\tfrac{x}{|x|}.\end{equation}
$v^\pm$ can be called the {\em asymptotic normalized
  velocity  operator}. It
 is a vector of commuting self-adjoint operators  (on the space
 $1_{\mathrm c}(H)L^2(\R^d)$) satisfying
\begin{equation}
(v^\pm)^2=1_{\mathrm c}(H),\ \ \ [v^\pm,H]=0.
\end{equation}
We say that $\breve W^\pm$ is {\em an outgoing/incoming
 wave operator associated with $H$}
 if it is  isometric and 
 satisfies
\begin{equation}
\breve W^\pm H_0=H\breve W^\pm,\ \ \ 
\breve W^\pm \hat p=v^\pm \breve W^\pm,
\end{equation}
where $\hat p=\frac{p}{|p|}$.

Note that if $\breve W_1^\pm$ and $\breve W_2^\pm$ are two wave operators
associated with a given $H$, then there exists a function
 $\psi^\pm$ such that
\begin{equation}
\breve W_1^\pm=\breve W_2^\pm\e^{\i\psi^\pm(p)}.
\end{equation}
Therefore, scattering cross-sections, which are usually
considered to be the only  measurable quantities in scattering theory, are
insensitive  to the choice of a wave operator.

It is easy to show that   $W^\pm$, $W_\sr^\pm$,
 $W_\dol^\pm$ are all wave operators in the sense of the above
 definition. We also note that for the wave operators $W^\pm$ the
 corresponding generalized eigenfunctions, see \eqref{eq:eigenaa}, 
 jointly diagonalize $H$ and $v^{\pm}$.


\subsection{Low-energy asymptotics of classical orbits}

In the remaining part of the introduction we consider a more restricted class
of potentials. 
To simplify the presentation, in this introduction let us
  assume that the potential takes
the form
\begin{equation}
  \label{eq:1v}
 V(x)=-\gamma |x|^{-\mu} +O(|x|^{-\mu-\epsilon}),
\end{equation} where $\mu \in ]0,2[$ and $\gamma, \epsilon>0$. For
derivatives, assume that $\partial^\beta
\left(V(x)+\gamma |x|^{-\mu}\right)=O(|x|^{-\mu-\epsilon-|\beta|})$.
Compactly supported
singularities can be included.

For potentials satisfying (\ref{eq:1v}) we would like to extend the 
results described in Subsection \ref{s1.1}
 down to the energy $\lambda=0$.
To this end we change variables to ``blow
up'' the discontinuity at $\lambda=0$. This amounts to looking at
$\xi=\sqrt{2\lambda}\omega$ as depending on  two independent
variables $\lambda\geq 0$ and $\omega \in S^{d-1}$.
It is proven in \cite{DS1} that for any $\omega\in S^{d-1}$,
$\lambda\in[0,\infty[$ and $x$ from an appropriate outgoing/incoming region
    there exists a solution of the problem
\begin{equation}\label{eq:mixed conditions222}
\begin{cases}
\ddot y(t) =-\nabla V(y(t)),\\
\lambda=\tfrac {1}{2 }\dot y(t)^2 +V(y(t)),\\
y(\pm 1)=x,\\
\omega={\pm}\lim_{t\to {\pm}\infty}y(t)/|y(t)|.
\end{cases}
\end{equation}
One obtains a family $y^\pm
(t,x,\omega,\lambda)$ of solutions smoothly depending on
parameters. All scattering orbits are of this form.
Using these solutions one can construct a solution
$\phi^\pm(x,\omega,\lambda)$ to
the eikonal  equation
\begin{equation}
\frac12 \left(\nabla_x\phi^\pm(x,\omega,\lambda)\right)^2+V(x)=\lambda
\end{equation}
satisfying $\nabla_x\phi^\pm(x,\omega,\lambda)=\dot
y(\pm1,x,\omega,\lambda)$.  

\subsection{Low-energy asymptotics of wave and scattering matrices}

In the quantum case, we can use the new functions $\phi^\pm(x,\omega,\lambda)$
in the modifiers $J^\pm$, which lead to the definitions of the wave
operators $W^\pm$ and the scattering operator $S$. We can also improve on the
choice of the symbols $a^\pm(x,\xi)$ by assuming that in the incoming/outgoing
region they satisfy the appropriate  transport equations.

The main new result of
our paper about wave matrices can be summarized in the following theorem:
\begin{thm} There exists the  norm 
limit  of wave matrices at zero energy:
\[W^\pm(0)=\lim_{\lambda\searrow0} W^\pm(\lambda)\]
in the sense of operators in ${\mathcal B}(L^2(S^{d-1}),L^{2,-s}(\R^d))$,
where $s>\frac12+\frac\mu 4$.
\end{thm}
 The operator $W^\pm(0)$ can be called {\em the wave matrix at zero energy}.
We can   introduce
$w^\pm(\omega,0):=W^\pm(0)\delta_\omega$, called
the
 {\em generalized eigenfunction of $H$ at zero energy and
fixed asymptotic normalized velocity  $\omega$}. It  belongs to the weighted space
 $L^{2,-p}(\R^d)$ where
 $p>\frac{d}{2}+\frac{\mu}{2}-\frac{d\mu}{4}$. We shall also show
 weighted $L^2$--
 bounds on its  $\omega$-derivatives.

It is interesting to note that the behaviour of the generalized
eigenfunction $w^\pm(\omega,0)$
 depends strongly on the dimension. In dimension $1$ it is
unbounded, in dimension 2 it is  almost bounded and in dimension greater than 2
it  decays at infinity (without being square integrable).

The  main result of our paper about scattering matrices reads

\begin{thm} There exists the strong limit  of scattering
 matrices at zero energy
\[S(0)=\slim_{\lambda\searrow0} S(\lambda)\]
in the space ${\mathcal B}(L^2(S^{d-1}))$. This 
 limit $S(0)$  is unitary
on $L^2(S^{d-1})$.
\end{thm}

 We remark
 that neither $W(\lambda)$ nor
$S(\lambda)$ are smooth in $\lambda\geq 0$ at the threshold $0$, which
 can seem somewhat surprising given the fact that the boundary value
 of the resolvent
 $R(\lambda +\i0)=(H-\lambda-i0)^{-1}$ (interpreted as acting between appropriate 
weighted spaces)
has  this property (see \cite {BGS}
 for explicit expansions in the purely Coulombic case). 

\subsection{Geometric approach to scattering theory}
There exists an alternative approach to scattering theory, based on the study
of   generalized eigenfunctions. It allows us
 to characterize scattering matrices
by the spatial asymptotics of generalized eigenfunctions.
 It was used
in particular in
 Vasy \cite{V1} or \cite[Remark 19.12]{V2}. 
We shall study this approach, including the case of the zero energy,
   in Subsection \ref{Geometric scattering
matrices}.

\subsection{Low energy asymptotics of short-range and Dollard operators}

Let us stress that the existence of the limits of wave and scattering 
matrices at zero energy is made possible not only by appropriate assumptions
on the potentials, but also by the use of appropriate modifiers. Wave
matrices $W_\sr^\pm(\lambda)$
 defined by the standard short-range procedure, as well as
the Dollard modified wave operators $W_\dol^\pm(\lambda)$,  {\it do not} have this
property. They differ from our $W^\pm(\lambda)$ by a momentum dependent phase
factor that has an oscillatory behaviour as $\lambda\searrow0$. In
particular,
\begin{subequations}
  \begin{eqnarray}
W_\sr^\pm(\lambda)\ =&W^\pm(\lambda)\exp\left(\i
O(\lambda^{\frac12-\frac1\mu})\right),& 1<\mu<2;\label{osci1}\\ 
W_\dol^\pm(\lambda)\  =&\;  W^\pm(\lambda)\exp\left(\i
O(\lambda^{-\frac12}\ln \lambda)\right),& \mu=1;\\
W_\dol^\pm(\lambda)\  =&W^\pm(\lambda)\exp\left(\i
O(\lambda^{\frac12-\frac1\mu})\right),& \frac12<\mu<1.
\end{eqnarray}
\end{subequations}

We
 remark that oscillatory behaviour similar to (\ref{osci1}) was
  proved  in \cite {Y1} in the one-dimensional setting.

\subsection{Location of singularities of the zero energy 
 scattering matrix}

A recurrent idea of scattering theory
 is  the parallel behaviour of classical and quantum
systems. One of its manifestations is the relationship between scattering
 orbits at a given energy and the location of singularities of the scattering matrix.

In the case of positive energies the relationship is simple and well-known.
To describe it note that  scattering orbits
 of positive energy  have the  deflection angle that goes to zero
when the distance of the orbit to the center goes to infinity.
 In the quantum case this
 corresponds to the fact that the integral kernel 
of scattering matrices $S(\lambda)(\omega,\omega')$
 at positive energies $\lambda$ are smooth  for $\omega\neq\omega'$  and has a
 singularity at $\omega=\omega'$.

This picture changes at the  zero energy. For potentials considered in our
paper, the deflection angle of zero-energy orbits does not go to zero for
orbits
 far
from the center. The angle of deflection is small for small $\mu$ and goes to
infinity as $\mu$ approaches $2$.

For the strictly homogeneous potential, $V(r)=-\gamma
r^{-\mu}$, one can solve the equations of motion at zero energy.
The (non-collision) zero-energy orbits are given by  
the implicit equation (in polar coordinates)
  \begin{equation} \label{eq:pol-eqn}
\sin(1-\frac\mu2)\theta(t)= 
\Big ({r(t)\over r_\crit }\Big )^{-1+\frac\mu2},
 \end{equation} see \cite [Example 4.3]{DS1}. Whence  the deflection 
angle of such trajectories equals
$-\frac{\mu\pi}{2-\mu}$.  In particular, for attractive Coulomb potentials it
equals $-\pi$, which corresponds to the well-known fact that in this case
zero-energy orbits are parabolas 
(see
 \cite[p. 126]{Ne} for example). 

One of the main results of our paper is a quantum analogue of this fact:

\begin{thm} \label{thm:19}The integral kernel of
  the zero-energy scattering matrix
$S(0)(\omega,\omega')$ is smooth away from $\omega,\omega'$ satisfying
  $\omega\cdot\omega'=\cos  \frac{\mu\pi}{2-\mu}$.
\end{thm}

Note that this fact was known before in the case of Coulomb potentials, at least in dimension 
$d\geq
 3$. In this case $S(0)=\e^{\i  c}P$,   where $(P\tau)(\omega)=\tau(-\omega)$.
 Moreover in this case one can compute (using special functions, see Yafaev
 \cite {Y3} for an explicit formula)
 \begin{equation}
   \label{eq:S=P22}
   S_{\dol}(\lambda)=\e^{\i \lambda^{-1/2} \{C_1\ln \lambda+C_2+o(\lambda^0)\}}(P+o(\lambda^0)).
 \end{equation}

\subsection{Type of singularity of the scattering matrix}

Let $\Lambda$ be the operator on $L^2(S^{d-1})$ such that $\Lambda Y_l=(l+d/2-1)
Y_l$,
where $Y_l$ is a spherical harmonic of order $l$. Alternatively, it can be
introduced as follows:
\[\Lambda:=\sqrt{L^2+(d/2-1)^2},\]
where
\begin{equation*} 
L^2=\sum_{1\leq i<j\leq d}L_{ij}^2,\ \ \  \i L_{ij}=x_{i} \partial_{x_j}-x_{j} \partial_{x_i}.
\end{equation*}

Note that, for any $\theta$, the distributional kernel of
$\e^{\i\theta\Lambda}$ can be computed explicitly and its singularities appear
at $\omega\cdot\omega'=\cos\theta$.
This is expressed in the following result:

\begin{prop} $\e^{\i\theta \Lambda}$ equals
\begin{enumerate}\item $c_\theta I$, where
$I$ is the identity, if $\theta\in \pi2\Z$;
\item $c_\theta P$, where $P$ is  the parity operator,
 if $\theta\in
\pi(2\Z+1)$;
\item the  operator whose Schwartz kernel is of the form
$c_\theta(\omega\cdot\omega'-\cos
\theta +\i0
)^{-\frac d2}$ if   $\theta\in ]\pi 2k,\pi(2k+1)[$ for some $k\in \Z
$;
\item the  operator whose Schwartz kernel is of the form
$c_\theta(\omega\cdot\omega'-\cos
\theta -\i0
)^{-\frac d2}$ if   $\theta\in ]\pi( 2k-1),\pi 2k[$ for some $k\in \Z
$.
\end{enumerate}\end{prop}
 
For all $\theta$, the operator $\e^{\i\theta\Lambda}$ belongs to the
  class of Fourier
  integral operators of order 0  in the sense of H\"ormander \cite{Ho2,Ho4}.

The operator $\Lambda$ can be used to describe the leading asymptotics of the
scattering matrix at zero energy:

\begin{thm}\label{thm:singrad2}
If (\ref{eq:1v}) holds with $V$ being  spherically symmetric (up to a
compactly supported possibly singular term), now with  the
number $\epsilon$ obeying $\epsilon>1-\tfrac\mu 2$, then
\[S(0)=\e^{\i c}\e^{-\i\frac{\mu\pi}{2-\mu}\Lambda}+K,\]
where $K$ is compact.
\end{thm}

We shall prove Theorem \ref{thm:singrad2} by one-dimensional
WKB-analysis.

\subsection{Kernel of $S(0)$ as an explicit oscillatory integral}

In the case $V=-\gamma|x|^{-\mu} +O \big (|
x|^{-1-\tfrac \mu 2-\epsilon}\big ), \;\epsilon>0,$ it is possible to 
 represent the distributional
kernel of the scattering
matrix $S(0)$  (modulo a smoothing term)
in terms of a fairly explicit  oscillatory integral. 
This provides an alternative way to prove Theorem
\ref{thm:19} on the location of singularities of the scattering
matrix -- given the stronger conditions on the potential (we remark
that our proof of Theorem
\ref{thm:19} is rather abstract, see Subsection \ref{Propagation of singularities for zero-energy  generalized 
  eigenfunctions}).   
Moreover, although we shall not elaborate in this paper,  it is actually feasible
to prove a partial version of Theorem \ref{thm:singrad2} (more
precisely for the cases $\tfrac{2}
{2-\mu}\notin \N$)  using this fairly  explicit 
integral.


\subsection{Generalized eigenfunctions}

A  solution of the equation
\begin{equation}(-\Delta+V(x)-\lambda)u=0\label{lambda}\end{equation}
in $\bigcup_s L^{2,-s}(\R^d)$  will be called a {\em generalized eigenfunction
  with energy $\lambda$.} One of our results says that each generalized
  eigenfunction with positive or zero energy is of the  form
 $W^\pm(\lambda)\tau$, where $\tau$ is a distribution on the
  sphere $S^{d-1}$.

Such generalized eigenfunctions  are never
  square-integrable. A rough method
 to describe their behaviour for large $x$ is to
  use weighted spaces $L^{2,s}(\R^d)$ with appropriate $s$.
A more precise
  description is provided  by 
the  so-called {\em Besov spaces}. One of our results says 
 that  the range of
  (incoming and outgoing) wave matrices can be described precisely by an
  appropriate Besov space. One can also describe quite precisely 
their spatial asymptotics. In the case of zero energy, these results are new.

\subsection{Propagation of singularities for zero-energy  generalized 
  eigenfunctions} \label{Propagation of singularities for zero-energy  generalized 
  eigenfunctions}

It is well-known that some of the properties 
of solutions of PDE's of the form $P(x,D)u=0$
  can be explained by the behaviour of classical
hamiltonian
dynamics given by the principal symbol of $P$. One of the best known expressions
of this idea is H\"ormander's theorem about propagation of singularities.

Similar ideas are true in the case of Schr\"odinger operators.
 This is well understood for
positive energies. In the case of zero energy a similar analysis is
possible. It has an especially clean formulation if we assume that the
potential is $V(x)=-\gamma
|x|^{-\mu}$. Under this condition, the set of orbits of the 
classical system given by $h(x,\xi)$ is invariant with respect to an
appropriate scaling. This allows us to reduce the phase space.

In the quantum case, 
we introduce an appropriate concept of a {\em wave front set} adapted to the
solutions to (\ref{lambda}), different from H\"ormander's.
One of our main results describes a possible location of
this special
 wave front set for solutions to (\ref{lambda}) for $\lambda=0$
 -- the statement is very
similar to the statement of the original H\"ormander's theorem; it is
used in a proof of Theorem \ref{thm:19}.

\subsection{Sommerfeld radiation condition}

Another of our main results is a version of the Sommerfeld radiation
condition for zero energies. It says that given $v$ in a certain weighted
space
a solution $u$ of the equation $(H-\lambda)u=v$ satisfying appropriate outgoing/incoming
phase space localization is always of the form $u=R(\lambda\pm\i0)v$.

This somewhat technical result has a number of interesting applications. In
particular, we use it in  our proof that $S(0)$ can be expressed in terms of
an oscillatory integral, and also in  the  description of the asymptotics
of generalized eigenfunctions
at large distances.

\subsection{Organization of the paper}

The paper is organized as follows: In Section \ref{Conditions} we
impose conditions on the potential. In the case we allow the potential 
 to have  a
non-spherically symmetric term  we shall need  certain regularity properties
of the leading spherically symmetric term. These properties are stated
in Condition \ref{assump:conditions2}; they are fulfilled for the
example \eqref{eq:1v} discussed above.

In Section \ref{Preliminaries}
 we describe and extend some of 
results from our
previous papers. In particular, we recall
the construction of scattering phases in
  \cite{DS1} (given there under the same  conditions). We describe and
  to some extend
 the study of the  properties of these objects. 

In Section  \ref{Uniform resolvent estimates} we recall various microlocal
 resolvent estimates  from
\cite{FS} (slightly extended). We also introduce the concept of the scattering
 wave
 front set adapted to energy zero. We give its applications, in particular a
 result about the Sommerfeld radiation condition at zero energy.

In Section \ref{FIO}  we describe the modifiers used in
our paper. They are given by a WKB-type ansatz, which involves solving
 transport equations.
 
 In Section \ref{Wave matrices}
we  introduce wave operators  and wave matrices. We describe
  their low-energy asymptotics.

In Section \ref{Scattering matrices}  we introduce 
scattering operators and matrices. We 
analyse their low-energy asymptotics.

In Section  \ref{Generalized eigenfunctions} we study properties of
generalized eigenfunctions for non-negative  energies. 

In Section \ref{Propagation of singularities at zero energy}
 we restrict our attention
to potentials of the form \eqref{eq:1v}.
We 
 show the classical rule, $\omega\cdot \omega'= \cos
  {\mu\over 2-\mu}\pi$,  for the location of zero-energy
  singularities  (cf. Theorem \ref{thm:19}). 
We also show
a
  ``propagation of scattering singularities result'', see  Proposition
  \ref{prop:propa1}, on generalized zero-energy eigenfunctions. Under
  stronger conditions than \eqref{eq:1v} we represent the kernel of
  $S(0)$ as an explicit oscillatory integral.

In Section  \ref{Singularity of the kernel of the scattering matrix} we
study an  explicit Fourier 
integral operator on the unit sphere
-- the evolution operator for the wave equation -- and we show that it
coincides, modulo a compact term, with $S(0)$ (again, 
 under stronger conditions than \eqref{eq:1v}).

In Appendix \ref{Appendix} we present, in an
 abstract setting,   various  elements of stationary
scattering theory used in our paper.

\section{Conditions} 
\label{Conditions}
We shall consider a classical Hamiltonian $h=\tfrac {1}{2 }\xi^2+ V$ on ${\mathbb
R}^d\times {\mathbb R}^d$ 
where 
$V$ satisfies Condition
\ref{assump:conditions1} (in classical mechanics we can take  $V_3=0$)
and possibly Condition  \ref{assump:conditions2}
(both stated below).  Throughout the paper 
we shall use the
non-standard notation $\langle x \rangle$ for $x\in \R^d$ to denote a function $
\langle x \rangle= f(r); \;r=|x|$, where here $f\in C^\infty([0,\infty[)$
is taken convex, and obeys $f=\tfrac{1}{2}$ for $r<\tfrac{1}{4}$ and
$f=r$ for $r>1$. We shall often use the
 notation $\hat x=x/r$ for vectors $x\in
 \R^d\setminus \{0\}$. Let  $L^{2,s}=L^{2,s}({\mathbb
R}^d_x)=\langle x \rangle^{-s}L^2({\mathbb
R}^d_x)$ for any $s\in \R$ (the corresponding norm will be denoted by
 $\|\cdot \|_s$). Introduce also 
$L^{2,-\infty}(=L^{2,-\infty}(\R^d))=\cup_{s\in \mathbb
R}L^{2,s}$ and $L^{2,\infty}=\cap_{s\in \mathbb
R}L^{2,s}$.
The notation $F(s>\epsilon)$ denotes 
a smooth increasing function $=1$ for $s>\frac {3}{4}\epsilon$ and
 $=0$ for $s<\frac {1}{2}\epsilon$;
 $F(\cdot<\epsilon):=1-F(\cdot>\epsilon)$.
 The symbol $g$ will be used extensively; it stands for the function  $g(r)=\sqrt{2\lambda-2V_1(r)}$
 (for $V_1$ obeying Condition \ref{assump:conditions1} and
 $\lambda\in[0,\infty[$).  
\begin{cond}
\label{assump:conditions1}
The function $V$ can be written as a sum of three 
real-valued measurable functions, $V= V_1 + V_2+V_3$, such that: 
For some $\mu \in ]0,2[$ we
have

\begin{enumerate}[\quad\normalfont (1)]

   \item \label{it:assumption1} $V_1$ is a smooth negative function that only 
     depends on  the 
     radial variable $r$ in the region $r\geq 1$ (that is
     $V_1(x)=V_1(r)$ for $r\geq 1$). There exists $\epsilon_1 > 0$ such
     that $$V_1(r) \leq -\epsilon_1
r^{-\mu},\;r\geq 1.$$
   \item \label{it:assumption2}     For all $\gamma\in ({\mathbb N} \cup \{0\})^{d}$
there exists
$C_{\gamma} >0$ such that
$$
\langle x \rangle^{\mu+|\gamma|} |\partial^{\gamma} V_1(x)|
\leq C_{\gamma}.$$
\item \label{it:assumption3} 
     There exists $\tilde\epsilon_1 > 0$ such that 
     \begin{equation}
       \label{eq:virial}
       rV_1'(r)
\leq -(2-\tilde\epsilon_1) V_1(r),\;r\geq 1.
     \end{equation}

\item \label{it:assumption4} 
     $V_2=V_2(x)$ is smooth and there exists $\epsilon_2 > 0$ such that for all $\gamma\in ({\mathbb N} \cup \{0\})^{d}$
$$
\langle x \rangle^{\mu+\epsilon_2 +|\gamma|} |\partial^{\gamma} V_2(x)|
\leq C_{\gamma}.$$
   \item \label{it:assumption55} $V_3=V_3(x)$ is compactly supported.
\end{enumerate}
\end{cond}

 The following condition will be  needed only in the case $V_2\neq
 0$:
\begin{cond}
\label{assump:conditions2}
Let $V_1$ be given as in Condition  \ref
{assump:conditions1}  and $\alpha
:=\tfrac{2}{2+\mu}$. There exists $\bar\epsilon_1 > \max (0,
1-\alpha(\mu+2\epsilon_2))$ such that 
\begin{align}
 \label {eq:green's1}
  \limsup_{r\to \infty} r^{-1}V_1'(r)\Big (\int^r_1 (-2V_1(\rho))^{-\frac
  {1}{2}}\d\rho\Big )^2 &< 4^{-1}(1-\bar\epsilon_1^2),\\
\limsup_{r\to \infty} V_1''(r)\Big (\int^r_1 (-2V_1(\rho))^{-\frac
  {1}{2}}\d\rho\Big )^2 &< 4^{-1}(1-\bar\epsilon_1^2).\label{eq:green's2}
\end{align}
\end{cond}

 We notice that \eqref{eq:virial} and \eqref{eq:green's1}  tend to
 be somewhat strong conditions for $\mu{\approx}2$. On the other hand  Conditions
 \ref{assump:conditions1} and \ref{assump:conditions2} hold for all
 $\epsilon_2>0$ for the particular example
 $V_1(r)=-\gamma r^{-\mu}$ (with $\epsilon_1 =\gamma$, $\tilde
 \epsilon_1=2-\mu$  and some $\bar\epsilon_1<1-\alpha
\mu$). 

In quantum mechanics we consider $H = H_0+ V$, $H_0=\tfrac {1}{2
} p^2,\;p=-\i \nabla$,  on $\mathcal
H=L^2({\mathbb R}^d)$, and we need the following additional
condition. Clearly Condition \ref {assump:conditions3} (\ref
{it:assumption7}) assures that $H$ is self-adjoint. 
For  an elaboration of Condition \ref {assump:conditions3} (\ref
{it:assumption8}), see \cite {FS}; it guarantees that zero is not an
eigenvalue of $H$.  Condition \ref {assump:conditions3} (\ref
{it:assumption88}) is included here only for convenience of
presentation; with the other conditions there are no small positive eigenvalues, cf. \cite {FS}. 
\begin{cond} 
\label{assump:conditions3}
In addition to  Condition \ref {assump:conditions1}

\begin{enumerate} 
[\quad\normalfont (1)]

\item \label{it:assumption7} $V_3 (H_0+\i )^{-1}$ is a compact operator on
     $L^2({\mathbb R}^d)$.

   \item \label{it:assumption8} $H$ satisfies the unique continuation
      property at infinity.
\item \label{it:assumption88} $H$ does not have positive eigenvalues.
\end{enumerate}
\end{cond}

\section{Classical orbits}
\label{Preliminaries}
In this section we recall and extend the results of \cite{DS1} about low energy
classical orbits that we will
need in our paper.

\subsection{Scattering orbits at positive energies}
\label{Classical preliminaries}

We introduce for $R\geq 1$ and $\sigma>0$
 \begin{align*}
 &\Gamma^+_{R,\sigma}(\omega)=\{y\in {\mathbb R}^d\ |\ y\cdot \omega\geq (1-\sigma)|y|,\;|y|\geq R\};\; \omega\in S^{d-1},\\
   &\Gamma^+_{R,\sigma}=\{(y,\omega)\in {\mathbb R}^d \times
   S^{d-1}\ |\ y \in \Gamma^+_{R,\sigma}(\omega)\}.
 \end{align*}

\begin{lemma} Suppose that $V_1$ satisfies \eqref{symbol}. Let
  $\sigma\in ]0,2[$.
 Then there exists a decreasing function
$]0,\infty[\ni\lambda\mapsto R_0(\lambda)$ such that for all
   $|\xi|\geq\sqrt{2\lambda}$ and $x\in\Gamma_{R_0(\lambda),\sigma}^+
(\hat\xi)$ there
   exists a unique solution $y(t)=y^+(t,x,\xi)$ 
of the problem (\ref{eq:mixed conditions22}) such
   that $y(t)\in\Gamma_{R_0(\lambda),\sigma}^+(\hat\xi)$ for $t>1$.  If we set
\[F^+(x,\xi):=\dot y^+(1,x,\xi),\]
then $\rot_x F^+(x,\xi)=0$.
\label{posi}\end{lemma}

For any $\xi \neq 0$ we let $\lambda=2^{-1}\xi^2$, $\omega=\hat\xi$  and $R=R(\lambda)$.
For 
$(x,\omega)\in\Gamma^+_{R,\sigma}$ we choose
    a path
$[0,1]\ni l\mapsto\gamma(l)\in
\Gamma^+_{R,\sigma}(\omega)$ such that $\gamma(0)=R\omega$ and
    $\gamma(1)=x$.  We set
\[\phi^+(x,\xi):=\int_0^1F^+(\gamma(\tau),\xi)\cdot\frac{\d\gamma
(l)}{\d l}\d l+|\xi|R. \]
Note that $\phi^+(x,\xi)$ does not depend on the choice of the path
$\gamma$. For instance, we can take the interval joining these two points and
then 
  \begin{equation}
    \phi^+(x,\xi)=(x-R\omega)\cdot\int _0^1
    F^+(l(x-R\hat\xi)+R\hat\xi,\xi)\d l +|\xi|R.
\label{eqo1}  \end{equation} 
Another
possible choice is the radial interval from $R\omega$ to $|x|\omega$
and then the arc towards $x$:
\begin{equation}\phi^+(x,\xi)=\int_{R}^{|x|}F^+(l\omega,\xi)\cdot
  \omega \d l
+\int_0^{\arccos\omega \cdot \hat x}F^+(|x| v_\alpha,\xi)\cdot |x|\frac{\d v_\alpha
}{\d \alpha}\d \alpha+|\xi|R,\label{eqo2}
\end{equation}
where $v_\alpha:=\cos\alpha\omega+\sin\alpha\frac{\hat x-\omega\
  \omega\cdot\hat x}{\sqrt{1-(\omega\cdot\hat x)^2}}$.

The  phase function constructed above
  essentially coincides with the Isozaki Kitada (outgoing) phase
  function, 
  cf. \cite{Is1}, \cite[Definition 2.3]{IK1} or \cite[Proposition 2.8.2]{DG}. In
  particular, for any $\xi \neq 0$,  there are bounds   \begin{align}\label{eq:postca}
    \partial_\xi^\kappa \partial_x^\gamma(\phi^+
(x,\xi) -\xi\cdot x)&=O\big(|x|^{\delta-|\gamma|}\big )\;\text {for}\;|x|\to
    \infty,\\\delta&>\max (1-\mu,0).\nonumber   
  \end{align}
 These bounds are not  uniform in  $\xi\neq 0$,  they are
however uniform on compact subsets
 of $\R^d\setminus \{0\}$.


\subsection{Scattering orbits at low  energies}

Let us now recall  some results about scattering orbits
taken  from
\cite{DS1}.

We assume
 Conditions \ref{assump:conditions1} and 
\ref{assump:conditions2} (only Condition \ref{assump:conditions1}  if $V_2=0$). 
The fact that our Condition \ref{assump:conditions1}
includes a possibly singular potential $V_3$ is irrelevant for this
subsection since by assumption this term is compactly supported. More
precisely we just need to make sure that the $R_0\geq 1$ in Lemma
\ref{lemma:mixed_2} stated below is taken so large that $V_3(x)=0$
for $|x|\geq R_0$, then \cite{DS1} applies.

\begin{lemma}
   \label{lemma:mixed_2}
There exist $R_0\geq 1$ and $\sigma_0>0$
such that for
all $R\geq R_0 $ and for all  positive 
$\sigma\leq \sigma_0$ the problem \eqref{eq:mixed conditions222} 
is
solved  for all data 
$(x,\omega)\in \Gamma^+_{R,\sigma}$ and $\lambda \geq 0$ 
by a unique
 function $y^+(t,x,\omega,\lambda),\;t\geq 1$, such that
$y^+(t,x,\omega,\lambda)\in \Gamma^+_{R,\sigma}(\omega)$ for all $t\geq 1$.
Define a vector field $F^+(x,\omega,\lambda)$
 on $\Gamma^+_{R_0,\sigma_0}(\omega)$ by
  \begin{equation}
    \label{eq:vector field}
    F^+(x,\omega,\lambda)=\dot y^+(t=1;x,\omega,\lambda).
  \end{equation} 
Then
\[\rot_x F^+(x,\omega,\lambda)=0.\]
\end{lemma}

Note that under the assumptions of Lemma \ref{lemma:mixed_2}, we can suppose
that $R_0(\lambda)$, introduced in Lemma \ref{posi}, equals $R_0$
 for all $\lambda>0$. We can define
$\phi^+(x,\omega,\lambda)$ on
 $(x,\omega,\lambda)\in\Gamma^+_{R,\sigma}\times[0,\infty[$.
For further reference let us record the analogues of (\ref{eqo1}) and
(\ref{eqo2}): 
  \begin{equation*}
    \phi^+(x,\omega,\lambda)=(x-R_0\omega)\cdot\int _0^1
    F^+(l(x-R_0\omega)+R_0\omega)\d l +\sqrt {2\lambda}R_0,
  \end{equation*} 
\[
\phi^+(x,\omega,\lambda)=\int_{R_0}^{|x|}F^+(l\omega,\omega,\lambda)\cdot
\omega \d l
+\int_0^{\arccos\omega \cdot \hat x}F^+(|x| v_\alpha,\omega,\lambda)\cdot \frac{\d v_\alpha
}{\d \alpha}\d\alpha +\sqrt {2\lambda}R_0.\]

 We will add the subscript ``sph'' to all
objects  where $V$ is replaced by
 the  (spherically symmetric) potential $V_1$.
The following result  is proven in \cite{DS1}:

\begin{prop}\label{prop:mixed_2aaaa}
There exists $\breve \epsilon =\breve \epsilon(\mu,\bar \epsilon_1,\epsilon_2)>0$ and uniform
bounds
\begin{subequations}\label{group2}
\begin{equation}
  \label{eq:F's}
 F^+(x)-F_\sph^+(x)=O\big(|x|^{-\mu/2-\breve \epsilon }
\big).
\end{equation} In particular, for constants $C,\;c>0$ independent of $x,\;\omega
$ and $\lambda$
\begin{equation}
  \label{eq:F's2}
  \Big|\frac {F^+(x)}{|F^+(x)|}-\frac {F_\sph^+(x)}{|F_\sph^+(x)|}\Big |\leq
  C|x|^{-\breve \epsilon},
\end{equation} and 
\begin{align}
\label{eq:F_angle1}
 \frac {F^+(x)}{|F^+(x)|}\cdot \hat x  &\geq 1-C\big (1-\hat x  \cdot \omega \big )- C|x|^{-\breve \epsilon},\\ 
\label{eq:F_angle2}
 \frac {F^+(x)}{|F^+(x)|}\cdot \hat x&\leq 1-c\big (1-\hat x \cdot \omega \big )+ C|x|^{-\breve \epsilon},\\ 
\label{eq:F_angle3}
 \frac {F^+(x)}{|F^+(x)|}\cdot \omega &\geq 1-C\big (1-\hat x \cdot \omega \big )- C|x|^{-\breve \epsilon}. 
\end{align}

More generally (with the same $\breve \epsilon >0$), 
for all multiindices $\delta$ and $\gamma$ there are uniform bounds 
 \begin{align}\label{eq:derF222}\partial _\omega^\delta \partial _x^\gamma  
F^+(x)&=\langle x \rangle ^{-|\gamma|}O\left(
   g(|x|) \right ),\\
\label{eq:derF2222} \partial _\omega^\delta \partial _x^\gamma  \big (F^+(x)-F_\sph^+(x)\big
   )&=\langle x \rangle ^{-\breve \epsilon-|\gamma|}O\left(
   g(|x|) \right ).
\end{align}  
 
\end{subequations}

The vector field
 $F^+(x,\omega,\lambda)$, as well as all derivatives $\partial
 _\omega^\delta \partial _x^\gamma F^+$, are jointly continuous in the  variables $(x,\omega)\in \Gamma^+_{R_0,\sigma_0}$ and $\lambda \geq 0$.
\end{prop}

The problem \eqref{eq:mixed conditions222} in
the case of $t\to - \infty$ can also be solved.  
We introduce for $R\geq 1$ and $\sigma>0$
 \begin{align*}
 &\Gamma^-_{R,\sigma}(\omega)=\{y\in {\mathbb R}^d\;|\;y\cdot \omega\leq (\sigma-1)|y|,\;|y|\geq R\},\; \omega\in S^{d-1};\\
   &\Gamma^-_{R,\sigma}=\{(y,\omega)\in {\mathbb R}^d\times
   S^{d-1}|\;y \in \Gamma^-_{R,\sigma}(\omega)\}.
 \end{align*}

Mimicking the previous procedure, starting from the mixed problem \eqref{eq:mixed conditions222} in
the case of $t\to - \infty$, 
 we can  similarly construct a solution $\phi^-(x,\omega,\lambda)$ to the eikonal
equation in some $\Gamma^-_{R,\sigma}(\omega)$. This amounts to setting
\begin{equation}
  \label{eq:rel_phi+-}
 \phi^-(x,\omega,\lambda)=-\phi^+(x,-\omega,\lambda),\;x\in
 \Gamma^-_{R_0,\sigma_0}(\omega) =\Gamma^+_{R_0,\sigma_0}(-\omega). 
\end{equation}

\subsection{Radially symmetric potentials
}
\label{Classical preliminaries1}
In this subsection we assume that $V_2=0$, which means that the potential is
 spherically symmetric. More precisely, we assume that for $r\geq R_0$ 
\[|\partial_r^nV(r)|\leq c_nr^{-n-\mu},\ \ V(r)\leq-cr^{-\mu},\ c>0,\ \ 
rV'(r)+2V(r)<0.\]
Note that motion in such a potential is confined to a 2-dimensional
plane. In the case of the trajectory $y^+(t,x,\omega,\lambda)$, it is the plane
spanned by $\omega$ and $\hat x$.
It is also convenient to introduce the vectors $x^\perp:=
\frac{\omega-\cos\theta_1\hat x}{\sin\theta_1}$
and $\omega^\perp:=\frac{\hat x-\cos\theta_1\omega}{\sin\theta_1}$,
 where $\omega\cdot \hat x=\cos\theta_1$.
 Therefore,  we can restrict temporarily our attention to a
2-dimensional system. We will use the polar coordinates
$(r\cos\theta,r\sin\theta)$.
Note that
the energy $\lambda$ and the  angular momentum $L$ are preserved
quantities. Therefore, the Newton equations (for outgoing orbits) can be reduced to
\begin{equation}\label{eq:equations of motion}
\begin{cases}
\dot \theta=Lr^{-2},
\\
\dot r=\sqrt{2\lambda-2V(r)-L^2r^{-2}}.
\end{cases}\;
\end{equation}

\begin{lemma}\label{pri}
For some $\theta_0>0$,  for all $r_1\geq R_0$, $|\theta_1|\leq
\theta_0$  and $\lambda\geq 0$ we can find a solution of (\ref{eq:equations of
  motion}) 
satisfying 
\[r(1)=r_1, \ \dot{r}(1)>0,\ \lim_{t\to\infty}\theta(t)=0,\
\theta(1)=\theta_1.\] 
There exists a function $(r_1,\theta_1,\lambda)\mapsto
L(r_1,\theta_1,\lambda)\in {\mathbb R}$ specifying  the total angular
momentum 
 of the solution $y^+(t,x,\omega,\lambda)$. This function
$L$ is an odd function in $\theta_1$.
We have the following estimates:
\begin{subequations}
  \begin{align}
\partial_{r_1}^n \partial_{\theta_1^2}^m L^2&=O\left(
 r_1^{2-n} g(r_1)^{2}\right),\
 \
 n,m\geq0;
\label{estib2}\\
\partial_{r_1}^n \partial_{\theta_1^2}^m\frac{ L}{\theta_1}&= O\left(
 r_1^{1-n}g(r_1)\right),\ \
 n,m\geq0.
\label{estia2}
\end{align}
\end{subequations}
\end{lemma}


This allows
us to compute the initial velocity of the trajectory:
\[F^+(x,\omega,\lambda)=\sqrt{2\lambda-2V(r)-L^2/r^2} \hat
x-\frac{L}{r}x^\perp.\]
The function $\phi^+$ equals, with $r=|x|$ and $\cos\theta=\hat x\cdot
\omega$,
\begin{equation}\phi^+(x,\omega,\lambda)=\sqrt{2\lambda}R_0
+\int_{R_0}^{r} 
\sqrt{2\lambda-2 V(r')}\d r' +\int_0^{\theta}
L(r,\theta',\lambda)\d\theta'.
\label{phi}\end{equation}
Therefore, using also that $\nabla_\omega \theta =- \omega^\perp$, 
\begin{equation} \nabla_\omega\phi^+=
-L(r,\theta,\lambda)\omega^\perp.
\end{equation}
This gives the following estimates (in any dimension):
\begin{lemma}\label{lemma:jan} There exist constants $C,c>0$ such that 
\begin{subequations}
\begin{eqnarray}
|\hat x \cdot
 F^+(x)-g(|x|)|&\leq&C(1-\hat x\cdot\omega)g(|x|),\\
|F^+(x)-\hat x\ \hat x\cdot F^+(x)|&\leq& C
\sqrt{1-\hat x\cdot\omega}g(|x|),\\
|\nabla_\omega \phi^+|&\geq& c
\sqrt{1-\hat x\cdot\omega}g(|x|)|x|,\label{eq:ole0}\\
\partial_\omega^\delta\partial_x^\gamma\phi^+&=&\langle
x\rangle^{1-|\gamma|}O(g(|x|))\label{eq:ole}.\end{eqnarray}
\end{subequations}
\end{lemma}
We calculate for $\lambda>0$:
\begin{eqnarray*}
\nabla_\xi F^+
&=&(2\lambda)^{-\frac12}\nabla_{\omega}F^++
(2\lambda)^{\frac12}\partial_\lambda F^+\otimes\omega,\\
\nabla_\omega F^+
&=&L\partial_{\theta}L(2\lambda-2V(r)+L^2r^{-2})^{-\frac12}
r^{-2}\omega^\perp\otimes \hat x+
\partial_{\theta}Lr^{-1}\omega^\perp\otimes
x^\perp-\frac{L}{r}\nabla_\omega x^\perp,\\
\partial_\lambda F^+&=&
(2\lambda-2V(r)-L^2r^{-2})^{-\frac12}(1-L\partial_\lambda
Lr^{-2})\hat x
-\partial_\lambda Lr^{-1}x^\perp.
\end{eqnarray*}

Specifying to $x$ parallel to $\omega$ and noting that $L(x,\hat x,\lambda)=0$,
we obtain
\begin{eqnarray}\label{estoaaa}
\nabla_{\xi} F^+&=&(2\lambda)^{1/2}\partial_\lambda(2\lambda-2
V(|x|))^{1/2}\hat x\otimes\hat x\nonumber
\\
&&-(2\lambda)^{-1/2}|x|^{-1}\partial_{\theta}L\  x^\perp\otimes x^\perp
\nonumber\\&=&
(2\lambda)^{1/2}(2\lambda-2 V(|x|))^{-1/2}\hat x\otimes\hat x
\\
&&+
(2\lambda)^{-1/2}|x|^{-1}
\Big (\int_{|x|}^\infty r^{-2}(2\lambda-2 V(r))^{-1/2}\d
r\Big )^{-1}x^\perp\otimes x^\perp,\nonumber\end{eqnarray}
cf. \cite[(4.5)]{DS1}. 

In an arbitrary dimension, the formula is the same except that the second term
is repeated $d-1$ times on the diagonal. Therefore,
\begin{equation}\label{esto22}
\det\left(\nabla_\xi\nabla_x \phi^+(x,\sqrt{2\lambda}\hat
  x)\right)^{1/2}=
(2\lambda)^{(2-d)/4}g(r)^{-1/2}
\Big (r^{-1}h(r)\Big )^{(d-1)/2},
\end{equation} where we have introduced the notation 
\begin{equation}
  \label{eq:8aammmi}
   h(r):=\big (\int_{r}^\infty r'^{-2}g(r')^{-1}\d
r'\big )^{-1}.
\end{equation} 

Note the (uniform)
 bounds
\begin{equation} \label{eq:8aammm}
  crg(r)\leq h(r) \leq Crg(r).
\end{equation}

Whence, combining (\ref{esto22}) and (\ref{eq:8aammm}),
\begin{eqnarray}
c(2\lambda)^{(2-d)/4}g(r)^{(d-2)/2}&\leq&
\det\big (\nabla_\xi\nabla_x \phi^+(x,\sqrt{2\lambda}\hat
x)\big )^{1/2}\nonumber 
\\&\leq&
C(2\lambda)^{(2-d)/4}g(r)^{(d-2)/2}.\label{esto}\end{eqnarray}

 \section{Boundary values of the resolvent}
\label{Uniform resolvent estimates}
In this section we impose Conditions
\ref{assump:conditions1} and \ref{assump:conditions3}.
We shall recall (and extend) some resolvent estimates of \cite {FS}. They are
important tools used throughout our paper.

In Subsection \ref{Some other  estimates} we will also introduce the
notion of the scattering wave front set, which is well adapted to scattering
theory at various energies. We will return to this concept in particular
in Section \ref{Propagation of singularities at zero energy},
where we will prove a theorem about propagation of singularities for potentials
with a homogeneous principal part. A somewhat cruder version of this theorem
is given already in Subsection \ref{Some other  estimates} (valid, however,
for a more general class of potentials).

In Subsection 
\ref{Sommerfeld radiation condition}   we prove a version of  the Sommerfeld
radiation condition for the zero energy.

\subsection{Low energy resolvent estimates}
\label{Low energy resolvent estimates1}
Let $c$ be a function on the phase space $\R^d\times\R^d$.
The left and right Kohn-Nirenberg quantization of the
symbol $c$ are  the operators $\Opl(c)$ and $\Opr(c)$ 
acting  as 
\begin{align*}
 &(\Opl(c) f)(x) = (2\pi)^{-d/2} \int \e^{\i x\cdot \xi}
c(x,\xi) \hat f(\xi) \,\d \xi,\\ 
&(\Opr(c) f)(x) = (2\pi)^{-d} \int \!\!\!\!\int \e^{\i(x-y)\cdot \xi}
c(y,\xi) f(y) \,\d y \d \xi,
\end{align*} respectively. Notice that $\Opl(c)^*=\Opr(\bar c)$. In 
Proposition \ref{prop:resolvent_basic2} stated below we use for
convenience both of these quantizations, although they can be  used
interchangeably. Alternatively one can use Weyl quantization denoted
by $\Opw(c)$, cf.  \cite {FS}.  
We will often use the following ($\lambda$-dependent) symbols:
\begin{align}
\label{eq:def_a0_b}
a(x,\xi) = \frac{\xi^2}{g(|x|)^2} ,\;\;
b(x,\xi)  =  \frac{\xi}{g(|x|)} \cdot \frac{x}{\langle x \rangle}.
\end{align}

It is convenient to introduce the following
 symbol class: Let $c\in S(m, g_{\mu,\lambda})$, 
$g_{\mu,\lambda}=\langle
x \rangle^{-2}\d x^2+g^{-2}\d\xi^2$ and
 $m=m_{\lambda}=m_{\lambda}(x,\xi)$ be a  uniform weight
 function \cite{Ho3}. Here $\lambda \in [0,\lambda_0]$ (for an arbitrarily fixed $\lambda_0>0$)  is
 considered as a parameter; the function $m$ obeys bounds uniform
 in this parameter (see \cite [Lemma 4.3 (ii)] {FS} for details). For
 a uniform weight
 function $m$, the symbol class $S_{\unif}(m, g_{\mu,\lambda})$ is defined to be
 the set of parameter-dependent smooth symbols $c=c_{\omega,\lambda}$
 satisfying
 \begin{equation}
   \label{eq:sym_class}
  |\partial^\delta_\omega\partial^\gamma_x \partial^\beta_\xi c_{\omega,\lambda}(x,\xi)|\leq
  C_{\delta,\gamma,\beta}m_{\lambda}(x,\xi)\langle x \rangle^{-|\gamma|}g^{-|\beta|}.
 \end{equation} 

We notice that the ``Planck constant'' for this class is 
    $\langle x
    \rangle^{-1}g^{-1}$. 
The corresponding class of quantizations is
    denoted by $\Psi_{\unif}(m, g_{\mu,\lambda})$ (it does not depend on
    whether left or right quantization is used). Finally we remark
    that the  quantizations appearing in Proposition
    \ref{prop:resolvent_basic2} stated below belong to $\Psi_{\unif}(1,
    g_{\mu,\lambda})$, and hence they are bounded uniformly in $\lambda$ (these
    symbols are independent of $\omega$).

We can 
obtain  the following estimates by mimicking the proof of \cite [Theorem 4.1]{FS} (first for
the smooth case $V_3=0$, and then the general case by a resolvent
equation, see \cite [Subsection 5.1]{FS}; here the unique continuation assumption Condition
\ref{assump:conditions3} (\ref {it:assumption8}) comes into play).
 In particular,
 Proposition \ref{prop:resolvent_basic2}
(\ref{it:one2}) follows from \cite[Corollary 3.5]{FS} and a resolvent
identity (cf. \cite[(5.12)]{FS}). Similarly  Proposition
\ref{prop:resolvent_basic2} (\ref{item:C-02}) follows from \cite[Lemma
4.5]{FS} and  the proof of \cite[Lemma
4.6]{FS} (notice that it suffices to show the bounds
(\ref{eq:quantum_part2}) and (\ref{eq:quantum_part2222})  for $t=0$ due to
this proof), while Proposition
\ref{prop:resolvent_basic2} (\ref{some_label21}) follows from  \cite[Lemma
4.9]{FS} and the same minor modification of the proof of \cite[Lemma
4.6]{FS}. As for the continuity statement at the end of the proposition
we refer the reader to the end of this subsection.

 The
notation $R(\lambda+\i 0)$ refers to the limit of the resolvent
$R(\lambda+\i \epsilon)$ as $\epsilon\to 0^+$ in the sense of a form
on the Schwartz space $\mathcal S (\R^d)$, cf.  Remark \ref{remarks:nosmmmmo} \ref{it:csmoottaa}).

\begin{prop}
\label{prop:resolvent_basic2}
\begin{subequations}
\label{group12} 
Fix any $\lambda_0>0$. 
The following conclusions, {\rm
  (\ref{it:one2})--(\ref{it:last+})}, hold uniformly in $\lambda
\in [0,\lambda_0]$:
\begin{enumerate}[\normalfont (i)]
\item \label{it:one2}
For all $\delta> \tfrac {1}{2}$ 
there exists
$C>0$ such that
\begin{align}
  \label
  {eq:x_weights2}
\|\langle x \rangle^{-\delta} g^{1\over 2}&R(\lambda+\i 0) g^{1\over 2}\langle x \rangle^{-\delta} \| \leq C.
\end{align}
\item \label{item:C-02} There exists $C_0 \geq1$ 
 such that if
$\chi_{+} \in C^{\infty}({\mathbb R})$, 
    $\supp (\chi_{+})\subset ]C_0, \infty[$ and $\chi_{+}' \in
C_c^{\infty}({\mathbb R})$, then for all
$\delta> \tfrac {1}{2}$
    and all $s ,t \geq  0$ there exists $C>0$ such that 
\begin{align}\label{eq:quantum_part2}
&\|(\langle x \rangle g)^{s}\langle x \rangle^{t-\delta} g^{1\over 2}\Opl(\chi_{+}(a)) R(\lambda+\i 0) g^{1\over 2}\langle x \rangle^{-t-\delta}  (\langle x \rangle g)^{-s}\| \leq C,
  \\
\label{eq:quantum_part2222}
&\|(\langle x \rangle g)^{-s}\langle x \rangle^{-t-\delta} g^{1\over 2}R(\lambda+\i 0)\Opr(\chi_{+}(a)) g^{1\over 2}\langle x \rangle^{t-\delta} (\langle x \rangle g)^{s}\| \leq C.
\end{align}
\item \label{some_label21} Let $\bar \sigma >0$ and $\chi_{-} \in
  C^{\infty}_{c}({\mathbb R})$. Suppose
$
\tilde{\chi}_{-}, \tilde{\chi}_{+} \in C^{\infty}({\mathbb R})$
satisfy 
$$
\sup
\supp
\tilde{\chi}_{-} \leq  1-\bar \sigma,\; \inf \supp \tilde{\chi}_{+} \geq
\bar \sigma  - 1.$$ 
Then for all $\delta
 >\tfrac {1}{2}$ and all 
$s , t \geq  0$ there exists $C > 0$ such that 
\begin{align}
  \label{eq:PsDO_part21}
&\|(\langle x \rangle g)^{s}\langle x \rangle^{t-\delta} g^{1\over 2}\Opl(\chi_{-}(a)\tilde{\chi}_{-}(b))
R(\lambda+\i 0) g^{1\over 2}\langle x \rangle^{-t-\delta}  (\langle x \rangle g)^{-s}\| \leq C, \\
\label{eq:PsDO_part221}
&\|(\langle x \rangle g)^{-s}\langle x \rangle^{-t-\delta} g^{1\over 2}R(\lambda+\i 0)
\Opr(\chi_{-}(a)\tilde{\chi}_{+}(b)) g^{1\over 2}\langle x \rangle^{t-\delta} (\langle x \rangle g)^{s}\| \leq C.
\end{align}
\item \label{it:last} Suppose $\chi_{-}^1,  \chi_{-}^2 \in
  C^{\infty}_{c}({\mathbb R})$, $\tilde{\chi}_{-}$ and 
$\tilde{\chi}_{+}$ satisfy the
  assumptions from {\rm (\ref{some_label21})}  
 and in addition
$$
\sup \supp \tilde{\chi}_{-}<\inf \supp \tilde{\chi}_{+}.
$$
Then for all $s \geq  0$  there exists $C > 0$ such that
\begin{align}
\label{eq:disjoint_b}
\|\langle x\rangle^s
 \Opl(\chi_{-}^1(a) \tilde{\chi}_{-}(b)) R(\lambda+\i 0)
\Opr(\chi_{-}^2(a)\tilde{\chi}_{+}(b)) \langle x\rangle^s\| \leq C.
\end{align}

\item \label{it:last+}
Suppose $\chi_{+}$ is given as in {\rm (\ref{item:C-02})},  some functions $\tilde{\chi}_{+},
\tilde{\chi}_{-}, \chi_{-}$  are given as in  {\rm (\ref{some_label21})} 
and suppose
$$\dist(\supp \chi_{-}, \supp \chi_{+}) > 0.$$
Then for all $s \geq  0$
 there exists $C>0$
such that 
\begin{align}
  \label{eq:disjoint_a}
&\|\langle x\rangle^s
 \Opl(\chi_{+}(a)) R(\lambda+\i 0) \Opr(\chi_{-}(a)\tilde{\chi}_{+}(b))
 \langle x\rangle^s\| \leq C,\\ 
\label{eq:disjoint_a222}
&\|\langle x\rangle^s \Opl(\chi_{-}(a)\tilde{\chi}_{-}(b)) R(\lambda+\i 0)
\Opr(\chi_{+}(a)) \langle x\rangle^s\| \leq C.
\end{align}
\end{enumerate}
\end{subequations}

All the forms appearing in {\rm
        (\ref{it:one2})--(\ref{it:last+})}
 are continuous in $\lambda\geq0$. In fact the families of corresponding
  operators are continuous ${\mathcal B}(L^2(\R^d))$--valued functions.
\end{prop}

\begin{remarks}\label{remarks:nosmmmmo}
\begin{enumerate}[\quad\normalfont 1)]
\item \label{it:cstor}
    Although this will not be needed we have  in fact (\ref{item:C-02})
    with $C_0=1$; see Corollary  \ref{cor:renerg}  for a related
    result.
\item \label{it:csmoottaa} The paper \cite {FS}  contains a stronger
  version of the so-called limiting absorption principle than can be
  read from Proposition \ref{prop:resolvent_basic2} (\ref{it:one2}): For all $\delta> \tfrac {1}{2}$ 
there exists
$C>0$ such that
\begin{equation*}
\sup_{\lambda+\i \epsilon\in M}\|\langle x
\rangle^{-\delta} g^{1\over 2}R(\lambda+\i \epsilon) g^{1\over
  2}\langle x \rangle^{-\delta} \| \leq C;\;M:=[0,\lambda_0]\times \i \,]0,1],
\end{equation*} and the $\mathcal B (L^2(\R^{d})$--valued function $\langle x
\rangle^{-\delta-\frac\mu4} R(\zeta) \langle x
\rangle^{-\delta-\frac\mu4} $ is uniformly H\"older continuous in
$\zeta\in M$. The (well-known) positive energy analogue of this
assertion states that  for any positive
$\lambda_1<\lambda_0$ the $\mathcal B (L^2(\R^{d})$--valued function $\langle x
\rangle^{-\delta} R(\zeta) \langle x
\rangle^{-\delta} $ is uniformly H\"older continuous in
$\zeta\in M\setminus\{\Re \zeta<\lambda_1\}$; see \ref{it:csmoott3})
for a related remark.
\item \label{it:csmoott}  
  The paper \cite {FS} also contains an extension of Proposition
  \ref{prop:resolvent_basic2} to  powers of the resolvent, however this
  will not be useful in the forthcoming sections; see  Example 
  \ref{example:coulombsing} for a discussion. This is related to the
  fact  
  that our  classical
  constructions are not smooth in $\lambda$ at zero energy, cf. \cite[Remarks
  4.7 1)]{DS1}.
The collection of all estimates
  in Proposition
  \ref{prop:resolvent_basic2} (more precisely a  collection of similar estimates with a  complex spectral
  parameter) yields similar estimates for powers of the
  resolvent by a completely algebraic reasoning, cf. \cite[Appendix A]{FS}. 
\item \label{it:csmoott3}Assume that the potential satisfies Condition \ref{symbol}.
 Then all the bounds 
of Proposition  \ref{prop:resolvent_basic2} remain true uniformly in
$\lambda \in[\lambda_1,\lambda_0]$ for any positive
$\lambda_1<\lambda_0$ provided we replace
\begin{equation}
  \label{eq:newa's} a\to a:= {\xi^2\over 2\lambda},\;b\to
  b:={\xi\over \sqrt{2\lambda}}\cdot {x\over\langle x\rangle}\;\text{
    and }\;g\to 1.
\end{equation} (Under the stronger Conditions
  \ref{assump:conditions1} and \ref{assump:conditions3} the validity of this modification is a
  direct consequence of the bounds 
of Proposition  \ref{prop:resolvent_basic2}.) Also in this case the
families of associated 
  operators are norm continuous (now in $\lambda>0$ only).  
\end{enumerate}
\end{remarks}
\noindent{\em Proof of continuity statements in Proposition
  \ref{prop:resolvent_basic2}.} 
Due to Remark \ref{remarks:nosmmmmo} \ref{it:csmoottaa}) and the calculus of
pseudodifferential operators all appearing
forms in Proposition \ref{prop:resolvent_basic2} are  continuous in $\lambda\geq 0$.

Norm continuity of the corresponding
operator-valued functions also follows from Remark
\ref{remarks:nosmmmmo} \ref{it:csmoottaa}). This can be seen as
follows for $ B_\delta(\lambda):=\langle x \rangle^{-\delta} g^{1\over
  2}R(\lambda+\i 0) g^{1\over 2}\langle x \rangle^{-\delta}$
(appearing in (\ref{it:one2})): 

Pick $\delta'\in ]\frac12,\delta[$,
insert  for (small) $\kappa>0$ the identity $I=F(\kappa|x|<1)+F(\kappa|x|>1)$ on both
sides of $B_\delta(\lambda)$ and expand (into three terms). This yields
\begin{equation*}
  \|B_\delta(\lambda)-F(\kappa|x|<1)B_\delta(\lambda)F(\kappa|x|<1)\|\leq
  C\kappa^{\delta-\delta'}\|B_{\delta'}(\lambda)\|.
\end{equation*} Due to Proposition \ref{prop:resolvent_basic2}
(\ref{it:one2}) the right hand side is $O\big
(\kappa^{\delta-\delta'}\big )$
uniformly $\lambda\geq0$. On the other hand due to Remark
\ref{remarks:nosmmmmo} \ref{it:csmoottaa}) (and the calculus of
pseudodifferential operators) for fixed  $\kappa>0$ the  ${\mathcal B}(L^2(\R^d))$--valued function
$F(\kappa|x|<1)B_\delta(\cdot)F(\kappa|x|<1)$ is 
continuous. Hence $B_\delta(\cdot)$ is  a uniform limit
of continuous functions and therefore indeed continuous.

The other operator-valued functions  can
be dealt with in the same fashion.
\qed 
\subsection{Scattering wave front set}
\label{Some other  estimates}

The remaining subsections of Section \ref{Uniform resolvent estimates}  are devoted to a number of somewhat
technical
estimates on solutions to the equation $(H-\lambda)u=v$ for a fixed  $\lambda\geq0$. Although they are
proved under Conditions
  \ref{assump:conditions1} and \ref{assump:conditions3} we remark  that
  there are
  similar estimates under Condition \ref{symbol} for a fixed
  $\lambda>0$. The reader
may skip this material  on
the first reading.

Throughout  the remaining part of this section  we use the notation   $\langle \xi 
\rangle_1=(1+|\xi|^2)^{1/2}$ and $X=(1+|x|^2)^{1/2}$ for $\xi,x\in \R^d$.

With reference to the symbol class $
S_{\unif}(m,g_{\mu,\lambda})$ from  Subsection \ref{Low energy resolvent estimates1}
 clearly $h_1,h_2\in
S_{\unif}(m,g_{\mu,\lambda})$ 
with  $h_1:=\tfrac{1}{2}\xi ^2+V_1$, 
$h_2:=\tfrac{1}{2}\xi ^2+V_1+V_2$ and $m=g^2\langle \xi /g
\rangle^2_1$. In  the remaining part  of
Section \ref{Uniform resolvent estimates} we shall however only need a 
reminiscence of this symbol class given by disregarding the uniformity in
$\lambda\geq0$. Whence we shall consider 
symbols $c\in S(m, g_{\mu,\lambda})$ meaning, by definition,  that 
 \begin{equation}
   \label{eq:sym_classii}
  |\partial^\gamma_x \partial^\beta_\xi c(x,\xi)|\leq
  C_{\gamma,\beta}m(x,\xi)\langle x \rangle^{-|\gamma|}g^{-|\beta|}.
 \end{equation} 
 The corresponding class of 
    standard Weyl quantizations $\Opw(c)$ is
    denoted by $\Psi(m, g_{\mu,\lambda})$.

 It
is convenient to introduce the following constants:
\begin{equation}
  \label{eq:s_0}s_0=
\begin{cases} 
(1+\tfrac{\mu}{2})/2,\\
1/2,\end{cases}
\  s_1=
\begin{cases} 
1-\tfrac{\mu}{2},\\
1,\end{cases}
\  s_2=
\begin{cases} 
\mu, &\text{ for }\lambda =0,\\
0,&\text{ for }\lambda >0.
\end{cases}
\end{equation}
If $\epsilon>0$, then $\langle x\rangle^{-s_0-\epsilon}$ will be a typical weight that appears
in resolvent estimates. (Notice that in the uniform estimates of
Proposition \ref{prop:resolvent_basic2} the corresponding weight is
$g^{\frac12} \langle x\rangle^{-\frac12-\epsilon}$.) The weight 
$\langle x\rangle^{-s_1}$ plays the role of the ``Planck constant'' for
the class $\Psi(m, g_{\mu,\lambda})$.  Finally, $\langle
x\rangle^{-s_2}$ will appear in the ``elliptic regularity estimate''
of Proposition \ref{prop:enerest}. Clearly $s_0>s_2$ and $s_1>0$.

Let us decompose the normalized momentum $\xi/g$ as follows: 
\begin{equation}
   \label{eq:sym_basdef}{\xi\over g}=b{x\over \langle x\rangle}+\bar c,\;b:={x\over \langle x\rangle}
\cdot \frac \xi g\text{  and }\bar c :=\Big(I-\big |{x\over \langle x\rangle}\big\rangle \big\langle {x\over \langle x\rangle}\big|\Big)\frac \xi g.\end{equation} 
  Notice that $b$ was already defined in
 Subsection \ref{Low energy resolvent estimates1}, besides
   for $r=|x|\geq1$,  $b^2+\bar c ^2=a$ with $a$ also defined
in
 Subsection \ref{Low energy resolvent estimates1}. 
Moreover for  $r\geq1$ we have the identification 
 $b=\hat x\cdot
  {\xi\over g}\in\R$ and $\bar c=(I-|\hat x \rangle \langle \hat x|){\xi\over
    g} \in T_{\hat x}^*(S^{d-1})$ with $\hat x=x/r\in S^{d-1}$, which obviously constitute canonical coordinates for  ``the phase space''  ${\T}^*:=T^*(S^{d-1})\times \R=S^{d-1}\times \R^d$. This partly motivates the following definition:
    
  The wave front set $WF^s_{\sc} (u)$ of  a distribution $u\in L^{2,-\infty}$  is  the
subset of
 ${\T}^*$ given by the condition
\begin{align}
   \label{eq:WF^sa}
   &z_1=(\omega_1,\bar c_1,b_1) =(\omega_1,b_1\omega_1+\bar c_1)= (\omega_1,\eta_1) \notin WF^s_{\sc}(u) \nonumber\\
&\Leftrightarrow\\
   &\exists\;{\rm neighbourhoods }\;\mathcal{N}_{\omega_1}\ni
   {\omega_1},\;\mathcal{N}_{\eta_1}\ni {\eta_1}\ \ 
\forall \chi_{\omega_1}\in
   C^{\infty}_c(\mathcal{N}_{\omega_1}),\;\chi_{\eta_1}\in
   C^{\infty}_c(\mathcal{N}_{\eta_1}): \nonumber\\
&\Opw\big(\chi_{z_1}F(r>2)\big)u\in L^{2,s}\text{ where }\chi_{z_1}(x,\xi)=\chi_{\omega_1}({\hat x})\chi_{\eta_1}({b\hat x}+\bar c). \nonumber
 \end{align} Notice that this quantization is defined by the
 substitution ${b\hat x}+\bar c\to\xi /g$, cf. \eqref{eq:sym_basdef}.  Keep in  mind that the
 whole concept depends on the  given
 energy $\lambda\in[0,\infty[$ in consideration  
(through $g$, which enters in the definition of $b$ and
 $\bar c$). 

The above notion of wave front set is of course adapted to the problem
in hand. The classical definition is taylored to measure decay in momentum space; see for
example 
\cite[Chapter VIII]{Ho1}. Our definition concerns decay in position
space, and thus it  is more related to the  wave front set
 introduced in \cite[Section 7]{Me} (dubbed there as ``the 
  scattering wave front set''). 

Obviously  
\begin{equation*}
 u\in L^{2,s}  \Rightarrow WF^s_{\sc}(u)=\emptyset.
\end{equation*}
Conversely (by a compactness argument), if for some $\chi\in C^{\infty}_c(\mathbb{R}^d)$
\begin{equation}\label{eq:larxia}
  u-\Opw(\chi(\xi /g))u\in L^{2,s},
\end{equation}
then 
\begin{equation*}
WF^s_{\sc}(u)=\emptyset \Rightarrow u\in L^{2,s}.  
\end{equation*}

\begin{prop}\label{prop:enerest}
Let $\lambda \geq0$ and $s_2$ be defined in \eqref{eq:s_0}. 
Let $u\in L^{2,-\infty}$, $v\in L^{2,s+s_2}$ and
$(H-\lambda)u=v$. Then the estimates \eqref{eq:larxia} and 
\begin{equation}
  \label{eq:Chara}
WF^s_{\sc}(u)\subseteq \{z\in {\T}^*|b^2+\bar c ^2=1\} 
\end{equation} hold. 

More generally, suppose $u\in L^{2,-\infty}$, $g^{-1}v\in L^{2,s}$ and
$(H-\lambda)u=v$. Then the  following estimates hold:
\begin{subequations}
\begin{align}
  \label{eq:highes}
 &\text{For all }\epsilon>0:\;g\Opw (F(b^2+\bar c^2-1>\epsilon)
)u\in L^{2,s},\\
  \label{eq:highesb}
 &  \text{For all }\epsilon>0,\;g\Opw\left(\langle\xi/g\rangle_1^2
F(b^2+\bar c^2-1>\epsilon)\right)
u\in
  L^{2,s},\\
\label{eq:highesc}&\text{For all }\epsilon>0:\;g\Opw \big (F(1-b^2-\bar c^2 >\epsilon)\big )u\in L^{2,s},\\&
WF^s_{\sc}(gu)\subseteq \{z\in {\T}^*\ |\ b^2+\bar c ^2=1\}.\label{eq:Charab}
\end{align} 
\end{subequations}
\end{prop}
\proof Obviously  \eqref{eq:highesb} is stronger than \eqref{eq:highes}. Notice also that \eqref{eq:highes} 
in some sense is stronger than Proposition \ref{prop:resolvent_basic2}
(\ref{item:C-02}) (involves  weaker weights). It is also obvious that 
 \eqref{eq:Charab} is a consequence of \eqref{eq:highesb} and \eqref{eq:highesc}. 

The proof of \eqref{eq:highesb} given below is somewhat similar  to the
proof of the analogue of Proposition \ref{prop:resolvent_basic2}
(\ref{item:C-02}) 
given in \cite{FS}. For convenience we have  divided the proof into four steps. For the calculus of
pseudodifferential operators, used tacitly below,  we refer to 
 \cite[Theorems 18.5.4, 18.6.3, 18.6.8]{Ho1} (the reader might find it more convenient to consult \cite{FS} for an  elaboration).

 The bounds \eqref{eq:highesc} can be proved by mimicking Steps III and IV below. We note that the complication due to high energies, cf. Step II below, is absent. For this reason \eqref{eq:highesc} is somewhat easier to establish than \eqref{eq:highesb} and we shall leave the details of proof to the reader.
 
 \noindent{\bf Step I}. At various points in the proof of \eqref{eq:highesb} we need to control the possibly existing local singularities of the potential $V_3$. This is done in terms of the following elementary bounds:
 \begin{subequations}
 \begin{align}\label{eq:T1}
 T_1&:=\langle
  x\rangle^{t'} g^{-1}V_3(H-\i)^{-1}g^{-1}\langle x\rangle^{-t} \in
  \mathcal{B}(L^2),\;t,t'\in \R;\\
\widetilde T_1&:=\langle
  x\rangle^{t'} g^{-1}V_3(1+p^2)^{-1}g^{-1}\langle x\rangle^{-t} \in \mathcal{B}(L^2),\;t,t'\in \R;\label{eq:tilT1}\\
  \label{eq:T2}
 T_2&:=\langle
  x\rangle^{t} (1+p^2)g(H-\i)^{-1}g^{-1}\langle x\rangle^{-t} \in
  \mathcal{B}(L^2),\;t\in \R.
\end{align} 
\end{subequations}

\noindent{\bf Step II}. Suppose $gu\in L^{2,t}$ for some fixed $t\leq s$. We shall prove that
then $Agu\in L^{2,t}$ for all $A\in \Psi(\langle {\xi/g}\rangle_1^2,
  g_{\mu,\lambda})$, more precisely, that 
\begin{equation}
  \label{eq:redenerg}
  \text{for all } A\in \Psi(\langle {\xi/g}\rangle_1^2,
  g_{\mu,\lambda}):\;
  \|Agu\|_t\leq C (\|gu\|_t+\|g^{-1}v\|_s).
\end{equation} 

For any such an operator $A$ and any $m\in \R$, we decompose
\begin{equation}\label{eq:rei}
  \langle x\rangle^{t}A= B_{m}\langle x\rangle^{t}\Opw (\langle {\xi/
    g}\rangle_1^2) +R_{m},
\end{equation} where $B_{m}\in \Psi(1, g_{\mu,\lambda})$ and
$R_{m}\in \Psi(\langle {\xi/
    g}\rangle_1^2\langle x\rangle^{-m}, g_{\mu,\lambda})$. 

Now, cf. \cite[proof of Lemma 4.5]{FS}, 
\begin{align}\label{eq:reii}
 \Opw(\langle  {\xi} /{g}\rangle^2_1)&=g^{-1}p^2 g^{-1}+
  \Opw( a_1)\nonumber\\&=2g^{-1}(H-\lambda)g^{-1}+
  \Opw( a_2)-2g^{-2}V_3;\\ a_1&=1 -|\nabla g^{-1}|^2+4^{-1}\Delta
  g^{-2},\;  a_2=a_1+1-2g^{-2}V_2\in
  S(1,
  g_{\mu,\lambda})\nonumber.
\end{align}

We substitute (\ref{eq:reii}) in (\ref{eq:rei}), expand into
altogether four terms and apply the resulting sum to the state $gu$. The contribution from the first
term of  (\ref{eq:reii}) is estimated as
\begin{equation*}
  \|B_{m}\langle
  x\rangle^{t}2g^{-1}(H-\lambda)g^{-1}(gu)\|\leq C_1\|g^{-1}v\|_t\leq C_2\|g^{-1}v\|_s.
\end{equation*}

Similarly, the contribution from the second
term of  (\ref{eq:reii}) is estimated as
\begin{equation*}
  \|B_{m}\langle
  x\rangle^{t}\Opw( a_2)gu\|\leq C\|gu\|_t.
  \end{equation*}

As  for the third term of  (\ref{eq:reii}) we use (\ref{eq:T1}) with $t=t'$ to bound 
\begin{align*}
  2 \|B_{m} &\langle
  x\rangle^{t}g^{-2}V_3gu\|\leq
  2\|B_{m}\|\,\|T_1\langle
  x\rangle^{t}g(H-\i)u\|\\&\leq C_1\big ( \|gv\|_t+\|(\lambda-\i)gu\|_t\big)\leq
  C_2(\|gu\|_t+\|g^{-1}v\|_s).
\end{align*}

To treat the contribution from the second term of (\ref{eq:rei}) we
note that
\begin{equation*}
  \Psi(\langle {\xi/
    g}\rangle_1^2\langle x\rangle^{-m}, g_{\mu,\lambda})\subseteq \Psi(\langle {\xi}\rangle_1^2\langle x\rangle^{2-m}, g_{\mu,\lambda}).
\end{equation*} 
Whence, using (\ref{eq:T2}) and choosing $m= 2-t$,
\begin{equation*}
  \| R_{m} gu\|\leq
  C_1\|T_2\langle
  x\rangle^{t}g(H-\i)u\|\leq
  C_2(\|gu\|_t+\|g^{-1}v\|_s).
\end{equation*} 

We  conclude (\ref{eq:redenerg}).

\noindent{\bf Step III}. Suppose $gu\in L^{2,t}$ for some fixed $t<
s$. Fix $s'\in ]t,t+1-\mu/2]$ with $s'\leq s$. We shall show that
\eqref{eq:highes} holds with $s$ replaced by $s'$. We set
$F_\epsilon:=F(b^2+\bar c^2-1>\epsilon)$.

We need a regularization in $x$-space given in terms of 
$\iota_\kappa=X_\kappa^{-\tfrac{2-\mu}{2}}$, where  for  $\kappa\in ]0,1]$ we
 let
\begin{equation} 
\label{eq:aai}
  X_\kappa:=(1+\kappa|x|^2)^{1/2}.
 \end{equation} 
  Mimicking \cite[proof of Lemma
 4.5]{FS},  for $R>1$  large enough we clearly have
 \begin{equation*}
   F_\epsilon^2F(r>R)^2\leq \frac 3\epsilon \Re
   \Big(\frac {2h_2-2\lambda}{g^2}\Big )F_\epsilon^2F(r>R)^2.
 \end{equation*} Let
 \begin{align*}D&=\Opw (d),\;d=\langle {\xi/g}\rangle_1^{-1}\langle {x}\rangle^{1-s'}=\hbar ^{-1}\langle {\xi/g}\rangle_1^{-1}g^{-1}\langle {x}\rangle^{-s'};\\
   P_\kappa&=\Opw (p_\kappa),\;p_\kappa=q_\kappa^2\Big (
   \frac 6\epsilon \Re (h_2-\lambda)-g^2\Big ),\;q_\kappa=\langle
   {x}\rangle^{s'}F_\epsilon\iota_\kappa F(r>R).
 \end{align*} Since $0\leq p_\kappa \in S(\hbar ^{-2}d^{-2},
  g_{\mu,\lambda})$, 
\begin{equation*}
  D^*P_\kappa D\geq -C 
  \end{equation*} uniformly in $\kappa$. 
Since  $0<d\in S(d,
  g_{\mu,\lambda})$, we can for any $m\in \R$ find  $e_m\in S(d^{-1},
  g_{\mu,\lambda})$ such that
  \begin{equation*}
  DE_m-I\in  \Psi(\langle {x}\rangle^{-2m},
  g_{\mu,\lambda});\;E_m=\Opw (e_m). 
  \end{equation*} Consequently, we have the uniform bound 
  \begin{equation*}
    P_\kappa\geq -CE_m^*E_m+R_m,\;R_m\in \Psi(\langle {\xi/g}\rangle_1^2 g^2\langle {x}\rangle^{2s'-2m},
  g_{\mu,\lambda}),\nonumber
  \end{equation*} and therefore by
  choosing $m=s'-t$ and by using (\ref{eq:redenerg}) 
  that the expectation
  \begin{equation}
    \label{eq:loi}
    \langle P_\kappa\rangle_{u}\geq -C((\|gu\|_t+\|g^{-1}v\|_s)^2.
  \end{equation}

On the other hand, for any $\delta\in ]0,1[$
\begin{equation}
    \label{eq:loi2}
    \langle P_\kappa\rangle_{u}\leq C((\|g
    u\|_t+\|g^{-1}v\|_s)^2-(1-\delta)\langle Q_\kappa^*Q_\kappa \rangle_{gu},\;Q_\kappa=\Opw (q_\kappa).
  \end{equation} Here we use that
  \begin{align}\label{eq:loi1}
    \Opw \big (q_\kappa^2
    \Re (h_2-\lambda)\big )&= \Re\big((Q_\kappa g)^*Q_\kappa g^{-1}(H-V_3-\lambda)\big)+R_\kappa,\nonumber\\
     \;\;\;R_\kappa&\in  \Psi(\langle {\xi/g}\rangle_1^2\hbar^2 \langle
     {x}\rangle^{2s'}g^2,
  g_{\mu,\lambda}))\subseteq \Psi(\langle {\xi/g}\rangle_1^2\langle
     {x}\rangle^{2t}g^2),\nonumber 
  \end{align} and the fact that $R_\kappa$ is bounded in $\kappa\in
  ]0,1]$ in the class $\Psi(\langle {\xi/g}\rangle_1^2\langle
     {x}\rangle^{2t}g^2)$. Notice that 
\begin{align*}
    &\frac 6\epsilon \langle \Re\big((Q_\kappa g)^*Q_\kappa g^{-1}(H-\lambda)\big)\rangle_{u}\\&\leq C\|
    Q_\kappa gu\|\,\|g^{-1}v\|_{s'}\leq \delta\|Q_\kappa gu\|^2+C_\delta\|g^{-1}v\|_s^2,
  \end{align*} and that the contributions from $V_3$ and the term
  $R_\kappa$ can be treated by  (\ref{eq:T1}) and (\ref{eq:redenerg}), respectively.

Now, combining (\ref{eq:loi}) and (\ref{eq:loi2}) we conclude that
\begin{equation*}
  \|Q_\kappa gu\|^2\leq C((\|g
    u\|_t+\|g^{-1}v\|_s)^2
\end{equation*} uniformly in $\kappa\in ]0,1]$. Letting $\kappa\to 0$
completes Step III.

\noindent{\bf Step IV}.  Note that \eqref{eq:highesb} is equivalent to the
following, seemingly stronger statement:
\begin{equation}
\text{For all }\epsilon>0,\;A\in \Psi(\langle {\xi/g}\rangle_1^2,
  g_{\mu,\lambda})\ \hbox{implies}\ \;Ag\Opw (F_{\epsilon})u\in L^{2,s}.
  \label{eq:highesc1}\end{equation}
We will show \eqref{eq:highesc1} by induction.

  By assumption, $gu\in L^{2,t}$ for a sufficently small $t\leq s$ and consequently, due to  Step II, it follows that $Agu\in L^{2,t}$ for all $A\in \Psi(\langle {\xi/g}\rangle_1^2,
  g_{\mu,\lambda})$. 
 Consider for all $k\in \N$ the following claim given in terms of $t_k:=\min (s, t+ (1-\mu/2)(k-1))$:

The bound/localization \eqref{eq:highesb} holds for all $\epsilon>0$ and all $A\in \Psi(\langle {\xi/g}\rangle_1^2,
  g_{\mu,\lambda})$ provided $u\to  u_{\epsilon}:=\Opw ( F_{\epsilon/2})u$ and
  $s$ is replaced by  $t_k$. (Notice that  this implies in particular  that
  the state $gu_{2\epsilon}\in L^{2,t_{k}}$ and,
 since  $\epsilon>0$ is arbitrary, that $gu_{\epsilon}\in L^{2,t_{k}}$.)

We have seen that this claim holds for $k=1$. So suppose $k>1$ and that the claim  is true for $k\to k-1$. To show the claim for $k$, we can assume that $t_{k-1}<s$. First, we notice that $v_{\epsilon}:=(H-\lambda)u_{\epsilon}$ obeys
the condition $g^{-1} v_{\epsilon}\in L^{2,t_k}$ due to the induction hypothesis,
   (\ref{eq:redenerg}), (\ref{eq:T1}) and  (\ref{eq:tilT1}). Notice  at this point  that
\begin{equation*}
[H-V_3-\lambda,\Opw ( F_{\epsilon/2})]\in \Psi(g^2\langle {\xi/g}\rangle_1^2\hbar, g_{\mu,\lambda}),  
\end{equation*} and that in fact (for any $m\in \R$)
\begin{align*}
[H-V_3-\lambda, \Opw ( F_{\epsilon/2})]&= gA g\Opw (
F_{\epsilon/4})+R_m,\\A&\in \Psi(\langle {\xi/g}\rangle_1^2\langle
{x}\rangle^{\mu/2-1}, g_{\mu,\lambda}),\;R_m\in \Psi(\langle
{\xi/g}\rangle_1^2\langle {x}\rangle^{-m}, g_{\mu,\lambda}). 
\end{align*} 
  Now, by Step III, \eqref{eq:highes}
applies to $u\to u_{\epsilon}$, $t\to s_{k-1}$
  and with $s$ replaced by $s'=t_k$. Next, by applying  Step II  to   the
  state $u\to \tilde u_{\epsilon}:=\Opw ( F_{\epsilon})u_{\epsilon}$  (note
  that as above  $g^{-1} (H-\lambda)\tilde u_{\epsilon}\in L^{2,t_k}$),
 we conclude that indeed  the bound \eqref{eq:highesb} holds with $u\to  u_{\epsilon}$ and $s$ replaced by  $t_k$. The induction is complete. 

Finally we obtain, using the above claim,  that the bound \eqref{eq:highesb}
holds without changing $u$ and with $s$ replaced by  $t_k$. Since clearly
$t_k=s$ for $k$ sufficiently large,  \eqref{eq:highesb} follows. 
\qed

The following corollary follows immediately  from
  Proposition \ref{prop:enerest}. At a fixed energy, it
 strengthens  
Proposition \ref{prop:resolvent_basic2}  {\rm(\ref{item:C-02})}.

\begin{cor} \label{cor:renerg}
Let $\chi\in C_c^\infty(\R)$, $\chi=1$ around
1.  Then for any $s>s_0$ we have (with $\lambda \geq0$,
and $s_0$ and $s_2$  as  given in
\eqref{eq:s_0})
\begin{equation}
\|\langle
x\rangle^{s-s_2} 
\Opw\left((a^2+1)(1-\chi(a)\right)R(\lambda\pm\i0)\langle
x\rangle^{-s}\|\leq  
  C.
\end{equation}
\end{cor}

The following proposition is 
 similar to  Proposition \ref{prop:propa1} stated later, although the flavour is
   somewhat  ``global''. These results (as well as their proofs) are modifications of 
\cite[Proposition 3.5.1]{Ho3} (and its proof), see also \cite{Me} and
\cite{HMV}. The condition (\ref{eq:8ai}) is similar to
(\ref{eq:highesb}); it implies that $WF^s_{\sc}(u)\subseteq
\{(b^2+\bar c ^2\leq 1\}$ and hence that $WF^s_{\sc}(u)$ is compact.

\begin{prop}  Let $\lambda \geq0$ and $s_0$ be defined in
\eqref{eq:s_0}.
  \label{prop:propa5aa} Suppose $u,v\in L^{2,-\infty}$, $(H-\lambda)u=v$,
  $s\in \R$, $k\in ]-1,1[$ and $\{b=k\}\cap
  WF^s_{\sc}(u)=\emptyset$. Suppose the following condition:
  \begin{equation}\label{eq:8ai}
  \text{For all }\delta>0,\;\Opw\left(\langle\xi/g\rangle_1^2
F(b^2+\bar c^2-1>\delta)\right)
u\in
  L^{2,s}.\end{equation}

  Define
  \begin{align}
    k^+&=\sup \{ \tilde k \geq k|\,\{b\in [k,\tilde k]\}\cap
  WF^s_{\sc}(u)=\emptyset\},\label{eq:2aa}\\\label{eq:1aa}k^-&=\inf  \{ \tilde k \leq k|\,\{b\in [\tilde k, k]\}\cap
  WF^s_{\sc}(u)=\emptyset\}.
  \end{align}
  
  Then 
\begin{align}
  k^+<1 &\Rightarrow \{b=k^+\}\cap
  WF^{s+2s_0}_{\sc}(v)\neq \emptyset\label{eq:y2},\\\label{eq:y1} k^->-1 &\Rightarrow \{b=k^-\}\cap
  WF^{s+2s_0}_{\sc}(v)\neq \emptyset.
\end{align} 
\end{prop}
\begin{proof} We shall only deal with the case of superscript "$+$"; the case of "$-$" is similar. For convenience we shall assume that $\epsilon_2 \leq 2-\mu$ and  divide the proof into two steps.
   
\noindent{\bf Step I}. We will first show the  following weaker
statement: Suppose $u\in L^{2,s-\epsilon_2/2}$, $v\in L^{2,s+2s_0}$
and  $(H-\lambda)u=v$ (in this case \eqref{eq:8ai} follows from
Proposition \ref{prop:enerest}).  Then 
 \begin{equation}
   \label{eq:bichi} k^+\geq1.
 \end{equation}

 Suppose on the contrary that $k^+<1$. By a compactness argument
we can  then   find a point in $WF^s_\sc (u)$ of the form $z_1=(\omega_1,\bar c_1, k^+)$. For $\epsilon >0$ chosen small enough (less than $(k^+-k)/2$ suffices here)
 \begin{equation}
\label{eq:nowavefronti}
   \{b\in ]k^+-2\epsilon,k^+[\}\cap WF^s_\sc(u) =\emptyset.
\end{equation} We can assume that $J:=]k^+-2\epsilon,k^++\epsilon[\subseteq ]-1,1[$. Pick a non-positive $f \in C^{\infty}_c(J)$ with $f'\geq 0$ on  $[k^+-\epsilon,\infty[$ and $f ( k^+)<0$,  
 and consider  for $K>0$ and $\kappa \in ]0,1]$ the symbol
\begin{equation} 
\label{eq:propobsi}
  b_\kappa=X^{s_0}a_\kappa,\; a_\kappa=X^sX_{\kappa}^{-\epsilon_2/2}F(r>2)\exp({-Kb})f ( b)F(b^2+\bar c^2<3);
 \end{equation} here $X_\kappa$ is defined by \eqref{eq:aai}.

 We compute the Poisson bracket
  \begin{eqnarray}
 \nonumber   \{h_2,b\}&=&
\frac{g}{r}\bar c^2+\frac{V_1'(b^2-1)}{g}-
\frac{x\cdot \nabla V_2}{g\langle x
\rangle}\\
&=&\frac gr\Big (\big (1-rV'_1g^{-2}\big )\bar c^2+rV'_1g^{-2}\big (b^2+\bar c^2-1\big )+O\big (r^{-\epsilon_2}\big )\Big)\label{eq:eqmotionb}\\
\label{eq:eqmotionbiiii}
  &=&\frac gr\Big (\big (1-rV'_1g^{-2}\big )\big (1-b^2\big
  )+g^{-2}2(h_2-\lambda)
+O\big (r^{-\epsilon_2}\big )\Big).
  \end{eqnarray} 
  
  We expand the right hand side of \eqref{eq:eqmotionbiiii} into three terms
  and notice that due to \eqref{eq:virial} the first term has the following
  positive lower bound on $\supp b_\kappa$: 
 \begin{equation*}\cdots \geq c\frac {g}{r};\;c=\frac {\tilde \epsilon _1}{2}\big(1-\sup \{t^2|t\in \supp f\}\big ).\end{equation*}

   First we fix $K$: A part of the Poisson bracket with
 $b_\kappa^2$ is
\begin{equation}
\label{eq:Poi1i}
  \{h_2,X^{2s+2s_0}X_{\kappa}^{-\epsilon_2}\}= \frac {g}{r}Y_\kappa  bX^{2s+2s_0}X_{\kappa}^{-\epsilon_2},
\end{equation}
where $Y_\kappa=Y_\kappa(r)$ is uniformly bounded in $\kappa$.
We pick $K>0$ such that for all $\kappa$ \[2Kc\geq |Y_\kappa  |+2\frac {r}{g}X^{-2s_0}\text{ on }
\supp b_\kappa.\]

From \eqref{eq:eqmotionbiiii}, \eqref{eq:Poi1i} and the properties of $K$ and $f$,  we conclude the following bound at  $\{ f'(b)\geq0\}$: 
\begin{equation*}
  \{h_2,b_\kappa^2\}\leq -2a_\kappa^2 +g^{-2}(h_2-\lambda)a_\kappa  O\big (r^s\big )+O\big (r^{2s}\big )(F^2)'(b^2+\bar c^2<3)+O\big (r^{2s-\epsilon_2}\big
  ).
\end{equation*}
  To use this bound effectively, we introduce  a partition of unity:  Let
  $f_1,f_2 \in C^{\infty}_c(J)$ be chosen such that  $\supp f_1\subseteq
  ]k^+-2\epsilon,k^+[$, $\supp f_2\subseteq ]k^+-\epsilon,k^++\epsilon[$ and
  $f_1^2+f_2^2=1$ on  $\supp f$. We multiply both sides  by
  $f_2^2\,(=1-f_1^2)$ and obtain after a rearrangement 
  \begin{align}
\label{eq:Poi3iyy}
  \{h_2&,b_\kappa^2\}\leq -2a_\kappa^2 +g^{-2}(h_2-\lambda)a_\kappa d_{\kappa}
  \nonumber\\+&K_1f_1^2F(b^2+\bar c^2<3)\langle
  {x}\rangle^{2s}-K_2(F^2)'(b^2+\bar c^2<3)\langle {x}\rangle^{2s} +K_3\langle
  {x}\rangle^{2s-\epsilon_2}, \\ &d_{\kappa}\in S(\langle {x}\rangle^{s}, 
  g_{\mu,\lambda});\nonumber
\end{align} here $K_1,K_2,K_3>0$ are independent of $\kappa$, and the symbols $d_{\kappa}$ are bounded in $\kappa$ in the indicated class.

We introduce  $A_\kappa={\Opw}(a
_\kappa)$, $B_\kappa={\Opw}(b
_\kappa)$  and the regularization $u_R=F(|x|/R<1)u$ in terms of a parameter $R>1$. First we compute 
\begin{equation}
  \label{eq:comm1i}
 \langle i[H, B_\kappa^2]\rangle _u= \lim_{R\to\infty}\langle i[H,
 B_\kappa^2]\rangle _{u_R}=-2\Im \langle v,B_\kappa^2 u\rangle.  
\end{equation}  
 Using \eqref{eq:comm1i} and the calculus, cf.
 \cite[Theorems 18.5.4, 18.6.3, 18.6.8]{Ho1}, we estimate 
\begin{equation}
  \label{eq:comm2i}
 |\langle i[H, B_\kappa^2]\rangle _u|\leq C_1\|v\|_{s+2s_0}\big (\|A_\kappa
 u\|+\|u\|_{s-\epsilon_2/2}\big)
\leq \tfrac{1}{2} \|A_\kappa
 u\|^2 +C_2.
\end{equation}

 On the other hand, using \eqref{eq:8ai}, \eqref{eq:nowavefronti} and 
 \eqref{eq:Poi3iyy},   we infer that 
\begin{align*}
  \langle i[H-V_3, B_\kappa^2]\rangle _u&=\lim_{R\to\infty}\langle i[H-V_3, B_\kappa^2]\rangle _{u_R}\\&\leq -2\|A_\kappa u\|^2+C_3\|(H-V_3-\lambda)u\|_{s+\mu}\|A_\kappa
 u\|+C_4,
\end{align*} and whence, using \eqref{eq:T1} to bound $\|(H-V_3-\lambda)u\|_{s+\mu}\leq C\big (\|v\|_{s+\mu}+\|u\|_{s-\epsilon_2/2}\big)$, that 
\begin{equation}
  \label{eq:comm3i}
 \langle i[H-V_3, B_\kappa^2]\rangle _u\leq -\tfrac{3}{2} \|A_\kappa
 u\|^2+ C_5.
\end{equation} 

Clearly another application of \eqref{eq:T1} yields
\begin{equation}
  \label{eq:hgyi}
 \langle i[V_3, B_\kappa^2]\rangle _u\leq  C_6.
\end{equation}

Combining \eqref{eq:comm2i}--\eqref{eq:hgyi} yields 
\begin{equation*}
 \|A_\kappa u\|^2\leq C_7=C_2+C_5+C_6,
\end{equation*} 
which in combination with  the property that $f ( k^+)<0$ in
turn gives a uniform bound 
\begin{equation}
  \label{eq:comm4i}
 \|X_{\kappa}^{-\epsilon _2/2}\Opw \big (\chi_{z_1}F(r>2)\big)u\|^2_s\leq C_8;
\end{equation}
 here $\chi_{z_1}$ signifies any phase-space localization
factor
of the form entering in \eqref{eq:WF^sa} supported in a sufficiently
small neighbourhood of  the point $z_1=(\omega_1,\bar c_1, k^+)$.

We let $\kappa\to 0$ in \eqref{eq:comm4i} and infer that $z_1 \notin
WF_{\sc}(u)$, which is a contradiction; whence \eqref{eq:bichi} is proven.

\noindent{\bf Step II}. We need to remove 
 the conditions of Step I,
$u\in L^{2,s-\epsilon_2/2}$ and $v\in L^{2,s+2s_0}$. This will be accomplished  by an iteration and modification of the procedure of Step I. 

Pick $t_1\in\R$ such that $v\in L^{2,t_1}$. Pick $t<s$ such that $u\in L^{2,t}$ and define $s_m=\min(s,t+m\epsilon_2/2)$ for $m\in \N$. Let correspondingly $k^+_m$ be given by \eqref{eq:2aa} with $s\to s_m$. Clearly 
\begin{equation}
  \label{eq:eq4i}
k^+_m \leq k^+_{m-1};\;m=2,3,\dots 
\end{equation}
 
 If $u\in L^{2,s_m-\epsilon_2/2}$ and $v\in L^{2,s_m+2s_0}$ then
 \eqref{eq:y2} with $k^+\to k^+_m$  and $s\to s_m$  follows from Step
 I. Although we shall not verify these conditions we remark that a
 suitable micro-local modification  will come into play in an
 inductive procedure, see (\ref{eq:abe2}) and (\ref{eq:abe3}) below. We shall indeed (inductively) show \eqref{eq:y2} with $k^+\to k^+_m$  and $s\to s_m$, i.e. that 
 \begin{equation}
  \label{eq:mimpli} k^+_m<1 \Rightarrow \{b=k^+_m\}\cap
  WF^{s_m+2s_0}_{\sc}(v)\neq \emptyset.
\end{equation} Notice that \eqref{eq:y2} follows  by using \eqref{eq:mimpli} for an  $m$  taken so large that $s_m=s$.

Let us consider  the start of induction given by $m=1$. In this case
obviously  $u\in L^{2,s_m-\epsilon_2/2}$. Suppose on the contrary that
(\ref{eq:mimpli}) is false. Then we consider the following case:
\begin{equation}
  \label{eq:abe}
 k^+_m<1 \text{ and } \{b=k^+_m,b^2+\bar c^2\leq 6\}\cap
  WF^{s_m+2s_0}_{\sc}(v)= \emptyset.
\end{equation}
 We let $\epsilon>0$, $J$ and $f$  be chosen as in Step I with $k^+\to
 k^+_m$.   Let $\tilde f \in
 C^{\infty}_c(]k^+-3\epsilon,k^++2\epsilon[)$ with $\tilde f=1$ on
 $J$. It follows from    \eqref{eq:abe}, possibly by taking
 $\epsilon>0$  smaller than needed in  Step I,  that
 \begin{equation}
  \label{eq:abe2}
 I_\epsilon v\in L^{2,s_m+2s_0};\;I_\epsilon  =\Opw \big (\tilde f(b)F(b^2+\bar c^2<6)\big ).
\end{equation} Next, we introduce the symbol $b_\kappa$ by
\eqref{eq:propobsi} (with $s\to s_m$) and proceed as in Step I. As for
the bounds \eqref{eq:comm2i}, we can replace $v$ by $I_\epsilon v$ up
to addition of a  term of the form $C\big(\|v\|_{t_1}^2+
\|u\|_{s_m-\epsilon_2/2}^2\big)$. Similarly we can verify
\eqref{eq:comm3i} and \eqref{eq:hgyi}  (using conveniently
\eqref{eq:tilT1}). So again we obtain \eqref{eq:comm4i} (with $s\to
s_m$), and therefore a contradiction as in Step I. We have shown \eqref{eq:mimpli} for $m=1$.

Now suppose $m\geq 2$ and that \eqref{eq:mimpli} is verified for $m-1$. We need to show the statement for the given $m$. Due to \eqref{eq:eq4i} and the induction hypothesis, we can assume that 
\begin{equation} \label{eq:abejj} k^+_m <k^+_{m-1}.
\end{equation}  Again we argue by contradiction assuming
\eqref{eq:abe}. We proceed as above noticing that  it follows from \eqref{eq:abejj} that in addition to \eqref{eq:abe2} we have
\begin{equation}
  \label{eq:abe3}
 I_\epsilon u\in L^{2,s_{m-1}};
\end{equation} at this point we   possibly need choosing $\epsilon>0$
even smaller (viz. $\epsilon<(k^+_{m-1}-k^+_m)/2$). By replacing $v$ by $I_\epsilon v$ and $u$ by $I_\epsilon u$  at various points in the procedure of Step I (using \eqref{eq:abe2} and \eqref{eq:abe3}, respectively) we obtain again a contradiction. Whence \eqref{eq:mimpli} follows.
\end{proof}

\begin{cor} \label{cor:renerg2}Let $s\in\R$, $u\in L^{2,-\infty}$, $v\in L^{2,s+2s_0}$,
$(H-\lambda)u=v$,  $k\in]-1,1[$ and $\{b=k\}\cap WF_\sc^s(u)=\emptyset$. Then
\begin{equation}
WF_\sc^s(u)\subseteq\{b=1\}\cup\{b=-1\}.\end{equation}
\end{cor}

\proof The condition \eqref{eq:8ai} is guaranteed by Proposition 
\ref{prop:enerest}. Then we apply Proposition \ref{prop:propa5aa}. \qed

\subsection{Wave front set bounds  of the boundary value of the resolvent}

 Proposition
\ref{prop:resolvent_basic2} implies that the symbol $R(\lambda\pm\i0)$  in
many cases can be treated as an  operator, although initially it is defined in terms of
a quadratic form. Notice that Remark \ref{remarks:nosmmmmo}
\ref{it:csmoottaa})  in one
situation gives a slightly  different and direct interpretation of
$R(\lambda\pm\i0)$ (as a limit of operators  and hence avoiding
quadratic forms). It will however be convenient to investigate possible other
 interpretations of states $R(\lambda\pm\i0)v$ (for which in
 particular Remark \ref{remarks:nosmmmmo}
\ref{it:csmoottaa})  does not apply) and study associated wave front
set bounds. The case of $R(\lambda-\i0)$ is  similar to that of 
 $R(\lambda+\i0)$ and will not be elaborated regarding proofs.

For sufficiently decaying states $v$ we have (using in
(\ref{item:lap1ii}) the slightly abused  notation
$a:=b^2+\bar c^2$ for generic points $z=(\omega,\bar c,b)
=(\omega,b\omega+\bar c)\in \T^*)$:
\begin{prop}
 \label{prop:proposi}  Let $s>s_0$ and $v\in L^{2,s}$.
 Then the following is true:
\begin{enumerate}[\normalfont (i)]
\item\label{item:lap1}
For any $t>s_0$
\[R(\lambda\pm \i 0)v=\lim_{\epsilon\searrow0}R(\lambda\pm \i\epsilon)v\text{ exists in } L^{2,-t}.\]
\item\label{item:lap1ii}
\[WF_\sc^{s-s_2}\left(R(\lambda\pm 
\i 0)v\right)\subseteq\{a=1\}.\]
\item For any $\epsilon>0$,
\begin{equation}
WF_\sc^{s-2s_0-\epsilon}
\left(R(\lambda\pm \i 0)v\right)\subseteq\{b=\pm 1\}.\label{labb}\end{equation}
\end{enumerate}
\end{prop}


\proof
{\bf Ad (i). } This statement follows from Remark \ref{remarks:nosmmmmo}
\ref{it:csmoottaa}); notice that the notation for the limit conforms
with  Proposition
\ref{prop:resolvent_basic2} (\ref{it:one2}).
\begin {comment}
By Proposition
\ref{prop:resolvent_basic2}
{\rm (\ref{it:one2})}, if $w\in L^{2,t}$, then
\[\langle w,R(\lambda+\i 0)v\rangle\leq C\|w\|_{t}\|v\|_{s}.\]
But $L^{2,-t}=(L^{2,t})^*$ (where $(L^{2,t})^*$ denotes the space of continuous
antilinear functionals on $L^{2,t}$). Therefore, 
$R(\lambda+\i 0)v\in L^{2,-t}$.
\end {comment}

{\bf Ad (ii).}  We have $(H-\lambda)u=v$. Therefore (ii) follows
from Proposition \ref{prop:enerest} (alternatively by using   Corollary \ref{cor:renerg}).

{\bf Ad (iii).}
Let $\chi_-\in C_c^\infty(\R)$ such that $\chi_-$ is zero around 1. Let  $\chi\in
C_c^\infty(\R)$. Then by Proposition \ref{prop:resolvent_basic2} {\rm
  \eqref{some_label21}}, for any $\epsilon>0$ 
\[\Opw(\chi(a)\chi_-(b))R(\lambda+\i0)v\in L^{2,s-2s_0-\epsilon}.\]
 \qed

Based completely on Proposition \ref{prop:resolvent_basic2} one can
give a meaning to $R(\lambda\pm\i0)v$ also for  some states $v$ with a
slower decay provided they have an appropriate phase space
localization. (In the statement below $C_0\geq1$ is  given in
agreement with Proposition
\ref{prop:resolvent_basic2}  {\rm\eqref{item:C-02}}.)

\begin{prop}\label{prop:proposi1} Let $s\leq s_0$ and $v\in
  L^{2,s}$. Suppose that for some $t>s_0$
and  $ k\in ]-1,1]$ (or $ k\in [-1,1[$) 
\begin{equation}
WF_\sc^{t}(v)\cap\{b< k,\, a< 2C_0\}  =\emptyset\;\;(\text{or }WF_\sc^{t}(v)\cap\{b> k,\, a< 2C_0\} =\emptyset).\label{labb2}\end{equation}
\begin{enumerate}[\normalfont (i)]
\item \label{item:r1}For any $\epsilon>0$ there  exists
\[R(\lambda+\i0)v=\lim_{\kappa\searrow0}R(\lambda+\i0)v_\kappa\text\;\;(R(\lambda-\i0)v:=\lim_{\kappa\searrow0}R(\lambda-\i0)v_\kappa)\text{  in } L^{2,s-2s_0-\epsilon},\]
 where $v_\kappa(x):=F(\kappa|x|<1)v(x)$.
\item\label{item:r2}
\[WF_\sc^{s-s_2}\left(R(\lambda+\i 0)v\right)\subseteq\{a=1\}\;\;(WF_\sc^{s-s_2}\left(R(\lambda-
\i 0)v\right)\subseteq\{a=1\}).\]
\item\label{item:r3}
 For any $\epsilon>0$
\begin{align}
WF_\sc^{t-2s_0-\epsilon}&(R(\lambda+\i0)v)\cap\{b< k,\,a\leq C_0\}=\emptyset\label{labb6}\\&(WF_\sc^{t-2s_0-\epsilon}(R(\lambda-\i0)v)\cap\{b> k,\,a\leq C_0\}=\emptyset).\nonumber
\end{align}
\end{enumerate}
\end{prop}


\proof
{\bf Ad (i).} Let $\chi\in C_c^\infty(]-\infty,2C_0[)$,
$\chi=1$ around $[0,C_0]$. Let $\chi_-\in C^\infty(\R)$ be chosen
such that  $\chi_-=1$ around $]-\infty,-1]$
and $\chi_-=0$ in  $[(k-1)/2,\infty[$. Then by the condition
(\ref{labb2}) and the calculus of pseudodifferential operators
\begin{equation*}
 \Opw(\chi(a)\chi_-(b))v_\kappa\longrightarrow\Opw(\chi(a)\chi_-(b))v\text{
   in } L^{2,t}\text{  as  } {\kappa\searrow0}.\end{equation*}
Whence by  Proposition \ref{prop:resolvent_basic2} (\ref{it:one2}), for any $\epsilon>0$,
\[u_1:=\lim_{\kappa\searrow0}R(\lambda+\i
0)\Opw(\chi(a)\chi_-(b))v_\kappa\text{ exists in } L^{2,-s_0-\epsilon}.\]

By Proposition
\ref{prop:resolvent_basic2}  {\rm\eqref{item:C-02}} we have
\begin{equation*}
u_2:=\lim_{\kappa\searrow0} R(\lambda+\i0)\Opw(1-\chi(a))v_\kappa\text{ exists in } L^{2,s-2s_0-\epsilon}.
\end{equation*}
By Proposition 
\ref{prop:resolvent_basic2}
{\rm  \eqref{some_label21}}
 we have
\begin{equation*}
u_3:=\lim_{\kappa\searrow0} R(\lambda+\i0)\Opw(\chi(a)(1-\chi_-(b))v_\kappa\text{
  exists in } L^{2,s-2s_0-\epsilon}.
\end{equation*}
But $s-2s_0\leq-s_0$.
Hence 
\begin{equation*}
R(\lambda+\i0)v:=\lim_{\kappa\searrow0} R(\lambda+\i0)v_\kappa=u_1+u_2+u_3
\in L^{2,s-2s_0-\epsilon}.
\end{equation*}

{\bf  Ad (ii).} This statement is proven as (ii) of the previous proposition.

{\bf Ad (iii).} 
Let $\chi^1,\,\chi^2\in C_c^\infty(]-\infty,2C_0[)$, $\chi^2=1$ around 
$[0,\max(\sup\supp \chi^1,C_0)]$. Let $\chi_-^1\in C_c^\infty(]-\infty,k[)$ and $\chi_-^2\in C^\infty(\R)$ such that $\chi_-^2=1$
around $]-\infty,\sup \supp \chi_-^1]$ and $ \supp \chi_-^2\subseteq ]-\infty,k[$. Then by the condition
\eqref{labb2}
\begin{equation*}\Opw(\chi^2(a)\chi_-^2(b))v\in L^{2,t}.\end{equation*}
  Whence, by Proposition
 \ref{prop:resolvent_basic2} (\ref{it:one2}),  noting that $t>s_0$, we obtain
\begin{equation*}
R(\lambda+\i0)\Opw(\chi^2(a)\chi_-^2(b))v\in L^{2,-s_0-\epsilon}
\end{equation*}
and
\begin{equation}
WF_\sc^{t-2s_0-\epsilon}\left(R(\lambda+\i0)
\Opw(\chi^2(a)\chi_-^2(b))v\right)\subseteq\{b=1\}.\label{labb3}
\end{equation}
By  Proposition
\ref{prop:resolvent_basic2}  {\rm \eqref{it:last}},
\begin{equation}
\Opw(\chi^1(a)\chi_-^1(b))R(\lambda+\i0)
\Opw(\chi^2(a)(1-\chi_-^2(b)))v\in L^{2,\infty},\label{labb4}\end{equation}
and by 
 Proposition
\ref{prop:resolvent_basic2}  {\rm \eqref{it:last+}},
\begin{equation}
\Opw(\chi^1(a)\chi_-^1(b))R(\lambda+\i0)
\Opw(1-\chi^2(a))v\in L^{2,\infty}.
\label{labb5}\end{equation}
Now
\eqref{labb3}--\eqref{labb5} yields
\begin{equation*}
\Opw(\chi^1(a)\chi_-^1(b))R(\lambda+\i0)
v\in L^{2,t-2s_0-\epsilon},
\label{labb59}\end{equation*} 
which implies
\eqref{labb6}.
\qed

We have yet another interpretation very similar to Proposition
\ref{prop:proposi} (\ref{item:lap1}):
\begin{prop}
 \label{prop:proposibb}  Fix real-valued  $\chi\in C_c^\infty(\R)$ and $ \tilde\chi\in
   C^\infty(\R)$ such that  $\inf
   \supp \tilde\chi>-1$ (or  $\sup
   \supp \tilde\chi<1$). Let $A:=\Opw(\chi(a)\tilde\chi(b))$. 
   Suppose $v\in
   L^{2,s}$ for some $s\leq s_0$.

 For any $\epsilon>0$ there exists
 \begin{align*}
   &R(\lambda+\i0)Av=\lim_{\kappa\searrow0}R(\lambda+\i\kappa)Av\text{
   in } L^{2,s-2s_0-\epsilon}\\&
(\text{or } R(\lambda-\i0)Av=\lim_{\kappa\searrow0}R(\lambda-\i\kappa)Av\text{
   in } L^{2,s-2s_0-\epsilon}). 
 \end{align*}
 Moreover  this limit agrees with
 the interpretation of Proposition \ref{prop:proposi1}
 (\ref{item:r1}).
\end{prop}
\begin{proof} We need to invoke an extended version of the  bound
  \eqref{eq:PsDO_part221}, see  \cite [Lemma 4.10]{FS}. First notice
  that the symbols   $g$, and hence also $a$ and  $b$, obviously depend on $\lambda$. Let
  $\zeta=\lambda+\i\kappa$ and define  $g_\zeta$, $a_\zeta$ and $b_\zeta$ by
  replacing $\lambda$ by $|\zeta|$ in the definition  of $g$ in
  Section \ref{assump:conditions1} and of  $a$ and  $b$ in 
  \eqref{eq:def_a0_b}, respectively. Now we have the following extension of the
  bound \eqref{eq:PsDO_part221}: 
  
For all $\delta>\tfrac {1}{2}$ and all 
$s , t \geq  0$, there exists $C > 0$ such that for all $\kappa\in
]0,1]$ 
\begin{equation}\label{eq:bunifa}
  \|(\langle x \rangle g_\zeta)^{-s}\langle x \rangle^{-t-\delta} g_\zeta^{1\over 2}R(\zeta)
  \Opw(\chi_{-}(a_\zeta)\tilde{\chi}_{+}(b_\zeta)) g_\zeta^{1\over 2}\langle x \rangle^{t-\delta} (\langle x \rangle g_\zeta)^{s}\| \leq C.
\end{equation}

Although this will not be needed,  the bound \eqref{eq:bunifa} is in fact
locally uniform in $\lambda\geq 0$.

We pick in \eqref{eq:bunifa} the functions $\chi_{-}$ and
$\tilde{\chi}_+$  in agreement with Proposition
  \ref{prop:resolvent_basic2} (\ref{some_label21}) such that in
  addition  $\chi_{-}=1$ around $[0,\sup \supp \chi]$ and $\tilde\chi_{+}=1$
  around $[\min (0,\inf \supp \tilde\chi), \infty[$. Using the bounds $g\leq g_\zeta$,
  $a_\zeta\leq a$ and $|b_\zeta|\leq |b|$ we then obtain that for any
  $m\in \R$
  \begin{equation}
    \label{eq:opeq}
    \big(\Opw(\chi_{-}(a_\zeta)\tilde{\chi}_{+}(b_\zeta))
    -1\big)A\in \Psi(\langle x\rangle^{m}, g_{\mu,\lambda}).
  \end{equation}
By combining Remark \ref{remarks:nosmmmmo} \ref{it:csmoottaa}),
\eqref{eq:bunifa} (with $s=0$,  $t=s_0-s+\frac\epsilon2$ and $\delta=\frac 12 +\frac\epsilon2$)  and \eqref{eq:opeq} we obtain the
uniform bound:  For all $\kappa\in
]0,1]$ 
\begin{equation}
    \label{eq:opeq2}
  \|\langle x \rangle^{-t-\delta} g^{1\over 2}R(\zeta)
Ag^{1\over 2}\langle x \rangle^{t-\delta} \| \leq C.  
  \end{equation}

Obviously we obtain from \eqref{eq:opeq2} and a density argument that
indeed there exists the  limit 
\begin{equation*}
u:=\lim_{\kappa\searrow0}R(\lambda+\i\kappa)Av\text{
   in } L^{2,s-2s_0-\epsilon}.  
\end{equation*} Since $u=R(\lambda+\i0)Av$
for $v\in L^{2,\infty}$ we are done (by using density and
interchanging limits).
\end{proof}

\subsection{Sommerfeld radiation condition}
\label{Sommerfeld radiation condition}

In this subsection we describe
 a version of the Sommerfeld radiation
 condition  close in spirit to \cite[Theorem 30.2.7]{Ho2}, \cite{Is2} and
 \cite{Me}.

We introduce for  $s>0$ Besov spaces $B_s$ and corresponding   duals $B_s^*$ 
 as in \cite{AH} (see
\cite[Section 14.1]{Ho2} for details about these spaces). 
 They consist of local $L^2$ functions with a certain
(norm) expression being finite.

Throughout this subsection we shall actually  only  use the duals $B_s^*$, 
  for which we can take the
norm squared to be
\begin{equation*}
  \|u\|^2_{B_s^*}:=\sup_{R>1}R^{-2s}\int_{|x|<R}|u|^2\d x.
\end{equation*} An equivalent norm is given by the square root of the expression 
\begin{equation*}
 \int_{|x|<1}|u|^2\d x+\sup_{R>1}R^{-2s}\int_{R/2<|x|<R}|u|^2\d x.
\end{equation*} In particular we see that for all $s,s'>0$ the map $X^{s'-s}:B_s^*\to B_{s'}^*$ is bicontinuous.

The subspace $B^* _{s,0}\subseteq  B^* _{s}$ is specified by the
additional condition 
\begin{equation*}
  \lim_{R\to \infty}R^{-2s}\int _{|x|<R}|u|^2\d x=0,
\end{equation*} or equivalently, 
\begin{equation*}
  \lim_{R\to \infty}R^{-2s}\int _{R/2<|x|<R}|u|^2\d x=0.
\end{equation*}

There are inclusions 
\begin{equation}
  \label{eq:9boncloseda}
  L^{2,-s}\subseteq B^* _{s,0}\subseteq  B^* _{s}\subseteq \cap_{s'>s}L^{2,-s'}.
\end{equation}

We introduce a notion of scattering wave front set
of  a distribution $u\in L^{2,-\infty}$ relative to  the Besov space
$B^*_{s,0}$,  $s>0$, say, denoted by
$WF(B^*_{s,0},u)$. It is the complement within $\T^*$ given by
replacing $WF^{-s}_{\sc}(u)\to WF(B^*_{s,0},u)$ and $L^{2,-s}\to B^*
_{s,0}$ in (\ref{eq:WF^sa}) (here (\ref{eq:WF^sa}) is considered with
$s\to -s$). Obviously (\ref{eq:9boncloseda}) implies the inclusions 
\begin{equation}\label{eq:WF^saib´nc}
WF^{-s}_{\sc}(u)\supseteq WF(B^*_{s,0},u)\supseteq WF^{-s'}_{\sc}(u);\;s'>s.
\end{equation}

\begin{prop}
  \label{prop:propa5} Suppose $v\in
  L^{2,s_0'}$ for some $s_0'>s_0$  (here  $s_0$ is given in
  (\ref{eq:s_0})). Then the equation $(H-\lambda)u=v$ has a unique
   solution $u\in L^{2,-\infty}$ obeying one of the following conditions:
\begin{enumerate}[\normalfont (i)]
\item \label{it:o1} $WF_{\sc}^{-s_0}(u)\subseteq \{b>-1\}$, 
\item \label{it:o12}$WF(B^*_{s_0,0},u)\subseteq \{b>0\}$.  
\end {enumerate} 
This  solution is given by  $u=R(\lambda+\i0)v\in
  L^{2,-s}$ for all  $s>s_0$ and $WF^{-s_0}_{\sc}(u)\subseteq \{b=1\}$. 

Similarly, under the same condition on $v$, the equation $(H-\lambda)u=v$ has a unique
   solution $u\in L^{2,-\infty}$ obeying one of the following conditions:
\begin{enumerate}[\normalfont (i)']
\item \label{it:o1b}  $WF_{\sc}^{-s_0}(u)\subseteq
   \{b<1\}$, 
\item \label{it:o12b} $WF(B^*_{s_0,0},u)\subseteq
   \{b<0\}$;  
\end {enumerate} 
 and this  solution is given by   $u=R(\lambda-\i0)v\in
  L^{2,-s}$ for all  $s>s_0$ and $WF^{-s_0}_{\sc}(u)\subseteq \{b=-1\}$. 
\end{prop}
\begin{proof} We shall only consider the first
  mentioned cases (\ref{it:o1}) or (\ref{it:o12}) (they will be treated in parallel); the
   other cases can be  treated similarly. 
By   Proposition \ref{prop:proposi},
 the function  
   $u=\tilde u:=R(\lambda+\i 0)v$ is a solution to $(H-\lambda)u=v$ enjoying the   stated
   properties (including (\ref{it:o1}) and  (\ref{it:o12})). 
Suppose in the  sequel that $u\in
  L^{2,-t}$ for some   $t>s_0$,  $(H-\lambda)u=v$ and 
   $WF^{s_0}_{\sc}(u)\subseteq \{b>-1\}$ or $WF(B^*_{s,0},u)\subseteq \{b>0\}$.  It remains to be shown  that $u=\tilde u$. 

\noindent{\bf Step I}. We shall show that $u\in
  L^{2,-s}$ for all  $s>s_0$. By  Proposition \ref{prop:enerest},
   \begin{align}
     \label{eq:8a}
     WF^{-s_0}_{\sc}(u)&\subseteq \{b^2+\bar c ^2=1\},\\A\Opw\big (F(b^2+\bar c^2>3)\big)u&\in L^{2,-s_0}\text{
     for all }A\in \Psi(\langle {\xi/g}\rangle_1^2,
  g_{\mu,\lambda}). \label{eq:8aA}
   \end{align}
It follows from (\ref{eq:WF^saib´nc}),  
    Propositions \ref{prop:enerest} and \ref{prop:propa5aa} and a
    compactness argument that 
   \begin{equation}\label{eq:8aB} WF^{-s}_{\sc}(u)\subseteq \{b=1\}\text{  for all  }s>s_0.\end{equation} 
   
Pick a real-valued decreasing $\psi \in C^{\infty}_c([0,\infty))$
   such that $\psi(r)=1$ in a small neighbourhood of $0$ and
   $\psi'(r)=-1$ if  $1/2\leq r\leq 1$. Let
   $\psi_R(x)=\psi\big (|x|/R\big )$; $R>1$. We also introduce
   \begin{equation*}
    \delta =\max \big(t-s_0, 2t-2s_0+\mu-2 \big ),
   \end{equation*}
  and check that \[\delta +s_0'\geq t,\; s_0 +\delta/2 +1-\mu/2 \geq t \text{ and }
   s_0 +\delta/2 <  t.\] 

By
   undoing the commutator we have on one hand that
   \begin{equation}
     \label{eq:uncomm}
   \langle \i [H, X^{-\delta}\psi_R]\rangle_u =-2\,\Im \langle v,X^{-\delta}\psi_R u\rangle,
   \end{equation}
 yielding the estimate
\begin{equation}
     \label{eq:uncomm2}
   |\langle \i [H, X^{-\delta}\psi_R]\rangle_u|\leq  C_1\|v\|_{s_0'}\|\,\|u\|_{-\delta -s_0'}\leq  C_2\|v\|_{s_0'}\|\,\|u\|_{-t}=O(R^0).
   \end{equation}

On the other hand 
\begin{align*}
     \label{eq:uncomm3}
   \i [H, X^{-\delta}\psi_R]&=\Re \big (g \langle x \rangle h_{\delta,R}
   \Opw(b)\big );\\h_{\delta,R}(x)&=-\delta
   X^{-2-\delta}\psi_R(x)+ X^{-\delta}(|x|R)^{-1}\psi'\big (|x|/R\big ),
   \end{align*} yielding by  using (\ref{eq:8aA}), (\ref{eq:8aB}) and the calculus (cf. \cite[Theorems 18.5.4, 18.6.3, 18.6.8]{Ho1}) 
\begin{equation*}
     \label{eq:uncomm8}
   \langle \i [H, X^{-\delta}\psi_R]\rangle_u = \Re \langle g \langle x \rangle h_{\delta,R}
   \Opw\big (bF(b>1/2)F(b^2+\bar c^2<6)\big)\rangle_{u} +O(R^0),
   \end{equation*} 
  which in turn (by the same arguments) implies that 
\begin{equation}
     \label{eq:uncomm9}
     \langle \i [H, X^{-\delta}\psi_R]\rangle_u \leq -\delta
     4^{-1}\langle g \langle x \rangle X^{-2-\delta}\psi_R \rangle_{u} +O(R^0).
   \end{equation} 

By combining  \eqref{eq:uncomm2} and \eqref{eq:uncomm9} we obtain 
\begin{equation}
     \label{eq:uncomm10}
     \langle g\langle x \rangle X^{-2-\delta}\psi_R \rangle_u \leq C,
   \end{equation} for some constant $C$ which is independent of $R>1$. Whence,
   letting $R\to \infty$ we see that $u\in L^{2,
   -t_{1}};\;t_{1}:=s_0+\delta/2$.  

More generally, we define for $k\in\N$ \[t_k=s_0+2^{-1}\max \big(t_{k-1}-s_0,
  2t_{k-1}-2s_0+\mu-2 \big ),\;t_0:=t,\]and  iterate the above procedure. We
  conclude that  $u\in 
  L^{2,-t_k}$, and hence  that indeed $u\in
  L^{2,-s}$ for all  $s>s_0$.

\noindent{\bf Step II}.
Due to Step I, it suffices to show that $u=0$
  is the only solution to the equation $(H-\lambda)u=0$ subject to the
  conditions $u\in
  L^{2,-s}$ for all  $s> s_0$ and  either $WF^{-s_0}_{\sc}(u)\subseteq
  \{b>-1\}$ or $WF(B^*_{s_0,0},u)\subseteq \{b>0\}$. In the following Steps
  III and IV we consider this problem.

\noindent{\bf Step III}. We shall show that $u\in {B^*_{s_0,0}}$. Under Condition  (\ref{it:o1}) the bound (\ref{eq:8aB}) holds for $s=s_0$ (by Proposition \ref{prop:propa5aa}) which implies that \[\text{There
exists }\epsilon>0\ \text{such that}\  WF(B^*_{s,0},u)\subseteq
\{b>\epsilon\}.\] 
 Under Condition  (\ref{it:o12}), we have the same conclusion due to (\ref{eq:8a}) and a compactness argument.
Next, we apply
the same scheme as in  Step I, now with $\delta=0$ and using a factor of 
$F(b>\epsilon)$ instead of a factor of  $F(b>1/2)$. This leads to 
\begin{equation*}
     R^{-1}\langle g \langle x \rangle |x|^{-1}\psi'\big (|\cdot |/R^{-1}\big )\rangle_u =o(R^0),
   \end{equation*} and hence 
   $u\in B^*_{s_0,0}$.

\noindent{\bf Step IV}. We shall show that $u=0$. For convenience we assume that $\epsilon_2 \leq 2-\mu$.  First,  letting   $s\in ]s_0-\epsilon_2/2,s_0[$  be given arbitrarily,  our goal  is to 
show that $u\in L^{2,-s}$. For that consider for  $\kappa\in ]0,1/2]$
\begin{equation} 
\label{eq:propobs8}
  b_\kappa=X^{s_0}a_\kappa;\; a_\kappa=\Big (\tfrac {X}{X_\kappa}\Big )^{-s}X_{\kappa}^{-s_0}F(-b>1/2)F(b^2+\bar c^2<3).
\end{equation} Here we use the regularization factor of \eqref{eq:aai}.
We calculate the Poisson bracket
\begin{equation*}
   \Big\{h_2, \Big(\tfrac {X}{X_\kappa}\Big )^{2s_0-2s} \Big\}=(1-\kappa)
   (2s_0-2s)\langle x\rangle X^{-1}X_\kappa^{-3}\Big (\tfrac {X}{X_\kappa}\Big )^{2s_0-2s-1}gb. 
\end{equation*} Obviously this is  negative  on the support of
   $b_\kappa$,  with the (uniform) upper bounds
   \begin{align*}\cdots &\leq -8^{-1}(2s_0-2s)\Big (\langle x\rangle X^{2s_0-2}g\Big )X_\kappa^{-2}\Big (\big (\tfrac {X}{X_\kappa}\big )^{-s}X_{\kappa}^{-s_0}\Big)^2 \\&\leq -cX_\kappa^{-2}\Big (\big (\tfrac {X}{X_\kappa}\big )^{-s}X_{\kappa}^{-s_0}\Big)^2,\;c>0.\end{align*} 

Similarly, by  \eqref{eq:eqmotionb}, 
\begin{align*}
   &\big\{h_2, F^2(-b>1/2)\big\}\\&=-{g\over r}(F^2)'(-b>1/2)\Big (\big (1-rV_1'g^{-2}\big )\bar c^2
   +\big (rV_1'g^{-2}\big )g^{-2}2(h_2-\lambda)+ O\big (r^{-\epsilon_2}\big )\Big ), 
\end{align*} where we expand the right hand side into a sum of three terms
and note  that the  first term is non-positive.

We introduce the quantizations $A_\kappa={\Opw}(a
_\kappa)$ and $B_\kappa={\Opw}(b
_\kappa)$, and  the states $u_R(x)=\psi_R(x)u(x)$, $R>1$. By Step III,
\begin{equation}\label{eq:0bn}
  \lim _{R\to \infty} \langle \i [H, B_\kappa^2]\rangle_{u_R}=0.
\end{equation}
On the other hand, due to the above considerations the expectation of
$\i [H, B_\kappa^2]$ in $u_R$ tends to be negative. Keeping the  precise upper bounds in mind,  we can let 
$R\to \infty$ (using  the calculus, \eqref{eq:T1}, to deal with a contribution from $V_3$ and \eqref{eq:0bn}) obtaining 
\begin{equation*}
  c\|X_\kappa^{-1}A_\kappa u\|^2\;\big (=\lim_{R\to \infty} c\|X_\kappa^{-1}A_\kappa u_R\|^2\big )\leq C,
\end{equation*}  where the constants $c$ (the one given above) and $C$ are positive and independent of
$\kappa$. Whence, letting $\kappa\to 0$, we conclude that 
\begin{equation}\label{eq:bor1}
  \Opw\Big( F(-b>1/2)F(b^2+\bar c^2<3)\Big)u\in L^{2,
  -s}.
\end{equation}

Upon replacing the factor $F(-b>1/2)$ in \eqref{eq:propobs8} by
$F(b>1/2)$, we can argue similarly and obtain
\begin{equation}\label{eq:bor2}
  \Opw\Big( F(b>1/2)F(b^2+\bar c^2<3)\Big)u\in L^{2,
  -s}.
\end{equation} 

In combination with Proposition \ref{prop:propa5aa}, the bounds
\eqref{eq:bor1} and \eqref{eq:bor2} and the fact that (\ref{eq:8a}) holds with $s_0$
replaced by $s$ (note this is trivial since,  by assumption, now $v=0$) yield that $u\in L^{2,
  -s}$.

Next, the above procedure can be iterated: Assuming that $u\in L^{2,
  -s}$ for all $s>t_k:= s_0-k\epsilon_2/2$ (for
  some 
  $k\in \N$), the  procedure leads to $u\in L^{2,-s}$ for all $s>t_{k+1}$.
  Consequently, 
  $u\in L^{2,s}$ for all $s\in \R$. In particular $u\in L^{2}$, and
  therefore $u=0$.
 \end{proof}

\section{Fourier integral operators}

\label{FIO}
In this section we construct and study  certain modifiers in the form
of Fourier integral operators; they will enter
in the construction of wave operators in Section \ref{Wave matrices}. 

\subsection{The WKB ansatz}\label{The first WKB ansatz}

Assume first that Condition \ref{symbol} holds. Fix
$\sigma_0\in]0,2[$. Recall 
from Lemma \ref{posi}
that there exists a decreasing function $]0,\infty[\ni\lambda\mapsto
R_0(\lambda) $ such that on the set
\[\left\{(x,\xi)\in \R^d\times (\R^d\setminus\{0\})\ |\ x\in\Gamma_{R_0(|\xi|^2/2),\sigma_0}^+(\hat\xi)\right\}\]
we can construct a solution $\phi^+$ of the eikonal equation satisfying the
(non-uniform in energy)  bounds \eqref{eq:postca}.

We fix $0<\sigma<\sigma'<\sigma_0$.
 Next we 
introduce  smoothed out characteristic functions

\begin{equation}\label{eq:chi^1}
  \chi_1(r)=
\begin{cases} 1, & \text{for}\; r\geq 2,\\
0, &\text{for}\; r\leq 1,
\end {cases}\; 
\end{equation}
and 
\begin{equation}\label{eq:chi^2}
  \chi_2(l)=
\begin {cases} 1, & \text{for}\; l\geq 1-\sigma,\\
0, & \text{for}\; l\leq 1-\sigma'.
\end {cases} \;
\end{equation}

Define
\[a_0^+(x,\xi):=\chi_2(\hat
x\cdot\hat \xi)\chi_1\left(|x|/R_0(|\xi|^2/2)\right).\] 
The basic idea of Isozaki-Kitada is to use the modifier
given by a
 Fourier integral operator $J_0^+$ on $L^2(\R^d)$ of the form
\begin{equation}
  \label{eq:int_ope9}
  (J_0^+f)(x)=(2\pi)^{-d/2}\int \e^{\i
    \phi^+(x,\xi)} a_0^+(x,\xi) \hat f(\xi)
  \d \xi,
\end{equation}
where
\[\hat f(\xi):=(2\pi)^{-d/2}\int \e^{-\i x\cdot \xi}f(x)\d x\]
denotes the (unitary) Fourier transform of $f$.

If we assume that the potentials satisfy Conditions
\ref{assump:conditions1} and  \ref{assump:conditions2}, then we can assume
that the function $R_0(\lambda)$ is the  constant $R_0$ given by Lemma 
\ref{lemma:mixed_2}. 
Thus in this case  the solution $\phi^+(x,\omega,\lambda)$ 
of the eikonal equation 
is defined in  $\Gamma^+_{R_0,\sigma_0}\times[0,\infty[$ (here
$\sigma_0$ is also given by Lemma 
\ref{lemma:mixed_2}; possibly it is much smaller than $2$), and  
 the amplitude $a_0$ is simply given by
\[a_0^+(x,\xi):=\chi_2(\hat x\cdot\hat \xi)\chi_1(|x|/R_0).\]

\subsection{The improved WKB ansatz}
\label{The WKB ansatz}

The modifier $J_0^+$ (and its incoming counterpart, say $J_0^-$) is sufficient only for the most basic purposes, such as the
existence of the outgoing  (incoming) wave operator. To study finer properties of wave operators it is
useful to use a more refined construction suggested by the WKB method.

This more refined construction is possible and useful
already under  Condition \ref{symbol}. However,
for simplicity of presentation, in the remaining part of the section we will
assume  that the potentials satisfy the more restrictive Conditions
\ref{assump:conditions1} and  \ref{assump:conditions2}.
These conditions allow us to extend this and related constructions (see
Subsection \ref{Fourier integral operators at fixed energies}) down to (and
including) 
$\lambda=0$. Therefore, it will be
 convenient to switch between the two notations 
$\phi^+(x,\xi)$ and 
$\phi^+(x,\omega,\lambda)$. This will be done tacitly in the
following, and in fact, we shall often 
slightly abuse notation by writing   $(x,\xi)\in \Gamma^+_{R_0,\sigma_0}$
instead of $(x,\omega,\lambda)\in \Gamma^+_{R_0,\sigma_0}\times[0,\infty[$.

The WKB method suggests to approximate the wave operator by 
a Fourier integral operator $J^+$ on $L^2(\R^d)$ of the form
\begin{equation}
  \label{eq:int_ope}
  (J^+f)(x)=(2\pi)^{-d/2}\int \e^{\i
    \phi^+(x,\xi)} a^+(x,\xi) \hat f(\xi)
  \d \xi,
\end{equation}
where the symbol $a^+(x,\xi)$ is
 supported  in  $\Gamma^+_{R_0,\sigma_0}$ and
 constructed  by an iterative procedure to make the difference
 $T^+:=\i (HJ^+-J^+H_0)$  small in an outgoing
region $\Gamma^+_{R,\sigma}$ for some $R>R_0$, $\sigma<\sigma_0$. We have 
\begin{equation}
  \label{eq:int_ope4}
  (T^+f)(x)=(2\pi)^{-d/2}\int \e^{\i
    \phi^+(x,\xi)} t^+(x,\xi) \hat f(\xi)
  \d \xi,
\end{equation} where
\begin{equation}
  \label{eq:sym_dif}
t^+(x,\xi)= \left((\nabla_x\phi^+(x,\xi)) \cdot 
\nabla_x +\tfrac {1}{2 }(\triangle_x\phi^+(x,\xi))\right)a^+(x,\xi)
 -\tfrac {\i}{2 }\triangle_x a^+(x,\xi). 
\end{equation}

As it is well-known from the WKB method, it is possible to improve on
the ansatz by putting (here we need $\xi\neq0$)
\begin{eqnarray}\label{eq:rega1}
a^+(x,\xi)&:=&\left(\det\nabla_\xi\nabla_x
\phi^+(x,\xi)\right)^{1/2} b^+(x,\xi),\\
t^+(x,\xi)&:=&\left(\det\nabla_\xi\nabla_x
\phi^+(x,\xi)\right)^{1/2} r^+(x,\xi).\label{eq:rega2}\end{eqnarray}
We have
\begin{eqnarray*}
\left((\nabla_x\phi^+(x,\xi)) \cdot 
\nabla_x 
+\tfrac {1}{2 }(\triangle_x\phi^+(x,\xi)\right)
\left(\det\nabla_\xi\nabla_x
\phi^+(x,\xi)\right)^{1/2}=0,\end{eqnarray*}
and therefore
\begin{eqnarray*}
 r^+(x,\xi)&=&
(\nabla_x\phi^+(x,\xi)) \cdot 
\nabla_x b^+(x,\xi)\\&&
 -
\left(\det\nabla_\xi\nabla_x
\phi^+(x,\xi)\right)^{-1/2}\tfrac {\i}{2 }\triangle_x 
\left(\det\nabla_\xi\nabla_x
\phi^+(x,\xi)\right)^{1/2} b^+(x,\xi).\end{eqnarray*}


It is useful to introduce 
\begin{eqnarray}\label{eq:4}
\zeta^+(x,\xi)&=&\ln \left(\det\nabla_\xi\nabla_x \phi(x,\xi)\right)^{1/2};\;\xi\neq0.\end{eqnarray}
Note that it satisfies the equation
\begin{eqnarray}
(\nabla_x\phi(x,\xi))\cdot \nabla_x\zeta^+(z,\xi)+\frac12\triangle_x\phi(x,\xi)=0.\label{equa}
\end{eqnarray}

\begin{prop}
For $(x,\xi) \in
 \Gamma_{R,\sigma}^+$, $\xi\neq0$,
\begin{align}
  \label{eq:psi_2}
 &\zeta^+ (x,\xi)=\tfrac {1}{2
 }\int _1 ^{\infty}
 \triangle_y\phi^+(y^+(t;x,\xi),\xi)\d t
 .\end{align} 
\end{prop}

\proof Both  $\zeta^+(x,\xi)$  and  the right
 hand side of (\ref{eq:psi_2}) 
 satisfy the first order  equation (\ref{equa}). Both go to zero as
 $|x|\to\infty$. In particular, they go to zero along the
 characteristics $t\to y^+(t,x,\xi)$. Therefore, they coincide. \qed

\begin{lemma}\label{lemmo}
There exist the uniform limits
\begin{eqnarray*}
\lim_{\lambda\searrow 0}\partial^\delta_\omega\partial^\gamma_x\left(\zeta^+(x,\xi)-
\zeta_\sph^+(x,\xi)\right).
\end{eqnarray*}
Besides, we have uniform estimates with
$\breve\epsilon$ given as in Proposition  \ref{prop:mixed_2aaaa} 
\begin{eqnarray*}
\partial_\omega^\delta\partial_x^\gamma
\left(\zeta^+(x,\xi)-
\zeta_\sph^+(x,\xi)\right)=O(|x|^{-|\gamma|-\breve\epsilon}),\
|\delta|+|\gamma|\geq0.
\end{eqnarray*}
\end{lemma}

\proof Below $\Div $ and $\nabla$ will always involve the derivatives
with respect to the first argument.
\begin{eqnarray*}
\zeta^+(x,\xi)-\zeta_\sph^+(x,\xi)&=&
\int_1^\infty\frac{1}{2}\left(
\Div F^+(y^+(t),\xi)-\Div F_\sph^+(y_\sph^+(t),\xi)\right)\d t \nonumber
\\
&=&\int_1^\infty \d t
 \frac{1}{2}\int_0^1 \nabla \Div F^+( y_l^+(t),\xi)\cdot \left(y^+(t)-y_\sph^+(t)\right)\d l\nonumber \\
&+&\int_1^\infty\frac{1}{2} \left(\Div F^+(y_\sph^+(t),\xi)-\Div F_\sph^+(y_\sph^+(t),\xi)
\right)\d t
\\
&=&I+II,\label{eq:2}\end{eqnarray*}
where $ y_l^+(t)= l y^+(t) +(1-l)y_\sph^+(t)$. 

Now $I$ can be estimated (cf. (\ref{eq:derF222}) and
\cite[(6.43)]{DS1}) by
\begin{eqnarray}\label{eq:1}
C_1\int_1^\infty |y^+|^{-2} g(|y^+|) t^{\alpha-\epsilon}\d t& \leq 
C_2\int_{|x|}^\infty |y^+|^{-2} |y^+|^{(\alpha-\epsilon)/\alpha}\d |y^+|\nonumber\\
& =O(|x|^{-\epsilon/\alpha})=O(|x|^{-\breve\epsilon}).
\end{eqnarray}  Here 
$\alpha =2/(2+\mu)$ and $\epsilon>0$ is specified in \cite [Subsection
6.1]{DS1}.
We used that \[\frac{\d |y^+|}{\d t}\geq c g(|y^+|),\ \ \ |y^+|\geq
ct^\alpha,\ \ c>0.\]
Splitting the time-integral as $\int_1^{T_0}\d t+
\int_{T_0}^{\infty}\d t$, the argument above yields  
(uniform) smallness of the
second term (provided $T_0$ is chosen big). As for the contribution
from the first term,  we can apply the dominated convergence theorem;
whence  we obtain the existence of
$\lim_{\lambda\searrow0}I$. 

Next  $\partial_\omega^\delta\partial_x^\gamma I$ 
is a sum 
 integrals of terms of the following form:
\begin{eqnarray*}
\partial_\omega^{\delta_1}\partial_x^{\gamma_1}y_l^+
\cdots 
\partial_\omega^{\delta_n}\partial_x^{\gamma_n} y_l^+
\partial_{y^+_l}^n\partial_\omega^\nu \nabla\Div F(y^+_l,\xi)\cdot 
\partial_\omega^\alpha\partial_x^\beta 
\left(y^+(t)-y_\sph^+(t)\right).\end{eqnarray*}
where 
$\delta_1+\cdots+\delta_n+\nu+\alpha=\delta$ and
$\gamma_1+\cdots+\gamma_n+\beta=\gamma$. 
This can be estimated (cf. (\ref{eq:derF222}) and
\cite[(4.41) and  (6.43)]{DS1}) by
\[C|x|^{-|\gamma|}|y^+|^{-2}g(|y^+|)t^{\alpha-\epsilon}.\]
We argue as above to obtain uniform bounds 
on  $\partial_\omega^\delta\partial_x^\gamma I$,
 as well as the existence of
$\lim_{\lambda\searrow0}\partial_\omega^\delta\partial_x^\gamma I$.

Now 
$II$ is bounded (cf. (\ref{eq:derF2222})) by 
\begin{eqnarray}\label{eq:3}
C_1\int_1^\infty|y^+|^{-1-\breve\epsilon} g(|y^+|)\d t\leq C_2
\int_{|x|}^\infty
|y^+|^{-1-\breve\epsilon}\d|y^+|=O(|x|^{-\breve \epsilon}).\end{eqnarray}
Then we apply the dominated convergence theorem as above,  and we obtain the
existence of 
$\lim_{\lambda\searrow0}II$.

 $\partial_\omega^\delta\partial_x^\gamma II$  is a sum of
 integrals  of terms of the
 form
\begin{eqnarray*}
\partial_\omega^{\delta_1}\partial_x^{\gamma_1} y^+
\cdots 
\partial_\omega^{\delta_n}\partial_x^{\gamma_n} y^+
\partial_{ y^+}^n\partial_\omega^\nu \left(
\Div F^+(y^+,\xi) - \Div F_\sph^+(y^+,\xi)\right)
,\end{eqnarray*}
where 
$\delta_1+\cdots+\delta_n+\nu=\delta$ and
$\gamma_1+\cdots+\gamma_n=\gamma$. 
This can be estimated (cf. (\ref{eq:derF2222}) and
\cite[(4.41) and  (6.43)]{DS1}) by
\[C|x|^{-|\gamma|}|y^+|^{-1-\breve\epsilon}g(|y^+|).\]
Then we can argue as above. 
\qed

Define
\[\tilde
\zeta^+(x,\omega,\lambda):=\zeta^+(x,\sqrt{2\lambda}\omega)-\ln(2\lambda)^{(2-d)/4}.
\]

\begin{prop}\label{prop:booh}
\begin{subequations}
\label{group98} 
\begin{enumerate}[\normalfont (i)]
\item 
There exist  (uniform) estimates
\begin{eqnarray}
  |\tilde\zeta^+ (x,\omega,\lambda)-\ln  g(|x|)^{(d-2)/2})|&\leq& 
C,\label{eq:psi_bound_der.1}\\
  \label{eq:psi_bound_der}\partial _\omega^\delta\partial _x^\gamma
  \tilde \zeta^+ (x,\omega,\lambda)&=& 
O\big (
   |x|^{-|\gamma|}\big ),\ \ \text { for }|\delta|+|\gamma|\geq 1.
\end{eqnarray}
\item 
There exist  (uniform)   estimates
\begin{eqnarray}  \label{eq:psi_bound_der1}&&
(2\lambda)^{(d-2)/4}
\partial _\omega^\delta\partial _x^\gamma
  \left(\det\nabla_\xi\nabla_x
\phi^+(x,\xi)\right)^{1/2},\nonumber\\
&=& g(|x|)^{(d-2)/2}
O\big (
   |x|^{-|\gamma|}\big ),\ \ \text { for }|\delta|+|\gamma|\geq 0.
\end{eqnarray}
\item
There exist the locally uniform limits
\[\partial_\omega^\delta\partial_x^\gamma \tilde\zeta^+ (x,\omega,0):=\lim_{\lambda\searrow0}
\partial_\omega^\delta\partial_x^\gamma \tilde \zeta^+(x,\omega,\lambda).\]
\end{enumerate}
\end{subequations}
\end{prop}

\proof Let us first prove the estimates
 (\ref{eq:psi_bound_der}) for $|\delta|=0$, $|\gamma|\geq1$ in the
 spherically symmetric case. 
$\partial_x^\gamma\zeta_\sph^+(x,\xi)$ is an integral of terms
 of the form
\[\partial_x^{\gamma_1} y\cdots 
\partial_x^{\gamma_n} y\partial_y^n\Div F_\sph^+(y^+(t),\xi),\]
where $\gamma_1+\cdots+\gamma_n=\gamma$. 
Using $\partial_x^\gamma
 y^+=O\big (|x|^{1-|\gamma|}g(|x|)g(|y^+|)^{-1}\big )$, cf. \cite [Proposition 4.9]{DS1}, these integrals are bounded
 by
\begin{eqnarray*}
&&C_1\int_1^\infty |x|^{-|\gamma|+n }g(|x|)^n g(|y^+|)^{-n+1}|y^+|^{-n-1}\d t\\
&\leq&
C_2\int_{|x|}^\infty |x|^{-|\gamma|+n} g(|x|)^n g(|y^+|)^{-n}|y^+|^{-n-1}\d |y^+|=
  O(|x|^{-|\gamma|}). \end{eqnarray*}
Thus
\begin{equation*}
 \partial _x^\gamma \zeta_\sph^+ (x,\xi)=
O\big (
   |x|^{-|\gamma|}\big )\ \text { for }|\gamma|\geq 1.
\end{equation*}
Clearly we can argue as above for $|\delta|>0$ as well. If
$|{\gamma}|=0$, we can use the formula (valid due to spherical symmetry)
\begin{equation*}
 \zeta_\sph^+ (x,R_\eta\xi)= \zeta_\sph^+ (R_\eta^{-1}x,\xi),
\end{equation*} for any $d$-dimensional rotation $R_\eta$. Clearly
this converts $\omega$-derivatives to $x$-derivatives, and
consequently we have shown (\ref{eq:psi_bound_der}) in the general
case.

Taking into account Lemma
\ref{lemmo} we obtain the estimates
 (\ref{eq:psi_bound_der}) in the general case (when $V$ is not necessarily
radial).

We have
\begin{eqnarray*}
\tilde\zeta_\sph^+(x,\xi)=\tilde\zeta_\sph^+(x,\sqrt{2\lambda}\hat x)
+\int_0^\theta\nabla_\omega
\tilde\zeta_\sph^+(x,\sqrt{2\lambda}\omega(l))\cdot \omega^\perp(l)\d  l,
\end{eqnarray*}
where $[0,\theta]\ni l\mapsto\omega(l)$ is the arc joining $\hat x$ and
$\omega$ and $\omega^\perp(l)$ is the tangent vector. Using (\ref{esto})
and 
(\ref{eq:psi_bound_der}) with $|\delta|=1$, $|\gamma|=0$ and   Lemma
\ref{lemmo} we obtain 
(\ref{eq:psi_bound_der.1}).

The above arguments in conjunction with the proof of  Lemma
\ref{lemmo} can be used to prove that there exist the limits
\[\lim_{\lambda\searrow0}\partial _\omega^\delta\partial _x^\gamma
 \tilde \zeta^+
 (x,\omega,\lambda), \ \ \ |\delta|+|\gamma|\geq1.\]

 We know from the explicit formula (\ref{esto22}) that
 $\lim_{\lambda\searrow0}\tilde\zeta_\sph^+ 
 (x,\hat x,\lambda)$ exists locally uniformly in $x$. Hence so does
 $\lim_{\lambda\searrow0}\tilde\zeta^+
 (x,\omega,\lambda)$ locally uniformly in $(x,\omega)\in\Gamma^+$.

As for the bounds (\ref{eq:psi_bound_der1}), we use (\ref{eq:psi_bound_der.1})
and 
(\ref{eq:psi_bound_der}). 
\qed

\subsection{Solving transport equations}
\label{Solving transport equations}

Introduce the operator
\begin{eqnarray*}
M&=&\left(\det\nabla_\xi\nabla_x
\phi^+(x,\xi)\right)^{-1/2}\tfrac {\i}{2 }\triangle_x 
\left(\det\nabla_\xi\nabla_x
\phi^+(x,\xi)\right)^{1/2}\\
&=&\e^{-\tilde\zeta^+(x,\xi)}\tfrac {\i}{2 }\triangle_x 
\e^{\tilde\zeta^+(x,\xi)}\\
&=&\tfrac {\i}{2 }\left(
\triangle_x +2\nabla_x\zeta^+(x,\xi)\cdot \nabla_x+\triangle_x
\zeta^+(x,\xi)+\nabla_x\zeta^+(x,\xi)^2\right). 
\end{eqnarray*} Notice that due to Proposition \ref{prop:booh} this
operator  is well-defined at $\lambda=0$ (more
precisely, for $(x,\omega,\lambda)\in \Gamma^+_{R_0,\sigma_0}\times\{0\}$).

We define inductively for $(x,\xi)\in\Gamma_{R,\sigma_0}^+$:
\begin{eqnarray*}
b_0^+(x,\xi)&:=&1;\\
b_{m+1}^+(x,\xi)&:=&\int_1^\infty Mb_m^+(y(t,x,\xi,t),\xi)\d
t.\end{eqnarray*}

\begin{prop}
There exist   the  following (uniform)   estimates:\begin{subequations}
\begin{eqnarray}\label{eq:5}
\partial _\omega^\delta\partial _x^\gamma b_m^+(x,\xi)&=
&O(|x|^{-m(1-\mu/2)-|\gamma|}) ,\\ \label{eq:6}
\partial _\omega^\delta\partial _x^\gamma Mb_m^+(x,\xi)&=
&O(|x|^{-2-m(1-\mu/2)-|\gamma|}).
\end{eqnarray}
\end{subequations}
\end{prop}

\proof For a given $m$, (\ref{eq:5}) easily implies  (\ref{eq:6}).

Integrating
$\partial _\omega^\delta\partial _x^\gamma Mb_m(x,\xi)$ we
can bound $\partial _\omega^\delta\partial _x^\gamma b_{m+1}(x,\xi)$
 by
\begin{eqnarray*}
&&\int_1^\infty |y^+|^{-2-m(1-\mu/2)-|\gamma|}\d t\\
&\leq& C_1\int_{|x|}^\infty |y^+|^{-2-m(1-\mu/2)-|\gamma|}g(|y^+|)^{-1}\d |y^+|\\
&\leq &C_2\int_{|x|}^\infty |y^+|^{-2-m(1-\mu/2)-|\gamma|+\mu/2}\d |y^+|
= O(|x|^{-(m+1)(1-\mu/2)-|\gamma|})
.\end{eqnarray*}
This shows the induction step. \qed

   
We set
\begin{eqnarray*}
b^+(x,\xi)&:=&\chi_2(\hat x\cdot\omega)\breve b^+(x,\xi),\ \ \ \breve b^+(x,\xi)=\sum_{m=0}^\infty
b_m^+(x,\xi)\chi_1(|x|/R_m)\end{eqnarray*}
for an appropriately chosen sequence $R_m\to\infty$ (this is an
example of the so-called
 Borel construction, cf. \cite [Proposition 18.1.3]
{Ho2}). There are (uniform) bounds
\begin{equation*}
  \partial _\omega^\delta\partial _x^\gamma b^+(x,\xi)
=O(|x|^{-|\gamma|}).
\end{equation*}


We introduce 
\begin{eqnarray}
r^+(x,\xi)&=&
\left(\nabla_x\phi^+(x,\xi)\cdot \nabla_x+M\right)b^+(x,\xi)\nonumber,\\
r_\pr^+(x,\xi)&=&
\chi_2(\hat x\cdot\omega)
\left(\nabla_x\phi^+(x,\xi)\cdot \nabla_x+M\right)\breve b^+(x,\xi),\nonumber\\
r_\bd^+(x,\xi)&=&r^+(x,\xi)-r_\pr^+(x,\xi).\label{eq:r2sum
}
\end{eqnarray}
(The subscript $\pr$ stands for the {\em propagation} and $\bd$ stands
for the {\em boundary}).

\begin{prop} \label{prop:brrrr}There exist (uniform) bounds 
\begin{eqnarray*}
\partial_\omega^\delta\partial_x^\gamma r_\pr^+(x,\xi) =O(|x|^{-\infty})
,\end{eqnarray*}
 and $r_\bd^+(x,\xi)$ is supported away from $\Gamma_{R_0,\sigma}^+$ and
\begin{eqnarray*}
\partial_\omega^\delta\partial_x^\gamma r_\bd^+(x,\xi) =O(g(|x|)|x|^{-1-|\gamma|}).
\end{eqnarray*}
\end{prop}

\subsection{Constructions in incoming region}
\label{Constructions in incoming region}

Using the phase function $\phi^-=\phi^-(x,\omega,\lambda)$ given in
\eqref{eq:rel_phi+-} we can construct a
symbol $a^-=\e^{\zeta^-} b^-$ with   
$t^-=\e^{\zeta^-}( r_\pr^-+r_\bd^-)$,
$r_\pr^-
=O(|x|^{-\infty})$ and
the  symbol
 $r_\bd^-=O\big (g(|x|)|x|^{-1}\big )$ vanishing  on a given $
\Gamma^-_{R,\sigma}  \subseteq
 \Gamma^-_{R_0,\sigma_0}$
 and obeying appropriate analogues of the conditions of
 the previous subsection.

Similar to \eqref{eq:int_ope} we consider the Fourier integral
operator $J^-$ on $L^2(\R^d)$ given by 
\begin{equation}
  \label{eq:int_ope2}
  (J^-f)(x)=(2\pi)^{-d/2}\int \e^{\i\phi^-(x,\xi)}a^-(x,\xi) \hat
  f(\xi) \d \xi.
\end{equation}

\subsection{Fourier integral operators at fixed energies}
\label{Fourier integral operators at fixed energies}

 For all $\tau\in L^2(S^{d-1})$ we introduce
\begin{eqnarray}
(J^{\pm}(\lambda)\tau)(x)&:=&
(2\pi)^{-d/2}
\int
  \e^{\i \phi^{\pm}(x,\omega,\lambda)}\tilde a^{\pm}(x,\omega,\lambda)
  \tau(\omega)\d \omega,\label{eq:jdef}\\
(T^{\pm}(\lambda)\tau)(x)&:=&
(2\pi)^{-d/2}
\int
  \e^{\i \phi^{\pm}(x,\omega,\lambda)}\tilde t^{\pm}(x,\omega,\lambda)
  \tau(\omega)\d \omega,\label{eq:tdef}
\end{eqnarray} 
where
\begin{eqnarray*} \tilde a^{\pm}(x,\omega,\lambda)&:=&
(2\lambda)^{(d-2)/4} a^{\pm}(x, \sqrt{2\lambda}\omega),\\
 \tilde t^{\pm}(x,\omega,\lambda)&:=&
(2\lambda)^{(d-2)/4} t^{\pm}(x,\sqrt{2\lambda}\omega).\end{eqnarray*}

The functions $\tilde a^{\pm}$ and $\tilde t^{\pm}$ are continuous in
$(x,\omega,\lambda)\in \R^d\times S^{d-1}\times [0,\infty)$. This fact
will be very important in the forthcoming sections. Due to these
properties  we can {\it define} $J^{\pm}(\lambda)$ and $T^{\pm}(\lambda)$
at $\lambda=0$ by the expressions (\ref{eq:jdef}) and (\ref{eq:tdef}),
respectively.  We can split  $T^{\pm}
(\lambda)=T^{\pm}_\bd(\lambda)+T^{\pm}_\pr (\lambda)$ in agreement
with the decomposition (\ref{eq:r2sum
}) (cf. (\ref{eq:rega2})).

Throughout
this subsection 
$\breve \epsilon$ signifies the $\breve \epsilon>0$   appearing in
Proposition 
\ref{prop:mixed_2aaaa} (it is tacitly
assumed that $\breve \epsilon< 1-\mu/2)$). For the problems at hand we
can use coordinates for
$\omega\in S^{d-1}$ sufficiently close
to the $d$'th standard vector  $e_d\in \R^{d}$ specified as follows (using a
partition of unity in the $\hat x$--variable and a rotation of
coordinates this is without loss of generality):
\begin{equation}\label{perpy}\omega=\omega_\bot+\omega_d e_d;\;\omega_d=\sqrt
  {1-\omega_\bot^2},\;\omega_\bot\in \R^{d-1},\;|\omega_\bot| \text{ small}.\end{equation}
\begin{prop} There exist a (large) $R\geq R_0$ and a (small) $\tilde
  \sigma\in ]0,\sigma_0]$ such that for all $|x|\geq R$ there exists a
  unique $\omega\in S^{d-1}$ satisfying $\omega\cdot \hat x\geq
  1-\tilde\sigma$ (alternatively: $x\in \Gamma^+_{R,\tilde
    \sigma}(\omega)$) and
  $\partial_{\omega}\phi^+(x,\omega,\lambda)=0$.  We introduce
  the notation $\omega_{\crt}^+=\omega_{\crt}^+(x,\lambda)$ for this
  vector. It is smooth in $x$ and we have
\[\partial ^\gamma_x(\omega_{\crt}^+-\hat x)=O(|x|^{-\breve \epsilon-|\gamma|}).\]
 Let 
 \begin{equation}
   \label{eq:newph}
   \phi(x,\lambda)=\phi^+(x,\omega_\crt^+(x,\lambda),\lambda).
 \end{equation}
This function  solves  the eikonal equation
\begin{equation*}
(\partial_x \phi(x,\lambda))^2/2+V(x)=\lambda.
\end{equation*}
In the spherically symmetric case we have $\omega_\crt^+=\hat x$ and
\begin{equation}\label{eq:eik222}
\phi_\sph(x,\lambda)=\sqrt{2\lambda}R_0+\int_{R_0}^{|x|}\sqrt{2\lambda-2V(r)}\d
r.\end{equation}
\label{eiko}\end{prop}

The proposition is obvious in the case $V_2=0$, cf. (\ref{phi}). 
 The  general case follows by an application of the fixed point
 theorem, cf. the proof of the similar statement \cite[Lemma
4.1]{II}. At this point one needs some control of the Hessian; we
refer the reader to the proof of Theorem \ref{thm:kkk}.
 
Of course, there is an analogue of Proposition  \ref{eiko} in the $-$
case;  
we then need to replace $\phi^+$ with $\phi^-$, and $\hat x$ with $-\hat
x$. We obtain
$\omega_{\crt}^-(x,\lambda)=-\omega_{\crt}^+(x,\lambda)$. Note the
identity
\[\phi(x,\lambda)=-\phi^-(x,\omega_\crt^-(x,\lambda),\lambda).\]

\begin{thm}\label{thm:kkk}
Let $\tau\in C^\infty(S^{d-1})$. Then
\begin{align}&
 \big (J^\pm(\lambda)\tau\big )(x)\nonumber\\&=
(2\pi)^{-{\tfrac{1}{2}}}\e^{\mp\i \pi\tfrac{d-1}{4}}
g^{-\tfrac{1}{2}}(r,\lambda)r^{-\tfrac{d-1}{2}}\big
 (\e^{\pm\i\phi(x, \lambda)}\tau(\pm\hat x)+O(r^{-\breve
   \epsilon}\big )\big).\label{eq:lowbound eigb} 
\end{align}  Moreover (\ref{eq:lowbound eigb}) is uniform in  
 $(\hat x, \lambda) \in S^{d-1}\times [0,\infty[$. The same
 asymptotics holds for \[\pm g^{-1}\hat x\cdot pJ^\pm(\lambda)\tau(x).\]
 \end{thm}

\proof
We  invoke 
the method of
stationary phase 
(with a parameter given by the expression $h=h(r)$ of (\ref{eq:8aammmi})), cf. \cite[Theorem 7.7.6]{Ho1} or \cite[Theorem
4.3]{II}. For simplicity  we consider only the $+$ case and we abbreviate $\omega_{\crt}=\omega_{\crt}^+$.
 This method
yields (up to a minor point  that is resolved below) that 
\begin{align}\label{eq:lowbound eig}
 (J^+(\lambda)\tau\big )(x)=&(2\pi)^{-\frac d2}\e^{-\i \pi\tfrac{d-1}{4}}|\det
 (\partial^2_{\omega}\phi^+ (x,\omega_{\crt},
 \lambda)/2\pi)|^{-\tfrac{1}{2}}\nonumber\\&\times \e^{\i\phi^+(x,\omega_{\crt},
   \lambda)}\big (\tilde a^+(x,\omega_{\crt},
 \lambda)\tau(\omega_{\crt})+g^{\frac {d-2}2}O(r^{-\breve \epsilon})\big). 
\end{align}

Let us consider the Hessian. We first
  compute it  in the case $V_2=0$ choosing coordinates such that $\hat
  x= e_d$ and using \eqref{perpy}:
\begin{equation*}
  \partial^2_{\omega_\perp}\phi^+_{\sph}(\omega=\hat x)=-
  \partial_{\omega_\perp} 
  \partial_{\hat x} \phi^+_{\sph}(\omega=\hat x),
\end{equation*} and using the fact that 
\begin{equation}\label{eq:form1}
  \partial_{\omega_\perp}
  \partial_{\hat x}\phi^+_{\sph}(\omega=\hat x)= hI,
 \end{equation} cf. the computation (\ref{estoaaa}) (here $I$
 refers to the form on $TS^{d-1}_{\hat x=\omega}\times TS^{d-1}_\omega$ given
 by the Euclidean metric),
we obtain   that 
\begin{equation}\label{eq:lowbound eig3333}
  \partial^2_{\omega_\perp}\phi^+_{\sph}(\omega=\hat x)=-hI. 
\end{equation}
In particular the  critical point is non-degenerate in this case.

Since $\omega_\crt$ is a critical point, the second derivative has an
invariant geometric meaning. Therefore, we can drop the reference to the
special coordinates $\omega_\perp$ and we can  write simply 
$\partial_\omega^2$ for $\partial_{\omega_\perp}^2$ in the left  hand side of
\eqref{eq:lowbound eig3333}. The formula \eqref{eq:lowbound eig3333}
is then valid for all $\hat x \in S^{d-1}$.

The general
case is similar. In particular, after applying Proposition  \ref{prop:mixed_2aaaa}
and (\ref{eq:lowbound eig3333}),  we obtain
\begin{equation}
  \label{eq:8sammen}
 |\det
 (\partial^2_{\omega}\phi^+ (x,\omega_{\crt},
 \lambda))|=h^{d-1}(1+O(r^{-\breve \epsilon})). 
\end{equation}

We conclude  by combining (\ref{esto22}), Lemma \ref{lemmo},
Proposition \ref{prop:booh}, 
(\ref{eq:8sammen})  and the
construction of the symbol $\tilde a^+$ that  (\ref{eq:lowbound eig})
 and indeed  also (\ref{eq:lowbound eigb})
 hold. 

The second part of the theorem follows similarly.
\qed

 \section{Wave matrices}
\label{Wave matrices}

In this section we study (modified) wave matrices.
 We prove that  they
have  a limit at zero energy, in the sense of maps into an appropriate weighted
space. This implies asymptotic oscillatory formulas for the standard 
short-range and Dollard scattering matrices.

\subsection{Wave operators}
\label{Wave operators and S--matrix}

The following theorem is essentially well-known (follows from
(\ref{eq:postca})). It describes 
 a construction of modified wave operators similar to that of Isozaki-Kitada
\cite {IK1, IK2}. Notice, however, that the original construction involved
 energies strictly bounded
 away from zero. Notice also that the construction of $J^{\pm}$ in Section
 \ref{FIO}, although given under Conditions
\ref{assump:conditions1} and \ref{assump:conditions2}, in fact can  be
done under Condition \ref{symbol} as well.

\begin{thm}\label{thm:t11}
Suppose that $V$ satisfies Condition \ref{symbol}.  Then
\begin{align}  W^{\pm}f=\lim _{t\to {\pm}\infty}\e^{\i tH}J_0^{\pm}\e^{-\i tH_0}f=
\lim _{t\to {\pm}\infty}\e^{\i tH}J^{\pm}\e^{-\i tH_0}f;\;
  \hat f\in C_c(\R^d\setminus\{0\}). \label{eq:ooo1}
\end{align}
 The ``wave operator'' $W^\pm$ extends to an  isometric operator on
 $L^2(\R^d)$ satisfying  $HW^\pm=W^\pm H_0$, and its range is 
the absolutely continuous spectral subspaces of $H$. Moreover,
\begin{align}
  0=\lim _{t\to {\mp}\infty}\e^{\i tH}J_0^{\pm}\e^{-\i tH_0}f
=\lim _{t\to {\mp}\infty}\e^{\i tH}J^{\pm}\e^{-\i tH_0}f;\;
  \hat f\in C_c(\R^d\setminus\{0\}).
\label{eq:ooo}\end{align}
\end{thm}

\begin{remarks*} We know that $J_0^\pm1_{]\epsilon,\infty[}(H_0)$ and
 $J^\pm1_{]\epsilon,\infty[}(H_0)$ are
  bounded for any $\epsilon>0$, but we {\it do not} know if $J_0^\pm$ and
 $J^\pm$ are bounded (not even under Conditions
\ref{assump:conditions1} and \ref{assump:conditions2}).
 This is the reason for restricting the choice of vectors in
  (\ref{eq:ooo1}) and   (\ref{eq:ooo}). An alternative, and
  equivalent, definition of $W^\pm$ as a  bounded  operator on $L^2(\R^d)$ is
  the following: \[W^{\pm}=\slim _{\epsilon \searrow0}\slim _{t\to {\pm}\infty}\e^{\i tH}J^{\pm}1_{]\epsilon,\infty[}(H_0)\e^{-\i tH_0}.\] 
\end{remarks*}

The following general  fact serves as the basic formula in stationary
scattering theory, see Appendix \ref{Appendix} for a derivation.
 
\begin{lemma}\label{lemma:t12} Suppose there are densely defined operators $\breve
  J^\pm$ and $\breve T^\pm$ on $L^2(\R^d)$ such that $\breve J^\pm1_{]\epsilon,\infty[}(H_0)$ and $\breve
  T^\pm1_{]\epsilon,\infty[}(H_0)$ are
  bounded for any $\epsilon>0$ and that $\breve T^\pm f=i(H\breve
  J^\pm-\breve J^\pm H_0)f$ for any $f\in L^2(\R^d)$ with 
$\hat f\in C_c(\R^d\setminus\{0\})$. Suppose there exists 
\[\breve W^\pm f:=\lim_{t\to\pm\infty}\e^{\i tH}\breve J^\pm\e^{-\i tH_0}f,\;
  \hat f\in C_c(\R^d\setminus\{0\}).\]

Then
we have the following  formula
\begin{eqnarray}\label{wave9}
\breve W^\pm f
 &=&\lim_{\epsilon\searrow0}
\int(\breve J^\pm+\i R(\lambda\mp \i
\epsilon)\breve T^\pm)\delta_\epsilon(\lambda) f \d \lambda,
\end{eqnarray} 
where $\delta_\epsilon(\lambda)=\frac{R_0(\lambda+\i\epsilon)-R_0(\lambda-\i\epsilon)}{2\pi\i}=
\frac{\epsilon}{\pi}\left((H_0-\lambda)^2+\epsilon^2\right)^{-1}$.
\end{lemma}

 \subsection{Wave matrices at positive energies}\label{Wave matrices at positive  energies}

For any $s\in {\mathbb R}$ we recall the definition of weighted spaces
$L^{2,s}({\mathbb R^d}):=(1+x^2)^{-s/2}L^{2}({\mathbb R^d})$.

Let $\Delta_\omega$ denote the Laplace-Beltrami operator on the sphere
$S^{d-1}$. For  $n\in{\mathbb R}$ we define the Sobolev spaces on the sphere
$L^{2,n}(S^{d-1}):=(1-\Delta_\omega)^{-n/2}L^{2}(S^{d-1})$.

For $\lambda>0$  we introduce ${\mathcal
F}_0(\lambda)$  by
\begin{equation*}
{\mathcal F}_0(\lambda)f(\omega)=(2\lambda)^{(d-2)/4} \hat f(\sqrt
{2\lambda}\omega).
\end{equation*} 
Let  $s>\frac12$ and $n\geq0$.
Note that  ${\mathcal
F}_0(\lambda)$ is a bounded operator in the space $\mathcal B
(L^{2,s+n}(\R^d),L^{2,n}(S^{d-1}))$ 
 and depends continuously on $\lambda>0$.
Likewise,
  ${\mathcal F}_0(\lambda)^*\in
\mathcal B (L^{2,-n}(S^{d-1}),L^{2,-s-n}(\R^d))$ \and it also
depends continuously on
$\lambda>0$. Note also that the operator
 \begin{equation}\label{decom1}
 \int \oplus {\mathcal F}_0(\lambda)\, \d
\lambda:L^2(\R^d)\to \int_{0}^\infty \oplus L^2(S^{d-1})\;\d \lambda
 \end{equation}   is unitary;
  consequently the operators ${\mathcal
  F}_0(\lambda)$  diagonalize the operator $H_0$. Finally,
\begin{eqnarray}\label{free22}
{\rm s}-\lim_{\epsilon\searrow0}\delta_\epsilon(\lambda) 
&=& {\mathcal F}_0(\lambda)^*{\mathcal
  F}_0(\lambda)\text{ in }{\mathcal B}(L^{2,s}(\R^d)),L^{2,-s}(\R^d)).\end{eqnarray}

Due to the limiting
absorption principle we have the following partial analogue of (\ref{free22}) for the full
Hamiltonian, defined under  Condition \ref{symbol}:  Let $s>\frac12$ and 
\begin{equation}
\delta_\epsilon^V(\lambda)
:=\frac{R(\lambda+\i\epsilon)-R(\lambda-\i\epsilon)}{2\pi\i}.
\label{eq:sfs}\end{equation}
Then there exists
\begin{equation}
\delta^V(\lambda):={\rm s}-\lim_{\epsilon\searrow0}\delta_\epsilon^V(\lambda)\text{ in }{\mathcal B}(L^{2,s}(\R^d)),L^{2,-s}(\R^d)).
\label{eq:sfs1}\end{equation}
The operator-valued function $\delta^V(\cdot)$ is a strongly continuous function of
$\lambda>0$.

If  Conditions
\ref{assump:conditions1}--\ref{assump:conditions3} are true then
we can extend the definition of $\delta^V(\lambda)$ to include $\lambda=0$ if we
 demand that $s>\frac12+\frac\mu4$, and the corresponding
 operator-valued function will be a strongly continuous (in fact, norm continuous) function of
$\lambda\geq0$, 
cf. Remark \ref{remarks:nosmmmmo} \ref{it:csmoottaa}). 

In the remaining part of this section we shall assume that the
positive parameter $ \sigma'$   in
 (\ref{eq:chi^2}) is sufficiently  small (this requirement can be
 fulfilled uniformly in  $\lambda\geq0$). Notice that the condition
 conforms well with Lemma \ref{lemma:mixed_2}; we need it
 at various points, see for example the proof of Lemma \ref{lemma: statio}.

Formally, we have
$J^\pm(\lambda)=J^\pm{\mathcal F}_0(\lambda)^*$ and 
$T^\pm(\lambda)=T^\pm{\mathcal F}_0(\lambda)^*$. This suggests that \eqref{wave9} can
be used to define wave operators at a fixed energy. This idea is used in the
following theorem (which is  essentially well-known).
\begin{thm} \label{thm:1}
 Suppose that the potential satisfies Condition \ref{symbol}. Let $\epsilon>0$, $n\geq0$ and $\lambda>0$. Then
\begin{eqnarray}\label{wave22}
W^{\pm}(\lambda)&:=&J^{\pm}(\lambda)+\i R(\lambda{\mp}\i0)
T^{\pm}(\lambda)\end{eqnarray}
defines a bounded operator in $ {\mathcal
B}(L^{2,-n}(S^{d-1}),L^{2,-\frac12-\epsilon-n}(\R^d))$, which
 depends   continuously on
$\lambda>0$. It  depends only on the splitting of the potential $V$
 into $V_1$ and $V_3$ (but does not depend on the details of the
 construction of $J^\pm$). 
 For all $f\in L^{2,\frac12 +\epsilon}({\mathbb R}^d)$ and $g\in C_c(]0,\infty[)$, we have
\begin{equation}\label{eq:7usefbbb}
W^{\pm}g(H_0)f=\int_0^\infty g(\lambda)W^{\pm}(\lambda)
{\mathcal F}_0(\lambda) f \d\lambda.
\end{equation}
Moreover, 
\begin{equation}
  \label{eq:7usef}
  W^\pm(\lambda)W^\pm(\lambda)^*= 
 \delta^V(\lambda).\end{equation}
  We set 
\[w^{\pm}(\omega,\lambda)=W^{\pm}(\lambda)\delta_\omega,\]
 where 
$\delta_\omega$ denotes the delta-function at $\omega\in S^{d-1}$.
Then for
all multiindices $\delta$ the function 
\[S^{d-1}\times]0,\infty[\ni(\omega,\lambda)\mapsto \partial^\delta_\omega w^{\pm}(\omega,\lambda)\in
 L^{2,-p}(\R^d);\;p>|\delta|+d/2,\]
is continuous.
\end{thm}

\begin{remark*}The operator $W^\pm(\lambda):\mathcal D'(S^{d-1})\to
  L^{2,-\infty}$ is  called the  {\em wave matrix at the 
 energy $\lambda$}. Its range consists of generalized eigenfunction
at the energy $\lambda$. The function $ w^{\pm}(\omega,\lambda)$
(which belongs to $W^\pm(\lambda)L^{2,\frac12 -p}(S^{d-1})$ for $p>\frac d2$)  is called the {\em generalized
 eigenfunction at the energy $\lambda$ and  outgoing (or incoming)
 asymptotic normalized velocity  $\omega$}.
  \end{remark*}

Let us explain  the  steps of a proof of
  Theorem \ref{thm:1} (in the case of ``$+$''--superscript only);  our (main) results
  contained  in
Theorems \ref{thm:4a} and  \ref{thm:4}  will be  proved by a parallel  procedure.

First  one
 introduces
  a partition of unity of the form
\begin{align}
  \label{eq:partitionleft}
  I &=\Opr(
 {\chi}_{+}(a))+
\Opr({\chi}_{-}(a)\tilde{\chi}_{-}(b))+\Opr({\chi}_{-}(a)
\tilde{\chi}_{+}(b))\nonumber\\&=:\Opr(\chi_1)+\Opr(\chi_2)+\Opr(\chi_3) .
\end{align} Here $a$ and  $b$ are the symbols introduced in
(\ref{eq:newa's}) (rather
than in \eqref{eq:def_a0_b} since we do not here impose Conditions
\ref{assump:conditions1}--\ref{assump:conditions3})
 and 
 ${\chi}_{+}$ is a real-valued function  as in Proposition 
\ref{prop:resolvent_basic2} (\ref{item:C-02}) such that $\chi_+(t)=1$ for
$t\geq2C_0$, and ${\chi}_{-}=1-{\chi}_{+}$. 
 Moreover, 
$\tilde{\chi}_{-},\tilde{\chi}_{+}\in C^\infty(\R) $ are real-valued
functions obeying 
$\tilde{\chi}_{-}+\tilde{\chi}_{+}=1$ and 
\begin{align}
  \label{eq:supp tilde chi1}
 &\supp \tilde{\chi}_{-}
 \subseteq(-\infty, 1-\bar \sigma],\\&\supp \tilde{\chi}_{+} \subseteq
 [1-2\bar \sigma,\infty[.  \label{eq:supp tilde chi2}
\end{align}
 The number  $\bar \sigma$ needs to be taken (small) positive depending on the  parameter $\sigma$
 of   Subsection \ref{The first WKB ansatz}. (For the proof of Theorems \ref{thm:4a}
 and  \ref{thm:4} to be elaborated on later we refer at this
point to 
 \eqref{eq:small_epsi} for the
 precise requirement.)  

The proof of Theorem  \ref{thm:1} is based on the following lemma:

\begin{lemma}\label{lem:2a}
Suppose that the potential satisfies  Condition \ref{symbol}. 
\begin{enumerate}[\normalfont (i)]
\item \label{it:p2}For  all  $n\geq 0$ and   $\epsilon>0$,
$
 J^+(\lambda)$ is a continuous function in $\lambda > 0$ with values in
  $ {\mathcal
 B}(L^{2,-n}(S^{d-1}),L^{2,-{1\over 2} -\epsilon-n}(\R^d)).$

\item \label{it:p3}For  all $n\in \R$ and $\epsilon>0$,
$
 T_\bd^+(\lambda)$ is a continuous function in $\lambda > 0$ with values in
  $ {\mathcal
 B}(L^{2,-n}(S^{d-1}),L^{2,{1\over 2} -\epsilon-n}(\R^d)).$

\item \label{it:p4}For all $m,n\in \R$, 
$\Opr(\chi_3)T_\bd^+(\lambda)$ is a continuous  function in $\lambda > 0$ with values in
$ {\mathcal
 B}(L^{2,-n}(S^{d-1}),L^{2,m}(\R^d))$.

\item \label{it:p5}
For  all $m,n\in \R$,
  $T_\pr^+(\lambda)$ is a continuous function in $\lambda > 0$ with values in     
     $ {\mathcal
 B}(L^{2,-n}(S^{d-1}),L^{2,m}(\R^d)).$
\end{enumerate}
\end{lemma}

More general statements than Lemma \ref{lem:2a}
(\ref{it:p2})--(\ref{it:p5}) will be given and proven in the context
of treating small energies (see  Lemma \ref{lem:22a}); these
statements are
  under
Conditions \ref{assump:conditions1} and \ref{assump:conditions2}. Let us here use (\ref{it:p2})--(\ref{it:p5}) in
an

\noindent{\it Outline of a proof of Theorem \ref{thm:1}.} The
expression (\ref{wave22}) is a well-defined element of $ {\mathcal B}(L^{2,-n}(S^{d-1}),L^{2,-{1\over 2}
  -\epsilon-n}(\R^d))$ due to the positive energy  version of Proposition
\ref{prop:proposi1} and Lemma  \ref{lem:2a}; this is for any
$\epsilon>0$  and $n\geq 0$. (Notice that 
(\ref{labb2}) holds for any $t\in \R$ by  Lemma
\ref{lem:2a}.) Effectively, this argument is based on the following  
scheme (to
be used below):
We insert the right hand side of (\ref{eq:partitionleft}) to the right
of the resolvent in (\ref{wave22}) and expand into three terms. Whence,
by 
using Remark \ref{remarks:nosmmmmo} \ref{it:csmoott3}) and Lemma
\ref{lem:2a}, we see that $W^+(\lambda)$ is a sum
of four well-defined operators in $ {\mathcal B}(L^{2,-n}(S^{d-1}),L^{2,-{1\over 2}
  -\epsilon-n}(\R^d))$, hence well-defined.

Next note that $
\lambda\mapsto W^+(\lambda)$ is norm continuous, due to the norm continuity
of each of
the above mentioned four  operators, which
in turn may be seen by combining  the continuity statements of Remark \ref{remarks:nosmmmmo} \ref{it:csmoott3}) and Lemma
\ref{lem:2a}. 

The statement on the independence of details of construction of
$J^\pm$ is based on the positive energy  version of Proposition
\ref{prop:propa5}; the interested reader will realize this by using
arguments  from the proof
of Lemma \ref{prop:qs4} stated later.

The formula (\ref{eq:7usefbbb}) can be verified by combining
(\ref{wave9}) with arguments used above, see Appendix
\ref{Appendix} for an abstract approach. The identity
(\ref{eq:7usef}) is a consequence of (\ref{eq:7usefbbb}). 

Finally, due to the fact that $ \partial^\delta_\omega\delta_\omega\in L^{2,{1\over 2}-p}(S^{d-1})$ for
$p>|\delta|+{d\over 2}$ (with continuous dependence of  $\omega\in S^{d-1}$),
we conclude that indeed $\partial^\delta_\omega w^+(\omega,\lambda)\in L^{2,-p}(\R^d)$
with a continuous dependence of $\omega$ and $\lambda$.

\subsection{Wave matrices at low energies}\label{Wave matrices at low energies}

Until the end of this section we assume that Conditions
\ref{assump:conditions1}--\ref{assump:conditions3} are true. The main  new result of this section 
is expressed in the following two theorems which concern the
low-energy behaviour of the wave matrices of Theorem
\ref{thm:1}:

\begin{thm}\label{thm:4a}
For $s>\frac12+\frac\mu4$ and $n\geq0$,
\begin{eqnarray}W^\pm (0):=J^\pm(0)+\i R(\mp\i0) T^\pm(0)\label{wave2}
\end{eqnarray}
defines a bounded operator in  
 $ {\mathcal
B}(L^{2,-n}(S^{d-1}),L^{2,-s-n(1-\mu/2)}(\R^d))$. 
 It  depends only on the splitting of the potential $V$
 into $V_1+V_2$ and $V_3$ (but does not depend on the details of the
 construction of $J^\pm$). We have
 \begin{equation}
   \label{eq:delt2}
   W^\pm(0)W^\pm(0)^*=\delta^V(0).
 \end{equation}

If we set
\[w^\pm(\omega,0)=W^\pm(0)\delta_\omega,\]
then we obtain an element of $L^{2,-p}(\R^d)$ with
$p>\frac{d}{2}+\frac\mu2-\frac{d\mu}{4}$ depending continuously on
 $\omega$. In fact, more generally, $\partial_\omega^\delta
w^\pm(\omega,0)\in L^{2,-p}(\R^d)$ with
$p>(|\delta|+\frac{d}{2})(1-\frac{\mu} {2})+\frac{\mu}{2}$ with  continuous
dependence  on
$\omega$.
\end{thm}

\begin{thm}\label{thm:4}
For  all $\epsilon>0$  and  $n\geq 0$, 
\begin{equation}
  \label{eq:x_weights29bc}
 (\langle x \rangle g)^{-n} \langle x \rangle^{-\frac12 -\epsilon}
  g^{1\over 2}W^\pm(\lambda)
\end{equation}
 is a continuous  ${\mathcal
 B}(L^{2,-n}(S^{d-1}),L^{2}(\R^d))$--valued function in
$\lambda\in[0,\infty[$.  

For  all $\epsilon>0$  and all multiindices $\delta$,
 the
function
\[S^{d-1}\times[0,\infty[\ni(\omega,\lambda)
\mapsto
 (\langle x \rangle g)^{-|\delta|+\frac{1}{2}-\frac{d}{2}} \langle x \rangle^{-\frac12 -\epsilon}
  g^{1\over 2}
 \partial_\omega^\delta w^\pm(\omega,\lambda)\in L^{2}(\R^d)\]
is continuous. \end{thm}

The following corollary interprets Theorem \ref{thm:4} in terms of the usual
weighted spaces:

\begin{cor}
Let $n\geq0$. We have
\[W^\pm(0)=
\lim_{\lambda\searrow0}W^\pm(\lambda)
\]
in the sense of  operators  in 
 $ {\mathcal
B}(L^{2,-n}(S^{d-1}),L^{2,-\tilde s_n}(\R^d))$,
where\\ $\tilde
s_n>\frac12+n+\max\left(0,\frac{\mu}{4}-n\frac\mu2 \right)$.
For all multiindices $\delta$, the
function
\[S^{d-1}\times[0,\infty[\ni(\omega,\lambda)
\mapsto \partial_\omega^\delta w^\pm(\omega,\lambda)\in L^{2,-\tilde p}(\R^d)\]
is continuous, with $\tilde p>\frac{d}{2}+|\delta|$ for $d\geq2$ and $\tilde p
>\frac12+|\delta|+\max(0, (1-2|\delta|)\frac\mu4)$ for $d=1$.
\end{cor}

The proof of Theorems  \ref{thm:4a} and \ref{thm:4} 
is based on the following analogue of Lemma \ref{lem:2a}  (for convenience we 
focus as before on   the case of ``$+$''--superscript only). The symbol $\chi_3$
appearing in the statement (\ref{it:p40}) below is specified  as
before, i.e. by 
(\ref{eq:partitionleft}) and the subsequent discussion.

\begin{lemma}\label{lem:22a} 
\begin{subequations}
  \begin{enumerate}[\normalfont (i)]
\item \label{it:p20}For  all  $n\geq 0$  and  $\epsilon>0$, 
\begin{equation}
  \label{eq:x_weights29bc9}
 (\langle x \rangle g)^{-n} \langle x \rangle^{-\frac12 -\epsilon}
  g^{1\over 2}J^+(\lambda)
\end{equation} is a continuous ${\mathcal
 B}(L^{2,-n}(S^{d-1}),L^{2}(\R^d))$--valued  function in
$\lambda\in[0,\infty[$.
\item \label{it:p30}For  all $n\in \R$  and  $\epsilon>0$, 
\begin{equation}
  \label{eq:x_weights29c}
 (\langle x \rangle g)^{-n} \langle x \rangle^{\frac 12 -\epsilon} g^{-\frac
    12}T_{\bd}^+(\lambda)
\end{equation} is a continuous ${\mathcal
 B}(L^{2,-n}(S^{d-1}),L^{2}(\R^d))$--valued  function in
$\lambda\in[0,\infty[$.
\item \label{it:p40}For all $m,n\in \R$, 
\begin{equation}
  \label{eq:x_weights29d}
 \langle x \rangle^m \Opr(\chi_3) T_{\bd}^+(\lambda)
\end{equation} is a continuous ${\mathcal
 B}(L^{2,-n}(S^{d-1}),L^{2}(\R^d))$--valued function in
$\lambda\in[0,\infty[$.
\item \label{it:p50}
For  all $m,n\in \R$, 
 \begin{equation}
  \label{eq:x_weights29e}
 \langle x \rangle^m T_{\pr}^+(\lambda)
\end{equation} is a continuous ${\mathcal
 B}(L^{2,-n}(S^{d-1}),L^{2}(\R^d))$--valued function in
$\lambda\in[0,\infty[$.
\end{enumerate}
\end{subequations}
\end{lemma}

Later on we will actually need a slightly  stronger bound than the one
of Lemma 
\ref{lem:22a} (i) with $n=0$, which we state below (referring to
notation of (\ref{eq:s_0}) and (\ref{eq:9boncloseda})):
\begin{lemma}
  \label{lemma: statio} For all $\tau\in L^2(S^{d-1})$,
  $J^+(\lambda)\tau\in B_{s_0}^*$. In fact, with a bounding constant
  independent of $\lambda\geq0$,
 \begin{equation*}
    g^{\frac 12}J^+(\lambda)\in \mathcal{B} (L^2(S^{d-1}),B_{\frac 12}^* ).
  \end{equation*} 
\end{lemma}
\begin{proof}
  We need to  bound  the operator
  $P_R:=R^{-1}J^+(\lambda)^*g1_{\{|x|<R\}}J^+(\lambda)$ independently of
  $R>1$ and $\lambda\geq 0$. Writing $P_R=R^{-1}\int _0^R\d r\int_{S_r}Q_{r}\d x$ with $S_r=\{|x|=r\}$, it
  thus suffices to bound the operator $\int_{S_r}Q_{r}\d x$ independently of
  $r>0$ and $\lambda\geq 0$.

\noindent{\bf Step  I}. Analysis of  $\int_{S_r}Q_{r}\d x$. The kernel of $Q_{r}$ is given by
\begin{equation*}
 Q_{r}(\omega,\omega')= \e^{\i\big (\phi^+(x,\omega',\lambda)-\phi^+(x,\omega,\lambda) \big )}a(x,\omega,\omega',\lambda),
\end{equation*} where
\begin{equation*}
 a(x,\omega,\omega',\lambda)=(2\pi)^{-d}g(|x|,\lambda)\overline {\tilde a^+(x,\omega,\lambda)} \tilde a^+(x,\omega',\lambda).
\end{equation*} For simplicity, we shall henceforth omit the
superscript $+$, $r>0$  and $\lambda\geq 0$ in the notation.

Our goal is to show that $\int_{S_r}Q_{r}\d x$ is a PsDO on
$L^2(S^{d-1})$ with symbol $b(\omega,\omega', z)$ obeying uniform
bounds (uniform in $r>0$  and $\lambda\geq 0$) 
\begin{equation}
  \label{eq:8symbnd}
 |\partial_\omega^{\beta_1}\partial_{\omega'}^{\beta_2}  \partial_z^\alpha
 b|\leq C_{\beta_1,\beta_2,\alpha}\langle z\rangle^{-|\alpha|}.
\end{equation}
Clearly this would  prove  the lemma.

We can use a partition of unity on $S^{d-1}$, and therefore we
can assume that the vectors $\omega, \omega'$ and $\hat x$ are close
to the $d$'th standard vector  $e_d\in \R^{d}$. Consequently, we can 
use coordinates
\begin{align}\label{eq:coor1}
  \omega&=\omega_\bot+\omega_d e_d,\;\omega_d=\sqrt
  {1-\omega_\bot^2},\\
 x&=x_\bot+x_de_d,\;x_d=\sqrt {r^2-x_\bot^2}.\label{eq:coor2}
\end{align}

Next we write
\begin{equation*}
  \phi(x,\omega')-\phi(x,\omega)=
  (\omega_\bot-\omega'_\bot)\cdot z,\;z=-\int_0^1\partial_{\omega_\bot}\phi(x,s(\omega'-\omega)+\omega)\d s.
\end{equation*}

\noindent{\bf Step  II}. 
We shall show that the map
\begin{equation}
  \label{eq:toe}
  S_r\supset \mathcal{U}\ni x\to Tx=z\in \R^{d-1}\text{ is a
    diffeomorphism onto its range}. 
\end{equation} Here $\mathcal{U}$ is an open neighbourhood of $e_d$ 
containing the supports of $a(\cdot, \omega,\omega')$.

To this end we investigate the bilinear form
$\partial_x\partial_{\omega}\phi(x,\omega)$ on  $TS_
x^{d-1}\times TS_\omega^{d-1}$. Note that 
\begin{equation}\label{eq:form2}
  \partial_x
  \partial_\omega \phi^+_{\sph}(\hat x=\omega)= r^{-1}hI,
 \end{equation} cf. \eqref{eq:form1}.
 
 In the coordinates (\ref{eq:coor1}) and (\ref{eq:coor2}), the
 identity (\ref{eq:form2}) reads for $z_{\sph}=(Tx)_{\sph}$ (here we
 consider the case where $V_2=0$)
 \begin{equation}\label{eq:form3}
  \partial_{x_j}
 z_{\sph,\;i}(\omega=\omega'=\hat x)=
-r^{-1}h\big (\delta_{ij}+\omega_d^{-2}\omega_i\omega_j\big ),\;i,j\leq
d-1.
 \end{equation} 

Due to (\ref{eq:8aammm}), Proposition \ref{prop:mixed_2aaaa} and
(\ref{eq:form3}) we obtain the more general result
\begin{equation}\label{eq:form4}
  \partial_{x_j}
 z_i=
-r^{-1}h\big (\delta_{ij}+\omega_d^{-2}\omega_i\omega_j+O(\sigma')
+O(r^{-\breve \epsilon})\big ),\;i,j\leq
d-1.
 \end{equation} Here $O(\sigma')$ refers to a term obeying
 $|O(\sigma')|\leq C\sigma'$, where $ \sigma'>0$ is given in
 (\ref{eq:chi^2}) (assumed  to be small).  

In particular, $T$ is a local diffeomorphism with inverse  determinant
\begin{equation}\label{eq:det1}
  |\partial_{x_j}
 z_i|^{-1}=
(-r^{-1}h)^{1-d}(\omega_d')^2 \big (1+O(\sigma')
+O(r^{-\breve \epsilon})\big ).
 \end{equation} 

For a later application we note the uniform bounds  
\begin{equation}\label{eq:det2}
  \partial_\omega^{\beta_1}\partial_{\omega'}^{\beta_2}
  \partial_x^\alpha |\partial_{x_j}
 z_i|^{-1}=g^{1-d}r^{-|\alpha|}O(r^{0}).
 \end{equation} 

Also, $T$ is injective: Suppose $Tx^1=Tx^2$, then
\begin{align*}
0&= \int_0^1 \partial_{x_j}
 z_i(s(x^1-x^2)+x^2)(x^1_j-x^2_j)\d s\\&=-r^{-1}h\big(
 (\delta_{ij}+\omega_d^{-2}\omega_i\omega_j)+O(\sigma')
+O(r^{-\breve \epsilon})\big )(x^1_j-x^2_j).
      \end{align*} Using the  invertibility of the matrix
      $\delta_{ij}+\omega_d^{-2}\omega_i\omega_j$, it follows
      that $x^1=x^2$.

\noindent{\bf Step  III}. Analysis of  symbol $b$. Due to Step II, we
can change coordinates and obtain that $\int_{S_r}Q_{r}\d x$ is a PsDO
with a symbol $b=|\partial_{x_j}
 z_i|^{-1}a$. It remains to show (\ref{eq:8symbnd}). For zero indices
 $\beta_1=\beta_2=\alpha=0$, we obtain the bound by combining
 Proposition \ref{prop:booh} and (\ref{eq:det1}). For derivatives,
we note the bounds
\begin{equation}
  \label{eq:8symbndd1}
 |\partial_\omega^{\beta_1}\partial_{\omega'}^{\beta_2}  \partial_x^\alpha
 z|\leq C_{\beta_1,\beta_2,\alpha}gr^{1-|\alpha|},
\end{equation} which by a little bookkeeping yields 
\begin{equation}
  \label{eq:8symbndd2}
 |\partial_\omega^{\beta_1}\partial_{\omega'}^{\beta_2}  \partial_z^\gamma
 x|\leq C_{\beta_1,\beta_2,\gamma}r\langle z\rangle^{-|\gamma|}.
\end{equation}

Another bookkeeping using Proposition \ref{prop:booh},
(\ref{eq:det1})  and
(\ref{eq:8symbndd2}) yields (\ref{eq:8symbnd}).

\end{proof}

\noindent{\em Proof of Lemma \ref{lem:22a}.} 
 We drop the superscript ``$+$'' and the parameter
  $\lambda$ in the notation. We first prove  uniform boundedness on
  any compact interval 
$[0,\lambda_1]$.

 \noindent{\bf Re (\ref{it:p20}).} 
We replace $J=J(\cdot)$ by  $J\chi(\omega)$, where
 $\chi\in C^\infty(S^{d-1})$ with a sufficiently small support. We can
 assume that $n$ is a non-negative integer. 
 Instead of studying $J(1-\Delta_\omega)^{n/2}$, it
then suffices to study $J\partial_\omega^\nu$ for
$|\nu|\leq n$.

Integrating by parts, we observe that the corresponding  integral kernel equals
\begin{equation*}
C\partial_\omega^\nu \big (\e^{\i\phi(x,\omega)}\tilde
a(x,\omega)\big )
=\e^{\i\phi(x,\omega)}\tilde a_\nu(x,\omega),
\end{equation*}
where $\tilde a_\nu$ is a  linear combinations of terms of the form
\[\partial_\omega^{\nu_1}\phi(x,\omega)\cdots
\partial_\omega^{\nu_k}\phi(x,\omega )\partial_\omega^{\nu_0}
\tilde a(x,\omega),\]
with $\nu_0+\nu_1+\cdots+\nu_k=\nu$. Thus, using that
$|\partial_\omega^\delta\phi|\leq C\langle x\rangle g$
(cf. \eqref{eq:ole}) and Proposition \ref{prop:booh}, we obtain
\begin{equation}
 \label{eq:x_weights29bb}
\partial_\omega^\delta\partial_x^\gamma\tilde a_\nu(x,\omega)=O(\langle x\rangle^{n-|\gamma|} g^{n+\frac{d-2}2}).
\end{equation}
Then we follow the proof of Lemma \ref{lemma: statio}.

\noindent{\bf Re (\ref{it:p30}).} Assume first that $n\geq
0$. Then we follow the same scheme as
above. The bound on the relevant kernel needs to be replaced by 
\begin{equation}
 \label{eq:x_weights29bbb}
\partial_\omega^\delta\partial_x^\gamma
\tilde t_\nu(x,\omega)=
O(\langle x\rangle^{n-1-|\gamma|} g^{n+\frac d2}),
\end{equation} cf. 
Proposition \ref{prop:brrrr}. Using \eqref{eq:x_weights29bbb} we can proceed as
before.

Assume next that $n<0$. We can assume that $n$ is a negative integer. For fixed $x$, we decompose
$\omega=\omega_{\perp}+\sqrt{1-\omega_{\perp}^2}\hat x$, where
$\omega_{\perp}\cdot x=0$. By \eqref{eq:ole0}, we have the uniform lower
bound
\begin{equation}
  \label{eq:low71}
  |\nabla _{\omega_{\perp}}\phi(x,\omega)|\geq c|x|g \text{ for }\hat x\cdot
  \omega\leq 1-\sigma, 
\end{equation} and by \eqref{eq:ole}  the uniform upper
bounds
\begin{equation}
  \label{eq:low72}
  |\partial_{\omega_{\perp}}^\delta\phi(x,\omega)|\leq C|x|g.
\end{equation} 
We apply the non-stationary method based on the identity
\[\left(\i\frac{\nabla_{\omega_\perp}\phi}{|\nabla_{\omega_\perp}\phi|^2}\cdot 
\nabla_{\omega_\perp}\right)^{-n}\e^{\i\phi^+(x,\omega)}
=\e^{\i\phi(x,\omega)}.\]
After  performing $-n$ integrations
by parts, the bounds
(\ref{eq:low71}) and (\ref{eq:low72}) yield 
\[T\chi\tau=\sum_{|\nu|\leq-n}\int
\tilde t_\nu(x,\omega)\partial_{\omega_\perp}^\nu\tau(\omega)\d\omega,\]
where the functions $\tilde t_\nu$ also satisfy the bounds
\eqref{eq:x_weights29bbb}. Then  we proceed as before.

\noindent{\bf Re (\ref{it:p40}).}  
The kernel of $\Opr(\chi_3) T_{\bd}(\cdot)$ is given by the integral
\begin{equation*}
  \int \d \xi \e^{\i x\cdot \xi}\int \e^{\i (\phi(y,\omega)-y\cdot
  \xi)}k(\omega, y,\xi)\d y,\;k(\omega, y,\xi)=(2\pi)^{-3d/2}\chi_3(y,\xi)\tilde t_\bd(y,\omega).
\end{equation*} It suffices to show that
\begin{equation}\label{eq:inred}
  |\partial_\xi^\beta\partial_\omega^\delta\int \e^{\i (\phi(y,\omega)-y\cdot
  \xi)}k(\omega, y,\xi)\d
  y|\leq C_{\beta,\delta}
\text{  uniformly in }\xi, \omega\text{ and  }\lambda.
\end{equation} Notice
that the symbol $k$ is compactly supported in $\xi$. First we observe
that (using notation of Subsection \ref{Low energy  resolvent estimates1})
\[k=k_{\omega,\lambda}\in S_{\unif}(g^{\frac d2}\langle x \rangle^{-1},
g_{\mu,\lambda}).\]
  We can substitute $k\to k=F(|y|>2\bar R)k(\omega, y,\xi)$. 

Next we
integrate by parts, writing first 
\begin{equation*}
 \Big(\i \frac {\xi-\nabla_y\phi  }{|\xi-\nabla_y\phi|^2}\cdot \nabla _y
  \Big)^\ell \e^{\i (\phi(y,\omega)-y\cdot
  \xi)}=\e^{\i (\phi(y,\omega)-y\cdot
  \xi)}.
\end{equation*}

We need to argue that $\xi-\partial_y\phi\neq 0$ on the support of the involved
symbol. For that we recall the following elementary inequality valid for all $z_1,z_2 \in
\R^d$ and $ \kappa_1,\kappa_2>0$:  
\begin{equation}
  \label{eq:elem_eq}
  |z_1-z_2|^2\geq\min (\kappa_1^2/2,\kappa_2-\kappa_2^2/2) (|z_1|^2+|z_2|^2),
\end{equation} provided one of the following three conditions holds:
\begin{equation*}
  |z_2|\leq (1-\kappa_1)|z_1|,\;|z_1|\leq
   (1-\kappa_1)|z_2|\;\text{or}\;z_1\cdot z_2 \leq (1-\kappa_2) |z_1||z_2|.
\end{equation*} 

Now, on the support of the symbol $k$ we have $(1-\sigma')|y|\leq y\cdot \omega \leq
(1-\sigma)|y|$, cf. \eqref{eq:chi^2}. We use these inequalities  in \eqref{eq:F_angle1} and
\eqref{eq:F_angle2}, yielding 
\begin{equation*}
 1-C\sigma'- C|y|^{-\breve \epsilon} \leq\frac
   {\nabla_y\phi(y,\omega)}{|\nabla_y\phi(y,\omega)|}\cdot \hat y
\leq 1-c\sigma+ C|y|^{-\breve \epsilon},
\end{equation*} which in turn (if $\bar R$ is taken  large enough) 
 implies that 
\begin{equation}
  \label{eq:F_angle4}
 1-2C\sigma'\leq\frac
   {\nabla_y\phi(y,\omega)}{g(|y|)}\cdot \frac {y}{\langle y \rangle}
\leq 1-\frac
   {c}{2}\sigma.
\end{equation}  

We claim that there exists a small $c'=c'(\sigma,\sigma')>0$ such that 
\begin{equation}
  \label{eq:basicest}
  \big |\xi -\nabla_y\phi(y,\omega)\big | \geq c' \Big (|\xi| + \big |\nabla_y\phi(y,\omega)\big | \Big )
\end{equation}  on the support of $k$ (showing in
particular that $\xi-\partial_y \phi\neq 0$). 

Obviously, \eqref{eq:basicest} follows from \eqref{eq:elem_eq} with 
\begin{equation*}
 z_1=\tfrac
{\xi}{g(|y|)}\text{ and } z_2=\tfrac
{\nabla_y\phi(y,\omega)}{g(|y|)},
\end{equation*}
provided one of the above
three conditions hold. If all of those conditions fail, so that
intuitively 
$z_1\approx z_2$, we can replace $z_2$ in \eqref{eq:F_angle4} by
$z_1$ yielding
\begin{equation}
  \label{eq:F_angle5}
 1-3C\sigma'\leq b(x,\xi)
\leq 1-\frac
   {c}{3}\sigma.
\end{equation}  Here we applied  \eqref{eq:elem_eq}  for some  $\kappa_1$
and $\kappa_2$, depending on $\sigma$ and $\sigma'$.  
Now,  the second inequality of \eqref{eq:F_angle5} is violated on the
support of $\tilde{\chi}_{+}(b(y,\xi))$, provided that $\bar \sigma>0$ of
\eqref{eq:supp tilde chi2} is chosen such that 
\begin{equation}
  \label{eq:small_epsi}
  2\bar \sigma< \frac{c}{3}\sigma.
\end{equation}

We have shown the bound \eqref{eq:basicest} on the support of the
symbol $k$, and therefore in
particular on the support of the relevant symbol, after performing the
$y$--integrations by parts. The estimate (\ref{eq:inred}) follows.

\noindent{\bf Re (\ref{it:p50}).} 
First we assume that $n\geq0$. Integrating by parts in $\omega$,
as in the proof of
 (\ref{it:p20}), and using 
Proposition \ref{prop:brrrr}, which says that
  $t_\pr$ with all its derivatives is
 $O(\langle x\rangle^{-\infty})$, we obtain that $\langle x\rangle^m
 T_\pr(\lambda)$ is in ${\mathcal B}(L^{2,-n}(S^{d-1}),L^2(\R^d))$ for any
 $m$.  The case
 $n<0$ then follows trivially.

Let us now prove the continuity. Consider for instance 
 (\ref{it:p20}).
Let $\tau\in C^\infty(S^{d-1})$ and set
\[J_{n,\epsilon}(\lambda):=(\langle x\rangle g)^{-n}\langle
 x\rangle^{-\frac12-\epsilon}g^{\frac12}J(\lambda).\]
Clearly, for (small) $\kappa>0$,
\begin{equation}
J_{n,\epsilon}(\lambda)\tau=
F(\kappa|x|<1)J_{n,\epsilon}(\lambda)\tau+
F(\kappa|x|>1)\langle x\rangle^{-\epsilon/2}J_{n,\epsilon/2}(\lambda)\tau.
\label{labb8}\end{equation}

We know that $J_{n,\epsilon/2}(\lambda)$ is bounded uniformly in
$\lambda$. Hence the second term on the right of \eqref{labb8} is
$O(\kappa^{\epsilon/2})$.

 We know that $a(x,\omega,\lambda)$, 
 $\phi(x,\omega,\lambda)$ and $g(x,\lambda)^{\pm1}$ are continuous down to
 $\lambda=0$. The first term on the right of  \eqref{labb8} involves only
 variables in a compact set. 
Therefore it is continuous in $\lambda$. Hence
 $J_{n,\epsilon}(\lambda)\tau$ is continuous as the uniform limit of continuous
   functions. 

By the uniform bound, which we proved before,
we conclude that
 $J_{n,\epsilon}(\lambda)$   is strongly continuous in ${\mathcal
  B}(L^{2,-n}(S^{d-1}),L^2(\R^d))$. 

Now  
\[J_{n,\epsilon}(\lambda)=g^{\epsilon/2}J_{n+\epsilon/2,\epsilon/2}(\lambda)
(1-\Delta_\omega)^{-\epsilon/4}(1-\Delta_\omega)^{\epsilon/4},\] 
  where $g^{\epsilon/2}$ is strongly continuous,
$J_{n+\epsilon/2,\epsilon/2}(\lambda)$ is strongly continuous in\\
${\mathcal  B}(L^{2,-n-\epsilon/2}(S^{d-1}),L^2(\R^d))$,  
$(1-\Delta_\omega)^{-\epsilon/4}$ is a compact operator on $L^{2,-n-\epsilon/2}(S^{d-1})$ and 
$(1-\Delta_\omega)^{\epsilon/4}$ is a unitary element of  
${\mathcal  B}(L^{2,-n}(S^{d-1}),L^{2,-n-\epsilon/2}(S^{d-1}))$.   
We invoke the general fact that the product of
   a strongly continuous operator-valued
 function   and a compact operator is norm
  continuous. Whence we obtain the norm continuity
of  $J_{n,\epsilon}(\lambda)$    in ${\mathcal
  B}(L^{2,-n}(S^{d-1}),L^2 (\R^d))$. 

The proof of the norm continuity of the operators in the remaining parts of
the lemma is similar.\qed

\noindent{\em Outline of a proof of Theorems \ref{thm:4a} and
  \ref{thm:4}.} The proof goes along the lines of the proof of Theorem
\ref{thm:1}. In particular this amounts to 
 inserting  the right hand side of (\ref{eq:partitionleft}) to the right
of the resolvent in (\ref{wave22}) and expanding into three terms. Next,
using Proposition  \ref{prop:resolvent_basic2} and Lemma \ref{lem:22a},
we conclude that $W^+(\lambda)$ is
well-defined as a sum of four  operators, say   $T_j(\lambda)$.  In
fact, all  of the four maps
\begin{equation*}
[0,\infty[\ni\lambda\to (\langle x \rangle g)^{-n} \langle x \rangle
^{-\frac12 -\epsilon} g^{1\over 2}T_j(\lambda)\in  {\mathcal
B}(L^{2,-n }(S^{d-1}),L^{2}(\R^d))
\end{equation*} are continuous.

For the independence of $W^+(\lambda)$ of cutoffs, we use Propositions 
 \ref{prop:proposi1}  and \ref{prop:propa5} in the same way as in the
arguments for deducing (\ref{eq:7}) stated below.

The formula \eqref{eq:delt2} follows by combining \eqref{eq:7usef},
Remark \ref{remarks:nosmmmmo} \ref{it:csmoottaa}) and  the shown
continuity properties of $W^+(\lambda)$ and $W^+(\lambda)^*$.
\qed

\begin{lemma} \label{prop:qs4}  For 
 any $\lambda\geq0$,  $ R(\lambda{\pm}\i0)
T^{\pm}(\lambda)$ is well-defined as a map from  ${\mathcal
  D}'(S^{d-1})$ to $L^{2,-\infty}$ and
  \begin{equation}
      \label{eq:7}
 0=J^{\pm}(\lambda)+\i R(\lambda{\pm}\i0)
T^{\pm}(\lambda).
    \end{equation}
\end{lemma}

\proof Note that we can extend Lemma \ref{lem:22a} as follows: Let $\chi_-\in
C_c^\infty(\R)$ and $\tilde\chi_-\in C_c^\infty(\R)$ with
$\supp\tilde\chi_-\subseteq]-\infty,2\bar\sigma-1[$ for some small
$\bar\sigma>0$. Then,
for all $m,n\in \R$, 
\begin{equation*}
  \label{eq:x_weights29d9}
 \Opr(\chi_-(a)\tilde\chi_-(b))
 T_{\bd}^+(\lambda),\;
 \Opr(\chi_-(a)\tilde\chi_-(b))J^+(\lambda)\in {\mathcal
 B}(L^{2,-n}(S^{d-1}),L^{2,m}(\R^d)),
 \end{equation*}
    cf. (\ref{eq:F_angle5}) (recall the standing hypothesis of this subsection that the
positive parameter $ \sigma'$   in
 (\ref{eq:chi^2}) is sufficiently  small).

Therefore, for all $\tau\in{\mathcal D'}(S^{d-1})$ and $s\in \R$,
\begin{equation}
\big( WF_\sc^s(T^+(\lambda)\tau)\cup WF_\sc^s(J^+(\lambda)\tau)\big )\cap\{b<\bar\sigma-1\}=\emptyset.\label{qs1}
\end{equation} 

By the definition of $T^+(\lambda)$, 
\begin{equation}(H-\lambda)J^+(\lambda)\tau=-\i T^+(\lambda)\tau=-\i(H-\lambda)R(\lambda{+}\i0)T^+(\lambda)\tau.\label{qs2}\end{equation}  
Notice that due to (\ref{qs1}) and Proposition \ref{prop:proposi1} (\ref{item:r3}),
 the vector $u=R(\lambda{+}\i0)T^+(\lambda)\tau$ is in fact  well-defined
 and  
\begin{equation}
 WF_\sc^s(u)\cap\{b<\bar\sigma-1\}=\emptyset.\label{qs3}
\end{equation}  

Using (\ref{qs1})--(\ref{qs3}) and
Proposition \ref{prop:propa5}, we conclude that the generalized eigenfunction
satisfies 
\begin{equation}
J^+(\lambda)\tau+\i R(\lambda+\i0) T^+(\lambda)\tau=0. 
\end{equation}
\qed

\begin{remark*}
 There exists an alternative time-dependent proof of 
Lemma \ref{prop:qs4} that avoids the use of Proposition
\ref{prop:propa5}:
 Due to  (\ref{eq:ooo})
\begin{equation*}
0=\lim_{\epsilon\searrow0}\int(J^{\pm}+\i
  R(\lambda\pm\i\epsilon)T^\pm)\delta_\epsilon(\lambda) f\d \lambda, \ \ \hat f\in
  C_c(\R^d\backslash\{0\}),\end{equation*} cf. Lemma
\ref{lemma:t12} or Appendix \ref{Appendix}. The right hand is given by 
\begin{equation*}
\int(J^{\pm}+\i
  R(\lambda\pm\i 0)T^\pm)\delta_0(\lambda) f\d \lambda,\end{equation*}
cf. Appendix \ref{Appendix}.
Whence, by 
a density argument, \eqref{eq:7} follows. \qed
\end{remark*}

We complete this subsection by discussing a certain refined mapping
property of $W^\pm(\lambda)$. Besides its own interest its  application
(see Corollary
\ref{cor:extendmatrices-at} stated below) will be needed in Section
\ref{Generalized eigenfunctions}. The result 
 is related to the fact that the 
continuity in $\lambda$ of the operators in (\ref{eq:x_weights29bc}) and
(\ref{eq:x_weights29bc9}) is  proven only for $n\geq0$ while 
the continuity in $\lambda$ of the operator  in
(\ref{eq:x_weights29c}) is valid for all $n\in\R$.

 \begin{thm}\label{thm:middl} Fix real-valued  $\chi, \tilde\chi_-\in
   C_c^\infty(\R)$ and $\chi_+\in
   C^\infty(\R)$ such that  $\supp \tilde\chi_-\subset]-1,1[$, $\chi_+'\in
   C^\infty_c(\R)$  and $\supp
   \chi_+ \subset]C_0,\infty[$ . Let
   $\tilde A:=\Opw(\chi(a)\tilde\chi_-(b))$ and $A_+:=\Opw(\chi_+(a))$ for
   $\lambda\geq0$. 
   For   all $n\in \R$, $\epsilon>0$  and  with $A= \tilde A$ or $A= A_+$,  
\begin{equation}\label{eq:wwiii}
 W^\pm_{n,\epsilon}(\lambda):=(\langle x \rangle g)^{-n} \langle x \rangle^{-\frac12 -\epsilon}
  g^{1\over 2}AW^\pm(\lambda)
\end{equation}
 is a continuous  ${\mathcal
 B}(L^{2,-n}(S^{d-1}),L^{2}(\R^d))$--valued function in
$\lambda\in[0,\infty[$.  
 \end{thm}
 \begin{proof} With reference (\ref{eq:sym_class}) (this  class of
   symbols is used
   extensively in \cite{FS})
   \begin{equation*}
    B(\lambda):=(\langle x \rangle g)^{-n} \langle x \rangle^{-\frac12 -\epsilon}
  g^{1\over 2}Ag^{-{1\over 2}}\langle x \rangle^{\frac12
    +\frac\epsilon 2}(\langle x \rangle g)^{n} 
  \in \Psi_{\unif}(\langle x \rangle^{-\frac\epsilon 2}, g_{\mu,\lambda}).
   \end{equation*}
 Whence, by the calculus, $B(\lambda)\in {\mathcal B}(L^2(\R^d))$ with a bound
 locally independent of $\lambda\geq0$, and  in fact $B(\cdot)$ is norm
 continuous. By using  this continuity and Theorem \ref{thm:4}, we
 conclude that it
 suffices to consider the case  $n<0$.

\noindent{\bf Re  $A= \tilde A$.}  
Since the construction of
$W^+(\lambda)$ is independent of the (small) parameters
$\sigma$ and $\sigma'$ in   (\ref{eq:chi^2}), we  can take them smaller
(if needed) to assure that
\begin{equation}
  \label{eq:supprtj}
 \sup \supp \chi_-<1-3C\sigma'. 
\end{equation}
Here we
refer to the left hand side of (\ref{eq:F_angle5}). 

Now, to show that $W^+_{n,\epsilon}(\lambda)$ is an element of
${\mathcal B}(L^{2,-n}(S^{d-1}),L^{2}(\R^d))$, we consider for
$\lambda>0$ the two terms of (\ref{wave22}) separately (if $\lambda=0$
we use instead \eqref{wave2}): The contribution from the first term
(i.e. from $J^+(\lambda)$) has better mapping properties than
specified, cf. Lemma \ref{lem:22a} (\ref{it:p40}). In fact, using
(\ref{eq:supprtj}) we can mimic the proof of Lemma \ref{lem:22a}
(\ref{it:p40}) to handle this contribution. As for the contribution
from the second term (i.e. from $\i R(\lambda- \i 0)T^+(\lambda)$), we
combine Lemma \ref{lem:22a} (\ref{it:p30}) and (\ref{it:p50}) and
Proposition \ref{prop:resolvent_basic2} (\ref{some_label21}).

By the same arguments, continuity in $\lambda\geq 0$ is valid for the
contribution from each of the mentioned two terms,
hence for $W^+_{n,\epsilon}(\lambda)$.

\noindent{\bf Re  $A= A_+$.}  Again we consider for
$\lambda>0$ the two terms of (\ref{wave22}) separately (if $\lambda=0$
we use instead \eqref{wave2}). The contribution from the first term
$J^+(\lambda)$  has again better mapping properties than
needed. More precisely,  we have the following analogue of   Lemma
\ref{lem:22a} (\ref{it:p40}):

For all $m\in \R$ the family of operators $\langle x \rangle^m  A_+J^+(\lambda)$
constitutes  a continuous ${\mathcal
 B}(L^{2,-n}(S^{d-1}),L^{2}(\R^d))$--valued function of
$\lambda\in[0,\infty[$.

To show this, we can again follow the proof of Lemma \ref{lem:22a}
(\ref{it:p40}). It suffices to show locally uniform boundedness in the
indicated topology and we may replace  $A_+\to \Opr(\chi_+(a))$. 
The kernel of $\Opr(\chi_+(a))J^+(\lambda)$ is given by the integral
\begin{align*}
  \int \d \xi \e^{\i x\cdot \xi}&\int \e^{\i (\phi(y,\omega)-y\cdot
  \xi)}k_{\omega,\lambda}(y,\xi)\d
y,\\&k_{\omega,\lambda}(y,\xi)=(2\pi)^{-3d/2}\chi_+(\xi^2/g\big
(|y|,\lambda)^2\big )\tilde a^+(y,\omega,\lambda).
\end{align*} It suffices to show that
\begin{equation}\label{eq:inrediiii}
  |\partial_\xi^\beta\partial_\omega^\delta\int \e^{\i (\phi(y,\omega)-y\cdot
  \xi)}k_{\omega,\lambda}(y,\xi)\d
  y|\leq C_{\beta,\delta}\japxi^{-d-1}
\text{  uniformly in }\xi, \omega\text{ and  }\lambda. 
\end{equation} For that we notice that 
\[k=k_{\omega,\lambda}\in S_{\unif}(g^{\frac d2-1},
g_{\mu,\lambda}).\]
  It suffices to show \eqref{eq:inrediiii} with  $k\to k=F(|y|>2\bar
  R)k_{\omega,\lambda}(y,\xi)$. 

  Next we
integrate by parts, writing  first
\begin{equation*}
 \Big(\i \frac {\xi-\nabla_y\phi  }{|\xi-\nabla_y\phi|^2}\cdot \nabla _y
  \Big)^\ell \e^{\i (\phi(y,\omega)-y\cdot
  \xi)}=\e^{\i (\phi(y,\omega)-y\cdot
  \xi)},
\end{equation*}  and then we invoking the uniform bounds 
\begin{equation}
  \label{eq:basicestiiii}
  C|\xi| \geq  \big |  \xi -\nabla_y\phi(y,\omega)\big | \geq c \big (|\xi| + \big |\nabla_y\phi(y,\omega)\big | \big ),
\end{equation}  which are valid on the support of $k$ (provided $\bar
R$ is chosen sufficiently large). Clearly, we obtain
\eqref{eq:inrediiii} by this procedure if $\ell $ (i.e. the number of integrations by
parts) is chosen
sufficiently large.

As for the contribution
from the second term $\i R(\lambda- \i 0)T^+(\lambda)$, we
combine Lemma \ref{lem:22a} (\ref{it:p30}) and (\ref{it:p50}) and
Proposition \ref{prop:resolvent_basic2} (\ref{item:C-02}).
\end{proof}

We can extend the identities 
(\ref{eq:7usef}) and \eqref{eq:delt2} (which below corresponds to $s=0$) as
follows: 
\begin{cor}Let $\chi, \chi_-\in
   C_c^\infty(\R)$ be given as in Theorem \ref{thm:middl}.  Fix
   $\lambda\geq0$.  Let again 
   $\tilde A:=\Opw(\chi(a)\chi_-(b))$. For   all $\delta>\frac12$  and
   $s\leq0$,  there exists the strong limit
\begin{equation}
  \label{eq:7usef111}
  {\rm
    s}-\lim_{\epsilon\searrow0}\,g^{\frac12}\delta_\epsilon^V(\lambda)\tilde
  Ag^{\frac12}=g^{\frac12}\delta^V(\lambda)\tilde Ag^{\frac12}=g^{\frac12}W^\pm(\lambda)W^\pm(\lambda)^*\tilde
  Ag^{\frac12}\end{equation}
in ${\mathcal
  B}(L^{2,s+\delta}(\R^d),L^{2,s-\delta}(\R^d))$.
\label{cor:extendmatrices-at}
\end{cor}
\begin{proof} 
  It follows from Proposition \ref{prop:proposibb} that indeed there
  exists the limit
\begin{equation*}
  B:={\rm
    s}-\lim_{\epsilon\searrow0}\,g^{\frac12}\delta_\epsilon^V(\lambda)\tilde
  Ag^{\frac12}\text{
  in }{\mathcal
  B}(L^{2,s+\delta}(\R^d),L^{2,s-\delta}(\R^d)).
\end{equation*}

Let $n=s/s_1$, where $s_1$ is  given as in (\ref{eq:s_0}). 
 Due to Theorem \ref{thm:middl}, \[W^\pm(\lambda)^*\tilde
 Ag^{\frac12}=\big (g^{\frac12}\tilde AW^\pm(\lambda)\big )^*\in {\mathcal
 B}(L^{2,s+\delta}(\R^d),L^{2,n}(S^{d-1})),\] 
 and, due to Theorem,
\ref{thm:4} 
\[g^{\frac12}W^\pm(\lambda)\in {\mathcal
 B}(L^{2,n}(S^{d-1}),L^{2,s-\delta}(\R^d)).\]

We have  shown that \[{g^{\frac12}W^\pm(\lambda)W^\pm(\lambda)^*\tilde
  Ag^{\frac12}\in \mathcal
  B}(L^{2,s+\delta}(\R^d),L^{2,s-\delta}(\R^d)).\]

Since \[Bv=g^{\frac12}W^\pm(\lambda)W^\pm(\lambda)^*\tilde Ag^{\frac12}v\text{
for }v\in L^{2,\infty},\]
 cf. \eqref{eq:7usef} and \eqref{eq:delt2}, we are done by a density
 argument.
\end{proof} 


\subsection{Asymptotics of  short-range wave matrices}
\label{Asymptotics of the short-range wave matrices}
 Clearly, if $\mu>1$, there exists
\begin{equation}\label{short}
  W^{\pm}_{\sr}f=\lim _{t\to {\pm}\infty}\e^{\i tH}\e^{-\i tH_0}f,
\end{equation} 
which is the usual definition of wave operators in the short-range case.
In the case $\mu \in ]1,2[$, we can compare  our wave matrices with the wave
matrices defined by \eqref{short}.

 Recall $\hat p:=p/|p|$.
\begin {thm} \label{thm:shortr1} For $\mu \in ]1,2[$, the operators 
  \begin{align*}
   &\psi_\sr^{+}(p):= \i \int_{R_0}^\infty (|p| -F^+(l\hat p,p,p^2/2)\cdot
   \hat p)\;\d l,\\
&\psi_\sr^{-}(p):=-\i \int_{R_0}^\infty (|p|+F^+(-l\hat p,-p,p^2/2)\cdot \hat p)\;\d l
\end{align*}
are well-defined. 
If $V_2=0$, then $\psi_\sr^\pm(p)=\psi_\sr^\pm(|p|)$ with 
 \begin{equation*}
\psi_\sr^\pm(|p|)
={\pm \i\int_{R_0}^\infty \big(|p|-\sqrt{p^2-2V_1(r))}\big)\;\d r}. 
  \end{equation*} 

We have 
\begin{subequations}
 \begin{align}
    \label{eq:Uni3}
   &W^+_{\sr}=W^+ \e^{\i\psi_\sr^+(p)},\\\label{eq:Uni4}
&W^-_{\sr}=W^-\e^{\i\psi_\sr^-(p)}.
\end{align} 
\end{subequations}
  Whence in particular,  for all $\lambda>0$,
  \begin{subequations}
  \begin{align}
    \label{eq:Uni3bb}
   &W^+_{\sr}(\lambda)=W^+ (\lambda)\e^{\i\psi_\sr^+(\sqrt{2\lambda}\cdot)},\\\label{eq:Uni4bb}
&W^-_{\sr}(\lambda)=W^-(\lambda)\e^{\i\psi_\sr^-(\sqrt{2\lambda}\cdot)}.
\end{align}  
  \end{subequations}

\end {thm} 
\begin{proof} One can readily
   show the theorem from well-known properties  of the free
   evolution and the fact that 
   \begin{equation}
     \label{eq:ph+}
     \phi^+(x,\omega,\lambda)+\int_{R_0}^\infty (\sqrt
   {2\lambda}-F^+(l\omega,\omega,\lambda)\cdot
   \omega)\;\d l=\sqrt{2\lambda}\omega\cdot x+o(|x|^0),
   \end{equation} which in turn follows from  \cite[(4.50)]{DS1}  and a
   change a contour of integration. The asymptotics is locally
   uniform in $(\omega,\lambda)\in S^{d-1}\times ]0,\infty[$.
 \end {proof}

\begin{remark} \label{remark:oscill}
   $\psi_\sr^\pm $  is indeed oscillatory. Notice that for
   $V_1(r)= -\gamma r^{-\mu}$,  as $\lambda \to 0^+$, we have
   \begin{eqnarray*}
\psi_\sr^+(\sqrt{2\lambda})&=&\int_{R_0}^\infty \Big(\sqrt
   {2\lambda}-\sqrt{2(\lambda+\gamma r^{-\mu})}\Big)\;\d r \\
&=&(2\lambda)^{\frac12-\frac1\mu}\int_{R_0(2\lambda)^{\frac1\mu}}^\infty
(1-\sqrt{1+2\gamma
   s^{-\mu}} )\d s\\
&=&(2\lambda)^{\frac12-\frac1\mu}\int_0^\infty
(1-\sqrt{1+2\gamma
   s^{-\mu}} )\d s+O\left(\lambda^0 \right),
\end{eqnarray*}
cf. \cite[(7.11)]{Y1}. See Remark
 \ref{remark:oscill2} for a similar  result. 
\end {remark}

\subsection{Asymptotics of  Dollard-type wave matrices}
\label{Asymptotics of the Dollard-type  scattering matrix}
For $\mu >\frac12$ and $\mu+\epsilon_2>1$,
 the Dollard-type wave operators are given by
\begin{align*}
  &W^{\pm}_{\dol}f=\lim _{t\to {\pm}\infty}\e^{\i tH}U_{\dol}(t)f,
\end{align*} where 
\begin{equation*}
 U_{\dol}(t)=\e^{-\i \int_{0}^t(p^2/2+V_1(sp)1_{\{|sp|\geq R_0\}})\;\d
 s}. 
\end{equation*}

We have the following analogue of Theorem \ref{thm:shortr1}.
\begin {thm} \label{thm:dol1} For $\frac12<\mu<2$, $\epsilon_2 <1$
  and $\mu+\epsilon_2>1$, the operators 
   \begin{align*} 
   &\psi_\dol^{+}(p)= \i \int_{R_0}^\infty (|p|-F^+(l\hat p,\hat p,p^2/2)\cdot
   \hat p -|p|^{-1}V_1(l))\;\d l,\\
&\psi_\dol^{-}(p)= -\i \int_{R_0}^\infty (|p|+F^+(-l\hat p,-\hat p,p^2/2)\cdot
   \hat p-|p|^{-1}V_1(l))\;\d l
\end{align*}
are well-defined. If $V_2=0$, then $\psi_\dol^\pm(p)=\psi_\dol^\pm(|p|)$ and
 \begin{equation*}
 \psi_{\dol}^\pm(|p|)=\pm\i \int_{R_0}^\infty (|p|-\sqrt{p^2-2V_1(r)}-|p|^{-1}
V_1(r))\;\d r.
  \end{equation*}

  We have
  \begin{subequations}
    \begin{align}
    \label{eq:Uni3a}
   &W^+_{\dol}=W^+ \e^{\i\psi_\dol^+(p)},\\
&W^-_{\dol}=W^-\e^{\i\psi_\dol^-(p)}.
\end{align} 
  \end{subequations}
 Whence in particular,  for all $\lambda>0$
 \begin{subequations}
   \begin{align}
    \label{eq:Uni3bba}
   &W^+_{\dol}(\lambda)=W^+ (\lambda)\e^{\i\psi_\dol^+(\sqrt{2\lambda}\cdot)},\\\label{eq:Uni4bba}
&W^-_{\dol}(\lambda)=W^-(\lambda)\e^{\i\psi_\dol^-(\sqrt{2\lambda}\cdot)}.
\end{align}
 \end{subequations}

\end {thm} 
\begin{proof} First we notice that $\psi_\dol^\pm$
 are well-defined  due to the
   fact that
   \[(F^+-F_\sph^+)(l\omega,\omega,\lambda)=O(l^{-\delta})\]
 for any $\delta<\min(\mu+\epsilon_2,2\mu)$,
   and hence integrable. Here $F_\sph^+$ refers to the $F^+$ for the case
   $V_2=0$, whence $F_\sph^+
(l\omega,\omega,\lambda)=g(l,\lambda)\omega$. For this estimate,
   we refer to \cite[Remarks 6.2 2)]{DS1} and the proof of \cite[Lemma
   6.4]{DS1}.  It appears
   stronger at the price of  not being 
   uniform in (small) $\lambda$. There is an extension of this
   estimate that allows us to integrate along  the line segment joining
     $x$ and $R\omega$ and taking the limit:
   \begin{align}\label{eq:ph+233}
     \int_x^{\infty \omega}
     (F^+-F_\sph^+)(\bar x,\omega,\lambda)\cdot d\bar x &=\lim_{R\to \infty} \int_x^{R\omega}
     (F^+-F_\sph^+)(\bar x,\omega,\lambda)\cdot \d \bar x \nonumber\\&=o(|x|^0).
   \end{align}
Introduce the auxillary phases
   \begin{align*}
     &\phi_{\dol}^{\pm}(x,\omega,\lambda)=\sqrt
   {2\lambda}x\cdot \omega\mp
   (2\lambda)^{-\tfrac{1}{2}}\int_{R_0}^{\pm x\cdot \omega}
   V_1(l)\;\d l,\\
&\phi_{\aux }^{\pm}(x,\omega,\lambda)=\phi_{\aux}^{\pm}=\phi_{\dol}^{\pm}-\int_x^{\pm \infty \omega}
     (F_\sph^+-\nabla _x\phi_{\dol}^{\pm})\cdot \d \bar x, 
\end{align*}
and  corresponding modifiers
\begin{equation*}
  (J_{\sharp}^{\pm}f)(x)=(2\pi)^{-d/2}\int
   \e^{\i \phi_{\sharp}^{\pm}(x,\xi)}\chi(x,\pm \hat \xi)\hat f(\xi) \d \xi;\;\xi=\sqrt
   {2\lambda}\omega.
\end{equation*} Here we can take the function $\chi$ of the form $\chi(x,\omega)=\chi_1(|x|/R)\chi_2(\hat x\cdot
\omega)$ with $\chi_1$ and $\chi_2$  given as in
   (\ref{eq:chi^1}) and (\ref{eq:chi^2}), respectively.

By the stationary phase method, \cite[Theorem 7.7.6]{Ho1},  one derives  the following asymptotics in
$L^2(\R^d)$ for any state $f$ with $\hat f\in
C^\infty_c(\R^d\setminus \{0\})$: 
\begin{equation*}
  U_{\dol}(t)f\asymp J_{\dol}^{\pm}\e^{\i tH_0}f\asymp
  J_{\aux}^{\pm}\e^{\i tH_0}f \text{ as } t\to \pm \infty.
\end{equation*} 

Next we notice the following analogue of  \eqref{eq:ph+},
cf. \eqref{eq:ph+233}: 
   \begin{align*}
     \phi^{\pm}(x,\omega,\lambda)&+\int_{R_0}^\infty 
     (\nabla \phi_{\dol}^+-F^+)({\pm}l\omega,{\pm}\omega,\lambda)\cdot
     \omega\;\d l\\&=\phi_{\aux}^{\pm}(x,\omega,\lambda)+o(|x|^0).
   \end{align*} Again this asymptotics is locally
   uniform in $(\omega,\lambda)\in S^{d-1}\times ]0,\infty[$.
\end {proof}
\begin{remark} 

\label{remark:oscill2}The first factor on the right hand side of
  \eqref{eq:relSV=02} is  oscillatory. Let us state the following
   asymptotics for the special case  where 
   $V_1(r)= -\gamma r^{-\mu}$ for $r\geq R_0$:
\begin{eqnarray*}
\psi_\dol^+(\sqrt{2\lambda})
&=&\int_{R_0}^\infty \Big(\sqrt
   {2\lambda}-\sqrt{2(\lambda+\gamma r^{-\mu})}+(2\lambda)^{-\frac12}\gamma
   r^{-\mu} \Big)\;\d r \\
&=&(2\lambda)^{\frac12-\frac1\mu}\int_{R_0(2\lambda)^{\frac1\mu}}
(1-\sqrt{1+2\gamma
   s^{-\mu}} +\gamma s^{-\mu})\d s.
\end{eqnarray*}
For $\lambda\searrow0$, this behaves as
\[\begin{array}{rll}
&(2\lambda)^{\frac12-\frac1\mu}C_\mu
+O\left(\lambda^{-\frac12}\right),&\frac12
    <\mu<1;\\[2.5ex]
&-\gamma(2\lambda)^{-\frac12}\ln 2\lambda +(2\lambda)^{-\frac12}C_1
+O\left(1\right),&\mu=1;\\[2.5ex]
&(2\lambda)^{-\frac12}\frac{R_0^{1-\mu}\gamma}{\mu-1}
+O\left(\lambda^{\frac12-\frac1\mu}\right),&1<\mu<2.
\end{array}\]
Here\begin{eqnarray*}
C_\mu&:=&
\int_0^\infty
(1-\sqrt{1+2\gamma
   s^{-\mu}} +\gamma s^{-\mu})\d
    s,\\
C_1&:=&\int_1^\infty
(1-\sqrt{1+2\gamma
   s^{-1}} +\gamma s^{-1})\d s+\int_0^1
(1-\sqrt{1+2\gamma
   s^{-1}})\d s
-\gamma\ln R_0.\end{eqnarray*}
\end {remark}

\section{Scattering matrices}
\label{Scattering matrices}

In this section we study (modified) scattering matrices. We prove that they
have a limit at zero energy. This implies low energy 
oscillatory asymptotics for the
standard short-range and Dollard scattering matrices.

\subsection{Scattering matrices at positive  energies}
\label{Scattering matrices at positive  energies}

The scattering operator commutes with $H_0$, which is diagonalized by the
direct integral decomposition (\ref{decom1}).
Because of
 that, the general theory of decomposable operators says that there exists a
measurable family $]0,\infty[\ni\lambda\mapsto S(\lambda)$, with $S(\lambda)$
    unitary operators on  $L^2(S^{d-1})$ defined for almost all $\lambda$, such
    that 
\begin{equation}
  S\simeq\int_0^\infty \oplus S(\lambda)\;\d \lambda,\label{decom2}
\end{equation} using the decomposition (\ref{decom1}).

The following theorem is (essentially) well-known: 

\begin{thm}\label{thm-s} Assume Condition  \ref{symbol}.
  Then
  \begin{subequations}\begin{eqnarray}
  \label{eq:Smatrix}
 S(\lambda)
&= &-2\pi J^+(\lambda)^{*}T^-(\lambda)
   +2\pi \i T^{+}(\lambda)^*R(\lambda+\i 0)T^-(\lambda)\\
&=&-2\pi W^{+}(\lambda)^*T^-(\lambda)\label{eq:Sgod} 
\end{eqnarray}  
  \end{subequations}
 defines a unitary operator on $L^2(S^{d-1})$ depending
strongly continuously on $\lambda>0$. Moreover, (\ref{decom2}) is true.
Furthermore, for all $n\in \R$ and $\epsilon >0$,  \[S(\lambda)\in {\mathcal
  B}(L^{2,n}(S^{d-1}),
L^{2,n-\epsilon}(S^{d-1})),\]  depending norm continuously on $\lambda>0$.
(Hence in particular $S(\lambda)$ maps $C^\infty(S^{d-1})$ into itself.)
\end{thm}

For a derivation
of the formula (\ref{eq:Smatrix}) we refer the reader to Appendix
\ref{Appendix}.  For the remaining part of the theorem we
refer  the reader to the proof of Theorem \ref{thm:S(0)1} stated
below (one can 
use Theorem \ref{thm:1} and  Lemma \ref{lem:2a} as substitutes for
Theorem \ref{thm:4} and Lemma \ref{lem:22a}, respectively).

\subsection{Scattering matrices at low energies}
\label{Asymptotics of scattering matrix at low energies}

Until the end of this section we assume that Conditions
\ref{assump:conditions1}--\ref{assump:conditions3} are true. The main  new result of this section 
is the following theorem:

\begin {thm} \label{thm:S(0)1}
 The  result of Theorem \ref{thm-s} is true for all
 $\lambda\in[0,\infty[$. Specifically, if we define
\begin{subequations}
\begin{eqnarray} \label{eq:bndSS} S(0)&=& -2\pi J^+(0)^{*}T^-(0)
 +2\pi \i T^{+}(0)^*R(+\i 0)T^-(0)\\
&=&-2\pi W^{+}(0)^*T^-(0)\label{eq:Sgod2} ,\end{eqnarray}
\end{subequations}
  then $S(0)$ is unitary, $\slim_{\lambda
  \searrow0}S(\lambda)=S(0)$ in the sense of  ${\mathcal
  B}(L^{2}(S^{d-1}))$
and
 $\lim_{\lambda
  \searrow0}S(\lambda)=S(0)$ in the sense of
 ${\mathcal
  B}(L^{2,n}(S^{d-1}), L^{2,n-\epsilon}(S^{d-1}))$ for any $n\in \R$ and $\epsilon >0$.\end{thm}

\begin{proof}  First we notice that the expression
 \[S(\lambda)=-2\pi W^{+}(\lambda)^*T^-(\lambda)\in {\mathcal
  B}(L^{2,n}(S^{d-1}), L^{2,n-\epsilon}(S^{d-1}))\text{ for }n>0\]
has a norm continuous dependence of $\lambda\geq0$.
 Indeed, fix $n>0$ and $\epsilon\in ]0,n]$, and pick 
$\epsilon_1,\epsilon_2\in \R$ such that 
 $\epsilon\frac{\mu}{2}<\epsilon_1<\epsilon$ and
 $\epsilon_2=\frac12(\epsilon-\epsilon_1)$.  We write
\begin{align}
&W^+(\lambda)^*T^-(\lambda)\label{eq:qs5} \\&=
\left(
W^+(\lambda)^*g^{\frac12}\langle x\rangle^{-\frac12-\epsilon_2}
(\langle x\rangle g)^{-n+\epsilon}\right)\left( g^{-\epsilon}
 \langle x\rangle^{-\epsilon_1 }\right)\left((\langle x\rangle g)^n
\langle x\rangle^{\frac12-\epsilon_2}
g^{-\frac12}T^-(\lambda)\right).\nonumber
\end{align}
We shall use the analogues of Lemma \ref{lem:22a} (\ref{it:p30}) and
 (\ref{it:p50}) with 
 $T^+(\lambda)$ replaced by  $ T^-(\lambda)$ (proved in the same
 way).
The third factor on 
the right of \eqref{eq:qs5} is continuous in $\lambda$
with values in ${\mathcal
  B}(L^{2,n}(S^{d-1}), L^2(\R^d))$. The second factor is continuous in
$\lambda$
 as an
operator on $L^2(\R^d)$.         The first factor is continuous in $\lambda$
as an operator in
 ${\mathcal
  B}( L^2(\R^d),L^{2,n-\epsilon}(S^{d-1}) )$ due to Theorem \ref{thm:4}.
This proves the norm continuity of $ S(\lambda)$ in  ${\mathcal
  B}(L^{2,n}(S^{d-1}), L^{2,n-\epsilon}(S^{d-1}))$ for $n>0$.

Let us prove the same property for $n\leq0$ using a slight extension
of the above scheme: Notice that the positive sign  condition above  entered
only in the
 condition $n-\epsilon\geq0$ needed for applying Theorem
\ref{thm:4}. Since $n\leq0$ we  have $n-\epsilon<0$ and  therefore we need
a substitute for  Theorem
\ref{thm:4}. This is provided by Theorem \ref{thm:middl} and an  analogue of Lemma \ref{lem:22a} for $ T^-(\lambda)$. In fact, choose  for (small) $\bar \sigma>0$ real-valued  $\tilde\chi_-\in
   C_c^\infty(\R)$ and $\chi_+\in
   C^\infty(\R)$ such that  $\supp \tilde\chi_-\subset]-1,1[$, $
   \tilde\chi_-=1$ in $[\bar \sigma-1,1-\bar \sigma]$, $\supp
   \chi_+ \subset]C_0,\infty[$ and $\chi_+ =1$ in $[2C_0,\infty[$. Let
   $\chi=1-\chi_+$, 
   $\tilde A=\Opw(\chi(a)\tilde\chi_-(b))$,
   $A_+=\Opw(\chi_+(a))$ and $\bar A=\Opw(\chi(a)(1-\tilde\chi_-(b)))$.
We insert the identity $I=\tilde A+ A_++\bar
A$
\begin{equation}
  \label{eq:ide2}
  W^+(\lambda)^*T^-(\lambda)=((\tilde A+
A_+)W^+(\lambda))^*T^-(\lambda)+W^+(\lambda)^*(\bar AT^-(\lambda)).
\end{equation}
Due to Theorem \ref{thm:middl} the above argument can be repeated for
the first term on the right hand side, and if $\bar \sigma>0$ is chosen sufficiently small
we have the following analogue of Lemma \ref{lem:22a} (\ref{it:p40}) and
 (\ref{it:p50}) (here stated in combination): For all $m\in \R$ the family of operators $\langle x
 \rangle^m  \bar AT^-(\lambda)$
constitutes  a continuous ${\mathcal
 B}(L^{2,-n}(S^{d-1}),L^{2}(\R^d))$--valued function of
$\lambda\in[0,\infty[$. By choosing $m>\frac12+\frac\mu4$ and using  
   Theorem
\ref{thm:4} we conclude   norm continuity of the second term of \eqref{eq:ide2}.

But from the isometricity of $S$ we see that $S(\lambda)$ is isometric for
almost all $\lambda$ as a map on $L^2(S^{d-1})$.
 Therefore, it is isometric and strongly continuous as a map on $L^2(S^{d-1})$
 for all $\lambda \geq 0$.

By repeating this argument for $S^*$ (not to be elaborated on) we obtain that
 $S(\lambda)^*$ is isometric and strongly continuous in  $\lambda\geq
 0$ as a map on $L^2(S^{d-1})$. Whence  $S(\lambda)$ is unitary as a
 map on  $L^2(S^{d-1})$.\end{proof}

\begin{remark*}
  There is an alternative and completely stationary approach to
  proving the unitarity of the scattering matrices. In fact 
  taking  (\ref{eq:Sgod}) and (\ref{eq:Sgod2}) as  definitions the
  unitarity is a consequence of the formula (\ref{eq:Specaa}), which
  in turn can be verified directly along the lines of Section \ref{Generalized eigenfunctions}.
\end{remark*}
\subsection{Asymptotics of  short-range scattering matrices}
\label{Asymptotics of the short-range scattering matrix}
In the case $\mu \in ]1,2[$ we can compare 
$S(\lambda)$ with the $S$--matrix $S_{\sr}(\lambda)$ defined similarly
\begin{equation*}
  S_{\sr}=W_{\sr}^{+*}W_{\sr}^-\simeq \int_0^\infty \oplus S_{\sr}(\lambda)\;\d \lambda.
\end{equation*} 

Under the condition of radial symmetry Yafaev considered in \cite
{Y1}
the component of $S_{\sr}(\lambda)$ for each sector of fixed
angular momentum. He computed an explicit oscillatory behaviour as
$\lambda\to 0$. The following result is a consequence of Theorem
\ref{thm:shortr1}. In combination with Theorem \ref{thm:S(0)1}, it
yields oscillatory behaviour in a more general situation than
considered in \cite
{Y1}.

\begin {thm} \label{thm:shortr} For $\mu \in ]1,2[$, the operators $S_{\sr}$
   and $S$ are related by 
  \begin{equation}
    \label{eq:relS}
 S_{\sr}=\e^{-\i\psi_\sr^+(p)}
S\e^{\i\psi_\sr^-(p)}.  \end{equation}

In particular, for all $\lambda>0$,
\begin{equation}
    \label{eq:relS2i}
 S_{\sr}(\lambda)=\e^{-\i\psi_\sr^+(\sqrt{2\lambda}\cdot)}
S(\lambda)\e^{\i\psi_\sr^-(\sqrt{2\lambda}\cdot)},  \end{equation}
 and if 
  $V_2=0$  then 
 \begin{equation}
    \label{eq:relSV=0}
 S_{\sr}(\lambda)=e^{-\i 2\int_{R_0}^\infty \big(\sqrt
   {2\lambda}-\sqrt{2(\lambda-V_1(r))}\big)\;\d r}S(\lambda).  
  \end{equation} 
\end {thm}

\subsection{Asymptotics of  Dollard-type scattering matrices}
\label{Asymptotics of the Dollard-type  scattering matrix1}
For $\mu>\frac12$ and $\mu+\epsilon_2>1$, the Dollard-type 
 S--matrix is diagonalized as before: 
\begin{equation*}
  S_{\dol}=W_{\dol}^{+*}W_{\dol}^-\simeq \int_0^\infty \oplus S_{\dol}(\lambda)\;\d \lambda.
\end{equation*} 

We have the following analogue of Theorem \ref{thm:shortr}, cf. Theorem
\ref{thm:dol1}:
\begin {thm} \label{thm:dol} For $\frac12<\mu <2$, $\epsilon_2 <1$
  and $\mu+\epsilon_2>1$, the operators $S_{\dol}$
   and $S$ are related by 
  \begin{equation}
    \label{eq:relS2}
 S_{\dol}=\e^{-\i\psi_\dol^{+}(p)}S\e^{\i\psi_\dol^-
(p)}.  
  \end{equation}

In particular, for all $\lambda>0$,
\begin{equation}
    \label{eq:relS2aa}
 S_{\dol}(\lambda)=\e^{-\i\psi_\dol^{+}(\sqrt{2\lambda}\cdot)}S(\lambda)\e^{\i\psi_\dol^-
(\sqrt{2\lambda}\cdot)},  
  \end{equation} and if $V_2=0$ then
 \begin{equation}
    \label{eq:relSV=02}
 S_{\dol}(\lambda)=\e^{-\i 2\int_{R_0}^\infty \left(\sqrt
   {2\lambda}-\sqrt{2(\lambda-V_1(r)}-(2\lambda)^{-1/2}V_1(r)\right)\;\d
   r}S(\lambda).
  \end{equation} 
\end {thm}

\begin{example}\label{example:coulombsing}
 For the purely Coulombic case $V=-\gamma r^{-1}$ in dimension $d\geq
 3$ one can compute 
 \begin{equation}
   \label{eq:S=P}
   S(0)=\e^{\i c}P,\;c\in \R,
 \end{equation}
 where $(P\tau)(\omega)=\tau(-\omega)$. This formula can be
 verified using \eqref{eq:relSV=02} and  Remark
 \ref{remark:oscill2}, the explicit formula
 \cite[(4.3)]{Y3} for the Coulombic (Dollard) scattering matrix
 (slightly different from our definition),  asymptotics of the Gamma
 function (see for example the reference [3] of \cite{Y3}) and, for
 example, the stationary phase formula \cite[Theorem 7.7.6]{Ho1}
 (alternatively one can use the formula \cite[(3.4)]{Y3}). 

It follows from \eqref{eq:S=P} that the singularities of the kernel 
 $S(0)(\omega,\omega')$ in this particular case are located at $\{(\omega,\omega')\in
 S^{d-1}\times  S^{d-1}| \omega=-\omega'\}$. We devote Section
 \ref{Propagation of
   singularities at zero energy} to an extension of this
 result. In Section
 \ref{Singularity of the kernel of the scattering matrix} we provide a
 different proof of (\ref{eq:S=P}) (up to a compact term); this
 approach yields $c= 4\sqrt{2\gamma R_0}
 -\pi\frac{d-2}{2}$.

We also note  that for the  purely Coulombic case there is in fact a
complete asymptotic expansion $S(\lambda)\asymp \sum_{j=0}^\infty
S_j\lambda^{j/2}$. Here (of course) $S_0$ is given by
\eqref{eq:S=P}, and one can readily check that $S_1\neq 0$. In particular
we see that 
$S(\lambda)$ is not smooth at $\lambda=0$, cf. Remark
\ref{remarks:nosmmmmo} \ref{it:csmoott}). We refer to \cite {BGS} (and
references  cited therein) for explicit expansions of the
generalized purely Coulombic eigenfunctions at zero energy (for $d=3$); those are
also in $\sqrt {\lambda}$.
\end{example}

\section{Generalized eigenfunctions}
\label{Generalized eigenfunctions}

Throughout this section we impose  Conditions
\ref{assump:conditions1}--\ref{assump:conditions3}. 
For any $\lambda\geq0$, we define
\[\mathcal V^{-\infty}(\lambda)=\{u\in
   L^{2,-\infty}|(H-\lambda)u=0\}\subseteq \mathcal S'(\R^d).\]
Elements of $ \mathcal V^{-\infty}(\lambda)$ will be called generalized eigenfunctions
of $H$ at energy $\lambda$.
In this section we study all generalized eigenfunctions of $H$.

\begin{remark*}
Note that by Proposition \ref{prop:enerest}, for  any $u\in
V^{-\infty}(\lambda)$ and $s\in\R$,
\begin{equation}
WF_\sc^s(u)\subseteq\{b^2+\bar c^2=1\}.
\end{equation}
\end{remark*}

\subsection{Representations of generalized eigenfunctions}
\label{Representation of eigenfunctions}

In this subsection we show that all
generalized eigenfunctions can be represented by their incoming or outgoing
data.

\begin{thm}\label{thm-rep}
 For any $\lambda\geq 0$ the map
 \begin{equation*}
   W^\pm (\lambda):\mathcal D'(S^{d-1})\to \mathcal
   V^{-\infty}(\lambda) (\subseteq  L^{2,-\infty})
 \end{equation*} is continuous and bijective.
\end{thm}
\begin{proof}

\noindent{\bf Step  I}. Clearly $W^\pm (\lambda):\mathcal D'(S^{d-1})\to \mathcal V^{-\infty}(\lambda)$ is
 well-defined and continuous, cf. Theorem \ref{thm:4}.

  \noindent{\bf Step  II}. We show that $W^\pm (\lambda)$ is
  onto. Let $u\in \mathcal V^{-\infty}(\lambda)$ be given. Let
  \begin{equation}
  \label{eq:8part}
 \chi^\pm=\chi_-(a)\tilde{\chi}_\pm(b)+\frac 12\chi_+(a), 
\end{equation}
 where 
 ${\chi}_{+}=1-{\chi}_{-}$ is a real-valued function  as in Proposition 
\ref{prop:resolvent_basic2} (\ref{item:C-02}) such that $\chi_+(t)=1$ for
$t\geq2C_0$, and 
$\tilde{\chi}_{-},\tilde{\chi}_{+}\in C^\infty(\R) $ are real-valued
functions obeying 
$\tilde{\chi}_{-}+\tilde{\chi}_{+}=1$ and 
\begin {align}
  \label{eq:supp tilde chi1aaa}
 &\supp \tilde{\chi}_{-}
 \subseteq ]-\infty, 1/2[,\\&\supp \tilde{\chi}_{+} \subseteq
 ]-1/2,\infty[.  \label{eq:supp tilde chi2aa}
\end{align} 
 Now
\begin{eqnarray}
&&\lim_{\epsilon\downarrow0} R(\lambda\pm\i\epsilon)(H-\lambda)\Opr(\chi^\pm)u
\nonumber
\\
&=&\Opr(\chi^\pm)u\pm
\lim_{\epsilon\downarrow0}\i\epsilon R(\lambda\pm\i\epsilon)\Opr(\chi^\pm)u.
\label{argu}\end{eqnarray}
Note that  $\lim_{\epsilon\downarrow0}R(\lambda\pm\i\epsilon)\Opr(\chi^\pm)u$
exists, due to Propositions \ref{prop:enerest}, \ref{prop:proposi}, 
and   \ref{prop:proposibb}. Therefore the second term on the right of \eqref{argu} is
zero. Therefore, we have
\begin{equation}
      \label{eq:rep0a}
 0=\Opr(\chi^\pm)u-R(\lambda{\pm}\i0)
(H-\lambda)\Opr(\chi^\pm)u.
    \end{equation}

Adding the two equations of \eqref{eq:rep0a} yields
\begin{equation*}
  u=2\pi \i \delta^V(\lambda)(H-\lambda)\Opr(\chi^+)u,
\end{equation*} which in turn in conjunction with Proposition
\ref{prop:enerest}, \eqref{eq:7usef},  \eqref{eq:delt2}
 and Corollary \ref{cor:extendmatrices-at}
yields
\begin{equation}\label{eq:rep0ab}
 u=W^\pm(\lambda)\tau,\; \tau=\pm 2\pi\i
 W^\pm(\lambda)^*[H,\Opr(\chi^\pm)]u\in \mathcal D'(S^{d-1}). 
\end{equation}

\noindent{\bf Step  III}. We show that $ W^\pm (\lambda)$ is
injective. For convenience we shall only treat the case of superscript
$+$. By \eqref{eq:rep0ab} we need to show that for all    $\tau    \in \mathcal D'(S^{d-1})$
\begin{equation}
  \label{eq:8aa}
  \tau=2\pi\i
 W^+(\lambda)^*(H-\lambda)\Opr(\chi^+) W^+(\lambda)\tau.
\end{equation} By continuity it suffices to verify (\ref{eq:8aa}) for
$\tau\in C^\infty (S^{d-1})$.
This can be done as follows. Pick non-negative $f\in C^\infty_c(\R)$
with $\int_0^\infty f(s)\d s=1$, and let $F_R(t)=1-\int_0^{t/R}f(s)\d
s;\;R>1.$ We  write the right  hand side of (\ref{eq:8aa}) as
\begin{equation}\label{eq:8aa1}
  \wlim _{R\to \infty}2\pi\i
 W^+(\lambda)^*F_R(\langle x \rangle)(H-\lambda)\Opr(\chi^+)
 W^+(\lambda)\tau
\end{equation} and pull the factor $(H-\lambda)$ to the left. Thus
 (\ref{eq:8aa1}) equals
\begin{equation*}
  \wlim _{R\to \infty}
 2\pi R^{-1}W^+(\lambda)^*
f(\langle x \rangle/R)g\Opr(b\chi^+) W^+(\lambda)\tau.
\end{equation*} If $\lambda\geq0$, we insert (\ref{wave22}) for $W^+(\lambda)$ (if $\lambda=0$, we
use  instead \eqref{wave2}).
By Proposition \ref{prop:resolvent_basic2} (\ref{item:C-02}) and (\ref{some_label21})  and
Lemma \ref{lem:22a} (\ref{it:p30}) and (\ref{it:p50}), we
can  replace each  factor of
$W^+(\lambda)$ by  a  factor of $J^+(\lambda)$, cf. the proof of
Theorem \ref{thm:middl}. Moreover,  we can  replace
the   factor $\Opr(b\chi^+)$ by the operator $g^{-1}\hat x\cdot p$. Therefore, 
 (\ref{eq:8aa1}) becomes
\begin{equation}\label{eq:8aa3}
  \wlim _{R\to \infty}
 2\pi R^{-1}J^+(\lambda)^*
f(\langle x \rangle/R)\hat x
\cdot pJ^+(\lambda)\tau.
\end{equation}
By Theorem \ref{thm:kkk},  (\ref{eq:8aa3}) equals $\tau$.
The identity (\ref{eq:8aa}) follows.
\end{proof}

\begin{remarks*}\begin{enumerate}[\quad\normalfont 1)]
\item  A somewhat similar representation formula has been derived for
  representing positive solutions to  a PDE, see for example
  \cite{Mu}. This involves  the so-called
  Martin boundary. In our case, the  notion analogous to the
 ``Martin boundary'' would
 be $S^{d-1}$.
  \item For $V_3=0$, we have
  \begin{equation*}
    \mathcal V^{-\infty}(\lambda)=\{u\in{\mathcal S}'(\R^d) \ |\
    (H-\lambda)u=0\}, \end{equation*}
 and hence the set  $\mathcal V^{-\infty}(\lambda)$ is closed
  in $ \mathcal S'(\R^d)$  (with respect to  the weak-$*$ topology of 
$\mathcal S'(\R^d)$). Moreover, in this case  $W^\pm
(\lambda)$ maps $\mathcal D'(S^{d-1})$ bicontinuously onto $\mathcal
V^{-\infty}(\lambda)$).

In fact, suppose  $u\in{\mathcal S'(\R^d)}$ obeys  $(H-\lambda)u=0$. Then for some
$m\in \N$ we
have $\langle p\rangle^{-2m}u\in L^{2,-\infty}$. But $(H-\lambda+\i)^{-m}\langle
p\rangle^{2m}$ is bounded on any $L^{2,s}$. Whence, showing that indeed  $u\in \mathcal V^{-\infty}(\lambda)$,
\[\i^{-m} u=(H-\lambda+\i)^{-m}u
=(H-\lambda+\i)^{-m}\langle p\rangle^{2m}\big (\langle p\rangle^{-2m}u\big )\in L^{2,-\infty}.\]

\end {enumerate}
\end{remarks*}

\subsection{Scattering matrices -- an alternative construction}

The construction of scattering matrices given in Subsections 
\ref{Scattering matrices at positive  energies}
and \ref{Asymptotics of scattering matrix at low energies} involved 
a
detailed knowledge of appropriate
  operators, see
the proof of Theorem \ref{thm:S(0)1}. 
However, 
given the theory of wave matrices developed in Subsection 
\ref{Representation of eigenfunctions} and the basic formulas
(\ref{eq:7usef}) and \eqref{eq:delt2} for the spectral resolution, we
could have constructed the scattering matrix more easily.

Recall from Theorem \ref{thm-rep} that $W^\pm(\lambda):{\mathcal
  D}'(S^{d-1})\to L^{2,-\infty}$ is injective. Hence,
$W^\pm(\lambda)^*:L^{2,\infty}\to C^\infty(S^{d-1})$ has a dense range.

For $\tau\in L^2(S^{d-1})$ of the form $\tau=W^-(\lambda)^*v$ with $v\in
L^{2,\infty}$, we define $S(\lambda)\tau:= W^+(\lambda)^*v$. By
(\ref{eq:7usef}) and \eqref{eq:delt2},
we know that
\[\|W^+(\lambda)^*v\|^2=\|W^-(\lambda)^*v\|^2=\langle v,\delta^V(\lambda)
v\rangle.\]
Hence $S(\lambda)$ is indeed well-defined and isometric. But $W^\pm(\lambda)^*L^{2,\infty}$
is dense in $C^\infty(S^{d-1})$, and therefore also in
$L^2(S^{d-1})$. Whence 
$S(\lambda)$ extends to an isometric operator on 
 $L^2(S^{d-1})$. Reversing the role of $+$ and $-$, we obtain that $S(\lambda)$ 
is
actually unitary.
By construction, it satisfies
\begin{equation}\label{eq:Specaa}
  S(\lambda)W^-(\lambda)^*=W^+(\lambda)^*, \;\lambda \geq 0.
\end{equation}

\subsection{Geometric scattering
matrices}
\label{Geometric scattering
matrices}

The  following type of result was  proved for a class of
constant coefficient Hamiltonians  (with no potential) in [AH], and
generalized to Schr\"odinger operators with long-range potentials (for
a class including the one given by Condition \ref{symbol}) at positive
energies by \cite{GY}. It gives a characterization of the space 
$W^\pm(\lambda)L^2(S^{d-1})$, which in turn yields yet another
characterization  of the
scattering matrix $S(\lambda)$. 

Let $s_0=s_0(\lambda)$ be given as in
  (\ref{eq:s_0}),  and introduce in  terms of a dual Besov space 
\[\mathcal{V}^{-s_0}(\lambda):=B_{s_0}^*\cap \mathcal
   V^{-\infty}(\lambda) \] endowed with the topology  of
  $B_{s_0}^*$. 
 The statement (\ref{it:iv}) below 
 is given in terms of the phase function 
$\phi=\phi(x,\lambda)$ of (\ref{eq:newph}).

\begin{thm}\label{thm-rep2}
  \begin{enumerate}[\normalfont (i)]
\item \label{it:i}
  For all $\tau\in L^2(S^{d-1})$, \[WF^{-s_0}_{\sc}(W^{\pm}(\lambda)\tau)\subseteq
   \{b=-1\}\cup \{b=1\}.\]
\item \label{it:ii}
 The operator $W^{\pm}(\lambda)$ maps $L^2(S^{d-1})$
  bijectively  and bicontinuously  onto $\mathcal{V}^{-s_0}(\lambda)$.
\item \label{it:iii} The operator $W^{\pm}(\lambda)^*$ (defined a priori on
  $B_{s_0}^{**}\supseteq  B_{s_0}$)
maps  $B_{s_0}$ onto $ L^2(S^{d-1})$.
\item \label{it:iv}
 For all $\tau\in L^2(S^{d-1})$,
  \begin{align} 
\label{eq:Psi^-2}
  &W^-(\lambda)\tau( x)
-{\e^{\i
  \pi\tfrac{d-1}{4}}\e^{-\i\phi(x,\lambda)}\tau(-\hat
  x)+\e^{-\i \pi\tfrac{d-1}{4}}\e^{\i\phi(x,\lambda)}(S(\lambda)\tau)(\hat
  x)\over
  (2\pi)^{\tfrac{1}{2}}g^{\tfrac{1}{2}}(r,\lambda)r^{\tfrac{d-1}{2}}}\in
B^* _{s_0,0},\\
\label{eq:Psi^-2a}
  &W^+(\lambda)\tau( x)
-{\e^{-\i \pi\tfrac{d-1}{4}}\e^{\i\phi(x,\lambda)}\tau(\hat
  x)+\e^{\i
  \pi\tfrac{d-1}{4}}\e^{-\i\phi(x,\lambda)}(S(\lambda)^*\tau)(-\hat
  x)\over
  (2\pi)^{\tfrac{1}{2}}g^{\tfrac{1}{2}}(r,\lambda)r^{\tfrac{d-1}{2}}}\in
B^* _{s_0,0},\\
\label{eq:800}
&\|\tau\|^2_ {L^2(S^{d-1})}=\lim_{R\to \infty}R^{-{1}}\int
_{r<R}|\sqrt{\pi}g^{\tfrac{1}{2}}(r,\lambda)W^{\pm}(\lambda)\tau|^2\d
x.   
  \end{align}
\end {enumerate}
\end{thm}
\begin{proof}
 
 \noindent{\bf Re  (\ref{it:i}).} 
 Again we concentrate on the case of
 superscript $+$. Let $\tau\in L^2(S^{d-1})$ be
 given. We shall use the partition (\ref{eq:partitionleft}), as in the
 proof of Theorems \ref{thm:4a} and \ref{thm:4}, so let $\bar
 \sigma> 0$ be given as before, cf.   (\ref{eq:supp tilde chi1})
 and (\ref{eq:supp tilde chi2}). As for the partition functions
 (\ref{eq:8part}), we modify (\ref{eq:supp tilde chi1aaa}) and
 (\ref{eq:supp tilde chi2aa})  by  replacing here
 $\tilde{\chi}_{\pm}\to \tilde{\chi}_{\pm, \righ}$
\begin {align}
  \label{eq:supp tilde chi1aaab}
 &\supp \tilde{\chi}_{-,\righ}
 \subseteq ]-\infty, 1-\bar \sigma /4[,\\&\supp \tilde{\chi}_{+,\righ} \subseteq
 ]1-\bar \sigma /2,\infty[.  \label{eq:supp tilde chi2aab}
\end{align} 
 Then it follows from  Propositions  \ref{prop:resolvent_basic2} and \ref{prop:enerest}  and
 Lemmas \ref{lem:22a} and \ref{lemma: statio} 
 that
 \begin{equation}
   \label{eq:8bnnnd}
   \Opr(\chi^+_{\righ})W^{+}(\lambda)\in \mathcal{B} (L^2(S^{d-1}),B_{s_0}^* ).
 \end{equation} (The fact that this bound holds for $W^{+}(\lambda)\to
 J^{+}(\lambda)$ is indeed a consequence of Lemma \ref{lemma: statio}
 due to interpolation, cf. \cite[Theorem 14.1.4] {Ho2}, but it can also be
 proved concretely along the lines
 of  the proofs of Lemma \ref{lemma: statio} and Theorem \ref{thm:middl}.)

Since $\langle W^{+}(\lambda)\tau,\i[H,F_R
\Opr(\chi^+_{\righ})]W^{+}(\lambda)\tau\rangle=0$, we conclude from 
\eqref{eq:eqmotionbiiii} and (\ref{eq:8bnnnd}) that 
\begin{equation}
    \label{eq:800a}
\sup_{R>1} \Re \langle W^{+}(\lambda)\tau,
\Opw (F_R\chi_-(a)\tilde{\chi}_{\righ}'
(b)gr^{-1})W^{+}(\lambda)\tau\rangle\leq C\|\tau\|^2.
  \end{equation} Here we used the calculus of pseudodifferential operators,
  cf. \cite[Theorem 18.6.8]{Ho1}.  

In combination with Propositions \ref{prop:enerest} and \ref{prop:propa5aa}, we  conclude that
  \begin{equation}\label{eq:800ac}
  \{-1<b<1\}\cap
  WF^{-s_0}_{\sc}(W^{+}(\lambda)\tau)=\emptyset.  
  \end{equation}

\noindent{\bf Re  (\ref{it:ii}) (Boundedness).}

To proceed from here we change
(\ref{eq:supp tilde chi1aaab}) and
 (\ref{eq:supp tilde chi2aab}) as follows:
\begin {align}
  \label{eq:supp tilde chi1aaabb}
 &\supp \tilde{\chi}_{-,\lef}
 \subseteq ]-\infty, -1+\bar \sigma /2[,\\&\supp \tilde{\chi}_{+,\lef} \subseteq
 ]-1+\bar \sigma /4,\infty[.  \label{eq:supp tilde chi2aabb}
\end{align} 

 With these cutoffs we can show analogously that 
\begin{equation}
   \label{eq:8bnnndb}
   \Opr(\chi^{-}_{\lef})W^{-}(\lambda)\in \mathcal{B} (L^2(S^{d-1}),B_{s_0}^* ).
 \end{equation}
 
Using (\ref{eq:Specaa}),  this leads to 
\begin{equation}
   \label{eq:8bnnndc}
   \Opr(\chi^{-}_{\lef})W^{+}(\lambda)\in \mathcal{B} (L^2(S^{d-1}),B_{s_0}^* ).
 \end{equation}

 Finally, writing (with $\chi_{\rm {middle}}:=1-\chi^+_{\righ}-\chi^{-}_{\lef}$)
\begin{equation*}
W^+(\lambda)=\Opr(\chi^+_{\righ})W^+(\lambda)+ \Opr(\chi^{-}_{\lef})W^{+}(\lambda)+\Opr(\chi_{\rm {middle}})W^{+}(\lambda),\end{equation*}
we conclude from 
 (\ref{eq:9boncloseda}), (\ref{eq:8bnnnd}), (\ref{eq:800ac}) and 
 (\ref{eq:8bnnndc}) that 
indeed
\begin{equation}
  \label{eq:9bonclosed}
  W^{+}(\lambda)\in \mathcal{B} (L^2(S^{d-1}),B_{s_0}^* ).
 \end{equation} 

Whence  $W^{+}(\lambda)$ maps $L^2(S^{d-1})$
  continuously into $\mathcal{V}^{-s_0}(\lambda)$.

\noindent{\bf Re  (\ref{it:ii}) (Bijectiveness).}
 We shall
 show that $W^{+}(\lambda)$ maps $L^2(S^{d-1})$
 onto $\mathcal{V}^{-s_0}(\lambda)$. Using the expression (\ref{eq:rep0ab})  
for the inverse $\tau\in \mathcal{D}'(S^{d-1})$, mimicking
 the first part of Step III in the proof of Theorem \ref{thm-rep} and using the
Riesz' representation theorem (see for example \cite{Yo})  in conjunction with 
(\ref{eq:9bonclosed}), we obtain
 that indeed  $\tau\in L^2(S^{d-1})$. This
argument also shows that
\begin{equation}
  \label{eq:9bonclosedb}
  W^{+}(\lambda)^{-1}\in \mathcal{B} (\mathcal{V}^{-s_0}(\lambda),L^2(S^{d-1})).
 \end{equation}  

 \noindent{\bf Re  (\ref{it:iii}).} The result follows from
 (\ref{it:ii}) by the Banach's closed
  range theorem, see \cite{Yo}. 
 
\noindent{\bf Re  (\ref{it:iv}).} Let
\[u_{\pm,\tau}(x)=(2\pi)^{-{\tfrac{1}{2}}}\e^{\mp\i \pi\tfrac{d-1}{4}}g^{-{\tfrac{1}{2}}}(r,\lambda)r^{-{\tfrac{d-1}{2}}}\e^{\pm\i\phi(x,\lambda)}\tau(\pm
\hat x).\] Clearly  $u_{\pm,\tau}\in B^* _{s_0}$ with a continuous
dependence on $\tau$. We claim (with reference to (\ref{eq:8part}))  that 
   \begin{equation}
     \label{eq:8Bs}
     \Opr (\chi^{\pm})W^{\pm}(\lambda)\tau -u_{\pm,\tau}\in B^* _{s_0,0}.
   \end{equation} Notice that also the first term is in $ B^* _{s_0}$ with a continuous
dependence on $\tau$,
   cf. (\ref{eq:8bnnnd}) and (\ref{eq:800ac}), hence it suffices to show (\ref{eq:8Bs}) for
   $\tau\in C^\infty(S^{d-1})$, in which case the
   asymptotics follows from  Theorem
   \ref{thm:kkk}, cf.  Step
   III of the
   proof of Theorem \ref{thm-rep}.

Now, combining (\ref{eq:8Bs}) and the identity (\ref{eq:Specaa}),
we obtain
\begin{equation}
     \label{eq:8Bsb}
     \Opr (\chi^{+})W^{-}(\lambda)\tau -u_{+,S(\lambda)\tau},\;\Opr (\chi^{-})W^{+}(\lambda)\tau -u_{-,S(\lambda)^*\tau}\in B^* _{s_0,0}.
   \end{equation} 
By (\ref{eq:8Bs}) and (\ref{eq:8Bsb}),
\begin{equation*}
 W^{-}(\lambda)\tau- (u_{-,\tau}+u_{+,S(\lambda)\tau}),\;W^{+}(\lambda)\tau -(u_{+,\tau}+u_{-,S(\lambda)^*\tau})\in B^* _{s_0,0},
\end{equation*} showing (\ref{eq:Psi^-2}) and (\ref{eq:Psi^-2a}). 

As for (\ref{eq:800}) we use (\ref{eq:Psi^-2}) and (\ref{eq:Psi^-2a}); 
 notice that the cross terms  do not   contribute to the limit which
 can be seen by an integration by parts with respect to the variable
 $r=|x|$, invoking Proposition \ref{prop:mixed_2aaaa}.
\end{proof}

On the basis of Theorem \ref{thm-rep2}, we can characterize the scattering
matrix $S(\lambda)$ geometrically as follows:

\begin{cor}\label{cor:geomdef}
  For all $\tau^-\in L^2(S^{d-1})$, there exist a uniquely determined $u\in
  \mathcal{V}^{-s_0}(\lambda)$ and $\tau^+\in L^2(S^{d-1})$ such
  that
\begin{equation}
  \label{eq:Psi^-2b}
  u
-{\e^{\i
  \pi\tfrac{d-1}{4}}\e^{-\i\phi(x,\lambda)}\tau^-(-\hat
  x)+\e^{-\i \pi\tfrac{d-1}{4}}\e^{\i\phi(x,\lambda)} \tau^+(\hat
  x)\over
  (2\pi)^{\tfrac{1}{2}}g^{\tfrac{1}{2}}(r,\lambda)r^{\tfrac{d-1}{2}}}\in B^* _{s_0,0}.
\end{equation} We have   $ \tau^+=S(\lambda)\tau^-$,
$u=W^-(\lambda)\tau^-=W^+(\lambda)\tau^+$. 
\end{cor}
\begin{proof}
  The existence part (with $\tau^+=S(\lambda)\tau^-$) follows from
 (\ref{eq:Psi^-2}).

To show the
  uniqueness, suppose  that $u_i,\tau_i^+$, $i=1,2$, satisfy the
  requirements of \eqref{eq:Psi^-2b} with the same $\tau_-$.
Then for the difference, $u=u_1-u_2$, we have
 $(H-\lambda) u=0$ and $WF\left(B^*_{s_0,0},u\right)\subseteq \{b=1\}$. 
Hence by Proposition  \ref{prop:propa5},  
$u=0$.
\end{proof}

\begin{cor}\label{cor:geomdef2b} Let $d\geq2$ and
  $\lambda\geq0$. Suppose (in addition to Conditions
  \ref{assump:conditions1} and \ref{assump:conditions3}) that $V_2$ and $V_3$ are
  spherically symmetric and that  $\int_0^\infty r|V_3(r)|\,\d
  r<\infty$. (Condition
  \ref{assump:conditions2} is not needed since $V_2$ can be absorbed
  into $V_1$). Then there exists a  real-valued continuous function
  $\sigma_l(\cdot)$ such that for all spherical harmonics $Y$ of order $l$
we have  $S(\lambda)Y=\e^{\i2\sigma_l(\lambda)}Y$.

 Let $u_l(r)$ denote
  the regular solution of the reduced Schr\"odinger equation on the
  half-line $]0,\infty[$
\begin{equation*}
 -u'' +V_lu=0,\;V_l(r)=2(V(r)-\lambda)+\tfrac{(l+\tfrac d2-1)^2-4^{-1}}{r^2},\;l\geq0;
\end{equation*} where ``regular'' refers to the asymptotics $u(r)\asymp
r^{l+\tfrac{d-1}{2}}$ as $r\to 0$. Then  $\sigma_l(\cdot)$ is
uniquely determined $\mod 2\pi$  by the asymptotics
\begin{equation*}
  {u_l(r)\over r^{\tfrac{d-1}{2}}}-C{\sin \big (\int^r_{R_0}\sqrt{2(\lambda-V(r'))}\,\d r'+\sqrt{2\lambda
    R_0}-\tfrac{d-3+2l}{4}\pi+\sigma_l(\lambda)\big ) \over
  (\lambda-V(r))^{\tfrac{1}{4}}r^{\tfrac{d-1}{2}}}\in B^* _{s_0,0},
\end{equation*} where $C=C(l,\lambda)$ is a (uniquely  determined) positive constant.
 \end{cor}

\proof
Let $Y$ be a spherical harmonic of order $l$. Note that its parity 
 is  $(-1)^l$,  i.e. $Y(-\omega)=(-1)^lY(\omega)$.
Besides,
$u:=r^{-{\tfrac{d-1}{2}}}u_l(r)Y    (\hat x)$
solves $(H-\lambda)u=0$. We apply  Corollary \ref{cor:geomdef}
with this $u$ and with $\tau^-=Y$, so that $\tau^+=\e^{\i
  2\sigma_l(\lambda)}Y$. Then
\begin{eqnarray*}
&&\e^{\i
  \pi\tfrac{d-1}{4}}\e^{-\i\phi(x,\lambda)}\tau^-(-\hat
  x)+\e^{-\i \pi\tfrac{d-1}{4}}\e^{\i\phi(x,\lambda)} \tau^+(\hat
  x)\\
&=&\left(
\e^{\i
  \pi\tfrac{d-1}{4}-\i\phi(x,\lambda)+\i\pi l}
+\e^{-\i \pi\tfrac{d-1}{4}+\i\phi(x,\lambda)+\i2\sigma_l(\lambda)}\right)
Y(\hat x)
\\
&=&2\e^{\i\pi\frac{
  l}{2}+\i\sigma_l(\lambda)}\sin\left(\phi(x,\lambda)-\tfrac{d-3+2l}{4}\pi
+\sigma_l(\lambda)  \right)Y(\hat x).\end{eqnarray*}
We finish the proof 
using  \eqref{eq:eik222}.
\qed

\section{Homogeneous potentials -- location of singularities of  $S(0)$}
\label{Propagation of singularities at zero energy}

In this section 
we impose Conditions
\ref{assump:conditions1}--\ref{assump:conditions3} with  $d\geq 2$  and the condition $V_1(r)= -\gamma r^{-\mu}$
for $r\geq 1$ and hence $V(r)=-\gamma
r^{-\mu}+O(r^{-\mu-\epsilon_2})$, cf. \eqref{eq:1v}. Throughout the
section $g=g(\lambda=0)=\sqrt{-2V_1}$.

 Our goal is to prove a statement  about the
localization of the singularities of the (Schwartz) kernel
$S(0)(\omega,\omega')$. The purely Coulombic case for which $\mu=1$
and $d\geq 3$ was treated explicitly in
Example \ref{example:coulombsing}. Under an additional
 condition we can write
down a fairly  explicit integral that carries the singularities. 

In
Section \ref{Singularity of the kernel of the scattering matrix}  we
shall study the nature of these singularities (under the condition of
spherical  symmetry) using one-dimensional WKB-analysis.

\subsection{Reduced classical equations}
\label{Reduced equations}

Consider the classical system given by the Hamiltonian
 $h_1(x,\xi)=\frac12\xi^2-\gamma|x|^{-\mu}$ for $x\neq0$.
The equations of motion for $h_1(x,\xi)$ are invariant with respect to the transformation
\begin{equation}
\label{scal0}(x,\xi)\mapsto(\lambda x,\lambda^{-\mu/2} \xi),\ \ \ \lambda\in\R_+
,\end{equation}
 upon rescaling of time $t\mapsto t\lambda^{1+\mu/2}$. 

Let \[\T^*:=(\R^d\backslash\{0\})\times\R^d/\sim,\]
where $(x_1,\xi_1)\sim(x_2,\xi_2)$ iff there exists $\lambda>0$ such that 
$(x_1,\xi_1)\mapsto(\lambda x_2,\lambda^{-\mu/2} \xi_2)$.
Note that $\T^*$ can be conveniently identified with
$T^*(S^{d-1})\times\R$. 
We shall introduce
coordinates of $\T^*$  by setting $b=\hat x\cdot
  {\xi\over g}\in\R$ and $\bar c=(I-|\hat x \rangle \langle \hat x|){\xi\over
    g} \in T_{\hat x}^*(S^{d-1})$ with $\hat x\in S^{d-1}$. (At this point  we are slightly abusing the notation of  Subsection \ref{Some other  estimates}, however as noticed there the  $b$ and $\bar c$ given by \eqref{eq:sym_basdef}  agree with the above definition for $r\geq 1$.)
The equations of motion for the hamiltonian $h_1$ can
  be reduced to
$\T^*$. Introducing   the
  ``new time'' $\tau$ by
  ${\d\tau\over
 \d t}=g/r$ we have the following  system of reduced  equations of motion:
\begin{equation}\label{eq:new reduced eqns}
\begin{cases}
{\d\over \d\tau}\hat x=\bar c,\\
{\d\over \d\tau}\bar c=-(1-{\mu\over 2})b\bar c-\bar c^2\hat x,\\
{\d\over \d\tau}b=(1-{\mu\over 2})\bar c^2 +{\mu\over
  2}(b^2+\bar c^2-1).
\end{cases}\;
\end{equation}  (Notice that the
  last equation follows from (\ref{eq:eqmotionb})).
The maximal solution of \eqref{eq:new reduced eqns} that passes 
 $z=(\hat x,b,\bar c)\in {\mathbb T}^*$ at $\tau=0$ is denoted by $\gamma(\tau,
z)$.

Beside \eqref{eq:new reduced eqns}, we shall   consider a related dynamics given
by the equations
\begin{equation}\label{eq:reduced eqns}
\begin{cases}
{\d\over \d\tau}\hat x=\bar c,\\
{\d\over \d\tau}\bar c=-(1-{\mu\over 2})b\bar c-\bar c^2\hat x,\\
{\d\over \d\tau}b=(1-{\mu\over 2})\bar c^2.
\end{cases}\;
\end{equation} 

The (maximal) solution of the system
\eqref{eq:reduced eqns} that 
passes $z=(\hat x,b,\bar c)\in\T^*$
 at $\tau=0$ will be denoted by $\gamma_0(\tau,
z)$. Clearly the equation $\bar c=0$ defines the
fixed points, and the system is complete.

Notice that the surface $h_1^{-1}(0)$ in the coordinates $(\hat x,b,\bar c)$
corresponds to the condition $b^2+\bar c^2=1$. This surface is preserved both
by the flow $\gamma$ and $\gamma_0$, and on this surface both flows coincide.

Note that the flow $\gamma_0$ is exactly solvable. 
The variable $b$ is always 
increasing and $k=b^2+\bar c^2$ is a conserved quantity; of
course the relevant value is $k=1$. 
For non-fixed points we can compute its dependence on the modified time
\begin{equation}\label{eq:8bbbbb}
b(\tau)=\sqrt{k}\tanh \sqrt{k}(1-\tfrac\mu2)(\tau-\tau_0).
\end{equation}

 Values $k\neq 1$ correspond
in this picture to replacing the coupling constant $\gamma\to
k\gamma$. More precisely, if  $k=b^2+\bar c^2$ for a solution to
\eqref{eq:reduced eqns}, we can
define $r(\tau)= r_0\exp (\int_0^\tau b d\tau')$, introduce
$t=\int _0^\tau \tfrac{r}{g(r)}d\tau'$ and check  that indeed 
\begin{equation}\label{eq:HAloesn}
\begin{cases}x(t)=r\hat x,\\
\xi(t)=g(r)(b\hat x+\bar c),
\end{cases}\;
\end{equation} defines a zero energy solution to Hamilton's equations
with $V\to kV$. The equation $b=0$ corresponds to  a 
turning point (at which $|x(t)|$ has the smallest value).

Clearly, it follows from (\ref{eq:8bbbbb}) that 
$\lim_{\tau\to\infty}b=\sqrt k$, $\lim_{\tau\to-\infty}b=-\sqrt k$. Upon 
writing $\hat x(\tau)\cdot\hat x(\infty)=\cos \theta(\tau)$ for some
monotone continuous  function $\theta(\cdot)$, we obtain from
\eqref{eq:pol-eqn}  that
\begin{equation}
  \label{eq:vinkel}
  |\theta(\infty)-\theta(-\infty)|=\tfrac{2}{2-\mu}\pi.
\end{equation}

\subsection{Propagation of singularities}
\label{Quantum bounds}

We will use the scattering wave front set at zero energy, introduced in
Subsection  \ref{Some other  estimates}. 
The following proposition is somewhat similar to  H\"ormander's theorem about
  propagation of singularities adapted to scattering at the zero
  energy. It is a ``local'' version of  Proposition \ref{prop:propa5aa}  which takes into account the
  fact that in the case of a homogeneous potential we can use the dynamics in
  the reduced phase space. Again the proof is a modification of that of
\cite[Proposition 3.5.1]{Ho3}, see also \cite{Me} and \cite{HMV}.

\begin{prop}  \label{prop:propa1}
 Suppose $u,v\in L^{2,-\infty}$, $Hu=v$,
  $s\in \R$, $z\in{\mathbb T}^*$ and $z\not\in
  WF^s_{\sc}(u)$. Define
  \begin{align*}\tau^+&:=\sup \{\tau \geq 0  |\, \ {\gamma_0
(\tilde\tau,z)}\notin WF_{\sc}^s(u) \text { for all }\tilde \tau \in [0,\tau]\},
    \\ \tau^-&:=\inf  \{ \tau \leq 0|\, \ {\gamma_0
(\tilde\tau,z)}\notin WF_{\sc}^s(u) \text { for all }\tilde \tau \in [\tau,0]\}.
  \end{align*}
If $\tau^+<\infty$, then $\gamma_0(\tau^+,z)\in
  WF^{s+2s_0}_{\sc}(v)$. If $\tau^->-\infty$, then $\gamma_0(\tau^-,z)\in
  WF^{s+2s_0}_{\sc}(v)$.
\end{prop}
\begin{proof} The
  proof is  similar to the one of Proposition
  \ref{prop:propa5aa}. We shall only deal with the case of  forward
  flow; the
  case of superscript "$-$" is similar (actually  it follows from the
  case of "$+$" by time reversal invariance). For convenience, we shall
  assume that $\epsilon_2 \leq 2-\mu$. 

\noindent{\bf Step I}. We will first show the  following weaker
statement: 
Suppose $u\in L^{2,s-\tfrac{\epsilon_2}{2}}$, $v\in L^{2,s+2s_0}$ and  $Hu=v$. 
  Then
  \begin{equation}\label{eq:nowconjrc}
   {\gamma_0
(\tau,z)}\notin WF_{\sc}^s(u)\text{   for all }\tau\geq 0. 
  \end{equation}

  Suppose on the contrary that (\ref{eq:nowconjrc}) is false. Then we
  obtain from  Proposition \ref{prop:enerest} that the flows of (\ref{eq:new reduced eqns}) and (\ref{eq:reduced
  eqns}), starting at $z$, coincide.
 Letting  $\gamma(\tau)=\gamma(\tau,z)$, it thus needs to be shown that
\begin{equation}
  \label{eq:maxtau}
\tau^+ := \sup \{\tau \geq 0 | {\gamma
(\tilde\tau)}\notin WF_{\sc}^s(u) \text { for all }\tilde \tau \in [0,\tau]\}= \infty.  
\end{equation}

Suppose on the contrary that $\tau^+$ is finite. Then $\gamma
(\tau^+)$ is not a fixed point. Consequently, we can  pick a slightly smaller
$\tilde\tau^{+} <\tau^+$ and  a transversal $(2d-2)$--dimensional submanifold at $\gamma(\tilde\tau^{+})$, say
$\mathcal{M}$, such  that with $J=]-\epsilon+\tilde\tau^{+} ,
  \tau^++\epsilon[$,   
for some small $\epsilon>0$,
the map
\begin{equation*}
 J\times \mathcal{M} \ni
 (\tau,m) \to  \Psi (\tau,m )= \gamma(\tau-\tilde\tau^{+}, m) \in {\mathbb
T}^*
\end{equation*}
is a diffeomorphism onto its range.

We pick $\chi \in C^{\infty}_c(\mathcal{M})$ supported in a small
neighbourhood of 
$\gamma(\tilde\tau^{+})$  such that $\chi (\gamma(\tilde\tau^{+}))=1$
and 
\begin{equation}
\label{eq:nowavefront}
  \Psi(]-\epsilon+\tilde\tau^{+} , \tilde\tau^{+}] \times \supp \chi) \cap WF^s(u) =\emptyset.
\end{equation}
 We pick  a  non-positive function
$f \in
C^{\infty}_c(J)$ such that $f'\geq 0$
on a neighbourhood of $[ \tilde\tau^{+} , \tau^++\epsilon)$ and  $f( \tau^+)<  0$.

Let $f_K(\tau)=\exp (-K\tau)f(\tau)$ for $K>0$, and $X_{\kappa}=
(1+\kappa r^2)^{1/2}$ for $\kappa \in ]0,1]$. We consider the symbol 
\begin{equation} 
\label{eq:propobs}
  b_\kappa=g^{-1/2}X^{1/2}a_\kappa;\; a_\kappa=X^sX_{\kappa}^{-\epsilon_2/2}F(r>2)(f_K
  \otimes \chi)\circ \Psi^{-1}.
\end{equation} 
First we fix $K$. A part of the Poisson bracket with
 $b_\kappa^2$ is
\begin{equation}
\label{eq:Poi1}
  \{h_2,g^{-1}X^{2s+1}X_{\kappa}^{-\epsilon_2}\}= r^{-1}Y_\kappa bX^{2s+1}X_{\kappa}^{-\epsilon_2},
\end{equation}
where $Y_\kappa=Y_\kappa(r)$ is uniformly bounded in $\kappa$.
We fix $K$ such that $2K\geq |Y_\kappa b|+2$
on $\supp b_\kappa$.

We  compute
\begin{equation}
  \label{eq:Poi2}
 \{ h_1,(f_K \otimes \chi)\circ \Psi^{-1}\} =\tfrac{g}{r}\big (\big [\frac
 {\d} {\d\tau }f_K\big ] \otimes \chi\big )\circ \Psi^{-1}. 
\end{equation}

From \eqref{eq:Poi1} and \eqref{eq:Poi2}, and by the choice of $f$ and
$K$,  we conclude that  
\begin{equation}
\label{eq:Poi3}
  \{h_2,b_\kappa^2\}\leq -2a_\kappa^2 +O\big (r^{2s-\epsilon_2}\big
  )\text{ at   }
 \mathcal P \subseteq \T^*
\end{equation}
given by  $$\mathcal P =\Psi(\{\tau \in J| f'(\tau)\geq0\}\times \supp \chi).$$

Introducing  $A_\kappa={\Opw}(a
_\kappa)$ and $B_\kappa={\Opw}(b
_\kappa)$, we have 
\begin{equation}
  \label{eq:comm1}
 \langle i[H, B_\kappa^2]\rangle _u=-2\Im \langle v,B_\kappa^2 u\rangle, 
\end{equation} 
and we estimate the right hand side using the calculus of
pseudodifferential operators, cf.   \cite[Theorems 18.5.4, 18.6.3,
18.6.8]{Ho1},   to
obtain the uniform bound 
\begin{equation}
  \label{eq:comm2}
 |\langle i[H, B_\kappa^2]\rangle _u|\leq C_1\|v\|_{s+2s_0}\|A_\kappa
 u\|+C_2
\leq \|A_\kappa
 u\|^2 +C_3.
\end{equation}

 On the other hand, using \eqref{eq:nowavefront} and 
 \eqref{eq:Poi3},  we infer that 
\begin{equation}
  \label{eq:comm3}
 \langle i[H-V_3, B_\kappa^2]\rangle _u \leq -2\|A_\kappa u\|^2+C_4.
\end{equation} 

An  application of  \eqref{eq:T1} yields
\begin{equation}
  \label{eq:hgyi2}
 \langle i[V_3, B_\kappa^2]\rangle _u\leq  C_5.
\end{equation}

Combining \eqref{eq:comm2} -- \eqref{eq:hgyi2} yields 
\begin{equation*}
 \|A_\kappa u\|^2\leq C_6=C_3+C_4+C_5,
\end{equation*} 
which  in
turn gives a uniform bound 
\begin{equation}
  \label{eq:comm4}
 \|X_{\kappa}^{-\epsilon_2/2}\Opw\big(\chi_{\gamma (\tau^+)}F(r>2)\big)u\|^2_s\leq C_7.
\end{equation}
Here $\chi_{\gamma (\tau^+)}$ signifies a  phase-space localization
factor
of the form entering in \eqref{eq:WF^sa} supported in a sufficiently
small neighbourhood of  the point $\gamma
(\tau^+)$.

We let $\kappa\to 0$ in \eqref{eq:comm4} and infer that $\tau^+\notin
WF_{\sc}^s(u)$, which is a contradiction. We have proved
(\ref{eq:maxtau}) and hence (\ref{eq:nowconjrc}).

\noindent{\bf Step II}.
To relax the assumptions on  $u$ and  $v$  used in
Step I, 
we modify  the
  above proof (using localization) in an iterative procedure very
  similar to Step II of the
  proof of Proposition \ref{prop:propa5aa}.

Pick $t<s$ such that $u\in
L^{2,t}$ and define $s_m=\min(s,t+m\epsilon_2/2)$ for $m\in \N$. Let
correspondingly $\tau^+_m$ be given as $\tau^+$, upon replacing $s\to s_m$. Clearly, 
\begin{equation}
  \label{eq:eq4i7}
\tau^+_m \leq \tau^+_{m-1};\;m=2,3,\dots 
\end{equation}
We shall  show that 
 \begin{equation}
  \label{eq:mimpli7} \tau^+_m<\infty  \Rightarrow
  \gamma_0(\tau^+_m,z)\in 
  WF^{s_m+2s_0}_{\sc}(v).
\end{equation} We are done by using \eqref{eq:mimpli7} for an  $m$  taken so large that $s_m=s$.

Let us consider  the start of induction given by $m=1$, in which  case
obviously  $u\in L^{2,s_m-\epsilon_2/2}$. Suppose on the contrary that
(\ref{eq:mimpli7}) is false. Then we consider the following case:
\begin{equation}
  \label{eq:abe7}
 \tau^+_m<\infty  \text{ and } \gamma_0(\tau^+_m,z)\notin 
  WF^{s_m+2s_0}_{\sc}(v).
\end{equation}

It  follows from  \eqref{eq:abe7} and an ellipticity argument that
$b^2+\bar c^2=1$ at $\gamma_0(\tau^+_m,z)$ (using that $ \gamma_0(\tau^+_m,z)\notin
WF^{s_m+\mu}_{\sc}(Hu)$). Consequently  we can henceforth  use the flow of (\ref{eq:new reduced
  eqns}),
$\gamma(\tau)=\gamma(\tau,\cdot)$, exactly as in  Step I.
 
 We let $\epsilon>0$,  $J$,  $f$, $f_K$, $\chi$ and $\Psi$  be chosen as in Step I with
 $\tau^+\to \tau^+_m$ and $\tilde\tau^+\to \tilde\tau^+_m$.   Let
 $\tilde f \in C^{\infty}_c(]\tilde\tau^+_m-2\epsilon,\tau^+_m+2\epsilon[)$
 with $\tilde f=1$ on  $J$.  Similarly, let $\tilde\chi \in
 C^{\infty}_c(\mathcal{M})$ be supported in a small
neighbourhood of 
$\gamma(\tilde\tau^{+}_m)$  such that $\tilde\chi
(\gamma(\tilde\tau^{+}_m))=1$ in  a 
neighbourhood of $\supp \chi$.

It follows from    \eqref{eq:abe7}, possibly by shrinking the supports
of $\tilde f$ and $\tilde\chi$,  that
 \begin{equation}
  \label{eq:abe27}
 I_\epsilon v\in L^{2,s_m+2s_0},\;I_\epsilon  =\Opw \big
 (F(r>2)(\tilde f_K
  \otimes \tilde\chi)\circ \Psi^{-1}\big ).
\end{equation} Next, we introduce the symbol $b_\kappa$ by \eqref{eq:propobs}
(with $s\to s_m$) and proceed as in Step I. As for
the bounds (\ref{eq:comm2}),  we can replace $v$ by $I_\epsilon v$ up
to addition of a  term that is bounded uniformly in $\kappa$.
Clearly,  we can verify
(\ref{eq:comm3}) and (\ref{eq:hgyi2}). So again we obtain \eqref{eq:comm4} (with $s\to
s_m$), and therefore a contradiction as in Step I. We have shown \eqref{eq:mimpli7} for $m=1$.

Now suppose $m\geq 2$ and that \eqref{eq:mimpli7} is verified for
$m-1$. We need to show the statement for the given $m$. Due to
\eqref{eq:eq4i7} and the induction hypothesis, we can assume that 
\begin{equation} \label{eq:abejj7} \tau^+_m <\tau^+_{m-1}.
\end{equation}  Again we argue by contradiction assuming
\eqref{eq:abe7}. We proceed as above noticing that  it follows from \eqref{eq:abejj7} that in addition to \eqref{eq:abe27} we have
\begin{equation}
  \label{eq:abe37}
 I_\epsilon u\in L^{2,s_{m-1}};
\end{equation} at this point we   possibly  need to  shrink the supports
of $\tilde f$ and $\tilde\chi$ 
even more  (viz. taking $\epsilon<(\tau^+_{m-1}-\tau^+_m)/2$). By replacing $v$ by $I_\epsilon v$ and $u$ by $I_\epsilon u$  at various points in the procedure of Step I (using \eqref{eq:abe27} and \eqref{eq:abe37}, respectively) we obtain again a contradiction. Whence \eqref{eq:mimpli7} follows.
\end{proof}

\begin{remark}\label{remark:unifprop}  Suppose $u\in
L^{2,t_1}$,  $v\in
L^{2,t_2}$ and  $Hu=v$. Suppose $z_0\not\in
  WF^s_{\sc}(u)$  for  some  $s> t_1$. 
  Fix $\tilde\tau^+\in ]0,\infty[$ and suppose that  $\gamma_0(\tau,z_0)\not\in
  WF^{s+2s_0}_{\sc}(v)$ for all $\tau\in [0,\tilde\tau^+]$. Write
  $\gamma_0(\tilde\tau^+,z_0)=(\omega_1, \bar c_1, b_1)=(\omega_1,\eta_1)$. Then there exist neighbourhoods $\mathcal{N}_{\omega_1}\ni
   {\omega_1}$ and $\mathcal{N}_{\eta_1}\ni {\eta_1}$ such that for all $\chi_{\omega_1}\in
   C^{\infty}_c(\mathcal{N}_{\omega_1})$ and $\chi_{\eta_1}\in
   C^{\infty}_c(\mathcal{N}_{\eta_1})$ we have
   $\Opw\big(\chi_{z_1}F(r>2)\big)u\in L^{2,s}$. Here
   $\chi_{z_1}(x,\xi)=\chi_{\omega_1}({\hat x})\chi_{\eta_1}(\xi/g).$ Notice
   that this conclusion is already contained in Proposition
   \ref{prop:propa1}; however  the  above proof
  yields  an additional  bound: 

First, writing $z_0=(\omega_0, \eta_0)$, we can pick any
similarly defined localization factor, say denoted by $\chi_{z_0}$,
with  $\chi_{\omega_0}=1$ and $\chi_{\eta_0}=1$ around the points
$\omega_0$ and $\eta_0$, respectively, and such that 
$\Opw\big(\chi_{z_0}F(r>2)\big)u\in L^{2,s}$ (this is by
assumption). Next we pick a small neighbourhood
  $U$ of
  $\gamma_0([0,\tilde\tau^+],z_0)\subset \T^*$ and $\chi\in C^\infty_c
  (U)$ with $\chi=1$ around this orbit segment. If $U$ is small enough
  we have (again by
assumption) that $\Opw\big(\chi_{\gamma_0} F(r>2)\big)v\in L^{s+2s_0}$,
  $\chi_{\gamma_0}(x,\xi):=\chi(\hat x, \xi/g)$. Now, there are   neighbourhoods $\mathcal{N}_{\omega_1}\ni
   {\omega_1}$ and $\mathcal{N}_{\eta_1}\ni {\eta_1}$ depending only
   on $\chi_{z_0}$ and $\chi _{\gamma_0}$ such that for all $\chi_{\omega_1}\in
   C^{\infty}_c(\mathcal{N}_{\omega_1})$ and $\chi_{\eta_1}\in
   C^{\infty}_c(\mathcal{N}_{\eta_1})$ we have
\begin{multline*}
  \|\Opw\big(\chi_{z_1}F(r>2)\big)u\|_s\\\leq C\big
  (\|\Opw\big(\chi_{z_0}F(r>2)\big)u\|_s+\|u\|_{t_1}+\|\Opw \big(\chi _{\gamma_0}F(r>2)\big)v\|_{s+2s_0}+\|v\|_{t_2}\big ),
\end{multline*}
where the constant $C$ only depends on the various localization factors.
  \end{remark}

\subsection{Location of singularities of the kernel of the scattering matrix}
\label{Quantum singularities}

In this subsection we describe the location of the singularities of the scattering matrix at zero
energy. 

\begin{thm}  \label{thm:sings} 
Suppose that  $V_1(r)= -\gamma r^{-\mu}$ 
 for $r\geq 1$. Then  
 the kernel $S(0)(\omega,\omega')$ is smooth
  outside the set $\{(\omega,\omega')|\;\omega\cdot \omega'= \cos
  {\mu\over 2-\mu}\pi\}$.
\end{thm}

To analyse $S(0)(\omega,\omega')$ we shall use the representation
(\ref{eq:bndSS}), which we write (formally) as
\begin{equation*}
S(0)(\omega,\omega')=-2\pi\langle j^+(\cdot,\omega),
v^-(\cdot,\omega')\rangle+
2\pi\i
\langle v^+(\cdot,\omega),
R(+\i0)v^-(\cdot,\omega')\rangle,\end{equation*}
where
\begin{eqnarray*}
j^\pm(x,\omega)&=&(2\pi)^{-d/2}\big (\e^{\i \phi^\pm}
\tilde a^\pm\big )(x,\omega,0),\\
v^\pm(x,\omega)&=&(2\pi)^{-d/2}\big (\e^{\i \phi^\pm}
\tilde t^\pm\big )(x,\omega,0).
\end{eqnarray*}
Let $\phi_\sph^+$ denote the solution of the eikonal equation for the
potential $V_1$ at zero energy, cf. (\ref{phi}). It is given by 
\begin{equation}\label{eq:tilll}
\phi_{\sph}^+(x,\omega)
=\frac{\sqrt{2\gamma}}{ 1-\mu/2}\left( r^{1-\mu/2}\cos (1-\mu/2)
\theta-R_0^{1-\mu/2}\right),
\end{equation} 
where $ \cos \theta=\hat x \cdot
\omega$. Using $x^\perp=\frac{\omega-\hat x\cos\theta}{\sin\theta}$ and
$\nabla_x\theta=-\frac{x^\perp}{r}$, we can also compute
\begin{eqnarray*}
F_\sph^+(x,\omega)&=&\nabla_x\phi_\sph^+(x,\omega)\\
&=&\sqrt{2\gamma}r^{-\mu/2}
\left(\hat x\cos(1-\mu/2)\theta+x^\perp\sin(1-\mu/2)\theta\right).
\end{eqnarray*}

\begin{lemma} 
  \label{lemma:wavefas}For all  $s\in \R$, $\omega\in S^{d-1}$ and
  multiindices  $\delta$, 
  \begin{align}
    &WF^{s}_{\sc}(\partial_\omega^\delta
 v^{\pm} (\cdot ,\omega))\nonumber\\&\subseteq \Big\{z=(\hat x,
    \bar c, b)\in \T^*| 1-\sigma'\leq 
\pm \hat x\cdot \omega \leq 1-\sigma,\;b\hat x+\bar c=\pm \frac{F_{\sph}^+(\hat x,\pm \omega)}{({2\gamma})^{1/2}}\Big\},\label{eq:wacom}\\
    \nonumber&WF^{s}_{\sc}(\partial_\omega^\delta
j^{\pm} (\cdot ,\omega))\\&\subseteq\Big \{z=(\hat x,
    \bar c, b)\in \T^*| 1-\sigma'\leq
\pm \hat x\cdot \omega
 ,\;b\hat x+\bar c=\pm \frac{F_{\sph}^+(\hat x,\pm \omega)}{({2\gamma})^{1/2}}\Big\}.\label{eq:wacom1}
  \end{align} Suppose in addition that $\chi_+\in
   C^\infty(\R)$, $\chi_+'\in
   C^\infty_c(\R)$  and $\supp
   \chi_+ \subset]1,\infty[$. Then
   \begin{equation}
     \label{eq:+bound}
     \Opw(\chi_+(a))\partial_\omega^\delta v^{\pm} (\cdot ,\omega),\;\Opw(\chi_+(a))\partial_\omega^\delta j^{\pm} (\cdot ,\omega)\in L^{2,s}.
   \end{equation}

\end{lemma}
\begin{proof} Only the ``$+$'' case  needs to be considered (can be
    seen by  complex conjugation). 
  Upon multiplying by  a localization operator  supported outside
  of the right hand side of (\ref{eq:wacom}), we need to demonstrate that
  the result is in $L^{2,s}$,  
 cf. the definition (\ref{eq:WF^sa}). Using the
 right  Kohn-Nirenberg
 quantization (instead of the
 Weyl quantization) this can be done by integrating by parts
  in  explicit integrals, exactly as in the proofs of Lemma
  \ref{lem:22a} (\ref{it:p4}) and Theorem \ref{thm:middl}. The
  arguments for (\ref{eq:wacom1}) and (\ref{eq:+bound})   are the
  same, in particular, (\ref{eq:+bound})  
 follows from the proof of Theorem \ref{thm:middl}.
\end{proof}

\noindent{\em Proof of Theorem \ref{thm:sings}. } Due to  Proposition
\ref{prop:proposi1} and Lemma \ref{lemma:wavefas} 
 we are allowed to act by $R(+\i 0)$ 
on $\partial_{\omega'}^{\delta'}v^{-} (\cdot ,\omega'))$. In fact, for all $\tau\in C^\infty(S^{d-1})$
 \begin{equation}\label{eq:8sss}
  R(+\i 0)T^-(0)\tau =\int_{S^{d-1}}
  R(+\i0)v^-(\cdot,\omega')\tau(\omega')\,\d \omega'.
 \end{equation}

Using the representation
(\ref{eq:bndSS}) interpreted as a form on $C^\infty(S^{d-1})$, and
(\ref{eq:8sss}), we have
$S_\kappa(0)\to S(0)$  as $\kappa\searrow0$,
 where the kernel of $S_\kappa(0)$ is the well-defined smooth
expression
\begin{align*}
&S_\kappa(0)(\omega,\omega')\\&=-2\pi\langle j^+(\cdot,\omega),F(\kappa|\cdot|<1)
v^-(\cdot,\omega')\rangle+
2\pi\i
\langle v^+(\cdot,\omega),
F(\kappa|\cdot|<1)R(+\i0)v^-(\cdot,\omega')\rangle.\end{align*}
It remains to be shown that $S_\kappa(0)(\cdot,\cdot)$ has a limit in 
$C^\infty\big (\{\omega\cdot\omega'\neq\cos\frac{\mu\pi}{2-\mu}\}\big )$.

 By
integration by parts,  it follows that the first term  has a limit, in fact
in $C^\infty\big (S^{d-1}\times S^{d-1}\big )$, cf. the proof of Lemma \ref{lemma:wavefas}.  Whence we only
look at the second term.

By Lemma \ref{lemma:wavefas} and Proposition \ref{prop:mixed_2aaaa}, for all $s$
\begin{align}
&WF^{s}_{\sc}(\partial_\omega^\delta
 v^{+} (\cdot ,\omega))\subseteq \{\bar
 c\neq 0,\;b^2+\bar c^2=1\}\cap \nonumber \\&
\{z\ |\ \lim_{\tau\to+\infty}\hat x(\tau)=\omega,\ \ \hbox{where}\
 \gamma_0(\tau,z)=(\hat x(\tau),b(\tau),\bar c(\tau))\}.\label{eqna1}
\end{align} Here $\gamma_0(\tau,z)$ refers to the flow defined by (\ref{eq:reduced eqns}).

By Propositions  \ref{prop:proposi1} and \ref{prop:propa1}, for all  $s$,
\begin{align}
&WF^{s}_{\sc}(
R(+\i0) \partial_{\omega'}^{\delta'}v^{-} (\cdot ,\omega'))\nonumber\\
&\subseteq\{\gamma_0(\tau,z)\ |\ 
\tau\geq0,\ z\in WF_\sc^s(\partial_{\omega'}^{\delta'}v^-(\cdot,\omega'))\}\cup\{\bar c=0,\;b>0\}
\nonumber\\
&\subseteq
\{z\ |\ \lim_{\tau\to-\infty}\hat x(\tau)=-\omega'\}\cup \{\bar c=0,\;b>0\}.\label{eqna2}
\end{align}
By invoking \eqref{eq:vinkel}, we see that
the sets on the right hand side of 
 (\ref{eqna1}) and  (\ref{eqna2})  are  disjoint away from $\{\omega\cdot\omega'\neq\cos\frac{\mu\pi}{2-\mu}\}$.
Hence also \[WF^{s}_{\sc}(\partial_\omega^\delta
 v^{+} (\cdot ,\omega))\cap WF^{s}_{\sc}(
R(+\i0) \partial_{\omega'}^{\delta'} v^{-} (\cdot
,\omega'))=\emptyset,\]  which implies, upon taking $s=0$ and  using
(\ref{eq:+bound}) and a suitable partition of unity,  that
\[\langle\partial_\omega^\delta v^+(\cdot,\omega,0),
R(+\i0)\partial_{\omega'}^{\delta'}v^-(\cdot,\omega',0)\rangle\]
is well-defined. 

By the same arguments 
\begin{multline*}
 \partial_\omega^\delta \partial_{\omega'}^{\delta'}\langle v^+(\cdot,\omega,0),
F(\kappa|\cdot|<1)R(+\i0)v^-(\cdot,\omega',0)\rangle\\\to \langle\partial_\omega^\delta v^+(\cdot,\omega,0),
R(+\i0)\partial_{\omega'}^{\delta'}v^-(\cdot,\omega',0)\rangle
  \end{multline*}
locally uniformly in
$\{\omega\cdot\omega'\neq\cos\frac{\mu\pi}{2-\mu}\}$. Notice that the
bound (\ref{eq:+bound}) is uniform in $\omega$; a similar statement
 is valid  for the
bounds underlying  (\ref{eq:wacom}), and we also need at this point
to invoke Remark \ref{remark:unifprop}. \qed

\begin{remarks} \label{remarks:smoloca}
\begin{enumerate}[\quad\normalfont 1)]
\item \label{it:smoloca1} The somewhat
  abstract procedure of the proof of Theorem \ref{thm:sings} does not
  provide information about the nature of the singularities at
  the cone $\omega\cdot \omega'= \cos
  \tfrac {\mu}{2-\mu}\pi$. In the study of the  singularities at the diagonal 
  of the kernel of  scattering matrices for positive energies
  (see \cite {IK2} and \cite {Y2}) it
  is important  that the eikonal and transport equations can be solved in
  sufficiently big sectors. In combination with   resolvent estimates this allows one
  to put the
  singularities in a  rather explicit term similar to the
  first one on the right hand side of (\ref{eq:Smatrix}). A very similar
  procedure can
 be used (at least for $V_2=0$)  for $S(0)(\omega,\omega')$ provided $\mu<1$. 
However, for
  $\mu\in [1,2[$ there is a ``glueing
  problem'' due to the fact that in order to apply resolvent estimates  in
  this case the constructed { solutions to the eikonal equations
  $\phi^\pm$}
 need to be extended,
  viz. as to including some $\theta >\tfrac {\pi}{2-\mu}$. 
 Therefore,
  multivalued $\phi^\pm$ are needed. We devote Subsection
  \ref{Quantum singularities for $V_2=0$} to a discussion of this
  question.
\item \label{it:smoloca3} Under Condition \ref{symbol}, it follows essentially
 by the same method of proof that,  for
  $\lambda>0$, the kernel $S(\lambda)(\omega,\omega')$
 is smooth
  outside the set $\{(\omega,\omega')|\;\omega= \omega'\}$; for that
  we use \eqref{eq:reduced eqns} with 
  $\mu=0$. See \cite[Chapter 19]{V2} for a related result and
  procedure. 
\item \label{it:smoloca2} There is a discrepancy between our results and the main result of \cite{Kv}. The idea of
  \cite{Kv} is to use a partial wave analysis to obtain an asymptotic expression of the scattering
  amplitude for $\lambda\to 0$ (with  the assumption of
  radial symmetry and under the short-range condition
  $\mu>1$). Unfortunately  \cite[(17)]{Kv} is incompatible with Theorems
   \ref{thm:S(0)1}, \ref{thm:shortr} and \ref{thm:sings}.
\end{enumerate}
\end{remarks}

\subsection{Distributional kernel of $S(0)$ as an oscillatory integral}
\label{Quantum singularities for $V_2=0$}

In addition to the previous assumption $V_1(r)= -\gamma r^{-\mu}$
for $r\geq 1$, we shall here assume that $V_2=0$, see though Remark
 \ref{remarks: extwe} \ref{it:ex1}). We shall explain  a
procedure which in principle  allows us to calculate the
singularities of the kernel $S(0)(\omega,\omega')$;  a fairly  explicit
{\it  oscillatory integral}
 will be specified.  Using this integral we  derive below  the location
 of the singularities of
 $S(0)$  by the method of
non-stationary phase, which gives an
alternative proof of  Theorem
\ref{thm:sings} (under the condition that  $V_2=0$).

We shall improve on the representation (\ref{eq:bndSS}) for $S(0)$.
Notice that the functions $\tilde a^+$ and $\phi^+$
 used up to now are supported near the forward region
$\cos \theta=\hat x \cdot \omega \approx 1$ only. Now we shall take advantage
of 
the fact that the expression \eqref{eq:tilll}  defines a
solution to the eikonal equation  for all values of $\theta$. We
shall consider a cut-off at larger values of  $\theta$, in fact
slightly to the left of  the critical angle $
\theta  = (1-\mu/2)^{-1}\pi$. The basic idea is similar to the one applied
in the study  of
the kernel of  scattering matrices for positive energies, cf. Remark
\ref{remarks:smoloca} \ref{it:smoloca1}). If we can extend the
construction of the phase and amplitude as indicated above,
 then we can apply a ``two-sided''  resolvent estimate to deal with the
second term on the right hand side of (\ref{eq:bndSS}), i.e. to show
that it contributes by a smooth kernel; in our case the appropriate ``two-sided'' estimate
  is given by \eqref{eq:disjoint_b}. 

Now besides the problem of extending the phase up to $
\theta  = (1-\mu/2)^{-1}\pi$, there is obviously the issue  of
well-definedness, since $\theta$ as a function of $x$ is multi-valued;
for the case of positive energies this problem does not occur since the
cut-off in this case occurs before the  angle $
\theta  = \pi$.
We have
\begin{equation}
( J^+\tau)(x)=(2\pi)^{-d/2}
 \int_{S^{d-1}}
  \big (\e^{\i \phi^+}\tilde a^{+}
\big )(x,\omega,0)\tau(\omega)\d \omega.\label{eq:qe}
\end{equation} In fact, in the present  spherically  symmetric case the dependence of the
  variables $x$ and $\omega$ is through $r=|x|$ and $\hat x \cdot
  \omega$ only. Writing 
  \begin{equation*}
    \omega=\cos \theta \,\hat x+\sin \theta \,\widetilde\omega,
  \end{equation*} where $\widetilde\omega \cdot \hat x=0$, \eqref{eq:qe}
can be written as 
\begin{equation}
(2\pi)^{-d/2}  \int _{S^{d-2}}\d \widetilde\omega \int_0^\pi
  \big (\e^{\i \phi}\tilde a
\big )(r,\theta)\tau(\cos \theta \,\hat x+\sin \theta
  \,\widetilde\omega)\sin ^{d-2}\theta\,\d \theta;\label{eq:qe1}
\end{equation} for convenience we dropped the superscript.  The phase $\phi$ is given by 
  \eqref{eq:tilll}, and using  this  expression and the
  orbit  \eqref{eq:pol-eqn}, we
  can extend the support  of $\tilde a$
  by solving transport equations as in Subsection \ref{Solving transport equations}; the cut-off is now taken
  slightly to the  left of   $
\theta  = (1-\mu/2)^{-1}\pi$. More precisely, the cut-off is defined as
  follows: First pick $L\in \N$ such that $(1-\mu/2)L<1$ while
  $(1-\mu/2)(L+1)\geq 1$. We shall assume that the analogue of $\sigma'$ for the construction of $J^-$, entering in
  \eqref{eq:chi^2} for the construction of $J^+$,  is so small that 
  \begin{equation}
    \label{eq:simasmall} (1-\mu/2)(L\pi+\cos ^{-1}(1-\sigma'))<\pi. 
  \end{equation}  Next the version of  \eqref{eq:chi^2}  that we need
  is given in terms of the $\sigma$ of the  construction of $J^-$ as
  follows: Choose angles $\pi L<\theta_0<\theta'_0<\pi (L+1)$ such
  that $(1-\mu/2)\theta'_0<\pi$ and $(1-\mu/2)\big (\theta_0
  +\cos^{-1}(1-\sigma)\big )>\pi.$ Introduce a smoothed out
  characteristic function 
\begin{equation}\label{eq:chi^2222}
  \chi_2(s)=
\begin {cases} 1, & \text{for}\; s\leq \theta_0,\\
0, & \text{for}\; s\geq \theta_0';
\end {cases} \;
\end{equation} and with this choice the new cut-off function takes the
(essentially same) form  $\chi =\chi_1(r)\chi_2(\theta)$.

The extended $\tilde a$ has 
  similar properties as before due to the cut-off.  Whence we are lead
   to consider the following  modification of the expression \eqref{eq:qe1}: 
\begin{equation*}
  \int _{S^{d-2}}\d \widetilde\omega \int_0^\infty
  f(r,\theta)\tau(\cos \theta \,\hat x+\sin \theta
  \,\widetilde\omega)\,|\sin ^{d-2}\theta|\,\d \theta;\,f=(2\pi)^{-d/2}\e^{\i
    \phi}\tilde a,
\end{equation*}  where the $\theta$--integration (due to the cut-off)  effectively takes
  place on the interval $[0, (1-\mu/2)^{-1}\pi]$. The next step is to
  change variable, writing for $\theta$ in  intervals of the form
  $(2k\pi, (2k+1)\pi]$, 
  \begin{equation*}
    \cos \theta \,\hat x+\sin \theta
  \,\widetilde\omega=\cos \psi \,\hat x+\sin \psi
  \,\widetilde\omega;\, \psi =\theta-2k\pi,
  \end{equation*} while on intervals of the form
  $(2k+1)\pi, (2k+2)\pi]$,
\begin{equation*}
    \cos \theta \,\hat x+\sin \theta
  \,\widetilde\omega=\cos \psi \,\hat x+\sin \psi
  \,(-\widetilde\omega);\, \psi =(2k+2)\pi-\theta,
  \end{equation*} respectively; here $k\in \N \cup\{0\}$.  Whence we consider the expression  
\begin{equation*}
  \int_{S^{d-1}}F(r,\psi)
  \tau(\omega)\d \omega,
\end{equation*} where 
\begin{equation*}
  F(r,\psi)=\sum _{k=0}^\infty \big \{f(r,\psi+2k\pi)+f(r,(2k+2)\pi-\psi)\big \},
\end{equation*}
 and as above 
\begin{equation*}
    \omega=\cos \psi \,\hat x+\sin \psi
    \,\widetilde\omega\text{ with }\widetilde\omega \cdot \hat x=0
    \text{ and  }\psi\in [0,\pi],
  \end{equation*} i.e. $\psi=\cos^{-1}\hat x \cdot \omega$.

We claim that $F(r,\psi)$ is smooth in $x$ and $\omega$. Notice that
this is not an  obvious fact,
  since although the function $\psi=\cos^{-1}\hat x \cdot
\omega$ is continuous, it  has a cusp singularity  at $ \hat x \cdot
\omega=\pm 1$. However, as can easily verified, $\psi^2$ is smooth at $
\hat x \cdot \omega= 1$ and $(\pi-\psi)^2$ is smooth at $ \hat x \cdot
\omega= -1$, respectively. Moreover,  $f(r,\psi)$ and
$f(r,\psi+2(k+1)\pi)+f(r,(2k+2)\pi-\psi)$ are in fact  smooth
functions  of $\psi^2$ near $ \hat x \cdot
\omega= 1$, and similarly
$f(r,\psi+2k\pi)+f(r,(2k+2)\pi-\psi)=f(r,(2k+1)\pi-(\pi-\psi))+f(r,(2k+1)\pi+(\pi-\psi))$
is a smooth
function of $(\pi-\psi)^2$ at $ \hat x \cdot
\omega= -1$.

Recall that we have the
representation (\ref{eq:Sgod2})
\begin{equation}
S(0)(\omega,\omega')=
-2\pi\langle w^+(\omega,0),\e^{\i\phi^-}\tilde t^-(\cdot,\omega',0)\rangle,
\label{eq:repo}\end{equation}
 where   $w^+(\omega,0)$ is the generalized eigenfunction 
of Theorem \ref{thm:4a}.

Define
$w=w(x,\omega)=F(r,\psi)-R(-\i 0)HF$. Due to Proposition
\ref{prop:propa5}, Proposition \ref{prop:resolvent_basic2}
(\ref{some_label21}) and Lemma \ref{lem:22a} (\ref{it:p40}),  this $w$ agrees with the eigenfunction
$w^+(\omega,0)$, cf. the proof of Lemma
\ref{prop:qs4}. Therefore, our (extended) version of (\ref{eq:bndSS}) reads
\begin{equation}
S(0)(\omega,\omega')=
-2\pi\langle F,\e^{\i\phi^-}\tilde t^-(\cdot,\omega',0)\rangle
+2\pi\langle R(-\i 0)HF,\e^{\i\phi^-}\tilde t^-(\cdot,\omega',0)\rangle.
\label{eq:repo1}\end{equation}
  As indicated above, the contribution to
$S(0)(\omega,\omega')$ from the second term on the right hand side of \eqref{eq:repo1}
 is smooth in $\omega$ and $\omega'$, if we use
a cut-off sufficiently close (but to the left of)   the critical angle
$
\theta  = (1-\mu/2)^{-1}\pi$; this is  
indeed  accomplished by using \eqref{eq:chi^2222} as cut-off function.

We conclude that the singularities of the kernel
of $S(0)$ are the same as those of the kernel
 of the operator $\widetilde S(0)$ given by
\begin{equation*}\langle \tau_1,\widetilde S(0)\tau_2\rangle=
  -2\pi\Big \langle \int F(r,\psi)
  \tau_1(\omega)\d \omega,\int \big (\e^{\i
  \phi^-}\tilde t^-\big )(\cdot,\omega',0)\tau_2
  (\omega') \d \omega' \Big \rangle.
\end{equation*} Whence (formally)
\begin{equation}\label{eq:regggg}
  \widetilde S(0)(\omega,\omega')=
-2\pi\int \overline {F(r,\psi)}\big (\e^{\i
  \phi^-}\tilde t^-\big )(\cdot,\omega',0)\,\d x.
\end{equation}

Next we introduce the variable $\theta'=\cos^{-1}\hat x \cdot
(-\omega')\in [0,\pi/2)$; we can represent  $\phi^-(x, \omega',0)=
-\phi(r,\theta')$, cf. \eqref{eq:rel_phi+-}. The  integrand  on the
right hand side of 
(\ref{eq:regggg}) is 
given as $ \sum _{k=0}^\infty f_k$, where $f_k$ has the form
  \begin{align}\label{eq:phaaa}
  &\e^{-\i\big
  (\phi(r,\psi+2k\pi)+\phi(r,\theta')\big )}g(r,\psi+2k\pi,
  \theta')\nonumber\\&+
\e^{-\i\big (\phi(r,(2k+2)\pi-\psi)+\phi(r,\theta')\big )}g(r,(2k+2)\pi-\psi, \theta').
\end{align}

Let us argue that the  integral \eqref{eq:regggg}  is well-defined in $\{\omega\cdot \omega'\neq \cos
  {\mu\over 2-\mu}\pi\}$, in agreement with Theorem
  \ref{thm:sings}. The argument is based on  the method of non-stationary
  phase. 
First we notice that the cusp singularities at $\psi =0$ and 
  $\psi =\pi$ correspond to non-stationary points. More precisely, we
  can write \begin{equation*}
    x=r(\cos \psi\,\omega+\sin \psi\,\widetilde {\hat x}),
  \end{equation*} and perform the $x$--integration as
  \begin{equation} \label{eq:regggg2}  
\int \cdots \d x= \int_0^\pi  \sin^{d-2}\psi\,\d \psi\int_{S^{d-2}}\d \widetilde {\hat
  x}\int_0^\infty  \cdots \,r^{d-1}\d r  .
  \end{equation} 
Now on the support  of $g$ the factor $\cos (1-\mu/2)\theta'\geq \cos
\theta'\geq 1-\sigma'$,
while the factors $\cos (1-\mu/2)(\psi+2k\pi)$ and $\cos
(1-\mu/2)((2k+2)\pi-\psi)$ stay sufficiently away from $-1$ (given
that  $\psi \approx 0$ or  
  $\psi \approx \pi$) to ensure that the sum of phases does not
  vanish; here we use \eqref{eq:simasmall}. Thus the phases of $f_k$ are nonzero near  the
  $\psi$--endpoints of integration, and consequently integration by parts with respect to $r$
  regularizes the integral \eqref{eq:regggg} (upon first substituting
  \eqref{eq:regggg2} and localizing near  the $\psi$--endpoints). 

 By the same reasoning as above, depending on whether $L$ is even or
  odd  (viz. $L=2l$ or $L=2l+1$), only the integral of one term of
  \eqref{eq:phaaa} (and only with $k=l$) 
  carries singularities. We first look at the 
  case for which only $\e^{-\i\big
  (\phi(r,\psi+2l\pi)+\phi(r,\theta')\big )}g(r,\psi+2l\pi,
  \theta')$ contributes by  singularities.
Clearly, for a stationary point
\begin{equation}
  \label{eq:1p}
  \cos ((1-\mu/2)(\psi+2l\pi))+ \cos ((1-\mu/2) \theta ')=0,
\end{equation} which leads to the condition
\begin{equation}
  \label{eq:1pp}
  \cos (\psi+\theta ')=\cos (\tfrac{2}{2-\mu}\pi).
\end{equation} 

There are three cases to consider.

\noindent{\bf Case  I}. $\omega=-\omega'$. In this case $\theta'=\psi$,
so that 
\begin{align}
  \label{eq:2p}
  &{\d \over \d \psi}\big
  (\phi(r,\psi+2l\pi)+\phi(r,\theta')\big )\nonumber\\&=-\sqrt {2\gamma}
  r^{1-\mu/2}
  \big(\sin (1-\mu/2)(\psi+2l\pi)+ \sin (1-\mu/2) \psi \big )<0.
\end{align} Whence there are no stationary points.

\noindent{\bf Case  II}. $\omega=\omega'$. In this case $\theta'=\pi-\psi$
so that \eqref{eq:1pp} reads 
\begin{equation*}
  \omega\cdot \omega'=1= -\cos (\tfrac{2}{2-\mu}\pi)=\cos (\tfrac{\mu}{2-\mu}\pi).
\end{equation*} This agrees with the ``rule''  of Theorem
  \ref{thm:sings}. 

\noindent{\bf Case  III}. $\omega\neq C \omega'$. In dimension $d\geq 3$ the vectors $\widetilde {\hat
x}=\pm y/|y|$ where 
$y= \omega'-\omega'\cdot \omega \;\omega$ are the  only possible critical
points of  the map
\begin{equation*}
 S^{d-2}\ni \widetilde {\hat
x}\to \theta' =\cos^{-1} (-(\cos \psi\,\omega+\sin \psi\,\widetilde
{\hat x})\cdot \omega')\in \R. 
\end{equation*} Consequently, for any stationary point, $\hat x$ must
belong to the plane spanned by $\omega$ and $\omega'$ (like for $d=2$). Let us introduce
the angle $\gamma =\cos ^{-1}\omega\cdot (-\omega')$. There  are three
possible relationships  to be considered  a) $\gamma=|\psi-\theta'|$, b)
$\gamma=\psi+\theta'$ and c) $\gamma=2\pi -(\psi+\theta')$. For a),
$\theta'=\psi\mp \gamma$ can be substituted into the sum of phases
and we compute as in \eqref{eq:2p}. Again there will not be any
stationary point.  For b) we can use \eqref{eq:1pp} to compute
\begin{equation*}
  \omega\cdot \omega'=-\cos \gamma= -\cos (\tfrac{2}{2-\mu}\pi)=\cos (\tfrac{\mu}{2-\mu}\pi),
 \end{equation*} which  agrees with the ``rule''  of Theorem
  \ref{thm:sings}. Similarly, for c) we compute 
\begin{equation*}
  \omega\cdot \omega'=-\cos \gamma=-\cos (\psi+\theta')= -\cos (\tfrac{2}{2-\mu}\pi)=\cos (\tfrac{\mu}{2-\mu}\pi).
\end{equation*}

Next we look at the 
  case for which only $\e^{-\i\big
  (\phi(r,2(l+1)\pi-\psi)+\phi(r,\theta')\big )}g(r,2(l+1)\pi-\psi,
  \theta')$ contributes to  singularities. 
The stationary point is given by
\begin{equation}
  \label{eq:1pr}
  \cos ((1-\mu/2)(2(l+1)\pi)-\psi)+ \cos ((1-\mu/2) \theta ')=0,
\end{equation} which leads to the condition
\begin{equation}
  \label{eq:1ppr}
  \cos (\psi-\theta ')=\cos (\tfrac{2}{2-\mu}\pi).
\end{equation} 

Again there are three cases to consider.

\noindent{\bf Case  I}. $\omega=-\omega'$. In this case $\theta'=\psi$,
so that 
\begin{equation*}
  \omega\cdot \omega'=-1= -\cos (\tfrac{2}{2-\mu}\pi)=\cos (\tfrac{\mu}{2-\mu}\pi), 
\end{equation*} which  agrees with Theorem
  \ref{thm:sings}.

\noindent{\bf Case  II}. $\omega=\omega'$. We have $\theta'=\pi-\psi$,
so that 
\begin{align}
  \label{eq:2pr}
  &{\d \over \d \psi}\big
  (\phi(r,2(l+1)\pi-\psi)+\phi(r,\theta')\big )\\&=\sqrt{2\gamma}r^{1-\mu/2}
  \big(\sin (1-\mu/2)(2(l+1)\pi-\psi)+ \sin (1-\mu/2) (\pi-\psi) \big )>0;\nonumber
\end{align} whence there are no stationary points.

\noindent{\bf Case  III}. $\omega\neq C \omega'$. As in the previous
  ``Case III'',   for any stationary point  the vector $\hat x$ must
belong to the plane spanned by $\omega$ and $\omega'$. Again we
define $\gamma =\cos ^{-1}\omega\cdot (-\omega')$, and  there  are three
possible relationships  to be considered:  a) $\gamma=|\psi-\theta'|$, b)
$\gamma=\psi+\theta'$ and c) $\gamma=2\pi -(\psi+\theta')$. For a),
\begin{equation*}
  \omega\cdot \omega'=-\cos \gamma=-\cos (\psi-\theta')= -\cos (\tfrac{2}{2-\mu}\pi)=\cos (\tfrac{\mu}{2-\mu}\pi),
\end{equation*} which  agrees with Theorem
  \ref{thm:sings}. For b) and c) we compute as in \eqref{eq:2pr};
  there are no stationary points.
 \begin{remarks}
    \label{remarks: extwe} 
\begin{enumerate}[\normalfont 1)]
\item \label{it:ex1}  For the above considerations (on the location of singularities),
 it is not strictly needed that $V_2=0$. In fact we can include a
 $V_2$ as in Condition \ref{assump:conditions1} with $\epsilon_2>
 1-\frac12\mu$ and solve transport equations as before using the
 same phase function (the one determined by  $V_1$ only).
\item \label{it:ex2} Suppose in addition to \ref{it:ex1})
 that $V_2$ is spherically symmetric. Then the operators $T=S(0)$, as well as
 $T=\tilde S(0)$, obey that $RTR^{-1}=T$ for all $d$-dimensional
 rotations $R$. This means that the kernel
 $T(\omega,\omega')$ of these
 operators is a function of $\omega\cdot\omega'$ only. Using the
 stationary phase method it is feasible for $\tfrac{\mu}{2-\mu}\notin
 \Z$ to write (as a possible continuation of the
above analysis)  the
singular part of the kernel of $\tilde S(0)$ as a sum 
of terms of  the form  
$(\omega\cdot\omega'-\nu\pm\i
  0)^{-\frac s2}a(\omega\cdot \omega')$
 (at least for poly-homogeneous $V_2$); we shall not elaborate.
The   next section 
 is devoted to an alternative
approach that we find more elementary,  and 
 besides,  by that  method we can extract the singular part in the exceptional cases $\tfrac{\mu}{2-\mu}\in
 \Z$ too.
\end{enumerate}
\end{remarks}

\section{Homogeneous potentials -- type
 of singularities of $S(0)$}

\label{Singularity of the kernel of the scattering matrix}
In this section we shall compute the main contribution 
to the  scattering matrix $S(0)$
for a potential 
homogeneous of degree $\mu$ (plus a lower order term), see Subsection \ref{Leading asymptotics
   of kernel} for precise conditions.
It will turn out to be  the evolution
operator for the wave equation on the sphere at time
$\frac{-\mu}{2-\mu}\pi$.  We devote Subsections \ref{prop:wav1}--\ref{thm:wav2} to a study of this  operator. In particular, we will  compute explicitly its distributional kernel and determine the
location of its singularities. We assume throughout the section that
$d\geq2$.

\subsection{Evolution operator of the wave equation on the sphere}

For any $1\leq i<j\leq d$, define the corresponding angular momentum operator
\[L_{ij}:=-\i(x_i\partial_{x_j}-x_j\partial_{x_i}).\]
Set
\[L^2:=\sum_{1\leq i<j\leq d}L_{ij}^2,\ \
\Lambda:=\sqrt{L^2+(d/2-1)^2}.\]
Note that $\Lambda$ is a self-adjoint operator on $L^2(S^{d-1})$ and
its eigenfunctions with eigenvalue $l+d/2-1$
 are $l$th order
spherical harmonics for $l=0,1,\dots$.

For any $\theta$
 one can compute exactly the integral kernel of
 $\e^{\i\theta\Lambda}$. Although the result already  appears in the
 literature, see \cite[Chapter 4, (2.13)]{Ta}, we shall for the
 readers convenience give a complete derivation  (this proof is different
 from Taylor's). 
Note  that the operator appears naturally when we solve the  wave
 equation on the sphere, therefore  we call it the
evolution operator of the wave equation
 on the sphere.

First we need to introduce some notation about distibutions.
For any $\epsilon>0$ and $s\in\R$, the expression
\[\R\in y\mapsto (y\pm\i\epsilon)^{-\frac{s}{2}}\]
defines uniquely a function on a real line, which can be viewed as a
distribution in ${\mathcal S}'(\R)$. It is well known that for any
$\phi\in{\mathcal S}(\R)$ there exists a limit
\[\lim_{\epsilon\searrow0}\int(y\pm\i\epsilon)^{-\frac{s}{2}}\phi(y)\d y=:
\int(y\pm\i 0)^{-\frac{s}{2}}\phi(y)\d y
,\]
which defines a distribution in ${\mathcal S}'(\R)$. In the sequel we will
treat this distribution as if it were a function, denoting it by
$(y\pm\i 0)^{-\frac{s}{2}}$. Note that for $s,\epsilon >0$  we have
the identity
\begin{equation}
(y\pm\i\epsilon)^{-\frac{s}{2}}=
\frac{\e^{\mp\i\pi\frac{s}{4}}}{\Gamma(s/2)}\int_0^\infty
\e^{\i t(\pm y+\i\epsilon)}t^{\frac{s-2}{2}}\d t.\label{eq:8ssb}
\end{equation}

We shall  in this section show the following result:

\begin{prop}\label{prop:wav1}
\begin{enumerate}
\item If $\theta=\pi 2k$, $k\in\Z$, then $\e^{\i\theta \Lambda}=(-1)^{kd}$
  times  the identity.
\item If $\theta=\pi (2k+1)$, $k\in\Z$, then $\e^{\i\theta \Lambda}=
\e^{\i\pi(2k+1)(d/2-1)}P$, where $P$ is the parity operator.
\item If $\theta\in]\pi 2k,\pi(2k+1)[$, $k\in\Z$, 
then $\e^{\i\theta \Lambda}$ has the distributional kernel
\[\e^{\i\theta \Lambda}(\omega,\omega')
=(2\pi)^{-d/2}\sin\theta\,\Gamma(d/2)\e^{-\i\pi/2}
(-\omega\cdot\omega'+\cos\theta-\i 0)^{-d/2}.\]
\item If $\theta\in]\pi (2k-1),\pi2k[$, $k\in\Z$, 
then $\e^{\i\theta \Lambda}$ has the distributional kernel
\[\e^{\i\theta\Lambda}(\omega,\omega')
=(2\pi)^{-d/2}\sin\theta\,\Gamma(d/2)\e^{-\i\pi/2}
(-\omega\cdot\omega'+\cos\theta+\i0)^{-d/2}.\] 
\end{enumerate}\end{prop}

\subsubsection{Tchebyshev and Gegenbauer polynomials}

Recall that the Tchebyshev polynomials (of the first kind) are defined by the
identity 
\[T_n(\cos\phi):=\cos n\phi,\ \ \ n=0,1,\dots.\]
Let $|t|<1$.
The following generating function of Tchebyshev polynomials follows by an
elementary calculation:
\begin{eqnarray}\label{tche}
-\ln(1-2wt+t^2)&=&
\sum_{l=1}^\infty \frac{2t^l}{l} T_l(w).\end{eqnarray}
 Gegenbauer 
polynomials are
defined by the generating function \cite{M,AAR}
\begin{eqnarray}\label{gege}
\frac{1}{(1-2wt+t^2)^{(d-2)/2}}&=&
\sum_{l=0}^\infty t^l C_l^{(d-2)/2}(w).\end{eqnarray}
The left hand sides of
(\ref{tche}) and (\ref{gege})
 look different. But after simple manipulations (involving differentiation
of both sides) they become quite similar
\begin{eqnarray}
\frac{-t+t^{-1}}{(t-2w+t^{-1})^{\frac{d}{2}}}&=&\begin{cases}
T_0(w)+
\sum_{l=1}^\infty t^l 2T_l(w),&d=2; \\[3ex]
\sum_{l=0}^\infty t^{l+\frac{d}{2}-1}\frac{2l+d-2}{d-2}
C_l^{(d-2)/2}(w),&d\geq3.\end{cases}
\end{eqnarray}
By substituting
$t=\e^{\i\theta}$ for $\Im \theta>0$, we rewrite this as
\begin{eqnarray}
\frac{-\i2\sin\theta}{2^{d/2}(\cos\theta-w)^{\frac{d}{2}}}&=&\begin{cases}
 T_0(w)+
\sum_{l=1}^\infty \e^{\i l \theta}2 T_l(w),&d=2; \\[3ex]
\sum_{l=0}^\infty \e^{\i(l+\frac{d}{2}-1)\theta}\frac{2l+d-2}{d-2}
C_l^{(d-2)/2}(w),&d\geq3.\end{cases}
\label{tche3}\end{eqnarray}

\subsubsection{Projection onto $l$th sector of spherical harmonics}

It is well-known that the integral kernel of the projection onto $l$th sector
of spherical harmonics in $L^2(S^{d-1})$ can be computed explicitly.
This fact 
is usually presented in the literature
as the addition theorem for spherical harmonics, see
e.g. Theorem 2, Sect. 2 of \cite{M}.
 In the case $d=3$ it can also be found in \cite{Vl}.

\begin{prop}
Let $Y$ be an $l$th order spherical harmonic in
$L^2(S^{d-1})$.
\begin{enumerate}\item In the case  $d=2$,
\begin{eqnarray}\label{green3}
\int_{S^{1} }
\frac{1}{2\pi}
T_0(\hat x \cdot\hat y)Y(\hat y)\d\hat y&=&\delta_{l0} Y(\hat x);
\\
\int_{S^{1} }
\frac{1}{\pi}
T_n(\hat x \cdot\hat y)Y(\hat y)\d\hat y&=&\delta_{ln} Y(\hat x),\ \
n=1,2,\dots.\nonumber
\end{eqnarray}
\item In the case  $d\geq3$,
\begin{equation}\label{green2}
\int_{S^{d-1} }
\frac{(d-2+2l)\Gamma(d/2-1)}{4\pi^{d/2}}
C_n^{(d-2)/2}(\hat x \cdot\hat y)Y(\hat y)\d\hat y=\delta_{ln} Y(\hat x).
\end{equation}\end{enumerate}
\label{prop:projo}\end{prop}

\proof 
The case (\ref{green3}) is elementary.
In the proof below we restrict ourselves to $d\geq3$.

Let us first recall the formula for the Green's function in ${\mathbb
  R}^d$ for $d\geq3$:
\begin{eqnarray}
G_d(x)&=&
-\frac{\Gamma(d/2-1)}{4\pi^{d/2}|x|^{d-2}}=
-\frac{1}{s_{d-1}(d-2)|x|^{d-2}},
\end{eqnarray}
where $s_{d-1}=\frac{2\pi^{d/2}}{\Gamma(d/2)}$ is the area of $S^{d-1}$.
It satisfies
\[\Delta G_d=\delta_0,\]
where $\delta_0$ is Dirac's delta at zero.
Recall also the 3rd Green's identity: if $\Delta g=0$ and $\Omega$ is a
sufficiently regular domain containing $x$, then
\begin{equation}\label{green}
g(x)=\int_{\partial\Omega} g(y)\nabla_y G_d(x-y)\d\vec{s}(y)-
\int_{\partial\Omega} (\nabla g)(y) G_d(x-y)\d\vec{s}(y).\end{equation}

We extend $Y$ to ${\mathbb R}^d$ by setting $g(x)=|x|^lY(\hat x)$.
Note that 
\[\Delta g(x)=0,\ \ \ \hat x\nabla_x g( x)=lg(x). \]
By  (\ref{gege}),
 for $|x|<|y|$,
\begin{eqnarray*}
G_d(x-y)&=&
-\frac{\Gamma(d/2-1)}{4\pi^{d/2}}
\sum_{n=0}^\infty C_n^{(d-2)/2}
(\hat x\hat y)|x|^n|y|^{-d+2-n},\\
\hat y\cdot \nabla_y G_d(x-y)&=&
\frac{\Gamma(d/2-1)}{4\pi^{d/2}}
\sum_{n=0}^\infty (d-2+n)C_n^{(d-2)/2}
(\hat x\hat y)|x|^n|y|^{-d+1-n}.
\end{eqnarray*}
We apply (\ref{green}) to the unit ball, so that $|y|=1$ and $|x|<1$:
\begin{align}
&|x|^lY(\hat x)=
\int_{S^{d-1}}
g(\hat y)\hat y\cdot \nabla G_d(x-\hat y)\d\hat y-\int_{S^{d-1}}
(\hat y\cdot \nabla g)(\hat y) G_d(x-\hat y)\d\hat y\nonumber
\\
&=
\frac{\Gamma(d/2-1)}{4\pi^{d/2}}
\sum_{n=0}^\infty (d-2+n+l)\int_{S^{d-1}}Y(\hat y)
C_n^{(d-2)/2}
(\hat x\hat y)|x|^n\d\hat y.\label{green1}
\end{align}
Comparing the powers of $|x|$
on both sides of (\ref{green1}), we obtain (\ref{green2}).
\qed

\subsubsection{Proof of  Proposition \ref{prop:wav1}}
Let $Q_l^{d-1}$
 be the orthogonal projection onto $l$th order spherical harmonics on
$S^{d-1}$. 
We multiply  (\ref{tche3})   
 by $\Gamma(d/2)2^{-1}\pi^{-d/2}$, set $w=\omega\cdot\omega'$
 and use Proposition \ref{prop:projo}.
We obtain
\begin{eqnarray}
\frac{-\i\sin\theta\,\Gamma(d/2)}{(2\pi)^{d/2}(\cos\theta-\omega\cdot\omega')^{d/2}}
&=&\sum_{l=0}^\infty Q_l^{d-1}(\omega,\omega')\e^{\i(l+d/2-1)\theta}\nonumber\\
&=&\e^{\i\theta\Lambda}(\omega,\omega').\nonumber
\end{eqnarray}

Replace $\theta$ with $\theta+\i\epsilon$, where $\theta$ is real and
$\epsilon$ positive. For small $\epsilon$ we have
\[\cos(\theta+\i\epsilon)\approx\cos\theta-\i\sin\theta \epsilon.\]
Now $\sin\theta>0$ for $\theta\in]\pi 2k,\pi (2k+1)[$
 and $\sin\theta<0$ for $\theta\in]\pi( 2k-1),\pi 2k[$, which ends the
 proof for  the
  case $\theta\in\R\setminus\pi\Z$.
  
The case $\theta\in\pi\Z$ is obvious.

\subsection{Evolution operator
 of the wave equation on the sphere as a FIO}

Let 
$X$ be a smooth compact manifold of dimension $n$. Let us recall some basic
definitions related to Fourier integral operators on $X$, cf.  \cite{Ho4}.

We say that $X\times X\times\R^k\ni (x,x',\theta)\mapsto\phi(x,x',
\theta)$ is a
non-degenerate phase function if it is a function
homogeneous of degree 1 in $\theta$,  smooth and satisfying $\nabla\phi\neq0$
away from $\theta=0$,  and
 such that
\[\{(x,x',\theta)\in X\times X\times\R^k\ |\  \nabla_\theta \phi(x,x',\theta)=0\}\]
is a smooth manifold on which $\nabla\nabla_{\theta_1}\phi,\dots,
\nabla\nabla_{\theta_k}\phi$ are linearly independent.

Let 
$\chi$ be a smooth and homogeneous 
transformation on $\T^*X\setminus X{\times}\{0\}$.
 We say that it is associated to a
non-degenerate phase function 
$\phi$ iff
  two  pairs $(x,\xi),(x',\xi')\in \T^*X\setminus \{0\}{\times}X$ satisfy
 $\chi(x',\xi'):=(x,\xi)$ exactly when
\begin{eqnarray}\nonumber
\xi&=&\nabla_x\phi(x,x',\theta),\\\nonumber
\xi'&=&-\nabla_{x'}\phi(x,x',\theta),\\
0&=&\nabla_\theta\phi(x,x',\theta).\label{setof}
\end{eqnarray}
The transformation $\chi$ is
automatically canonical, that is, it preserves the symplectic form.

We say that a smooth function  $X\times X\times\R^k\ni
(x,x',\theta)\mapsto u(x,x',\theta)$ is an amplitude of order $m$ iff
\[\partial_x^\alpha\partial_{x'}^{\alpha'}\partial_\theta^\beta u=O(\langle
\theta\rangle^{m-|\beta|}).\]

Recall from
\cite{Ho4} that an operator $U$ from $C^\infty(X)$ to ${\mathcal D}'(X)$
is called a Fourier integral operator 
of order
\[m-\frac{n}{2}+\frac{k}{2}.\]
iff in local coordinate
patches its distributional kernel can be written as
\begin{equation}
 U(x,x')=\int
 \e^{\i\phi(x,x'\theta)}u(x,x',\theta)\,\d\theta,\label{fiio}\end{equation} 
where $\theta\in\R^k$ are auxiliary variables, the function $\phi$ is a
 non-degenerate phase function, and $u$ is an amplitude of order $m$.

If the phase of $U$ is associated to a canonical transformation $\chi$, we say
that $U$ itself is associated to $\chi$.
 We note that  in such a case there are conditions
under which we  have for all
$v\in{\mathcal D}'(X)$ (using here the notion of wave front set of a distribution, cf. \cite
[Section 2.5]{Ho4}) 
\[WF( Uv)\subseteq \chi(WF(v));\]
see \cite [Proposition 2.5.7 and Theorem 2.5.14]{Ho4} (these
conditions are fulfilled for the example $U=U_\theta$ given below).

\begin{thm}\label{thm:wav2}
The operator $U_\theta:=\e^{\i\theta \Lambda}$ is a FIO of order $0$.
\end{thm}

\proof If $\theta\in\pi\Z$, then $\e^{\i\theta\Lambda}$ is a so-called point
  transformation. But point  transformations given by
  diffeomorphisms of the underlying manifold are always FIO of order zero.

Assume that  $\theta\notin\pi\Z$. Consider e.g. the case
 $\theta\in]\pi2k,\pi(2k+1)[$. By \eqref{eq:8ssb} and Proposition
 \ref{prop:wav1}  the kernel of  $U_\theta$
 can then be written as
\begin{equation}
U_\theta(\omega,\omega')
=C\int_0^\infty
\e^{\i t (\omega\cdot\omega'-\cos\theta)}t^{\frac{d-2}{2}}\d t.\label{ffio}
\end{equation}
If we compare \eqref{ffio}
with  the definition of a FIO given above, we 
see that $ t (\omega\cdot\omega'-\cos\theta)$
is a non-degenerate phase function. We also 
have $n=d-1$,
$m=\frac{d-2}{2}$ and $k=1$. Thus $U_\theta$ is a FIO of  order
\[\frac{d-2}{2}-\frac{d-1}{2}+\frac12=0.\]\qed

Let us describe the canonical transformation associated to the FIO
$U_\theta$. Let $(\omega,\xi)\in\T^*(S^{d-1})$. It is enough to
  assume that $|\xi|=1$. Then the canonical transformation $\chi_\theta$
  associated to $U_\theta$ is given by 
$\chi_\theta(\omega',\xi')=(\omega,\xi)$, where
\begin{eqnarray*}
\omega&=&\omega'\cos\theta-\xi'\sin\theta,\\
\xi&=&\omega'\sin\theta+\xi'\cos\theta.
\end{eqnarray*}

\subsection{Main result}
\label{Leading asymptotics
   of kernel}

The main result of this section is

\begin{thm}\label{thm:singrad} Suppose  Conditions
\ref{assump:conditions1}--\ref{assump:conditions3} with  $d\geq 2$, the condition $V_1(r)= -\gamma r^{-\mu}$
for $r\geq 1$, $V_2$ is spherically symmetric and that the number $\epsilon_2$ of
Condition \ref{assump:conditions1} obeys $\epsilon_2>1-\tfrac\mu 2$
(viz. $\tfrac{\d ^k}{\d r^k}V_2(r)=O\big (r^{-1-
\tfrac \mu 2-\epsilon-k}\big);\;\epsilon>0$). Then
\[S(0)=\e^{\i c_0}\e^{-\i\frac{\mu\pi}{2-\mu}\Lambda}+K,\]
where $K$ is compact and
\[c_0=\tfrac {4\sqrt{2\gamma}} {2-\mu}R_0^{1-\frac\mu 2}+2\int^\infty_{R_0}\Big(
   \sqrt{-2V_1(r')}-\sqrt{-2V(r')}\Big )\,\d r'.
\]
\end{thm}

Let us remark, as a first reduction of the proof of Theorem
\ref{thm:singrad}, that we can assume that $V_3=0$. This can readily
be seen by using  resolvent equations  in the representation
formula (\ref{eq:bndSS}) and Proposition
\ref{prop:resolvent_basic2}. Another manifestation is provided by Subsection
\ref{Quantum singularities for $V_2=0$} (including Remarks
\ref{remarks: extwe} \ref{it:ex1})): Clearly the term $\tilde S(0)$
that carries the  singularities is independent of $V_3=0$.
 
 The analysis will go through under slightly weaker conditions than
needed for Theorem  \ref{thm:singrad} (given that  $V_3=0$). Specifically, we
shall in this subsection impose the following

\begin{cond}
\label{assump:conditions11} The potential $V$ splits into a sum of
three spherically symmetric terms $V=V_1+V_2+V_3$, where all terms
$V_1,\,V_2,$ and $V_3$ are real and continuous functions  on $]0,\infty[$, $V_1(r)= -\gamma r^{-\mu}$ for $r\geq 1$,
 $V_2=O\big (r^{-1-
\tfrac \mu 2-\epsilon}\big)$ for some $\epsilon>0$, $V_1$ and $V_2$ vanish in a
neighbourhood of $r=0$,   $V_3$ has bounded support and $|V_3(r)|\leq C r^{-2+\kappa}$
for some constants $C,\kappa>0$.
\end{cond}

Under Condition \ref{assump:conditions11}, we can  {\it define} the
phase shift $\sigma_l(0)$ as follows: Fix  $l\in \N
\cup\{0\}$ and  fix $R_0\geq 0$ so large that $V(r)<0$ for all $r> R_0$. Then all real 
solutions zero energy of  the reduced Schr\"odinger equation on the
  half-line $]0,\infty[$
\begin{equation}\label{eq:efffpot}
 -u'' +V_lu=0;\;V_l(r)=2V(r)+\tfrac{(l+\tfrac d2-1)^2-4^{-1}}{r^2}
\end{equation} obey 
\begin{equation*}
  {u(r)\over r^{\tfrac{d-1}{2}}}-C{\sin \big (\int^r_{R_0}\sqrt{-2V_1(r')}\,\d r'+D\big ) \over
  (-2V_1(r))^{\tfrac{1}{4}}r^{\tfrac{d-1}{2}}}F(r>1)\in B^* _{s_0,0}
\end{equation*} for some $C>0$ and $D\in \R$ (can be seen from the WKB-analysis given
 in the bulk of  Subsection \ref{WKB-analysis}).
The {\it regular} solution is defined by the requirement \[\lim _{r\to
  0}r^{-l-\tfrac{d-1}{2}}u(r)=1.\]
 (The existence and uniqueness of the
regular solution  is usually proven  by studying an
integral equation of Volterra type, cf. \cite {Ne}.)
 Now we  {\it define} in terms of the constant $D$ for the regular solution
 \begin{equation}
   \label{eq:phase shift22}
   \sigma_l(0)=D+\int^\infty_{R_0}\Big(
   \sqrt{-2V_1(r')}-\sqrt{-2V(r')}\Big )\,\d r'+\tfrac{d-3+2l}{4}\pi.
 \end{equation}
Note that  (\ref{eq:phase shift22}) and Corollary \ref{cor:geomdef2b}
are consistent; in particular this justifies  the joint use   of the 
symbol $\sigma_l(0)$ in (\ref{eq:phase shift22}) and Corollary
\ref{cor:geomdef2b}, see Remark \ref{remark:formb} for a related discussion.

We shall show the following asymptotics:
\begin{prop}\label{prop:scatphase} Under Condition
  \ref{assump:conditions11},  
the phase shift  obeys 
\begin{align}
 &\sigma_l(0)=-\tfrac {\mu\pi} {2(2-\mu)}l+\frac c2+o(l^0),\label{eq:pphh}\\
  &\frac c2=-\tfrac {\pi \mu(d-2)} {4(2-\mu)}+\tfrac {2\sqrt{2\gamma}} {2-\mu}R_0^{1-\frac\mu 2}+\int^\infty_{R_0}\Big(
   \sqrt{-2V_1(r)}-\sqrt{-2V(r)}\Big )\,\d r.\nonumber\end{align}
\end {prop}

Clearly Theorem \ref{thm:singrad} is a consequence of   Proposition \ref{prop:scatphase}.

\begin{remark} \label{remark:formb}Note that for $0<\mu<2$ and  in dimension
  $d\geq2$, 
  $V(x)=-\gamma|x|^{-\mu}$ is an infinitesimally  form
  bounded 
  perturbation of $-\Delta$. Therefore, $H_\mu:=\Delta-\gamma|x|^{-\mu}$ is
  well-defined and self-adjoint (even though 
 Condition \ref{assump:conditions3} (\ref{it:assumption7})  may fail). 
(Actually, $H_\mu$ extends to an analytic family of operators for
  $\Re\mu\in]0,2[$.)  It may be  tempting to claim that in the cases where
  operator boundedness fails 
  one can still follow the procedures of 
  Section \ref{Propagation of singularities at zero
    energy},  i.e. use the function 
$\phi^+_{\sph}$ in (\ref{eq:tilll}) as  the zero energy
solution of the eikonal equation to construct and analyse the zero energy scattering
  matrix. However, since the resolvent estimates from \cite{FS} are only
  derived for operator bounded potentials, these estimates would need
  to be reconsidered. On the other hand, since this potential
  $V(x)=-\gamma|x|^{-\mu}$, $0<\mu<2$, indeed fulfills  Condition
  \ref{assump:conditions11}, we can use (\ref{eq:phase shift22})  in
  this case to  
  define the zero energy scattering
  matrix  (by 
the formula $S(0)Y=\e^{\i2\sigma_l(0)}Y$  for
  spherical harmonics $Y$ of order $l$). 
 Based on this definition and Proposition \ref{prop:scatphase}  it is
natural to conjecture that it equals {\em exactly} $S(0)=\e^{\i
  c_0}\e^{-\i\frac{\mu\pi}{2-\mu}\Lambda}$, or alternatively, that the
terms  $o(l^0)$ in 
Proposition \ref{prop:scatphase} vanishes identically. We leave this as an open problem for
the interested reader.
\end{remark}

\subsection{One-dimensional WKB-analysis}
\label{WKB-analysis}

This subsection is devoted to the main part of the proof of 
 Proposition \ref{prop:scatphase}. It is based on detailed 1-dimensional analysis.

For convenience, let us note that the effective potential
$V_l$ of \eqref{eq:efffpot} for $V_2=V_3=0$
is given by 
\[V_l(r)=2V_1(r)+\tfrac{k(k+1)}{r^2}=
-2\gamma r^{-\mu}+\tfrac{k(k+1)}{r^2},\;k:=l+\tfrac{d-3}{2}.\] 
  Abusing 
slightly notation, we shall henceforth denote this  expression 
(whether $V_2=V_3=0$ or not) by
  $V_k$ and similarly $\sigma_k(0):=\sigma_l(0)$. 

In the case  $V_2=V_3=0$, 
there is a unique zero, say, denoted $r_0$,  of the effective potential
$V_k$. 
Explicitly, 
\begin{equation}
  \label{eq:zeroefff}
 V_k(r_0)=0\text{ for }r_0= \big (\tfrac{k(k+1)}{2\gamma}\big )^{\tfrac{1}{2-\mu}}.
\end{equation} For later applications, let us notice that 
\begin{equation}
  \label{eq:zeroefffa}
 V_k'(r_0)=-(2-\mu)\tfrac{k(k+1)}{r_0^3}.
\end{equation}  
Clearly $V_k$ is positive to the left of $r_0$ and  negative 
to the 
right of $r_0$.

\begin{prop}\label{prop:main1}
Under Conditions \ref{assump:conditions11}, the regular solution
(up to multiplication by  a positive constant) satisfies
\begin{equation}
  \label{eq:finalans}
u(r)= (-V_k)^{-\frac 14}(r)\Big(\sin \Big (\int_{r_0}^{r} \sqrt{-V_k(r')}\,\d
r'+\tfrac\pi 4+o(k^0)\Big) +O(r^{-\epsilon_k})\Big),
\end{equation} where $o(k^0)$ signifies a vanishing term that is
{\it independent}  of $r$ and $\epsilon_k>0$.
\end{prop}

\subsubsection{Scheme of proof of Proposition \ref{prop:main1}} \label{Scheme of proof}
 We shall first concentrate on the case where $V_2=V_3=0$; the general case
will be treated by the same scheme (to be discussed later). 

We introduce a partition of $]0,\infty[$ into four subintervals given
as follows in terms of
$\epsilon_1,\epsilon_2,\epsilon_3\in (0,1]$  to be fixed later:

\noindent {\bf 1.} $ I_1=]0,r_1],\;r_1=r_0
k^{-\tfrac{\epsilon_1}{2-\mu}}.$ 

\noindent {\bf 2.} $
I_2=]r_1,r_2],\;r_2=r_0(1-k^{-\epsilon_2})
.$ 

\noindent {\bf 3.} $
I_3=]r_2,r_3],\;r_3=r_0(1+k^{-\epsilon_3})
.$ 

\noindent {\bf 4.} $
I_4=]r_3,\infty[.$ 

In  each of the intervals $I_j$ where $ j=2,3$ or $ 4$,
 we shall specify a
certain model Schr\"odinger equation together with its
two linearly independent solutions $\phi_j^\pm$.
 In terms
of these, we can  construct exact solutions to the reduced equation 
\begin{equation}\label{eq:sch11}
 -u'' +V_ku=0
\end{equation} by the method of variation of parameters,
cf. for example \cite{HS}. Our  subject  of study is formulas for the regular solution
$u= u_k$. Specifically, in the interval $I_1$ we shall use a comparison
argument to get estimates of the regular solution at  $r=r_1$. Then we
shall use a  connection formula to get estimates of the
``coefficients'' $a_2^+$ and $a_2^-$ of  the ansatz 
\begin{equation}
  \label{eq:ansa}
 u=a_j^+\phi_j^++a_j^-\phi_j^- 
\end{equation} with $j=2$ at the same point $r=r_1$. Next, using
the ODE for $a_2^+$ and $a_2^-$ we shall  derive estimates of these
quantities at $r=r_2$. Proceeding similarly 
 we  shall consecutively represent $u$ by (\ref{eq:ansa}) on $I_3$ and
 $I_4$ using connection formulas at $r_2$ and $ r_3$,
 and eventually get estimates in the interval $I_4$, and whence derive
 the relevant asymptotics of $u$.

Suppose $\phi^-$ and $\phi^+$ solve the same one-dimensional
Schr\"odinger equation, say, 
\begin{equation*}
 -\phi'' +A\phi=0.
\end{equation*}
The variation of parameter method for the equations (\ref{eq:sch11})
and (\ref{eq:ansa}) yields
\begin{equation}\label{kap1}
\left[\begin{array}{cc}\phi^+&\phi^-\\
\frac{\d}{\d\tau}\phi^+&\frac{\d}{\d\tau}\phi^-
\end{array}\right]\frac{\d}{\d\tau}
\left[\begin{array}{c}a^+\\a^-\end{array}\right]
=(V-A)
\left[\begin{array}{cc}0&0\\
\phi^+&\phi^-
\end{array}\right]
\left[\begin{array}{c}a^+\\a^-\end{array}\right]
.\end{equation}
(We have
omitted the subscript $j$). We introduce the notation
$W(\phi^-,\phi^+)$ for the Wronskian
$W(\phi^-,\phi^+)=\phi^-\tfrac{\d}{\d r}\phi^{+}-\phi^+\tfrac{\d}{\d r}\phi^{-}$.
Then we write $B=V_k-A$ and  transform \eqref{kap1} into
\begin{equation*} 
 \dfrac{\d}{\d r}\Big ({a^+ \atop a^-}\Big )=N \Big ({a^+ \atop a^-}\Big ),
\end{equation*} where  
\begin{displaymath} 
 N= \frac {B}{W(\phi^-,\phi^+)}\left ( \begin{array}{cc}
\phi^-\phi^+ &(\phi^-)^2\\
-(\phi^+)^2&-\phi^-\phi^+ 
\end{array}\right ).
\end{displaymath}  

For a positive increasing continuous function $f$ on $I$ (to be specified),
 we introduce  the
matrix $T=\diag
(1, f^{-1})$. We compute 
\begin{displaymath} 
 TNT^{-1}= \frac {B}{W(\phi^-,\phi^+)}\left ( \begin{array}{cc}
\phi^-\phi^+ &f(\phi^-)^2\\
-f^{-1}(\phi^+)^2&-\phi^-\phi^+ 
\end{array}\right ). 
\end{displaymath} 
 Introducing the operator $(M_jz)(r)=\int_{r_{j-1}}^r N_j(r')z(r')\,\d
 r',\,j\geq2, $
acting on continuous functions $z(\cdot):I_j\to
\R^2$, the
above ODE is solved by
\begin{equation*} 
 \Big ({a^+_j \atop a^-_j}\Big )(r)- z_j=\sum _{m=1}^\infty M_j^m z_j;\;z_j=\Big ({a^+_j \atop a^-_j}\Big )(r_{j-1}).
\end{equation*} 
 Whence we have   the bound  
\begin{equation} \label{eq:decom8}
 \Big \|T_j(r) \Big \{\Big ({a^+_j \atop a^-_j}\Big )(r)- z_j\Big \} \Big \|\leq \sum _{m=1}^\infty \Big
\| \Big(\big (T_jM_jT_j^{-1}\big )^{m} T_jz_j\Big)(r)\Big \|;
\end{equation} to the right  $T_j$ is considered as an operator acting as $(T_jz)(r')=(T_j)(r')z(r')$.
 Using that $f_j$ is increasing, we can estimate
\begin{equation*}
  \|(T_jM_jT_j^{-1}z)(r)\|\leq \int
  _{r_{j-1}}^r\big \|(T_jN_jT_j^{-1})(r')\big \|\,\|z(r')\|\,\d r',
\end{equation*} which applied repeatedly in (\ref{eq:decom8}) yields the following bound  for
$r\in I_j$: 
\begin{align} \label{eq:finalest}
 &\Big \|T_j(r) \Big \{\Big ({a^+ \atop a^-}\Big )(r)- z_j\Big \} \Big \| \nonumber\\&\leq \Big \{ \Big (\exp \int
_{r_{j-1}}^r\big \|(T_jN_jT_j^{-1})(r')\big \|\,\d r'\Big )-1\Big \}\sup_{\tilde
  r\in I_j}\|T_j(\tilde r)z_j\|\nonumber\\&= \Big \{ \Big (\exp \int
_{r_{j-1}}^r\big \|(T_jN_jT_j^{-1})(r')\big \|\,\d r'\Big )-1\Big \}\,\|z_j\|.
\end{align}
 
We specify in the following $\phi_j^\pm$, $ B_j$  and $f_j$ for $j=2,3$ and $4$; in all cases $W(\phi_j^-,\phi_j^+)=1$:

\noindent{\bf Ad interval $I_2$.} We define 
\begin{subequations}
\begin{equation}
  \label{eq:81}
 \phi_2^\pm(r)=2^{-\frac12}V_k^{-\frac 14}\e^{\pm \int_{r_1}^r\sqrt {V_k}\,\d r'}, 
\end{equation}  compute
\begin{equation}
  \label{eq:81a}
  B_2=-\big (V^{-\frac 14}_k\big )''V^{\frac
    14}_k=-\tfrac5{16}\Big(\frac{V_k'}{V_k}\Big )^2+\tfrac1{4}\frac{V_k''}{V_k}
\end{equation} and let 
\begin{equation}
 f_2(r)=\frac{\phi_2^+(r)}{\phi_2^-(r)}=\e^{2 \int_{r_1}^r\sqrt {V_k}\,\d r'}.\label{eq:81b}
\end{equation}
\end{subequations}

\noindent{\bf Ad interval $I_3$.} We define (in terms of the Airy function,
cf. \cite {HS} and \cite [Definition 7.6.8]{Ho1})
\begin{subequations}
\begin{align}
\label{eq:838a}
 \phi_3^+(r)=&\sqrt \pi \zeta^{-1}\Ai \big
 (-\zeta^2(r-r_0)\big);\;\zeta:=|V'_k(r_0)|^{\frac16},\\
 \phi_3^-(r)=&\sqrt \pi \e ^{\tfrac{\pi\i}{6}}\zeta^{-1}\Ai \big (-\zeta^2\e ^{\tfrac{2\pi\i}{3}}(r-r_0)\big)\nonumber\\&+\sqrt \pi \e ^{-\tfrac{\pi\i}{6}}\zeta^{-1}\Ai \big (-\zeta^2\e ^{-\tfrac{2\pi\i}{3}}(r-r_0)\big),\label{eq:839a} 
\end{align}  compute
\begin{equation}
  \label{eq:83aa}
  B_3(r)=V_k(r)-\big (V_k(r_0)+V_k'(r_0)(r-r_0)\big )=\int
  _{r_0}^r(r-\tilde{ r}) V''_k(\tilde{r})\,\d \tilde{r}
\end{equation} and let 
\begin{equation}
  f_3(r)=\begin{cases}\exp \big (-\frac 43 \zeta^3 (r_0-r)^{\frac32}\big ),
&\text{ if
    }r<r_0;\\  1, &\text{ if
    }r\geq r_0. \end {cases}\label{eq:83ba}
\end{equation}
\end{subequations}

\noindent{\bf Ad interval $I_4$.} We define 
\begin{subequations}
\begin{align}
\label{eq:838}
 &\phi_4^+(r)=(-V_k)^{-\frac 14}\sin \Big (\int _{r_0}^r \sqrt {-V_k}\,\d
 r'+\frac \pi 4\Big),\\& \phi_4^-(r)=(-V_k)^{-\frac 14}\cos \Big (\int _{r_0}^r \sqrt {-V_k}\,\d
 r'+\frac \pi 4\Big)\label{eq:839}, 
\end{align}  compute
\begin{equation}
  \label{eq:83a}
  B_4=-\big ((-V_k)^{-\frac 14}\big )''(-V_k)^{\frac
    14}=-\tfrac5{16}\Big(\frac{V_k'}{V_k}\Big )^2+\tfrac1{4}\frac{V_k''}{V_k}
\end{equation} and let 
\begin{equation}
  f_4=1.\label{eq:83b}
\end{equation}
\end{subequations}

\subsubsection{Details of proof  of Proposition \ref{prop:main1}}
\label{ Details of proof} We start implementing the scheme outlined  in Subsubsection
\ref{Scheme of proof}. 

 In the interval $I_1$ we shall use  a standard comparison argument. With $V_k$
 replaced by $V=\tfrac{\tilde k(\tilde k+1)}{r^2}$, the regular solution is given by the
 expression $u=r^{\tilde k+1}$ and the corresponding Riccati equation
 \begin{equation}
   \label{eq:ric}
  \psi'=V-\psi^2 
 \end{equation}  is solved by $\psi=\tfrac{\phi'}{\phi}=\tfrac
 {\tilde k+1}{r}$. 

We fix
$\epsilon_1\in ]0,1]$ (actually $\epsilon_1>0$ can be chosen arbitrarily) 
 and notice  the following uniform bound in $r\in I_1$:
\begin{equation}
\label{eq:eps11} V_k(r)= \tfrac{k(k+1)}{r^2}\Big (1+O\big (k^{-\epsilon_1}\big
)\Big ).
\end{equation} 

Using (\ref{eq:eps11}), we can find $C>0$ such that with $k^\pm:=k(1\pm
Ck^{-\epsilon_1})$ and $V^\pm_k(r):=\tfrac{k^\pm(k^\pm+1)}{r^2}$ there are 
estimates
\begin{equation*}
 V_k(r)
\begin {cases}\leq V^+_k(r)\\ \geq  V^-_k(r)
\end {cases},\;r\in I_1.
\end{equation*}

Now, by using \cite[Theorem 1.8]{BR} and  the Riccati equation, it follows that the regular
solution $u$ of (\ref{eq:sch11}) is positive in  $I_1$, and that
$v:=\tfrac{u'}{u}$ obeys the bounds  
\begin{equation*}
  v(r)\begin {cases}\leq \tfrac
 {k^++1}{r}\\ \geq  \tfrac
 {k^-+1}{r}
\end {cases},\;r\in I_1.
\end{equation*}  
We conclude  the uniform bound
\begin{equation}
  \label{eq:8ric}
  v(r)=\tfrac
 {k+1}{r}\Big (1+O\big (k^{-\epsilon_1}\big
)\Big ),\;r\in I_1.
\end{equation} 

The connection formula at $r=r_1$ reads
\begin{equation}
  \label{eq:connec}
  c_j\Bigg ({1\atop v}\Bigg)_{r=r_{j-1}}=\Bigg ({a_j^+\phi_j^++a_j^-\phi_j^-\atop a_j^+(\phi_j^+)'+a_j^-(\phi_j^-)'}\Bigg)_{r=r_{j-1}},\;j=2.
\end{equation} 

Obviously, (\ref{eq:connec}) is solved for the coefficients by 
\begin{equation}
  \label{eq:connec2}
  \Bigg ({a_j^+\atop a_j^-}\Bigg)_{r=r_{j-1}}=\frac{c_j}{W(\phi_j^-,\phi_j^+)}\Bigg ({(-\phi_j^-)'+\phi_j^-v\atop (\phi_j^+)'-\phi_j^+v}\Bigg)_{r=r_{j-1}},\;j=2.
\end{equation}

 Next, from (\ref{eq:81}) we compute
 \begin{equation}\label{eq:derv2}
 (\phi_2^\pm)'=\Big (\pm \sqrt {V_k}-\frac14 \frac{V_k'}{V_k}\Big )\phi_2^\pm.  
 \end{equation} We substitute  these expressions and (\ref{eq:8ric})
 in the right hand side of 
 (\ref{eq:connec2}) and obtain 
\begin{equation}
  \label{eq:connec3}
  \Bigg ({a_2^+(r_1)\atop a_2^-(r_1)}\Bigg)=c_2\frac{2k}{r_1} \Bigg
  ({1+O\big (k^{-\epsilon_1} \big ) \atop O\big (k^{-\epsilon_1}\big)}\Bigg ).
\end{equation}

To apply (\ref{eq:finalest}), we notice that 
\begin{displaymath} 
 T_2N_2T^{-1}_2= {B_2}\phi^-_2\phi^+ _2\left ( \begin{array}{cc}
1&1\\
-1&-1 
\end{array}\right )={B_2}O\big (V_k^{-\frac 12}\big ). 
\end{displaymath}

Whence (for the first inequality below we assume  that the integral is bounded in
$k$ so that the inequality $\exp x-1\leq Cx$ applies -- this will be 
justified by (\ref{eq:reqeps2})),
\begin{align} \label{eq:finalest2}
 &\Big \|T_2(r_2) \Big \{\Big ({a^+_2 \atop a^-_2}\Big )(r_2)- \Big ({a^+_ 2 \atop a^-_2}\Big )(r_1)\Big \} \Big \| \nonumber\\& =\Big \{ \big (\exp \int
_{r_{1}}^{r_2}\Big |\Big (-\tfrac5{16}\Big(\frac{V_k'}{V_k}\Big
)^2+\tfrac1{4}\frac{V_k''}{V_k}\Big )O\big (V_k^{-\frac 12}\big) \Big
|\,\d r'\big )-1\Big \}O\Big (\frac{k}{r_1}\Big)\nonumber\\& \leq C_1  \tfrac{k}{r_1}r_0\int
_{r_1/r_0}^{r_2/r_0}\Bigg (\frac{\big
  (V_k'\big)^2}{V_k^{\frac52}}+\frac{\big |V_k''\big
  |}{V_k^{\frac32}}\Bigg )\,\d s\;(\text{changing variables }r'=r_0s)\nonumber\\& \leq C_2 r_1^{-1} \int
_{r_1/r_0}^{r_2/r_0}\Bigg (\frac{s^{-6}}{\big
    (s^{-2}-s^{-\mu}\big )^\frac52}+\frac{s^{-4}}{\big
    (s^{-2}-s^{-\mu}\big )^\frac32}\Bigg )\,\d s\nonumber\\&=C_2
  r_1^{-1} \Bigg (\int
_{1/2}^{r_2/r_0}\cdots \,\d s+\int
_{r_1/r_0}^{1/2}\cdots\,\d s\Bigg )\nonumber\\& \leq C_3  r_1^{-1} \max \Big (\int
_{1/2}^{r_2/r_0}(1-s^{2-\mu}\big )^{-\frac52}\,\d s,\int
_{r_1/r_0}^{1/2}s^{-1}\,\d s\Big )\nonumber\\& \leq C_4
k^{\frac32 \epsilon_2-1}\,\tfrac{k}{r_1};
\end{align} we need here
\begin{equation}
  \label{eq:reqeps2}
  \tfrac32 \epsilon_2-1<0.
\end{equation}

We conclude by combining  (\ref{eq:connec3}) and (\ref{eq:finalest2}):
\begin{equation}
  \label{eq:connec3a}
  \Bigg ({a_2^+(r_2)\atop a_2^-(r_2)}\Bigg)=c_2\frac{2k}{r_1} \Bigg
  ({1+O\big (k^{-\epsilon_1}\big ) +O\Big ( k^{\frac32 \epsilon_2-1}\Big) \atop O\big (k^{-\epsilon_1}\big)+\e^{2\int_{r_1}^{r_2}\sqrt {V_k}\,\d r'}O\Big ( k^{\frac32 \epsilon_2-1}\Big)}\Bigg ).
\end{equation}

Next we repeat the above procedure  passing from the interval $I_2$ to
$I_3$. 

The first issue is the connection formula (\ref{eq:connec}) with
$j=2$ replaced by $j=3$. The left hand side can be estimated using
(\ref{eq:derv2}), (\ref{eq:connec3a}) and the following estimates
(where (\ref{eq:reqeps2}) is used):
\begin{align}
  \label{eq:8est3a}
   \sqrt{V_k(r_2)}&=\frac{\sqrt{k(k+1)}}{r_2}\Big
  (1-\big(1-k^{-\epsilon_2}\big )^{2-\mu} \Big)^\frac12\nonumber\\
&=\frac{k}{r_0}(2-\mu)^{\frac12} k^{-\tfrac{\epsilon_2}{2}}\Big
  (1+O\big ( k^{-\epsilon_2}\big)\Big),\\ \frac{V_k'(r_2)}{V_k(r_2)}&=\frac{O\Big (\frac{k^2}{r_2^3}\Big )}{V_k(r_2)}=r_2^{-1}O\big ( k^{\epsilon_2}\big).\label{eq:8est3b}
\end{align} Notice that (\ref{eq:8est3a})  dominates
(\ref{eq:8est3b}) (by (\ref{eq:reqeps2}) again), so that 
\begin{equation*}
 \Big(\sqrt {V_k}-\frac14 \frac{V_k'}{V_k}\Big)(r_2) =(2-\mu)^{\frac12} \frac{k}{r_0}k^{-\frac{\epsilon_2}{2}}\Big
  (1+O\big ( k^{-\epsilon_2}\big)\Big).
\end{equation*} We conclude that 
\begin{align}
 v(r_2)&=\frac{(\phi_2^+)'(r_2)}{\phi_2^+(r_2)}\Big(1+O\big ( k^{\frac32 \epsilon_2-1}\big)\Big)\nonumber\\&=(2-\mu)^{\frac12} \frac{k}{r_0}k^{-\tfrac{\epsilon_2}{2}}\Big
  (1+O\big ( k^{-\epsilon_2}\big)+O\big ( k^{\frac32 \epsilon_2-1}\big)\Big).\label{eq:ddom}
\end{align} 
 
By  (\ref{eq:connec}) and (\ref{eq:connec2}) with
$j=2$ replaced by $j=3$, up to multiplication by a positive constant, 
\begin{equation}
  \label{eq:connec255}
  \Bigg ({a_3^+\atop a_3^-}\Bigg)_{r=r_{2}}=\Bigg ({(-\phi_3^-)'+\phi_3^-v\atop (\phi_3^+)'-\phi_3^+v}\Bigg)_{r=r_{2}}.
\end{equation} It remains to examine the asymptotics of
$\phi_3^\pm$ and their derivatives at $r_2$. For that we notice the
asymptotics as $r-r_0\to -\infty$, cf.   \cite [Appendix B]{HS} and \cite[(7.6.20)]{Ho1},
\begin{subequations}
\begin{align}
  \label{eq:asymp1-}
 \phi_3^+&=  \tfrac{\exp \big(
-\tfrac23 \zeta^3(r_0-r)^{\frac 32}\big)}{2\zeta^{\frac32}(r_0-r)^{\frac 14}}\Big(1+O\big( \zeta^{-3}(r_0-r)^{-\frac 32}\big)\Big),\\ 
 (\phi_3^+)'&=  \zeta^3(r_0-r)^{\frac12}\tfrac{\exp \big(
-\tfrac23 \zeta^3(r_0-r)^{\frac 32}\big)}{2\zeta^{\frac32}(r_0-r)^{\frac 14}}\Big(1+O\big( \zeta^{-3}(r_0-r)^{-\frac 32}\big)\Big),\label{eq:asymp1b-}\\
 \phi_3^-&=  \tfrac{\exp \big(
\tfrac23 \zeta^3(r_0-r)^{\frac 32}\big)}{\zeta^{\frac32}(r_0-r)^{\frac 14}}\Big(1+O\big( \zeta^{-3}(r_0-r)^{-\frac 32}\big)\Big),\label{eq:asymp1c-}\\
(\phi_3^-)'&=  -\zeta^3(r_0-r)^{\frac12}\tfrac{\exp \big(
\tfrac23 \zeta^3(r_0-r)^{\frac 32}\big)}{\zeta^{\frac32}(r_0-r)^{\frac 14}}\Big(1+O\big( \zeta^{-3}(r_0-r)^{-\frac 32}\big)\Big)\label{eq:asymp227-}.
\end{align}
\end{subequations}
  Since $\zeta^3(r_0-r_2)^{\frac 32}\asymp
 \sqrt{2-\mu}k^{1-\frac 32\tfrac{\epsilon_2}{2-\mu}}$,
 cf. (\ref{eq:zeroefffa}),   these  asymptotics
  are applicable. By the same
 computation, (\ref{eq:ddom}) can be rewritten as 
\begin{equation}
 v(r_2)=\zeta^3(r_0-r_2)^{\frac12} \Big
  (1+O\big ( k^{-\epsilon_2}\big )+O\big ( k^{\frac32 \epsilon_2-1}\big)\Big).\label{eq:ddom2}
\end{equation}
Whence, in 
 conjunction (\ref{eq:connec255}), we obtain 
(up to multiplication by a positive constant) 
\begin{equation}
  \label{eq:connec255uu}
  \Bigg ({a_3^+(r_2)\atop a_3^-(r_2)}\Bigg)=\Bigg ({\exp \big(
\frac23 \zeta^3(r_0-r_2)^{\frac 32}\big) \Big
  (1+O\big ( k^{-\epsilon_2}\big)+O\big ( k^{\frac32 \epsilon_2-1}\big)\Big)\atop \exp \big(
-\frac23 \zeta^3(r_0-r_2)^{\frac 32}\big)\Big
  (O\big ( k^{-\epsilon_2}\big)+O\big ( k^{\frac32 \epsilon_2-1}\big)\Big)}\Bigg).
\end{equation} 

Next, to apply (\ref{eq:finalest}) with $j=3$ we need 
the following asymptotics of
$\phi_3^\pm$ and their derivatives  as $r-r_0\to +\infty$, cf.   \cite [Appendix B]{HS} and \cite[(7.6.20) and (7.6.21)]{Ho1}:
\begin{subequations}
\begin{align}
  \label{eq:asymp1-9}
 \phi_3^+&=  \zeta^{-\frac32}(r-r_0)^{-\frac 14}\Big(\sin \big(
\tfrac23 \zeta^3(r-r_0)^{\frac 32}+\tfrac\pi 4\big)+O\big( \zeta^{-3}(r-r_0)^{-\frac 32}\big)\Big),\\ 
 (\phi_3^+)'&=  \zeta^{\frac32}(r-r_0)^{\frac 14}\Big(\cos \big(
\tfrac23 \zeta^3(r-r_0)^{\frac 32}+\tfrac\pi 4\big)+O\big( \zeta^{-3}(r-r_0)^{-\frac 32}\big)\Big),\label{eq:asymp1b-9}\\
 \phi_3^-&=  \zeta^{-\frac32}(r-r_0)^{-\frac 14}\Big(\cos  \big(
\tfrac23 \zeta^3(r-r_0)^{\frac 32}+\tfrac\pi 4\big)+O\big( \zeta^{-3}(r-r_0)^{-\frac 32}\big)\Big),
\label{eq:asymp227-9}\\(\phi_3^-)'&=  -\zeta^{\frac32}(r-r_0)^{\frac 14}\Big(\sin \big(
\tfrac23 \zeta^3(r-r_0)^{\frac 32}+\tfrac\pi 4\big)+O\big(
\zeta^{-3}(r-r_0)^{-\frac 32}\big)\Big).\label{eq:asymp1c-9} 
\end{align} 
\end{subequations} 
In particular,  
\begin{displaymath} 
 T_3N_3T^{-1}_3= B_3\zeta^{-2}O\big (k^{0}\big )\text{ uniformly in }r\in I_3. 
\end{displaymath} In conjunction with (\ref{eq:finalest}),
(\ref{eq:zeroefffa})  and the
fact that  
\begin{equation} \label{eq:2.der}
 V_k''(r)= O\Big (k^{-\tfrac{4+2\mu}{2-\mu}}\Big )\text{ uniformly in }r\in
 I_3,  
\end{equation} we obtain 
\begin{align} \label{eq:finalest28}
 &\Big \|T_3(r_3) \Big \{\Big ({a^+_3 \atop a^-_3}\Big )(r_3)- \Big
 ({a^+_ 3 \atop a^-_3}\Big )(r_2)\Big \} \Big \| \nonumber\\&
 \leq C_1\Big((r_3-r_0)^3+(r_0-r_2)^3\Big) k^{-\tfrac{4+2\mu}{2-\mu}}k^{\frac23 \tfrac{1+\mu}{2-\mu}}a_3^+(r_2)\nonumber\\&\leq C_2k^{\frac43-{3}\min(\epsilon_2,\epsilon_3)}a_3^+(r_2);
\end{align} here we need
\begin{equation}
  \label{eq:reqeps2333}
  \tfrac43-{3}\min(\epsilon_2,\epsilon_3)<0,
\end{equation} cf. (\ref{eq:reqeps2}). At this point  let us for
convenience take $\epsilon_3=\epsilon_2$, so that
(\ref{eq:reqeps2333}) simplifies and in conjunction with
(\ref{eq:reqeps2}) leads to the single requirement 
\begin{equation}
  \label{eq:reqeps2333b}
\tfrac23>\epsilon_2=\epsilon_3>\tfrac49.
\end{equation} 

We conclude that (up to multiplication by the  positive constant $a_3^+(r_2)$) 
\begin{equation}\label{eq:bo=ndii}
  \Bigg ({a_3^+(r_{3})\atop a_3^-(r_{3})}\Bigg)=\Bigg ({1+O\big ( k^{\tfrac43-{3}\epsilon_2}\big)\atop O\big ( k^{\tfrac43-{3}\epsilon_2}\big)}\Bigg).
\end{equation} 

Next we need to study the  connection formula passing from $I_3$ to
$I_4$; a little linear algebra takes it to the form 
\begin{displaymath} 
 \Big ({a^+ _4\atop a^-_4}\Big )= \left ( \begin{array}{cc}
W(\phi^-_4,\phi^+_3)&W(\phi^-_4,\phi^-_3)\\W(\phi^+_3,\phi^+_4)
&W(\phi^-_3,\phi^+_4)
\end{array}\right )\Big ({a^+ _3 \atop a^-_3}\Big ),\;r=r_{3}. 
\end{displaymath}  
So we need to compute the appearing Wronskians. To this end we note the following uniform asymptotics for $r\in
[r_0,r_3]$, which are readily obtained from (\ref{eq:zeroefffa}) and
(\ref{eq:2.der}) (recall that by now $\epsilon_3=\epsilon_2$):
\begin{subequations}
\begin{align}
  \label{eq:as1}
 V_k(r)&=V'_k(r_0)(r-r_0)\big(1+ O\big
 (k^{-\epsilon_2}\big)\big ),\\
\label{eq:as42}
 V'_k(r)&=V'_k(r_0)\big(1+ O\big (k^{-\epsilon_2}\big)\big ),\\\label{eq:as3}
\sqrt{-V_k(r)}&=\zeta^3(r-r_0)^\frac12\big(1+ O\big
 (k^{-\epsilon_2}\big)\big ),\\
\label{eq:as4}
 \int_{r_0}^r \sqrt{-V_k(r')}\,\d r'&=\frac23 \zeta^3(r-r_0)^\frac32\big(1+ O\big
 (k^{-\epsilon_2}\big)\big )\nonumber\\&=\frac23 \zeta^3(r-r_0)^\frac32+O\Big
 (k^{1-\frac52 \epsilon_2}\Big).
\end{align}
\end{subequations} Due to (\ref{eq:as3}) and (\ref{eq:as4}), the
asymptotics (\ref{eq:asymp1-9})--~(\ref{eq:asymp1c-9}) at the point $r=r_3$ can be written
in terms of \[\theta:=\int_{r_0}^{r_3} \sqrt{-V_k(r')}\,\d
r'+\tfrac\pi 4\] as 
\begin{subequations}
\begin{align}
  \label{eq:asymp1-9y}
 \frac{\phi_3^+(r_3)}{(-V_k(r_3))^{-\frac 14}}&=  \sin \Big(\theta+O\big
 (k^{1-\frac52
   \epsilon_2}\big)\Big)+O\big
 (k^{-\epsilon_2}\big)+O\big
 (k^{\frac32 \epsilon_2-1}\big),\\ 
  \frac{(\phi_3^+)'(r_3)}{(-V_k(r_3))^{\frac 14}}&=  \cos \Big(\theta+O\big
 (k^{1-\frac52
   \epsilon_2}\big)\Big)+O\big
 (k^{-\epsilon_2}\big)+O\big
 (k^{\frac32 \epsilon_2-1}\big),\label{eq:asymp1b-9y}\\
  \frac{\phi_3^-(r_3)}{(-V_k(r_3))^{-\frac 14}}&= \cos \Big(\theta+O\big
 (k^{1-\frac52
   \epsilon_2}\big)\Big)+O\big
 (k^{-\epsilon_2}\big)+O\big
 (k^{\frac32 \epsilon_2-1}\big), 
\label{eq:asymp227-9y}\\\frac{(\phi_3^-)'(r_3)}{(-V_k(r_3))^{\frac 14}}&= -\sin \Big(\theta+O\big
 (k^{1-\frac52
   \epsilon_2}\big)\Big)+O\big
 (k^{-\epsilon_2}\big)+O\big
 (k^{\frac32 \epsilon_2-1}\big).\label{eq:asymp1c-9y} 
\end{align} 
\end{subequations} 

Next, using that
\begin{equation*}
  \frac{-V'_k}{(-V_k)^{\frac32}}(r_3)=O\Big
 (k^{\tfrac 32 \epsilon_2-1}\Big),
\end{equation*} cf. (\ref{eq:as1}) and (\ref{eq:as42}), we obtain for
the functions 
$\phi^\pm_4$
\begin{subequations}
\begin{align}
  \label{eq:asymp1-9yi}
\phi_4^+(r_3)&=  (-V_k(r_3))^{-\frac 14}\sin \big(\theta\big),\\ 
 (\phi_4^+)'(r_3)&=  (-V_k(r_3))^{\frac 14}\Big(\cos \big(\theta\big)+O\big
 (k^{\frac32 \epsilon_2-1}\big)\Big),\label{eq:asymp1b-9yi}\\
 \phi_4^-(r_3)&= (-V_k(r_3))^{-\frac 14}\cos \big(\theta\big), 
\label{eq:asymp227-9yi}\\(\phi_4^-)'(r_3)&= -(-V_k(r_3))^{\frac 14}\Big(\sin \big(\theta\big)+O\big
 (k^{\frac32 \epsilon_2-1}\big)\Big).\label{eq:asymp1c-9yi} 
\end{align} 
\end{subequations} 

The matrix of  Wronskians is readily computed using
(\ref{eq:asymp1-9})--(\ref{eq:asymp1c-9}) and
(\ref{eq:asymp1-9y})--~(\ref{eq:asymp1c-9y}). In combination with
(\ref{eq:bo=ndii}), we obtain (using in the second step (\ref{eq:reqeps2333b})) 
\begin{align}
  \Bigg ({a_4^+(r_{3})-1\atop a_4^-(r_{3})}\Bigg)&=O\Big ( k^{-\epsilon_2}\Big)+O\Big ( k^{\frac32 \epsilon_2-1}\Big)+O\Big ( k^{\tfrac43-{3}\epsilon_2}\Big)+O\Big
 (k^{1-\frac52
   \epsilon_2}\Big)\nonumber\\ &=O\Big ( k^{-\epsilon_2}\Big)+O\Big ( k^{\frac32 \epsilon_2-1}\Big)+O\Big ( k^{\tfrac43-{3}\epsilon_2}\Big).\label{eq:bo=ndii2}
\end{align}

Now we estimate in $I_4$ using (\ref{eq:bo=ndii2}) (and mimicking
partially 
 (\ref{eq:finalest2}))
\begin{align}
 &\Big \|\Big ({a^+_4 \atop a^-_4}\Big )(r)- \Big ({a^+_ 4 \atop
   a^-_4}\Big )(r_3) \Big \| \nonumber\\& \leq C_1\Big \{ \big (\exp \int
_{r_{3}}^{r}\Big |\Big (-\tfrac5{16}\Big(\frac{V_k'}{V_k}\Big
)^2+\tfrac1{4}\frac{V_k''}{V_k}\Big )O\big ((-V_k)^{-\frac 12}\big) \Big
|\,\d r'\big )-1\Big \}\nonumber\\& \leq C_2 r_0\int
_{r_3/r_0}^{r/r_0}\Bigg (\frac{\big
  (-V_k'\big)^2}{(-V_k)^{\frac52}}+\frac{\big |-V_k''\big
  |}{(-V_k)^{\frac32}}\Bigg )\,\d s\;(\text{changing variables }r'=r_0s)\nonumber\\& \leq C_3 r_0^{\mu/2-1} \int
_{r_3/r_0}^{r/r_0}s^{\mu/2-2}\Big (\big
    (1-s^{\mu-2}\big )^{-\frac52}+\big
    (1-s^{\mu-2}\big )^{-\frac32}\Big )\,\d s\nonumber\\&\leq C_4 r_0^{\mu/2-1} 
 \Big (\int
_2^{\infty}s^{\mu/2-2}\,\d s+\int
_{r_3/r_0}^{2}\big
    (1-s^{\mu-2}\big )^{-\frac52}\,\d s\Big )\nonumber\\& \leq C_5
    r_0^{\mu/2-1} k^{\frac32 \epsilon_2}\nonumber\\&
    =O\Big ( k^{\frac32 \epsilon_2-1}\Big). \label{eq:finalest2y}
\end{align}

By the same type of   estimation  we also deduce that for fixed $k$ there exist 
$\epsilon_k>0$ and $a^\pm_4(\infty)\in \R$ such that 
\begin{equation*}
a^\pm_4(r)=a^\pm_4(\infty)+O(r^{-\epsilon_k}).
  \end{equation*}
 By applying (\ref{eq:finalest2y}) with $r=\infty$ in combination with
 (\ref{eq:bo=ndii2}) 
(and using  an elementary trigonometric formula), we
conclude that \eqref{eq:finalans} is true.

\noindent{\bf The general case.} It remains to prove  (\ref{eq:finalans}) 
 under
Condition \ref{assump:conditions11} (i.e. without assuming that $V_2=V_3=0$).
 All previous constructions and estimates carry over, so  below we consider only some additional estimates that  are  needed.
Denoting $U=2V_2+2V_3$, the functions $\phi^\pm_j$ and  $f_j$  and the potentials $A_j$ are exactly the same, while the potentials $B_j$ are  given as  the old $B_j$ plus $U$, $j=2,3,4$.

\noindent{\bf Ad interval $I_1$.} We notice that (\ref{eq:eps11}) is
valid (here with $V_k$ defined upon replacing $2V_1\to 2V=2V_1+U$). Whence we can  proceed exactly as before.

\noindent{\bf Ad interval $I_2$.} In addition to (\ref{eq:finalest2}) we need
the following estimation (assuming in the last step that $\frac\mu
2+\epsilon<1$): 
\begin{align} \label{eq:finalest2cc}
 &\int_{r_{1}}^{r_2}\big |UO\big (V_k^{-\frac 12}\big) \big
|\,\d r' O\Big (\frac{k}{r_1}\Big)\nonumber\\& \leq C_1  \tfrac{k}{r_1}r_0\int
_{r_1/r_0}^{r_2/r_0}\frac{r^{-1-\frac\mu2-\epsilon}}{V_k^{\frac12}}\,\d s\;(\text{changing variables }r=r_0s)\nonumber\\& \leq C_2 \tfrac{k}{r_1}\tfrac{r_0^{1-\frac\mu 2-\epsilon}}{k}\int
_{r_1/r_0}^{r_2/r_0}\frac{s^{-\frac\mu 2-\epsilon}}{\big
    (1-s^{2-\mu}\big )^\frac12}\,\d s\nonumber\\& \leq C_3
k^{-\frac{2\epsilon} {2-\mu}}\,\tfrac{k}{r_1}.
\end{align}

\noindent{\bf Ad interval $I_3$.} In addition to (\ref{eq:finalest28}) we need the following estimation 
\begin{align} \label{eq:finalest28ii}
&\int_{r_{2}}^{r_3}\big |U\zeta ^{-2} \big
|\,\d r' \nonumber\\&
 \leq C_1\int_{r_{2}}^{r_3}k^{\frac23 \tfrac{1+\mu}{2-\mu}}r'^{-1-\frac\mu
   2-\epsilon}\,\d r'\nonumber\\&\leq C_2k^{\frac23
   \tfrac{1+\mu}{2-\mu}}r_0^{-\frac\mu
   2-\epsilon}(k^{-\epsilon_2}+k^{-\epsilon_3})\nonumber\\&\leq C_3k^{\frac13
   -\epsilon_2-\tfrac{2\epsilon}{2-\mu}}. 
\end{align}  Due to (\ref{eq:reqeps2333b}) the right hand side of (\ref{eq:finalest28ii}) vanishes.

\noindent{\bf Ad interval $I_4$.} In addition to (\ref{eq:finalest2y}) we need
the following estimation: 
\begin{align} \label{eq:finalest2cciii}
 &\int_{r_{3}}^{r}\big |UO\big ((-V_k)^{-\frac 12}\big) \big
|\,\d r' \nonumber\\& \leq C_1  r_0^{-\epsilon}\int
_{r_3/r_0}^{r/r_0}\frac{s^{-1-\epsilon}}{(1-s^{\mu-2}\big )^\frac12}\,\d s\;(\text{changing variables }r'=r_0s)\nonumber\\& \leq C_2 k^{-\frac{2\epsilon} {2-\mu}}.
\end{align}

This ends the proof of  \eqref{eq:finalans}. \qed

\subsection{End of proof of  Proposition \ref{prop:scatphase}}

We  need the following elementary identity:
\begin{lemma} \label{lem:end-polar}Let $\mu<2$. Then
\begin{equation}
\int_1^\infty\left(\sqrt{r^{-\mu}-r^{-2}}-\sqrt{r^{-\mu}}\right)\d
r=\frac{2-\pi}{2-\mu}.\label{eq:mumu}\end{equation}
\end{lemma}

\proof
We first substitute $r=s^{\frac{1}{\mu-2}}$ and then $s=\sin^2\phi$. Thus the
left hand side of \eqref{eq:mumu} equals
\begin{eqnarray*}
&\frac{1}{2-\mu}\int_0^1 s^{-\frac32}\left(\sqrt{1-s}-1\right)\d s&=
\frac{2}{2-\mu}\int_0^{\frac\pi2}
\left(\frac{1-\cos\phi}{\sin^2\phi}-1\right)\d\phi\\
=&\frac{2}{2-\mu}
\left(
\frac{1-\cos\phi}{\sin\phi}-\phi\right)\Big|_0^{\pi/2}&=\frac{2-\pi}{2-\mu}
. 
\end{eqnarray*}\qed

{\it Proof of Proposition \ref{prop:scatphase}.}  Using Proposition \ref{prop:main1} we calculate
\begin{eqnarray*}
\sigma_k(0)&=&
\lim_{r\to\infty}\Big(\int_{r_0}^{r}\sqrt{-V_k(\tilde r)
}\d
\tilde r+\frac{\pi}{4}\\
&&-\int_{R_0}^r\sqrt{-2V(\tilde r)}\d \tilde r
+\frac{k\pi}{2}\Big)+o(k^0)\\
&=&\int_{r_0}^\infty
\Big(\sqrt{-V_k( r)
}-\sqrt{-2V_1( r)}\Big)\d
 r\\
&&+\int_{R_0}^\infty
\Big(\sqrt{-2V_1( r)
}-\sqrt{-2V( r)}\Big)\d
 r\\
&&-\int_{R_0}^{r_0}\sqrt{-2V_1( r)}\d r
+\frac{(k+\frac12)\pi}{2}+o(k^0).
\end{eqnarray*}
Now (using Lemma \ref{lem:end-polar})
\begin{align*}
 \int_{r_0}^\infty \big (\sqrt {-V_k}-\sqrt {-2V_1(r)}\big )\d r&
=\sqrt{k(k+1)}\int_1^\infty \big (\sqrt {r^{-\mu}-r^{-2}}-\sqrt 
  {r^{-\mu}}\big )\d r\\&=\sqrt{k(k+1)}\,\frac{2-\pi} {2-\mu};\\
  \int_{r_0}^{R_0}\sqrt {V_1(r)}\d r&=-\tfrac {2}{2-\mu}\sqrt{k(k+1)}+\tfrac {2\sqrt{2\gamma}}{2-\mu}R_0^{1-\tfrac\mu 2}.
  \end{align*}
Thus, \begin{eqnarray*}
&&\sigma_k(0) 
-\int_{R_0}^\infty
\Big(\sqrt{-2V_1( r)
}-\sqrt{-2V( r)}\Big)\d
 r\\
&=&
-\sqrt{k(k+1)}\frac{\pi}{2-\mu}
+\frac{(k+\frac12)\pi}{2}
+\frac{2\sqrt{2\gamma}}{2-\mu}R_0^{1-\frac{\mu}{2}} +o(k^0)\\
&=&-\frac{(k+\frac12)\pi\mu}{2(2-\mu)}
+\frac{2\sqrt{2\gamma}}{2-\mu}R_0^{1-\frac{\mu}{2}} +o(k^0).
\end{eqnarray*}
\qed

\appendix

\section{Elements of abstract scattering theory}
\label{Appendix}

Various versions of  stationary scattering theory
 can be found in the literature. In this appendix we give, in an
 abstract setting,  
a  self-contained
 presentation of its elements used in our paper. It
 is a version of the standard approach contained e.g. in \cite{Y4}, adapted to our
 paper. In our stationary formulas for the scattering operator we use
 in addition ideas due  to Isozaki-Kitada, see the proof of [IK2, Theorem 3.3].

\subsection{Wave operators}

Let $H_0$ and $H$ be two self-adjoint operators on a Hilbert space ${\mathcal
  H}$. We assume that $H_0$ has 
only continuous spectrum. Throughout this appendix, let
$\Lambda_n,\,n\in \N,$
be a sequence of compact  subsets   of $\sigma(H_0)$ such that $\Lambda_n$ is a
subset of the interior of $\Lambda_{n+1}$, and such that
$\sigma(H_0)\setminus \cup_n\Lambda_n$ has the  Lebesgue
 measure zero. Pick  a sequence $h_n \in C_c^\infty(\Lambda_{n+1})$ with
 $h_n=1$ on $\Lambda_n$. Let  ${\mathcal D}:=\cup_n\Ran1_{\Lambda_n}(H_0)$;
 it is dense in  ${\mathcal
  H}$.

We will write  
$R(\zeta)=(H-\zeta)^{-1}$ and  $R_0(\zeta)=(H_0-\zeta)^{-1}$ for
$\zeta\notin \sigma(H_0)$, and 
\[\delta_\epsilon(\lambda)=\frac{\epsilon}{\pi((H_0-\lambda)^2+\epsilon^2)}
=\frac{\epsilon}{\pi}R_0(\lambda-\i\epsilon)R_0(\lambda+\i\epsilon),\;\epsilon>0.\]

Note that if $I$ is an interval and  $f\in{\mathcal H}$, then
\begin{align}
  \label{eq:stone1}
  &\left\|\int_I\frac{\epsilon}{\pi}
R_0(\lambda-\i\epsilon)R_0(\lambda+\i\epsilon)f\d\lambda\right\|\leq \|f\|,\\&
\lim_{\epsilon\searrow0}
\int_I\frac{\epsilon}{\pi}
R_0(\lambda-\i\epsilon)R_0(\lambda+\i\epsilon)f\d\lambda = 1_I(H_0)f. \label{eq:stone2}
\end{align}

\begin{thm}\label{thm:wavo}
Suppose $J^\pm$ is a densely defined  operator whose domain contains ${\mathcal
  D}$ such that  $J^\pm _n:=J^\pm h_n(H_0)$ 
is  bounded for any $n$, 
 and 
\[\lim_{t\to\pm\infty}\|J^\pm  \e^{\i tH_0}f\|^2=\| f\|^2,\ \ 
f\in {\mathcal D} .\]
We also suppose that there exists the wave operator
\begin{equation}\label{wavo0}
W^\pm f
:=\lim_{t\to\pm\infty}\e^{\i tH}J^\pm \e^{-\i tH_0}f,\;
f\in {\mathcal D}
 .\end{equation}
Then 
\begin{enumerate}[\normalfont (i)]
\item \label{item:1a}
$W^\pm$
extends to  an isometric operator and $W^\pm H_0=H W^\pm$.
\item \label{item:2a} For any interval $I$ and $f\in {\mathcal D}$,
\begin{eqnarray}\label{wave9a}
W^\pm 1_I(H_0)f
 &=&\lim_{\epsilon\searrow0}
\int_I\frac{\epsilon}{\pi} R(\lambda\mp\i\epsilon)J^\pm
R_0(\lambda\pm\i\epsilon)f\d \lambda.\end{eqnarray}
\item \label{item:3a}For any continuous function $g:\R\to \C$ 
vanishing at infinity,  interval $I$  and  $f\in {\mathcal D}$,
\begin{eqnarray}
\label{wave9b}
W^\pm g(H_0)1_I(H_0)f&=&
\lim_{\epsilon\searrow0}
\int_I
\frac{\epsilon}{\pi}  g(\lambda)R(\lambda\mp\i\epsilon)J^\pm
R_0(\lambda\pm\i\epsilon)f\d \lambda.\end{eqnarray}
\item \label{item:4a} Suppose  in addition that $J^\pm $
 maps ${\mathcal D}$ into   ${\rm Dom} H$.
 Suppose that $T^\pm$ is a densely defined
 operator such that  
$T^\pm _n:=T^\pm h_n(H_0)$ is bounded for any $n$ and that
$T^\pm f=\i (HJ^\pm-J^\pm H_0)f$ for any $f\in {\mathcal D}$. Then 
we have the following modifications of \eqref{wave9a} and
 \eqref{wave9b}:
 \begin{eqnarray}\label{wave9a1}
W^\pm 1_I(H_0)f
 &=&\lim_{\epsilon\searrow0}
\int_I
(J^\pm+\i R(\lambda\mp\i\epsilon)T^\pm)\delta_\epsilon(\lambda)
f\d \lambda,\end{eqnarray}
\begin{eqnarray}
\label{wave9b1}
W^\pm g(H_0)1_I(H_0)f&=&
\lim_{\epsilon\searrow0}
\int_I g(\lambda)
(J^\pm+\i R(\lambda\mp\i\epsilon)T^\pm)\delta_\epsilon(\lambda)
f\d \lambda.\end{eqnarray}
\end{enumerate}\end{thm}

\proof
(\ref{item:1a}) is well-known.

Let us prove (\ref{item:2a}): By  (\ref{wavo0}),
\[W^\pm f=\lim_{\epsilon\searrow0}2\epsilon\int_0^\infty\e^{-2\epsilon t}
\e^{\pm\i tH}J^\pm \e^{\mp\i tH_0}f\d t.\]

 By the vector-valued Plancherel formula, we obtain
\begin{eqnarray}
W^\pm f
&=&\lim_{\epsilon\searrow0}
\int\frac{\epsilon}{\pi} R(\lambda\mp\i\epsilon)J^\pm
R_0(\lambda\pm\i\epsilon)f\d \lambda.\end{eqnarray}
Therefore,
\begin{eqnarray*}
W^\pm 1_I(H_0) f
&=&\lim_{\epsilon\searrow0}
\int_I\frac{\epsilon}{\pi} R(\lambda\mp\i\epsilon)J^\pm
R_0(\lambda\pm\i\epsilon)f\d \lambda\\
&&-
\lim_{\epsilon\searrow0}
\int_I\frac{\epsilon}{\pi} R(\lambda\mp\i\epsilon)J^\pm
R_0(\lambda\pm\i\epsilon)
1_{{\mathbb R}\backslash I}(H_0)
f\d \lambda\\&&+
\lim_{\epsilon\searrow0}
\int_{{\mathbb R}\backslash I}
\frac{\epsilon}{\pi} R(\lambda\mp\i\epsilon)J^\pm
R_0(\lambda\pm\i\epsilon)1_I(H_0)f\d \lambda.
\end{eqnarray*}

We need to show that the last two terms vanish. The proof for both terms is
identical. Consider the last one term. Let $f_1\in{\mathcal H}$ and pick an
$n$ so that $f=1_{\Lambda_n}(H_0)f$. Then (using (\ref{eq:stone1}) in
the last estimation)
\begin{eqnarray*}
&&\Big |
\int_{{\mathbb R}\backslash I}
\frac{\epsilon}{\pi}\langle f_1,
R(\lambda\mp\i\epsilon)J^\pm
R_0(\lambda\pm\i\epsilon)1_I(H_0)f\rangle\d \lambda \Big |\\
&\leq&
\|J^\pm_n\|
\left (\int_{{\mathbb R}\backslash I}
\frac{\epsilon}{\pi}\left\|
R(\lambda\pm\i\epsilon)f_1\right\|^2\d\lambda\right)^{\frac12}
\left (\int_{{\mathbb R}\backslash I}
\frac{\epsilon}{\pi}\left\|R_0(\lambda\pm\i\epsilon)1_I(H_0)f\right\|^2\d
\lambda\right )^{\frac12}\\&\leq&
C_\epsilon\|f_1\|;\;C_\epsilon:=\|J^\pm_n\|\left (\int_{{\mathbb R}\backslash I}
\frac{\epsilon}{\pi}\left\|R_0(\lambda\pm\i\epsilon)1_I(H_0)f\right\|^2\d
\lambda\right )^{\frac12}.
\end{eqnarray*}
 Due to  (\ref{eq:stone2}), $C_\epsilon\to 0$ as $\epsilon \to 0$. Whence (\ref{item:2a}) follows.

Let us prove (\ref{item:3a}): Let $f_1\in{\mathcal H}$ and pick an
$n$ so that $f=1_{\Lambda_n}(H_0)f$. Any continuous function $g$
vanishing at infinity can be uniformly approximated by $g_m$,
finite linear
combinations of characteristic functions of intervals.
By (\ref{item:2a}) and (\ref{eq:stone1}), 
\begin{eqnarray*}
W^\pm g_m(H_0) 1_I(H_0)f&=&
\lim_{\epsilon\searrow0}
\int_I 
\frac{\epsilon}{\pi} g_m(\lambda)R(\lambda\mp\i\epsilon)J^\pm
R_0(\lambda\pm\i\epsilon)f\d \lambda.\end{eqnarray*}
Now
\begin{eqnarray*}
\label{wave9c}
&&
\left|\int_I 
\frac{\epsilon}{\pi}\big (g_m(\lambda)-g(\lambda)\big )\langle f_1,R(\lambda\mp\i\epsilon)J^\pm
R_0(\lambda\pm\i\epsilon)f\rangle\d \lambda\right|\\
&\leq&
\|J^\pm_n\|
\left (\int
\frac{\epsilon}{\pi}\left\|
R(\lambda\pm\i\epsilon)
f_1\right\|^2\d\lambda \right )^{\frac12}
\left (\int\frac{\epsilon}{\pi}
\left\|R_0(\lambda\pm\i\epsilon)f\right\|^2\d \lambda \right )^{\frac12}
\sup|g_m-g|
\\
&\leq&C_m\|f_1\|;\;C_m:=\|J^\pm_n\|f\|\sup|g_m-g|.\end{eqnarray*} Since $C_m
 \to\  0$ we are done.

To prove (\ref{item:4a}),  we use (\ref{item:3a})  and the identity
\[R(\lambda\mp\i\epsilon)J^\pm=
(J^\pm+\i R(\lambda\mp\i\epsilon)T^\pm) R_0(\lambda\mp\i\epsilon).\]
\qed

\begin{remark*} In the context of our paper, we can take
  $\Lambda_n=[\frac1n,n] $.\end{remark*}

\subsection{Scattering operator}

Define
 the scattering
operator   by $ S:= W^{+*}W^-$. Clearly, $H_0S=SH_0$.

\begin{thm}\label{thm:t13}
Suppose that the conditions of Theorem \ref{thm:wavo} hold.
Let the operator $J^-$ satisfy
\begin{equation}\label{eq:8-}
\lim_{t\to+\infty}\e^{\i tH} J^-\e^{-\i tH_0}f=0,\;
 f\in {\mathcal D}.
\end{equation}
Then 
for all $f\in {\mathcal D}$
\begin{eqnarray}\label{scato}
 Sf&=&
-\lim_{\epsilon\searrow0}2\pi\int
\delta_\epsilon(\lambda)W^{+*}T^- \delta_\epsilon(\lambda)f\d \lambda.
\end{eqnarray}

\end{thm}

\proof
\begin{eqnarray*}
 W^-f&=&-\lim_{t\to +\infty}
\left(\e^{\i tH}J^-\e^{-\i tH_0}-
\e^{-\i tH}J^-\e^{\i tH_0}\right)f\\
&=&-\lim_{t\to +\infty}
\int_{-t}^t\e^{\i sH} T^-\e^{-\i sH_0} f \d s\\&=&
-\lim_{\epsilon\searrow0}
\epsilon\int_0^\infty\e^{-\epsilon t}\d t
\int_{-t}^t\e^{\i sH} T^-\e^{-\i sH_0} f \d s\\&=&
-\lim_{\epsilon\searrow0}\int\e^{-\epsilon|s|}\e^{\i
  sH} T^-\e^{-\i sH_0}f\d s.\end{eqnarray*}
Then we use the definition of $S$ and the intertwining property of $W^{+*}$
to obtain
\begin{eqnarray*}
 Sf&=&
-\lim_{\epsilon\searrow0}\int\e^{-\epsilon|s|}\e^{\i
  sH_0} W^{+*}T^-\e^{-\i sH_0}f\d s.\end{eqnarray*}
Finally, we use the vector-valued Plancherel theorem.
\qed

\subsection{Method of rigged Hilbert spaces applied to wave operators}
\label {Method of rigged Hilbert spaces applied to wave operators}

Consider a family of separable Hilbert spaces ${\mathcal H}$ and ${\mathcal V_s}$, $s>\frac12$,
such that  ${\mathcal V_s}$ is densely and continuously
embedded in ${\mathcal H}$, and similarly, ${\mathcal V_{s}}$  is
densely and continuously embedded
in ${\mathcal V_{t}}$ if  ${s}>{t}$. Let  ${\mathcal V}^*_s$ be the space dual to 
 ${\mathcal V}_s$, so that we have nested 
Hilbert spaces
\[{\mathcal V}_s\subseteq {\mathcal V}_t\subseteq {\mathcal H}\subseteq {\mathcal V}_t^*\subseteq {\mathcal V}_s^*;\;{s}>{t}.\]
We remark that ${\mathcal H}$ equipped with such a structure is sometimes called a {\em
  rigged Hilbert space}.

The following theorem allows us to introduce wave matrices:

\begin{thm}\label{wavo1} Fix   $s>t >\frac12$. Suppose that there exists
for almost all $\lambda$ the limit 
\begin{equation*}
  \slim _{\epsilon\to 0}\delta_\epsilon(\lambda)=:\delta_0(\lambda)\in \mathcal B ({\mathcal V}_{t},{\mathcal V}_{t}^*).
\end{equation*}
Suppose the conditions of Theorem
\ref{thm:wavo} and that 
the operators 
 $J^\pm_n$ and $R(\lambda\mp\i\epsilon)T^\pm _n$ with
$\lambda\in \Lambda_n$ and $\epsilon>0$    
 extend to elements of ${\mathcal B}({\mathcal V}^*_t,{\mathcal
   V}^*_s)$. Suppose that for fixed $n$ and almost everywhere in
$\Lambda_n$  there exists \[R(\lambda\mp\i0)T^\pm _n:=\slim _{\epsilon\searrow
  0}R(\lambda\mp\i\epsilon)T^\pm _n\in {\mathcal B}({\mathcal V}^*_t,{\mathcal
   V}^*_s).\]
Suppose furthermore that for any  $n$  there exists  $\epsilon_n>0$ such
that 
\begin{equation}\label{riggo1a11}
\sup_{\lambda\in\Lambda_n}
\sup_{\epsilon<\epsilon_n}\left\|\delta_\epsilon(\lambda)\right\|_{
{\mathcal V_t}\to{\mathcal V}_t^*},\;
\sup_{\lambda\in\Lambda_n}
\sup_{\epsilon <\epsilon_0}\left\|R(\lambda\mp\i\epsilon)T^\pm_n
\right\|_{
{\mathcal V}^*_t\to{\mathcal V}^*_s}<\infty.
\end{equation} 

 Let  $I$ be  an interval with $I\subseteq
  \Lambda_n$ for some $n$,  and let  $f\in  {\mathcal V}_t$ be given such
  that 
  $f=h_n(H_0)f$ (in particular this means that $f\in
  {\mathcal D}\cap{\mathcal V}_t$).   Then  (in terms  of an integral
  of a ${\mathcal
    V}^*_s$--valued function), for all  $g\in C^\infty (\R)$, 
\begin{equation}
\label{wave9bia}W^\pm g(H_0)1_I(H_0)f=\int_I g(\lambda) 
\big (J^\pm_n+\i R(\lambda\mp\i0)T^\pm_n\big )\delta_0(\lambda)f
\d\lambda.
\end{equation}
\end{thm}
\proof We can replace $T^\pm\to T^\pm _n$,  $J^\pm\to J^\pm _n$
 in the
integrand of (\ref{wave9b1}). Then, by the 
assumptions, 
it  has a  pointwise limit as an element  of ${\mathcal V}^*_s$.  Due
to (\ref{riggo1a11}), we can apply the dominated
convergence theorem.  \qed

\begin{remark*} In the context of our paper, we  take ${\mathcal
 V}_s:=L^{2,s}$. 
\end{remark*}
\subsection{Method of rigged Hilbert spaces applied to the scattering operator}

The method of rigged Hilbert spaces allows us to introduce scattering matrices:

\begin{thm}
Suppose that the conditions of Theorem \ref{wavo1} hold for some $s>t
>\frac12$. Suppose \eqref{eq:8-}. Fix $r>s$. Suppose that
for all $n\in \R$  and  $\epsilon >0$  the operators 
$T^-_n\delta_\epsilon(\lambda)\in \mathcal B ({\mathcal
  V}_{r},{\mathcal V}_{s})$ with a measurable  dependence on $\lambda \in \R$.
Suppose that for fixed $n$ and almost everywhere in
$\Lambda_n$  there exists
the limit 
\begin{equation*}
  \slim _{\epsilon\to
    0}T^-_n\delta_\epsilon(\lambda)=:T^-_n\delta_0(\lambda)\in
  \mathcal B ({\mathcal V}_{r},{\mathcal V}_{s}).
\end{equation*}
Suppose furthermore that for any  $n$  there exists  $\epsilon_n>0$ such
that 
\begin{equation}\label{riggo1a11xx}
\sup_{\lambda\in\R}
\sup_{\epsilon <\epsilon_n}\left\|T^-_n\delta_\epsilon(\lambda)
\right\|_{
{\mathcal V}_r\to{\mathcal V}_s}<\infty.
\end{equation} 

 Let  $I$ be  an interval with $I\subseteq
  \Lambda_n$ for some $n$, and let 
$f_1\in{\mathcal D}\cap{\mathcal V}_t$ and  $f_2\in{\mathcal
  D}\cap{\mathcal V}_r$ be given such
  that  $f_1=1_I(H_0)f_1$  
and $f_2=h_n(H_0)f_2$. Then 
\begin{eqnarray*}
  \label{eq:Smatrix1}
 \langle f_1,  Sf_2\rangle
&=& 
  -2\pi\int_I \langle f_1,\delta_0(\lambda) J^{+*}_n T^-_n
  \delta_0(\lambda) f_2\rangle \d\lambda\nonumber \\ 
   &&+2\pi \i \int_I \langle f_1,
\delta_{0}(\lambda) T^{+*}_nR(\lambda+\i 0)
 T^- _n\delta_{0}(\lambda)f_2\rangle \d\lambda.
\end{eqnarray*}
\end{thm}

\proof
Set
 $r_\epsilon(\lambda)
:=\frac{\epsilon}{\pi(\lambda^2+\epsilon^2)}$. We insert (\ref{wave9bia}) with
$g(\lambda)=r_\epsilon(\lambda-\lambda_1)$ into (\ref{scato}):
\begin{eqnarray*}
 \langle f_1,Sf_2\rangle&=&
-\lim_{\epsilon\searrow0}2\pi\int
\langle f_1,\delta_{\epsilon}(\lambda_1)
W^{+*}
T^- \delta_{\epsilon}(\lambda_1)f_2\rangle \d\lambda_1\\
		&=&
-\lim_{\epsilon\searrow0}2\pi \int\int_Ir_\epsilon(\lambda-\lambda_1)
\langle f_1,
\delta_0
(\lambda)\big (J^{+*}_n-\i T^{+*}_nR(\lambda+\i
0)\big )
T^- _n\delta_{\epsilon}(\lambda_1)f_2\rangle \\
			&=&
-\lim_{\epsilon\searrow0}2\pi\int_I
\langle f_1,
\delta_0(\lambda)\big (J^{+*}_n-\i T^{+*}_nR(\lambda+\i0)
\big ) T^- _n\delta_{2\epsilon
}(\lambda)f_2\rangle \d \lambda.
\end{eqnarray*}
In the last step we interchanged integrals using 
\eqref{riggo1a11xx} and the Fubini theorem,  and we used that 
\[\int\delta_{\epsilon}(\lambda_1)r_\epsilon(\lambda-\lambda_1)
\d\lambda_1=\delta_{2\epsilon}(\lambda).\]
Then we pass with $\epsilon\to0$  using \eqref{riggo1a11xx} and the dominated
convergence theorem.
\qed

\end{document}